\documentclass[12pt]{article}%
\usepackage[utf8]{inputenc}
\usepackage[nosort]{cite}
\usepackage[usenames, dvipsnames]{xcolor}
\usepackage{graphicx}
\usepackage{multicol}
\usepackage{amsfonts}
\usepackage{amssymb}
\usepackage{amsmath}
\usepackage{heck}
\usepackage{afterpage}
\usepackage{setspace}
\usepackage{verbatim}
\usepackage{color}
\usepackage{longtable}
\usepackage{float}
\usepackage{subcaption}
\usepackage{epsfig}
\usepackage{enumerate}
\usepackage{epstopdf}
\usepackage[enableskew, vcentermath]{youngtab}
\usepackage{adjustbox}
\usepackage{multirow}
\usepackage{tikz}
\usepackage[margin=1in]{geometry}
\usepackage{titletoc}
\usepackage{hyperref}
\usepackage[percent]{overpic}
\usepackage{gensymb}
\usepackage{mathtools}%
\usepackage{tikz-cd}
\usepackage{quiver}
\usepackage{booktabs}
\providecommand{\U}[1]{\protect\rule{.1in}{.1in}}

\newcommand{\UU}{\mathrm U}
\pdfoutput=1
\newsavebox{\mysavebox}

\hypersetup{colorlinks,citecolor=black,filecolor=black,linkcolor=black,urlcolor=black}
\usetikzlibrary{decorations.markings}

\numberwithin{equation}{section}

\newcommand{\R}{\mathbb R}

\newcommand{\RP}{\mathbb{RP}}
\newcommand{\CP}{\mathbb{CP}}
\newcommand{\HP}{\mathbb{HP}}
\usetikzlibrary{chains}
\allowdisplaybreaks

\newcommand{\ba}{\begin{eqnarray}}
\newcommand{\ea}{\end{eqnarray}}
\newcommand{\cA}{\mathcal{A}}

\newcommand{\cE}{\mathcal{E}}

\newcommand{\be}{\begin{equation}}
\newcommand{\ee}{\end{equation}}

\tikzstyle{startstop} = [rectangle, rounded corners, minimum width=3cm, minimum height=1cm,text centered, draw=black, fill=blue!10]
\tikzstyle{startstop} = [rectangle, rounded corners, minimum width=3cm, minimum height=1cm,text centered, draw=black, fill=blue!10]
\tikzstyle{io} = [trapezium, trapezium left angle=70, trapezium right angle=110, minimum width=3cm, minimum height=1cm, text centered, draw=black, fill=blue!30]
\tikzstyle{process} = [rectangle, minimum width=3cm, minimum height=1cm, text centered, draw=black, fill=orange!30]
\tikzstyle{decision} = [diamond, minimum width=3cm, minimum height=1cm, text centered, draw=black, fill=green!30]
\tikzstyle{arrow} = [thick,->,>=stealth]
\tikzset{->-/.style={decoration={
  markings,
  mark=at position #1 with {\arrow[scale=2.4]{>}}},postaction={decorate}}}
\makeatletter \@addtoreset{equation}{section} \makeatother

\usepackage{amsthm}
\usepackage[capitalize, noabbrev]{cleveref}
\newcommand{\Spin}{\mathrm{Spin}}
\newcommand{\Sp}{\mathrm{Sp}}
\newcommand{\SSO}{\mathrm{SO}}
\newcommand{\SL}{\mathrm{SL}}
\newcommand{\Mp}{\mathrm{Mp}}
\newcommand{\GL}{\mathrm{GL}}
\newcommand{\Z}{\mathbb Z}

\newcommand{\MTSpin}{\mathit{MTSpin}}
\newcommand{\inj}{\hookrightarrow}
\newcommand{\surj}{\twoheadrightarrow}
\newcommand{\MTSO}{\mathit{MTSO}}
\newcommand{\pt}{\mathrm{pt}}
\newcommand{\term}{\emph}
\newcommand{\BP}{\mathit{BP}}
\newcommand{\ko}{\mathit{ko}}
\newcommand{\ku}{\mathit{ku}}
\newcommand{\Sq}{\mathrm{Sq}}
\newcommand{\abs}[1]{\left\lvert #1 \right\rvert}
\newcommand{\ang}[1]{\left\langle #1 \right\rangle}
\newcommand{\set}[1]{\left\{ #1 \right\}}
\newcommand{\bl}{\text{--}}
\newcommand{\Q}{\mathbb Q}
\newcommand{\C}{\mathbb C}
\newcommand{\Sph}{\mathbb S}
\newcommand{\Ext}{\mathrm{Ext}}
\newcommand{\Hom}{\mathrm{Hom}}
\newcommand{\Det}{\mathrm{Det}}
\newcommand{\tH}{\widetilde H}
\newcommand{\KO}{\mathit{KO}}
\newcommand{\KU}{\mathit{KU}}
\newcommand{\paren}[1]{\left( #1\right)}
\newcommand{\tOmega}{\widetilde\Omega}
\newcommand{\id}{\mathrm{id}}

\newcommand{\halfQseven}{W^7_{1}} % 1/2 Q_7^4
\newcommand{\halfBott}{W_{1,8}} % 1/2 Bott manifold
\newcommand{\ninedimgen}{W^9_{1}} % quotient of L^4 x S^5 by an involution
\newcommand{\orangeseven}{W^7_{2}} % the other 7d generator

\newcommand{\tikzmarkinside}[2]{\tikz[overlay,remember picture,baseline=-0.5ex] \node (#1) at (#2) {};}

\usepackage{xspace}

\newcommand{\Pin}{\mathrm{Pin}}
\renewcommand{\O}{\mathrm O}

\usepackage{xparse}
\DeclareDocumentCommand{\shortexact}{s O{} O{} mmmm}{
\IfBooleanTF{#1}{
\begin{tikzcd}[ampersand replacement=\&]
	1 \& {#4}
	\&  {#5}
	\& {#6}
	\& 1#7
	\arrow[from=1-1, to=1-2]
	\arrow["#2", from=1-2, to=1-3]
	\arrow["#3", from=1-3, to=1-4]
	\arrow[from=1-4, to=1-5]
\end{tikzcd}
}{ % no star
\begin{tikzcd}[ampersand replacement=\&]
	0 \& {#4}
	\&  {#5}
	\& {#6}
	\& 0#7
	\arrow[from=1-1, to=1-2]
	\arrow["#2", from=1-2, to=1-3]
	\arrow["#3", from=1-3, to=1-4]
	\arrow[from=1-4, to=1-5]
\end{tikzcd}
}}

\usepackage{subcaption}
\usepackage{spectralsequences}
\usepackage{adamsmacros}
\newcommand{\AdamsTower}[1]{\DoUntilOutOfBounds{
	\class[#1](\lastx, \lasty+1)
	\structline[#1]
}}
\DeclareRobustCommand*{\RaiseBoxByDepth}{%
    \raisebox{\depth}%
}
\newcommand{\uQ}{\RaiseBoxByDepth{\protect\rotatebox{180}{$Q$}}}

\newtheorem{thm}[equation]{Theorem}
\newtheorem{lem}[equation]{Lemma}
\newtheorem{prop}[equation]{Proposition}
\newtheorem{cor}[equation]{Corollary}
\theoremstyle{definition}
\newtheorem{defn}[equation]{Definition}
\newtheorem{exm}[equation]{Example}
\theoremstyle{remark}
\newtheorem{rem}[equation]{Remark}

\crefname{thm}{Theorem}{Theorems}
\crefname{lem}{Lemma}{Lemmas}
\crefname{prop}{Proposition}{Propositions}
\crefname{defn}{Definition}{Definitions}
\crefname{exm}{Example}{Examples}
\crefname{rem}{Remark}{Remarks}
\crefname{cor}{Corollary}{Corollaries}

\begin{document}

\preprint{LMU-ASC 06/23\\ IFT-UAM/CSIC-23-7}

\date{February 2023}

\title{The Chronicles of IIBordia: \\[4mm] Dualities, Bordisms, and the Swampland}

\institution{PURDUE}{\centerline{$^1$Department of Mathematics, Purdue University, 150 N University St, West Lafayette, IN 47907, USA}}
\institution{MUNICH}{\centerline{$^2$Arnold Sommerfeld Center for Theoretical Physics, LMU, Munich, 80333, Germany}}
\institution{PENN}{\centerline{$^3$Department of Physics and Astronomy, University of Pennsylvania, Philadelphia, PA 19104, USA}}
\institution{PENNMATH}{\centerline{$^4$Department of Mathematics, University of Pennsylvania, Philadelphia, PA 19104, USA}}
\institution{HARVARD}{\centerline{$^5$Jefferson Physical Laboratory, Harvard University, Cambridge, MA 02138, USA}}
\institution{MADRID}{\centerline{$^6$Instituto de F\'{i}sica Te\'{o}rica IFT-UAM/CSIC, C/ Nicol\'{a}s Cabrera 13-15, 28049 Madrid, Spain}}

\authors{Arun Debray\worksat{\PURDUE}\footnote{e-mail: {\tt adebray@purdue.edu}},
Markus Dierigl\worksat{\MUNICH}\footnote{e-mail: {\tt m.dierigl@physik.uni-muenchen.de}},\\[4mm]
Jonathan J.\ Heckman\worksat{\PENN,\PENNMATH}\footnote{e-mail: {\tt jheckman@sas.upenn.edu}}, and
Miguel Montero\worksat{\HARVARD,}\worksat{\MADRID}\footnote{e-mail: {\tt miguel.montero@uam.es}}}

\longabstract{In this work we investigate the Swampland Cobordism Conjecture in the context of type IIB string
theory geometries with non-trivial duality bundle. Quite remarkably, we find that many non-trivial bordism classes
with duality bundles in $\Mp (2,\Z)$, a double cover of $\SL (2,\Z)$ related to fermions, correspond to asymptotic
boundaries of well-known supersymmetric F-theory backgrounds. These include $[p,q]$-7-branes, non-Higgsable
clusters, S-folds, as well as various lower-dimensional generalizations. These string theoretic objects break
the global symmetries associated to the non-trivial bordism groups, providing a strong test of the Cobordism Conjecture. Further including worldsheet orientation reversal promotes the duality group to the Pin$^+$ cover of $\GL(2,\Z)$. The corresponding bordism groups require a new non-supersymmetric ``reflection 7-brane'' and its compactifications to ensure the absence of global symmetries, thus providing an interesting prediction of the Cobordism Conjecture for non-supersymmetric type IIB backgrounds.

A major component of the present work is the explicit derivation of the involved bordism groups as well as their generators, which correspond to asymptotic boundaries of explicit string theory backgrounds. The main tool is the Adams spectral sequence, to which we provide a detailed introduction. We anticipate that the same techniques can be applied in a wide variety of settings.}

\maketitle

\thispagestyle{empty}
\begin{figure}[H]
\vspace*{-2cm}
\makebox[\linewidth]{\includegraphics[width = 1.15 \linewidth]{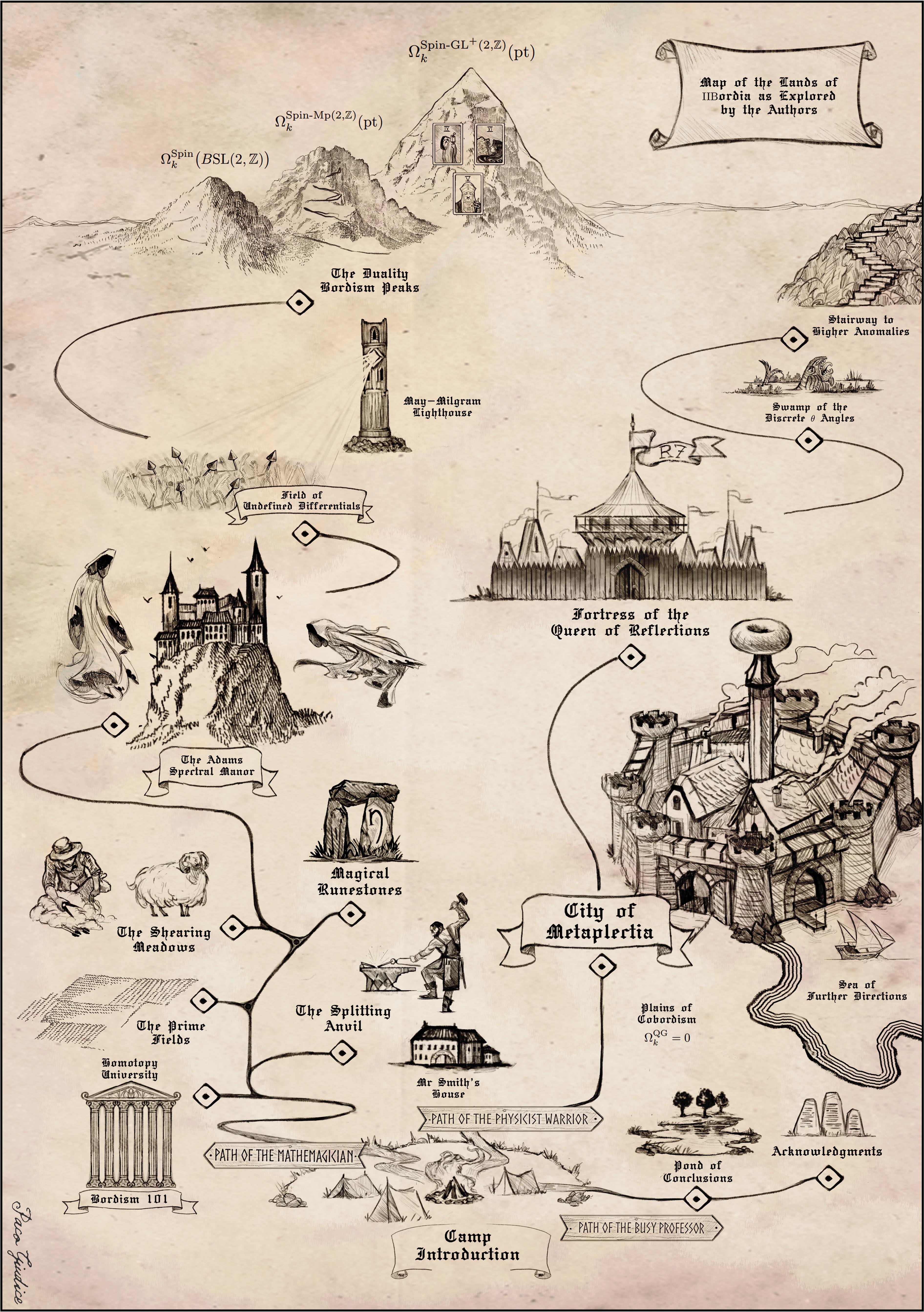}}
\end{figure}

\setcounter{tocdepth}{2}

\enlargethispage{\baselineskip}

\newpage
\tableofcontents

\newpage

\part{Introduction and summary}

\section{Introduction} \label{sec:intro}

One of the big challenges in the study of quantum gravity is that while various long distance / weak coupling limits are by
now well understood, the development of a UV complete microscopic formulation still remains elusive.
In the absence of a full theory of quantum gravity, it is natural to develop constraints from various complementary ``bottom up'' methods,
including the general principles of dualities and symmetries.

Broadly speaking, the Swampland program \cite{Vafa:2005ui, Ooguri:2006in}
(see \cite{Palti:2019pca, vanBeest:2021lhn, Agmon:2022thq} for reviews)
is an ambitious ``bottom up'' approach aimed at addressing precisely these issues.
This has led to an interconnected web of conjectural results, and has
also provided a new perspective on many previously considered string theory backgrounds.
With this in mind, it is natural to ask whether Swampland considerations can provide additional
insight into the structure of non-perturbative dualities. Conversely, we can use well-known string theoretic dualities to
test various Swampland conjectures.

The interplay between symmetries and quantum gravity has a long history, including the ``folk theorem''
that there are no global symmetries (see also \cite{Banks:1988yz, Banks:2010zn, Harlow:2018tng}).
In the context of the Swampland conjectures, this was significantly sharpened
as the Cobordism Conjecture (CC) of McNamara and Vafa \cite{McNamara:2019rup}
which asserts that the bordism group of quantum gravity is trivial:
\begin{align}
\Omega^{\text{QG}}_k = 0 \,.
\end{align}
In physical terms, bordism theory
arises because we are considering field configurations in the gravitational path integral which tend towards asymptotic boundary
conditions.\footnote{An important point to stress here is that a priori, one does not even need to specify a smooth interior geometry.
All that is required is that we have a notion of
asymptotically defined field configurations, as in the S-matrix of asymptotic flat space, or the AdS/CFT correspondence.}
The statement of the Cobordism Conjecture then is that in quantum gravity one can \textit{always} connect two asymptotic field configurations.
Of course, the precise definition of $\Omega^{\text{QG}}_k = 0$ would in some sense require one to first specify a definition of quantum gravity
and all its long distance limits. That being said, the elements of $\Omega^{\text{QG}}_k$ should be viewed as
the equivalence classes of $k$-dimensional compactifications of quantum gravity
that can be connected to each other via a domain wall. The statement above implies that the
Landscape of all quantum gravities is in fact connected via these domain walls.
The reason is that if the bordism group $\Omega^{\text{QG}}_k$ does not vanish,
the non-trivial bordism groups would act as conserved global charges.

Now, given a putative symmetry $\mathcal{G}$ of our theory,
we can equip our spacetime with a $\mathcal{G}$-bundle\footnote{Here, $\mathcal{G}$ in general contains information about internal symmetries, the tangential structure of the manifold, as well as their mixtures to twisted tangential structures.}, and in practice,
it often happens that $\Omega_{k}^{\mathcal{G}} \neq 0$. In other words, there are field configurations in the
gravitational path integral which cannot be topologically deformed into one another.
In such situations the CC makes a sharp prediction: In order for $\Omega^{\text{QG}}_k$ to be trivial, there must
exist new defects (i.e., some source for the symmetry $\mathcal{G}$) which supplement the spectrum of objects in the theory.
At this point, the conjecture has passed a number of non-trivial checks (see e.g.,
\cite{McNamara:2019rup, Ooguri:2020sua, Montero:2020icj, Dierigl:2020lai, Buratti:2021yia, Buratti:2021fiv, Blumenhagen:2021nmi, Blumenhagen:2022mqw, Andriot:2022mri, Angius:2022aeq, Blumenhagen:2022bvh, Angius:2022mgh, Velazquez:2022eco, McNamara:2022lrw,Dierigl:2022reg}),
providing a ``bottom up'' approach to the construction of
defects such as $p$-branes and particular gravitational instantons (see e.g., \cite{McNamara:2021cuo}).
In all known physical applications, the relevant bordism groups is simply a finite sum of free and torsional pieces:
\begin{equation}
\Omega_{k}^{\mathcal{G}} \simeq \Z^{n} \oplus \underset{i}{\bigoplus} (\Z / m_i \Z)^{n_i},
\end{equation}
in the obvious notation.\footnote{Compared with much of the physics literature, we have opted to write $\Z / m \Z$ to denote
the cyclic group with $m$ elements. This is to avoid any confusions with other notation
used in the text, including $\Z_{p}$ for the p-adic integers, for example.}

Now, bordism groups specify a generalized homology theory and as such are intrinsically Abelian objects. It is natural
to ask whether these considerations are compatible with the existence of non-Abelian duality groups, such as those which frequently arise in
string theory and various quantum field theories. In particular, reference \cite{Dierigl:2020lai} showed that the Cobordism Conjecture, when combined
with the famous $\SL(2,\Z)$ duality group of type IIB strings / F-theory successfully predicts the existence of $[p,q]$ 7-branes, one of the
core ingredients in non-perturbative F-theory vacua. In more technical terms, these are associated with $\Omega^{\text{QG}}_{1}$, namely
asymptotic one-dimensional boundary geometries.

Given this, we can ask whether the CC correctly predicts the known spectrum of objects of IIB / F-theory,
and conversely, whether we can use the CC to \textit{predict} new non-perturbative objects which cannot be accessed via other
(usually supersymmetric methods). The aim of the present paper will be to produce
a systematic classification of possible objects compatible with the
basic duality of type IIB / F-theory. The main tool we use to accomplish this is the calculation
of the corresponding bordism groups associated to twisted Spin structures involving the duality bundle.

Of course, this begs the question as to the precise notion of ``duality bundle'' which we will need to consider in this paper. To begin,
even specifying the duality group of IIB turns out to be remarkably subtle.
It is well-known that the type IIB supergravity action is invariant under an $\SL(2,\R)$ group, and taking into account quantization
of fluxes, this is reduced to $\SL(2,\Z)$, which is also the group of large diffeomorphisms of a torus, the starting point
for the geometric formulation of type IIB strings via F-theory \cite{Vafa:1996xn, Morrison:1996na, Morrison:1996pp}.
As noted in \cite{Pantev:2016nze,Tachikawa:2018njr}, the duality group must also be extended to an appropriate action
on the fermions of IIB supergravity, and this leads to $\Mp(2,\Z)$,
the metaplectic cover of $\SL(2,\Z)$.\footnote{It is analogous to the extension of SO$(d)$ to Spin$(d)$.}
Even this is not quite the full duality group, because one must also allow for the IIB worldsheet
orientation reversal and $(-1)^{F_L}$ left-moving spacetime fermion parity transformations.
Geometrically, these act as reflections on the F-theory torus. Indeed, for M- / F-theory duality to hold,
one must allow for such reflections, since M-theory implicitly makes sense on Pin$^+$ backgrounds.
Taking this into account, it follows (see \cite{Tachikawa:2018njr}) that the IIB duality group is most accurately
specified as the Pin$^+$ cover of $\GL(2,\Z)$. See Figure \ref{Fig:DualityGroups} for a depiction.
\begin{figure}
\centering
\includegraphics[width = 0.8 \textwidth]{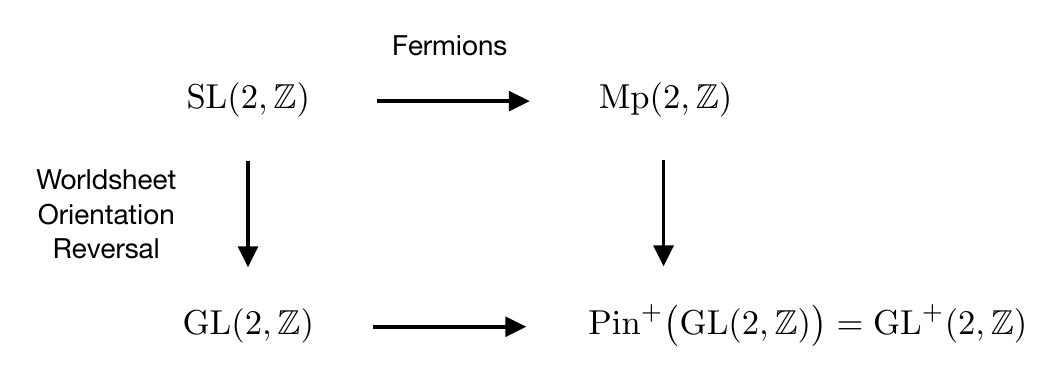}
\caption{Various extensions of the duality group of type IIB string theory.}
\label{Fig:DualityGroups}
\end{figure}

In physical terms, the duality group $\SL(2,\Z)$ tells us about the bosonic / non-chiral sector of IIB strings / F-theory. Taking into account
fermions we arrive at $\Mp(2,\Z)$. For supersymmetric F-theory backgrounds specified by genus-one fibered Calabi-Yau spaces,
the appearance of fermions can always be viewed as ``implicit,'' and so one expects the relevant objects predicted by the CC for $\Mp(2,\Z)$
to typically be supersymmetric. On the other hand, $\GL^{+}(2,\Z)$ includes reflections which, among other things can send a D3-brane to an anti-D3-brane.
As such, IIB field configurations which include such symmetries provides us access to \textit{non-supersymmetric} objects. Said differently,
we can use $\Mp(2,\Z)$-duality bordisms to test the CC, and $\GL^{+}(2,\Z)$-duality bordisms to make predictions for new, possibly
non-supersymmetric objects!

Given all of this, there are clearly three different notions of ``duality bundle''-bordisms
which will be of interest to us in the present work. First of all,
we can consider the bordism groups associated with SL$(2,\Z)$ dualities, and in which
the interpolating geometries also have a Spin structure. Formally these are captured by the bordism groups:
\begin{align}
\Omega^{\Spin}_k \big( B\SL(2,\Z) \big) \,,
\label{eq:approx1}
\end{align}
(see also \cite{Seiberg:2018ntt}), where the ``$B$'' refers to the classifying space of $\SL(2,\Z)$.

Accounting for the appearance of fermions, this naturally enlarges to $\Mp(2,\Z)$ duality transformations.
Here, an important subtlety is that ---as is well-known in many F-theory constructions--- there is no need for the
spacetime to retain a Spin structure. Rather, we can instead have a twisted Spin structure which is correlated
with duality group transformations,\footnote{This not an issue for $\SL(2,\Z)$ transformations since it is only
sensitive to the bosonic sector of the theory.} namely we consider manifolds with
$\big(\Spin(k) \times \Mp(2,\Z)\big) / (\Z / 2 \Z)$ structure. In this case, the relevant bordism groups
are given by:
\begin{align}
\Omega^{\text{Spin-Mp}(2,\Z)}_k (\pt) \,,
\label{eq:approx2}
\end{align}
see also \cite{Hsieh:2019iba, Hsieh:2020jpj}. Since $\Mp(2,\Z)$ is in some sense the ``minimal'' extension of $\SL(2,\Z)$ which
accommodates supersymmetry, it is natural to expect that for each such bordism generator, there is a corresponding
supersymmetric F-theory background which can be viewed as filling in the bulk of these bordism generators. We indeed find
that this is the case, namely at least for $k$ odd (so that the bulk space is even-dimensional), each generator of
$\Omega^{\text{Spin-Mp}(2,\Z)}_k (\pt)$ matches to a genus-one fibered Calabi-Yau threefold $X \rightarrow \mathcal{B}$,
where $\partial \mathcal{B}$ can be viewed as specify a generator of $\Omega^{\text{Spin-Mp}(2,\Z)}_k (\pt)$ with prescribed duality bundle.

The outcome of this analysis is that at least for Spin-$\Mp(2,\Z)$ bordisms,
the main building blocks of many F-theory compactifications, including $[p,q]$ 7-branes,
non-Higgsable clusters, S-folds, as well as lower-dimensional variants are all successfully predicted by the CC.
This amounts to an extremely non-trivial test of the CC, and it is remarkable that such minimal topological inputs
reconstruct these intricate structures.

Turning finally to the full IIB duality group $\GL^{+}(2,\Z)$, we also determine the corresponding
Spin-$\GL^{+}(2,\Z)$ bordism groups:
\begin{align}
\Omega^{\text{Spin-GL}^{+}(2,\Z)}_k (\pt) \,.
\label{eq:approx3}
\end{align}
In this case, the appearance of reflections means that we should not expect
supersymmetry to be retained for these generators, and in general we do not expect
there to be any genus-one fibered Calabi-Yau space which fills in the interior of such boundary geometries.
In this case, the CC amounts to a set of \textit{predictions} for objects which need not preserve supersymmetry.
Quite remarkably, we find that the generators are, in this case, primarily inherited from Spin-$\Mp(2,\Z)$
(namely they are supersymmetric objects), but that there is one additional non-supersymmetric ``reflection 7-brane'' \cite{Dierigl:2022reg},
as detected by $\Omega^{\text{Spin-GL}^{+}(2,\Z)}_1 (\pt)$. Starting from these ingredients,
we find that nearly all of the other generators are in some sense inherited from further compactification of these higher-dimensional
objects.\footnote{There are a few outlier cases at low dimensions which do not appear to be related in an obvious way to any of these other
objects, but this is the exception rather than the rule.}

In a ten-dimensional spacetime, the CC tells us about the spectrum of defects associated with $\Omega_{k}^{\mathcal{G}}$
for $1 \leq k \leq 9$. The bordism groups for $k = 10$ and $k = 11$ also specify important
physical data. For $k = 10$, $\Omega_{10}^{\mathcal{G}}$ detects parameters such as discrete theta angles,
and for $k = 11$, $\Omega_{11}^{\mathcal{G}}$ classifies candidate contributions to anomalies of the ten-dimensional theory. A detailed analysis of
type IIB duality anomalies was presented in \cite{Debray:2021vob}, where it was shown that a subtle duality anomaly in IIB backgrounds can
be cancelled by adding an additional topological term to the standard type IIB action.

The main computational tool we use to determine these bordism groups is the Adams spectral sequence.
Another aim of this work will be to provide an introduction to this tool. The Adams spectral sequence
has an (undeserved) reputation for being somewhat unwieldy, but in the case at hand it turns out to be rather tractable, once a number of intermediate
simplifications are implemented. In particular, one of the key ideas undergirding our results is the fact
that the non-Abelian duality groups of interest can be decomposed into corresponding amalgamated products,
namely they can be decomposed into some simpler generators, each of which has a
physical interpretation.\footnote{For example, 7-branes with a fixed value of the axio-dilaton.}
The procedure of simplifying the bundle structure is accomplished through
a procedure known as ``shearing'', and with these
elements in place, the calculation of the Adams spectral sequence vastly simplifies. Finding explicit bordism generators
is also somewhat challenging, but in many cases we can verify that we have found a primitive generator by determining
a suitable combination of $\eta$-invariants of the manifold (which in the cases of interest is a bordism invariant).

While we primarily emphasize the application of these bordism groups in the study of IIB strings / F-theory and the Swampland
program, it is clear that these results have broader applications both within string theory and quantum field theory, as well as
in pure math. In the context of string compactification / quantum gravity, one can reinterpret the different $\Omega_{k}^{\mathcal{G}}$'s
as classifying anomalies of a $k-1$ dimensional gravitational theory. Similarly, in the context of quantum field theory, these same
$\Omega_{k}^{\mathcal{G}}$'s specify mixed gravitational-duality anomalies, a topic which has recently been studied for example in
\cite{Seiberg:2018ntt, Hsieh:2019iba, Cordova:2019jnf, Cordova:2019uob}. In the context of pure math, many of the calculations we present are new,
and can be viewed as the starting point for developing a more systematic treatment of $\Omega_{k}^{\mathcal{G}}$ for
all $k$ (not just ``small values'').

This paper is huge because we provide all details behind the calculations.
It is our hope that this reference can be efficiently used by others to learn
the techniques behind the Adams spectral sequence and the calculation of bordism groups in general.

To make the paper manageable, we have organized it into three main parts, as well as a few Appendices:
\begin{itemize}

\item In the remainder of this Introductory part, we first give a brief introduction to bordism theory, since it
forms a core component of all of what follows. We then give a brief summary of the results and their physical
interpretation. We then turn to our conclusions and future directions. We also include a map to
the contents of the rest of the paper.

\item  In Part \ref{p:physics},  we give a physical interpretation of
each of the generators of $\Omega_{k}^{\mathcal{G}}$; we defer the derivation
of these groups to Part \ref{p:maths}. The CC asserts that there is a corresponding
spacetime defect associated with each such generator, and this leads us to the
(re)discovery of various interesting backgrounds of type IIB string theory.

\item Part \ref{p:maths} introduces the necessary mathematical tools to calculate
the various bordism groups discussed above. One of the main tools in this computation
is the Adams spectral sequence which is introduced in great detail. These techniques
are consecutively used in order to derive the bordism groups
$\Omega_{k}^{\mathrm{Spin}} \big(B\SL(2,\Z) \big)$, $\Omega_{k}^{\Spin\text{-}\Mp(2,\Z)}(\pt)$ and
$\Omega_{k}^{\Spin\text{-}\GL^{+}(2,\Z)} (\pt)$ for $k \leq 11$. Moreover, we describe the generating manifolds
for each of the non-trivial group factors. Wherever possible we also give a physical interpretation
of these mathematical manipulations.

\item In the Appendices we give some additional technical details on the structure of the duality groups,
$\eta$-invariants, and the May-Milgram theorem.
Additionally, we also comment on the use of the Smith homomorphism
which connects different duality bundle structures separated by
different dimensionality; the physical interpretation of this
can be viewed as duality defects inherited from a higher-dimensional bulk.
\end{itemize}

\section{Bordism: What is it?}
\label{sec:what}

We now proceed to a brief introduction to bordism theory. This forms the core mathematical
structure which we use in what follows. Bordism is an equivalence relation defined on the
set of $k$-dimensional manifolds equipped with a given tangential structure,
e.g., Spin structure, orientation, or similar notions (see e.g., \cite{stong2015notes} for a mathematical introduction and \cite{Garcia-Etxebarria:2018ajm,McNamara:2019rup} for a description in a physical context). Two $k$-dimensional manifolds $X_1$ and $X_2$ are in the same equivalence class,
i.e., define the same element in $\Omega_k^{\mathcal{G}}$ if $X_1 \sqcup (- X_2)$ arises as the boundary of a
$(k+1)$-dimensional manifold over which the structure extends, see Figure \ref{fig:BordST}. Here, the minus sign indicates
an orientation reversal for $X_2$.
\begin{figure}[t]
\centering
\includegraphics[width = 0.8 \textwidth]{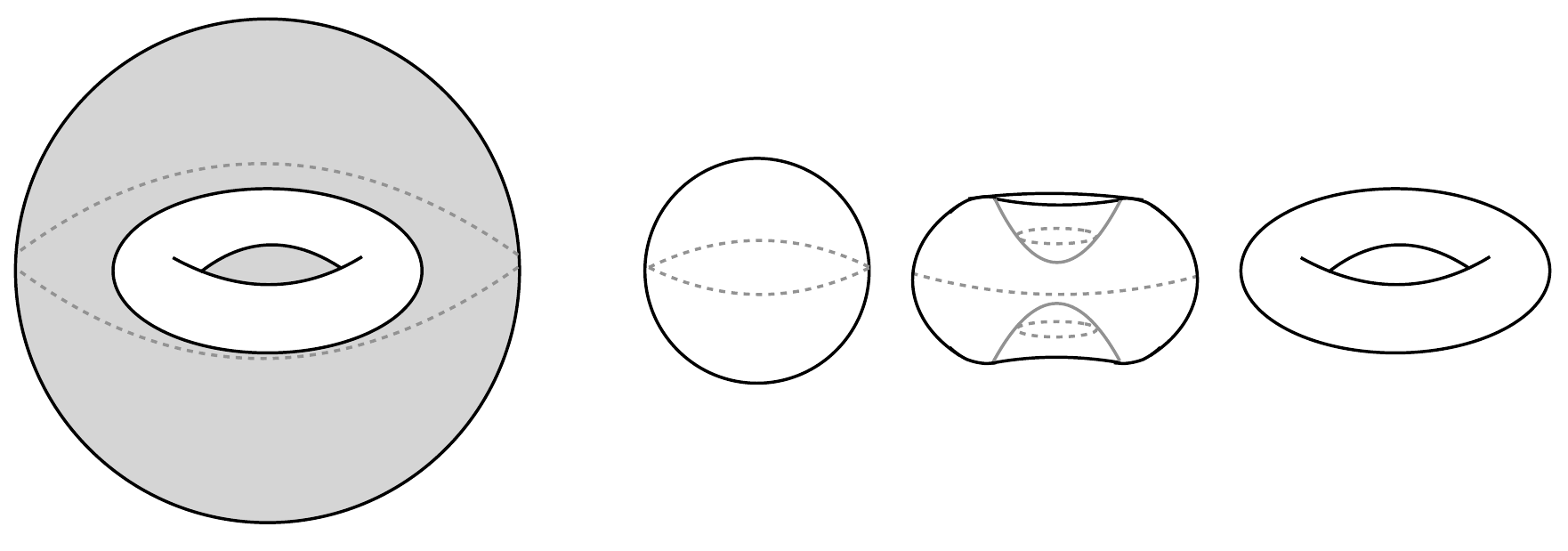}
\caption{Oriented bordism between a 2-sphere and a 2-torus (can be understood as a $3$-ball with an interior torus cut out).}
\label{fig:BordST}
\end{figure}

The disjoint union of manifolds (or equivalently, the connected sum since the two are bordant) naturally defines an Abelian additive operation turning $\Omega_k^{\mathcal{G}}$ into an Abelian group. Elements on the trivial class correspond to those manifolds that
are the boundary of a manifold in one higher dimension, see Figure \ref{fig:Bord2} for examples.
If some bordism class generates a finite subgroup $\Z/n\Z$, this implies that while
a single copy of a manifold representing a generating bordism class is not a boundary,
the disjoint union of $n$ copies of that manifold is, see Figure \ref{fig:Bord2}.
\begin{figure}
\centering
\includegraphics[width = 0.5 \textwidth]{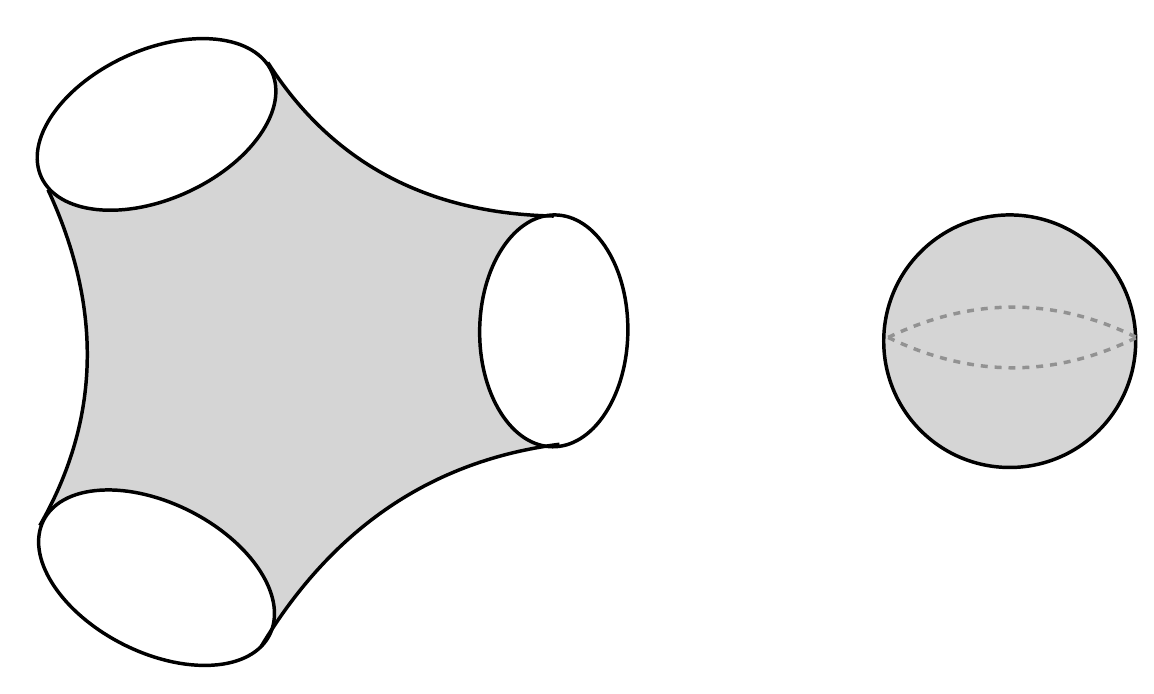}
\caption{Left: Depiction of a bordism group factor $\Z/3\Z$, where three copies of a generator bound a manifold. Right: Null-bordism of the 2-sphere, i.e., a 3-ball.}
\label{fig:Bord2}
\end{figure}

If $\mathcal{G}$ in addition to tangential structures $\xi$ also includes information about an internal symmetry $G$ in the theory, the defining data of the associated bordism group contains a map from the spacetime manifold $X$ into the classifying space of the symmetry group $BG$ specifying the associated gauge field backgrounds. The continuous deformation of the spacetime manifold with a given structure is then complemented with a continuous deformation of this classifying map and the associated bordism group is denoted as
\begin{align}
\Omega^{\xi}_k (BG) \,.
\end{align}
In general the tangential structure, e.g., the Spin structure, can mix with the internal symmetry group, defined by quotients of the form
\begin{align}
\frac{\Spin(k) \times G}{\Z / 2 \Z}
\end{align}
leading to twisted Spin structures, which we denote as
\begin{align}
\Omega^{\text{Spin-}G}_k (\pt) \,,
\end{align}
which do appear in our analysis once one includes fermions.

Importantly, bordism groups can be understood as a generalized homology theory \cite{atiyah_1961}.
In practice this means that the bordism groups are described as homotopy groups of a smash product (a homotopical
analogue of a tensor product) of the manifold $X$ under investigation with a spectrum $E$:
\begin{align}
h_{*} (X) = \pi_* \big(X \wedge E \big) \,.
\end{align}
For Spin bordism groups $\Omega^{\Spin}_k (\pt)$ the relevant spectrum is the Thom spectrum denoted by $\MTSpin$,
which in the presence of internal symmetries or twisted tangential structure needs to be generalized accordingly.
In order to derive various bordism groups one hence needs a good understanding of $\MTSpin$ and its generalizations.
For that it is often useful to decompose the spectrum into easier parts, all of which contain portions of the complete
information. The main simplification is that we use {\it $p$-local equivalences}
of $\MTSpin$ with other spectra, which probes the $p$-torsion part of the Spin bordism group.
In the present analysis we are faced with the analysis of bordism groups of the form
\begin{align}
\Omega^{\SSO}_k (B\Z/3\Z) \,, \quad \Omega_k^{\SSO}(BD_6)\,,\quad
\Omega^\Spin(B\Z/4\Z)\,,\quad
\Omega^{\text{Spin-}\Z/8\Z}_k (\pt) \,, \quad \Omega^{\text{Spin-}D_{16}}_k (\pt) \,,
\end{align}
where $D_{16}$ denotes the dihedral group with 16 elements.\footnote{Here, we really refer to the dihedral group and
not the binary dihedral group. It contains an element associated to reflections $R$ of order two and an element associated
to rotations $U$ of order $n$. They obey $R U R^{-1} = U^{-1}$. In particular, the reflection does not appear as a
certain power of the rotation such as for the binary dihedral case. See also Appendix \ref{app:groups}.}
In several cases we can extract the needed information from isomorphisms to other generalized (co)homology theories,
such as $BP$-homology and $ku$-homology groups. In many other cases, however, we will need to derive the results
using spectral sequence techniques, most prominently the Adams spectral sequence. Especially for large $k$ this
becomes increasingly complicated and we will present the details of these calculations explicitly in Part \ref{p:maths}.

Once the bordism groups are derived, we are faced with another challenge;
the determination of the associated generators, i.e., a set of $k$-manifolds
with the given structure that serve as representatives of all the bordism
classes we have found. To find these manifolds, we are aided by \emph{bordism invariants}:
Group homomorphisms from the bordism group to $\Z$ or $\Z/n\Z$,
which can help with the identification of generators. For instance, $\Omega_4^{\Spin} (\pt) =\Z$
has a bordism invariant given by integrating the first Pontrjagin class over $M_4$, a given four-manifold
\begin{align}
\int_{M_4} p_1 \,,
\end{align}
which is often referred to as the first Pontrjagin number. This means that we can take as a
generator of $\Omega_4^{\Spin} (\pt)$ any Spin 4-manifold with minimal Pontrjagin number. Many bordism invariants are like the Pontrjagin number, integrals of characteristic classes of the tangent bundle (as well as the gauge bundle, if present). However, this does not apply to all of them; some are detected by more exotic quantities, such as $\eta$-invariants of Dirac operators (see e.g.\cite{Witten:2015aba,Garcia-Etxebarria:2018ajm}), and it is not always easy to identify the representative. Thus, the detection as well as the construction of the relevant manifolds is an art more than a science,  but we will also describe it in great detail. It is noteworthy that certain types of manifolds naturally appear as generators of Spin-$D_{16}$ bordisms which can be described as fiber bundles of real projective spaces or lens spaces over a base real projective base. One certain class of these manifolds already appeared in \cite{Debray:2021vob}
and was referred to as `Arcanum XI' and will appear here in certain lower-dimensional variants
that we also refer to as Arcana.\footnote{Following the book of Thoth (see the summary in \url{https://en.wikipedia.org/wiki/Major_Arcana}),
the relevant arcana for us will be:
%Arcanum I (The Magus), Arcanum II (The Priestess), Arcanum III (The Empress), Arcanum IV (The Emperor),
Arcanum V (The Hierophant),
%Arcanum VI (The Lovers),  Arcanum VII (The Chariot), Arcanum VIII (Adjustment),
Arcanum IX (The Hermit),
%Arcanum X (TheWheel of Fortune),
Arcanum XI (Lust).}

\section{Summary of our results}
\label{sec:summary}

In this Section we briefly summarize the computational results of Part \ref{p:maths} concerning the bordism groups
\begin{align}
\Omega^{\Spin}_k \big(B \SL (2,\Z) \big) \,, \quad \Omega^{\text{Spin-Mp}(2,\Z)}_k (\pt), \quad \Omega^{\text{Spin-GL}^+(2,\Z)}_k (\pt) \,.
\end{align}
For $1 \leq k \leq 9$, the CC predicts a corresponding defect which ``fills in the bulk'' associated with
the corresponding bounding manifold equipped with an appropriate duality bundle. In the case of $k = 10$, the bordism groups are associated
with discrete theta angles, and in the case of $k = 11$ these provide a coarse classification of possible contributions
to duality anomalies. Our primary focus will be on Spin-$\Mp(2,\Z)$ and Spin-$\GL^{+}(2,\Z)$
bordisms since these are the cases where we also account for the contributions from fermionic degrees of freedom.
Finally, for $1 \leq k \leq 9$ we give a physical interpretation of defects that allow these non-trivial bordism
backgrounds to be bounded.

The bordism groups are given in Table \ref{the_big_table}.
\renewcommand{\arraystretch}{1.5}
\begin{table}[h!]
\begin{tabular}{c c c c}
\toprule
$k$ & $\Omega^{\Spin}_k \big(B\SL(2,\Z)\big)$ & $\Omega^{\Spin\text{-}\Mp(2,\Z)}_k (\pt)$ &
	$\Omega_k^{\Spin\text{-}\GL^+(2, \Z)} (\pt)$ \\
	& Section \ref{sl2_spin} & Section \ref{mp_spin} & Section \ref{gl_spin} \\
\midrule
$0$ & {\footnotesize $\Z$} & {\footnotesize $\Z$} & {\footnotesize $\Z$} \\ %\hline
$1$ & {\footnotesize $(\Z/2\Z) \oplus (\Z/4\Z) \oplus (\Z/3\Z)$} & {\footnotesize $(\Z/8\Z) \oplus (\Z/3\Z)$} & {\footnotesize $(\Z/2\Z)^{\oplus 2}$} \\ %\hline
$2$ & {\footnotesize $(\Z/2\Z)^{\oplus 2}$} & {\footnotesize $0$} & {\footnotesize $\Z/2\Z$} \\ %\hline
$3$ & {\footnotesize $(\Z/2\Z) \oplus (\Z/8\Z) \oplus (\Z/3\Z)$} & {\footnotesize $(\Z/2\Z) \oplus (\Z/3\Z)$} & {\footnotesize $(\Z/2\Z)^{\oplus 3} \oplus (\Z/3\Z)$} \\ %\hline
$4$ & {\footnotesize $\Z$} & {\footnotesize $\Z$} & {\footnotesize $\Z$} \\ %\hline
$5$ & {\footnotesize $(\Z/4\Z) \oplus (\Z/9\Z)$} & {\footnotesize $(\Z/2\Z) \oplus (\Z/{32}\Z) \oplus (\Z/9\Z)$} & {\footnotesize $(\Z/2\Z)^{\oplus 2}$} \\ %\hline
$6$ & {\footnotesize $0$} & {\footnotesize $0$} & {\footnotesize $0$}\\ %\hline
$7$ & {\footnotesize $(\Z/2\Z) \oplus (\Z/{32}\Z) \oplus (\Z/9\Z)$} & {\footnotesize $(\Z/4\Z) \oplus (\Z/9\Z)$} & {\footnotesize $ (\Z/2\Z)^{\oplus 3} \oplus (\Z/4\Z) \oplus (\Z/9\Z)$} \\ %\hline
$8$ & {\footnotesize $\Z \oplus \Z$} & {\footnotesize $\Z \oplus \Z$} & {\footnotesize $\Z\oplus\Z\oplus (\Z/2\Z)$} \\ %\hline
$9$ & {\footnotesize $(\Z/2\Z)^{\oplus 3} \oplus (\Z/4\Z) \oplus (\Z/8\Z)$} & {\footnotesize $(\Z/4\Z) \oplus (\Z/8\Z) \oplus (\Z/128\Z)$} & {\footnotesize $(\Z/2\Z)^{\oplus 8}$} \\ %\hline
 & {\footnotesize $\oplus (\Z/3\Z) \oplus (\Z/27\Z)$ } & {\footnotesize $\oplus (\Z/3\Z) \oplus (\Z/27\Z)$ } & \\ 	
$10$ & {\footnotesize $(\Z/2\Z)^{\oplus 5}$} & {\footnotesize $\Z/2\Z$} & {\footnotesize $(\Z/2\Z)^{\oplus 4}$} \\ %\hline
$11$ & {\footnotesize $(\Z/2\Z)^{\oplus 2} \oplus (\Z/8\Z)^{\oplus 2} \oplus (\Z/128\Z)$} & {\footnotesize $(\Z/2\Z)^{\oplus 2} \oplus (\Z/8\Z)$} & {\footnotesize $(\Z/2\Z)^{\oplus 9} \oplus (\Z/8\Z) $} \\ %\hline
 & {\footnotesize $\oplus (\Z/3\Z) \oplus (\Z/27\Z)$ } & {\footnotesize $\oplus (\Z/3\Z) \oplus (\Z/27\Z)$ } & {\footnotesize $\oplus (\Z/3\Z) \oplus (\Z/27\Z)$ } \\
\bottomrule
\end{tabular}
\caption{Bordism groups of the three approximations to the structure of IIB quantum gravity.}
\label{the_big_table}
\end{table}
It is apparent that they are far from trivial and thus demand the existence of various
defects to ensure that $\Omega^{\text{QG}}_k$ vanishes. To find these defects, it is useful to
analyze representatives of the generators. For the Spin-$\Mp(2,\Z)$
classes they are given in Table \ref{Mp_gens_table}.
\begin{table}[h!]
\centering
\begin{tabular}{c c c}
\toprule
$k$ & $\Omega_k^{\Spin\text{-}\Mp(2, \Z)} (\pt)$ & Generators\\
\midrule
$0$ & {\footnotesize $\Z$} & {\footnotesize $ \pt_+$} \\
$1$ & {\footnotesize $(\Z/8\Z) \oplus (\Z/3\Z)$} & {\footnotesize $( L^1_4 \,, L^1_3)$}\\
$2$ & $0$\\
$3$ & {\footnotesize $(\Z/2\Z) \oplus (\Z/3\Z)$} & {\footnotesize $(L^3_4 \,, L^3_3)$}\\
$4$ & {\footnotesize $\Z$} & {\footnotesize $E$}\\
$5$ & {\footnotesize $(\Z/2\Z) \oplus (\Z/{32}\Z) \oplus (\Z/9\Z)$} & {\footnotesize $(\widetilde{L}^5_4 \,, L^5_4 \,, L^5_3)$}\\
$6$ & {\footnotesize $0$} \\
$7$ & {\footnotesize $(\Z/4\Z) \oplus (\Z/9\Z)$} & {\footnotesize $(Q^7_4 \,, L^7_3)$}\\
$8$ & {\footnotesize $\Z\oplus\Z$} & {\footnotesize $(B \,, \HP^2)$}\\
$9$ & {\footnotesize $(\Z/4\Z) \oplus (\Z/8\Z) \oplus (\Z/128\Z) \oplus (\Z/3\Z) \oplus (\Z/27\Z)$} & {\footnotesize $( \widetilde{L}^9_4 \,,  \HP^2 \times L^1_4 \,, L^9_4 \,,  \HP^2 \times L^1_3 \,, L^9_3 )$} \\
$10$ & {\footnotesize $\Z/2\Z$} & {\footnotesize $X_{10}$}\\
$11$ & {\footnotesize $(\Z/2\Z)^{\oplus 2} \oplus (\Z/8\Z) \oplus (\Z/3\Z) \oplus (\Z/27\Z)$} & {\footnotesize $( \HP^2\times L^3_4 \,, X_{10} \times L^1_4 \,, Q_4^{11} \,, \HP^2 \times L^3_3 \,, L^{11}_3)$}\\
\bottomrule
\end{tabular}
\caption{Generators of Spin-$\Mp(2, \Z)$ bordism groups. Generators are listed in the same order as the
decomposition of the group.}
\label{Mp_gens_table}
\end{table}
The generators are given by
\begin{itemize}
	\item{Lens spaces: $L^{2n-1}_k = S^{2n-1} / (\Z/k\Z)$ \\ They naturally come equipped with a $\Z/k\Z$ bundle by the fibration
\begin{align}
\begin{split}
\Z/k\Z \hookrightarrow S^{2n-1} & \\
\downarrow \hspace{0.75cm} & \\
L^{2n-1}_k &
\end{split}
\end{align}
The duality bundle can be characterized in terms of this bundle by embedding $\Z/k\Z$ into the duality group.
Moreover, some of the lens spaces allow for various choices for tangential structures,
such as Spin-$\Z/8\Z$ structure for the $L^5_4$. The different choices are indicated via a tilde.
When these different tangential structures cannot be deformed into each other we expect several generating
manifolds with a similar geometric structure. This can be verified explicitly by calculating associated bordism invariants,
such as $\eta$-invariants for Dirac and Rarita-Schwinger fields.
}
\item{Lens space bundles: $Q^{2n-1}_4$ \\The spaces denoted as $Q^{2n-1}_4$ (not to be confused with quaternionic lens spaces) are lens space bundles
in the sense that they are the total spaces of the fibrations
\begin{align}
\begin{split}
L^{2n-3}_4 \hookrightarrow Q^{2n-1}_4 & \\
\downarrow \hspace{0.7cm} & \\
\mathbb{CP}^1 \quad &
\end{split}
\end{align}
The non-trivial fibration structure is obtained by embedding the covering $(2n-3)$-sphere of $L^{2n-3}_4$ in the fiber of a sum of line bundles $\mathcal{L}_1 \oplus \dots \oplus \mathcal{L}_{n-1}$ over $\mathbb{CP}^1$. For our purposes it will be enough to consider the line bundles
\begin{align}
\mathcal{L}_1 \oplus \mathcal{L}_2 \oplus \dots \oplus \mathcal{L}_{n-1} = H^{\pm 2} \oplus \underline{\mathbb{C}} \oplus \dots \oplus \underline{\mathbb{C}} \,,
\end{align}
where $H \cong \mathcal{O}(1)$ denotes the hyperplane bundle of $\CP^1$ and $\underline{\mathbb{C}} \cong \mathcal{O} \cong \mathcal{O}(0)$ the trivial line bundle. Again, these spaces are equipped with a Spin-$\Z/8\Z$ structure inherited from the lens space and demand a non-trivial duality bundle.
}
\item{Enriques surface: $E$ \\ It can be obtained by a fixed-point free $\Z / 2 \Z$ action on a K3 surface.}
\item{Bott manifold: $B$ \\ An arbitrary class in $\Omega_8^{\Spin} (\pt)$
is specified by the value of the index of the Dirac operator and its signature.
We shall specify a ``Bott manifold'' as any Spin 8-manifold with Dirac index equal
to $1$ (see also \cite{FH21}).
An example is an 8-manifold with a Spin$(7)$ structure.\footnote{As a brief aside, one can often produce Spin$(7)$ spaces by taking a
$\Z / 2 \Z$ quotient of a suitable Calabi-Yau fourfold \cite{Joyce:1999nk},
and the Dirac index on the Calabi-Yau fourfold is $2$.} }
\item{Quaternionic projective space: $\HP^2$ \\ This 8-manifold is the analog
of the more familiar $\mathbb{CP}^2$, replacing the complex numbers $\mathbb{C}$
by the quaternions $\mathbb{H}$ in the construction. It has Dirac index 0, and signature 1.}

\item{Milnor surface: $X_{10}$ \\ This is any closed 10-dimensional Spin manifold having non-trivial
$w_4 w_6$, where $w_i$ denotes the $i$th Stiefel-Whitney class of the tangent bundle.
It also appears as one of the generators of $\Omega^{\Spin}_{10}(\pt)$. For an explicit description of a Milnor
surface, see~\cite[\S 3]{Mi65}.}
%, see e.g., \cite{Diaconescu:2000wy}.}
\end{itemize}
As we will see below, for the construction of many of the above backgrounds
it is useful to geometrize the duality bundle using an auxiliary torus. This directly leads us to
F-theory \cite{Vafa:1996xn, Morrison:1996na, Morrison:1996pp} (see \cite{Heckman:2010bq, Taylor:2011wt, Weigand:2018rez} for reviews).
In particular the Spin-$\Mp(2,\Z)$ manifolds in  odd dimensions described above can often be realized as the boundary $\partial{\mathcal{B}}$
of the base manifold of a genus-one fibration $X \rightarrow \mathcal{B}$,
whose total space is given by a local, singular Calabi-Yau manifold.
\begin{figure}
\centering
\includegraphics[width = 0.6 \textwidth]{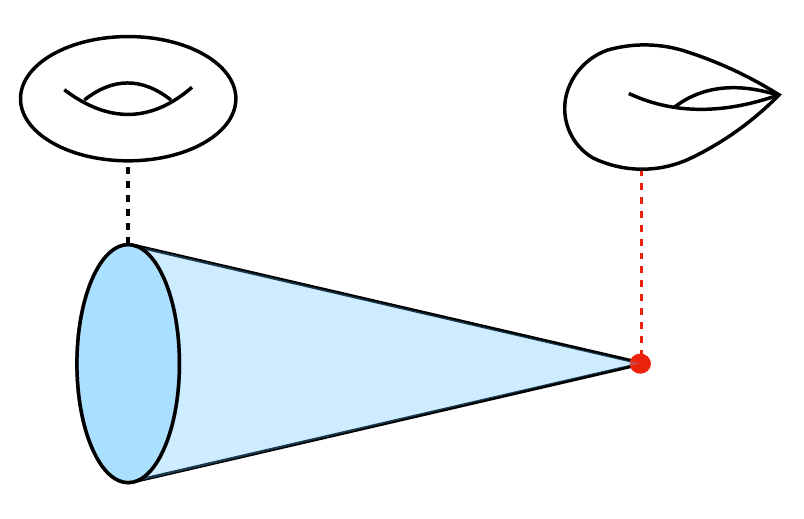}
\caption{The odd-dimensional Spin-$\Mp(2,\Z)$ manifolds often appear as the asymptotic boundary of the base manifold of a genus-one fibered, non-compact, singular Calabi-Yau manifold.}
\label{fig:conicalgeo}
\end{figure}
Thus, one naturally generates conical geometries with a central singularity indicating the presence
of various defects, see Figure \ref{fig:conicalgeo}. From this it is also apparent why we expect
these backgrounds to preserve at least part of the supersymmetry.

All of the manifolds discussed so far actually only contain duality bundles in $\Mp(2,\Z)$ $\subset$ $\GL^+(2,\Z)$, i.e.,  they correspond to bordism classes not involving reflections. Generators for those classes involving  $\GL^+(2,\Z)$ in their duality bundle are given in Table \ref{gl_gens_table}.
\begin{table}[h!]
\begin{tabular}{c c c}
\toprule
$k$ & $\Omega_k^{\Spin\text{-}\GL^+(2, \Z)} (\pt)$ & Generators\\
\midrule
$0$ & {\footnotesize $\Z$} & {\footnotesize $\pt_+$} \\
$1$ & {\footnotesize $(\Z/2\Z)^{\oplus 2}$} & {\footnotesize $(L^1_4 \,, S^1_R)$}\\
$2$ & {\footnotesize $\Z/2\Z$} & {\footnotesize $S^1_p \times S^1_R$} \\
$3$ & {\footnotesize $(\Z/2\Z)^{\oplus 3} \oplus (\Z/3\Z)$} & {\footnotesize $(L^3_4 \,, \RP^3 \,, \widetilde{\RP}^3 \,, L^3_3)$} \\
$4$ & {\footnotesize $\Z$} & {\footnotesize $E$}\\
$5$ & {\footnotesize $(\Z/2\Z)^{\oplus 2}$} & {\footnotesize $(L^5_4 \,, X_5)$}\\
$6$ & {\footnotesize $0$}\\
$7$ & {\footnotesize $(\Z/2\Z)^{\oplus 3} \oplus (\Z/4\Z)\oplus (\Z/9\Z)$} & {\footnotesize $(\RP^7 \,, \widetilde{\RP}^7 \,, \orangeseven  \,, \halfQseven \,, L^7_3)$}\\
$8$ & {\footnotesize $\Z\oplus \Z\oplus (\Z/2\Z)$} & {\footnotesize $(\halfBott \,, \HP^2 \,, \halfQseven \times S_p^1)$}\\
$9$ & {\footnotesize $(\Z/2\Z)^{\oplus 8}$} & {\footnotesize $(\halfQseven \times S_p^1 \times S_p^1 \,, X_9 \,,
	\widetilde X_9 \,, L^9_4 \,, W^9_1 \,,$} \\
& &	 {\footnotesize $\phantom{(} B\times L_4^1 \,, \HP^2 \times L^1_4 \,, \HP^2\times S_R^1 )$} \\
$10$ & {\footnotesize $(\Z/2\Z)^{\oplus 4}$} & {\footnotesize $(B\times L_4^1\times S_p^1 \,, W^9_1 \times S_p^1 \,, \HP^2 \times
	L_4^1 \times S_p^1\,, X_{10})$}\\
$11$ & {\footnotesize $(\Z/2\Z)^{\oplus 9} \oplus (\Z/8\Z) \oplus (\Z/3\Z) \oplus (\Z/27\Z)$} & {\footnotesize $(\RP^{11} \,, \widetilde{\RP}^{11} \,, X_{11} \,, \widetilde{X}_{11} \,, \HP^2 \times L^3_4 \,, \HP^2 \times \RP^3 \,, $} \\
& &	 {\footnotesize $\phantom{(} \HP^2\times \widetilde{\RP}^3 \,, X_{10} \times L_4^1 \,, X_{10}\times S_R^1 \,,  Q_4^{11} \,,   \HP^2 \times L^3_3 \,, L^{11}_3 )$} \\
\bottomrule
\end{tabular}
\caption{Generators of Spin-$\GL^+(2, \Z)$ bordism groups. Generators are listed in the same order as the
decomposition of the group, just as in \cref{Mp_gens_table}.}
\label{gl_gens_table}
\end{table}

In addition to the generators relevant for Spin-$\Mp(2,\Z)$ bordism, we also find:
\begin{itemize}
	\item{Circle with reflection bundle: $S^1_R$ \\ This describes a circle with a non-trivial duality
bundle implemented by a transition function with the reflection in $\GL^+ (2,\Z)$. The two classes
that we find at degree one tell us that there are two distinct compactifications of type IIB on a circle
with a duality bundle; these were recently described in \cite{Montero:2022vva}, and are related to each other by
turning on a discrete theta angle in the RR sector.}
	\item{Circle with periodic boundary conditions: $S^1_p$ \\
	Described by a circle without duality bundle but with periodic boundary conditions for fermions. This is the generator of $\Omega^{\Spin}_1 (\pt) = \Z/2\Z$.}
	\item{Real projective spaces: $\RP^{2n-1}$ and $\widetilde{\RP}^{2n-1}$ \\ These are the usual real projective space derived as $S^{2n-1} / (\Z/2\Z)$. They come equipped with natural $\Z/2\Z$ bundles via the fibration
\begin{align}
\begin{split}
\Z/2\Z \hookrightarrow S^{2n-1} & \\
\downarrow \hspace{0.8cm} & \\
\RP^{2n-1} \hspace{-0.2cm}&
\end{split}
\end{align}
For $\RP^{2n-1}$ the $\Z/2\Z$-bundle is associated to reflections in $\GL^+(2,\Z)$ whereas for
$\widetilde{\RP}^{2n-1}$ it is a combination of rotation and reflection. These are described by two different
embeddings of $\Z/2\Z$ into $\GL^+(2,\Z)$ as described below in Section \ref{Z2_to_D16} and Appendix
\ref{app:embdihedral}. A convenient way to implement this is an identification with the $\Z/2\Z$ action given by
worldsheet reflections $\hat{\Omega}$ in for $\RP^{2n-1}$ and an identification with the $\Z/2\Z$ action given by
worldsheet reflection in combination with an S-duality transformation $\hat{S}$ for $\widetilde{\RP}^{2n-1}$.}
	\item{Arcana: $X_i$ and $\widetilde{X}_i$ ($i \neq 10$) \\ The Arcana are spaces that can be defined as
	\begin{align}
	X_i = \big( S^{2k-1} \times S^{i + 1 - {2k}} \big) / D_4\,. %\big((\Z/2\Z) \oplus (\Z/2\Z)\big) \,.
	\end{align}
	Because $D_4\cong\Z/2\Z\times\Z/2\Z$, the $D_4$-action on $S^{i+1-2k}$ is equivalent to a pair of
	$\Z/2\Z$-actions, which make $X_i$ into a fiber bundle of a real projective space over another real projective space:
	\begin{align}
	\begin{split}
	\RP^{2k-1} \hookrightarrow X_i & \\
	\downarrow \hspace{0.1cm} & \\
	\RP & ^{i + 1 - 2k}
	\end{split}
	\end{align}
	Again, the non-trivial duality bundle described by the natural $D_4$-bundle of $X_i$ can be implemented by
	either an action of worldsheet reflections $\hat{\Omega}$ for $X_i$ or $\hat{S} \, \hat{\Omega}$ for
	$\widetilde{X}_i$ implemented by two different embeddings of $D_4$ into $D_{16}$, see Section \ref{arcana} and
	Appendix \ref{app:embdihedral}. For $i \in \{ 5, 9\}$ one has $k=2$ and for $i = 11$ one has $k = 3$.}
	\item{Lens space bundles over real projective space: $W^{2n-1}_1$ \\ These spaces can be obtained from the lens
	space bundles over spheres, including $Q^{2n-1}_4$ and the trivial bundle, by a further $\Z/2\Z$ quotient. In this way $W^7_1$ can be identified as ``half $Q^7_4$''. Since $\Z/2\Z$ acts as sign reversal in the base the resulting spaces are given by
	\begin{align}
	\begin{split}
	L^{k}_4 \hookrightarrow W^{n}_1 & \\
	\downarrow \hspace{0.35cm} & \\
	\RP^{m} \hspace{-0.1cm} &
	\end{split}
	\end{align}
	The associated duality bundles for these spaces contain both reflections as well as rotations. We use two
	examples: $W_1^7$, for which $k = 5$ and $m = 2$; and $W_1^9$, for which $k = 5$ and $m = 4$.}
	\item Prism space bundle over $(T^2\times S^2)/(\mathbb{Z}/2)$: $W_2^7$\\
In dimension 7, we have an additional generator, which is described as a bundle of the prism manifold $S^3/\Gamma$, where $\Gamma$ is the Dyciclic group of order 16:
\begin{align}
	\begin{split}
	S^3/\Gamma \hookrightarrow W^{7}_2 & \\
	\downarrow \hspace{0.35cm} & \\
	\frac{T^2\times S^2}{\mathbb{Z}/2} \hspace{-0.1cm} &
	\end{split}
	\end{align}
Here, the base is the quotient of $T^2\times S^2$ by the combined involution of the involution that leads to the Klein Bottle on $T^2$ and antipodal mapping on $S^2$. 	
The associated duality bundle also includes reflections, just as in the previous case.
	\item{Half-Bott manifold: $W_{1,8}$ \\ The half Bott manifold can be obtained by a $(\Z/2\Z) \oplus (\Z/2\Z)$ quotient of a complete intersection manifold in $\CP^7$.
We refer to this as a ``half-Bott'' manifold because quotienting by just one of the two $\Z / 2 \Z$'s
would produce a Bott manifold. The half-Bott is obtained after a further $\Z / 2 \Z$ quotient,
leading to a non-Spin manifold on which $\int [\hat{A}]_8$ (the class that measures the Dirac index on a Spin manifold)
is equal to $1/2$. Although the manifold is not Spin, it carries a Spin-${D_{16}}$ structure.}
\end{itemize}
For more details of the generators and their duality bundles see the explicit constructions in Sections \ref{mp_spin} and \ref{gl_spin}.

With all the bordism groups and their generators determined we can try to relate the specific backgrounds to
(extended) string theory objects. The CC
demands that all these classes are actually boundaries of some IIB backgrounds,
which necessarily involve singular geometries and duality bundles.
The generators can be viewed as the asymptotic geometry around the respective defects,
see Figure \ref{fig:exampleII} for a depiction.
\begin{figure}
\centering
\includegraphics[width = 0.5\textwidth]{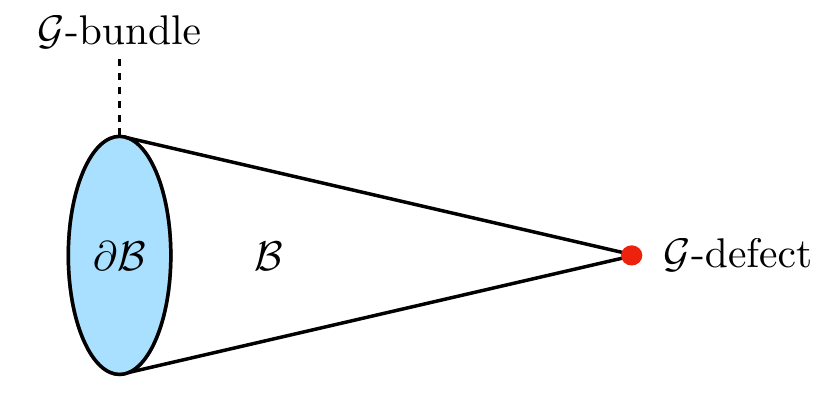}
\caption{Bulk / Boundary characterization of bordism generators and the corresponding
defects predicted by the CC. Viewing the bordism generator for manifolds
equipped with a structure defined by $\mathcal{G}$ as the boundary of some bulk space $\mathcal{B}$ equipped with a $\mathcal{G}$-bundle,
the CC predicts the existence of a $\mathcal{G}$-defect which allows the class to be trivialized in $\Omega^{\mathrm{QG}}_{k}$.}
\label{fig:exampleII}
\end{figure}
In the 10-dimensional string theory we therefore find that non-trivial bordism
classes $\Omega^{\mathcal{G}}_k$ are related to defects of codimension $(k+1)$.
Since type IIB string theory contains many objects of codimension
$2\ell$ with $\ell \in \{1, \dots, 4\}$ it is tempting to expect that some of these
objects are related to the defects predicted by the Cobordism Conjecture.
This is indeed correct as we will carefully analyze in Part \ref{p:physics}, with one potential caveat: while the backgrounds we construct are certainly consistent at the level of supergravity, it may be possible that type IIB string theory does not make sense on some of them. This will happen if, for instance, the worldvolume of a probe brane has an anomaly in the backgrounds, similarly to the situation in \cite{Witten:1996md}. In the language of \cite{McNamara:2019rup}, we would say that the bordism class is ``gauged'', and remove it from the discussion. We do not know if any of the classes we construct is ``gauged'' in this way, as a systematic analysis has not been carried out, unlike in the M-theory case \cite{FH21}. In this paper, we will assume that all the classes we construct uplift to IIB string theory backgrounds, and find the objects that would become their boundary if this is the case. It is worth noticing that in known examples, whenever a bordism class is gauged in the sense of \cite{McNamara:2019rup} (killing the class), it is often the case that a closely related ``cousin'' exists, having the same geometry but e.g., different flux, and which remains alive. An example is the class of $\RP^4$ in M-theory, which is itself inconsistent, but the inconsistency is removed if one includes a half-integer $G_4$ flux \cite{Witten:1996md}. So while the reader should keep in mind that we have not checked all anomalies of probe branes in the backgrounds we discuss, we believe these are unlikely to do much more than force the introduction of some fluxes.

For Spin-$\Mp(2,\Z)$ classes and odd $k \leq 9$, we find the identifications
in Table \ref{tab:Mpdefects}. Many of the Spin-$\Mp(2,\Z)$ configurations
can be associated with supersymmetric backgrounds for F-theory containing interesting objects,
such as non-Higgsable clusters \cite{Morrison:2012np} (see also \cite{Heckman:2013pva})
and S-folds \cite{Garcia-Etxebarria:2015wns, Aharony:2016kai} (see also \cite{Apruzzi:2020pmv, Heckman:2020svr, Giacomelli:2020gee}),
as well as generalizations thereof.

\begin{table}[h!]
\centering
\begin{tabular}{c c c c}
\toprule
\toprule
$k$ & $\Omega_k^{\Spin\text{-}\Mp(2, \Z)} (\pt)$ & Generators & Defects\\
\midrule
$1$ & {\footnotesize $(\Z/8\Z) \oplus (\Z/3\Z)$} & {\footnotesize $(L^1_4 \,, L^1_3)$} & {\footnotesize $[p,q]$-7-branes} \\
$3$ & {\footnotesize $(\Z/2\Z) \oplus(\Z/3\Z)$} & {\footnotesize $(L^3_4 \,, L^3_3)$} & {\footnotesize non-Higgsable clusters} \\
$5$ & {\footnotesize $(\Z/2\Z) \oplus (\Z/32\Z) \oplus (\Z/9\Z)$} & {\footnotesize $(\widetilde{L}^5_4 \,, L^5_4 \,, L^5_3)$} & {\footnotesize S-folds} \\
$7$  & {\footnotesize $\Z/4\Z$} & {\footnotesize $Q^7_4$} & {\footnotesize twisted compactification of S-folds} \\
 & {\footnotesize $\Z/9\Z$} & {\footnotesize $L^7_3$} & {\footnotesize S-string} \\
$9$ & {\footnotesize $(\Z/4\Z) \oplus (\Z/{128}\Z) \oplus (\Z/27\Z)$} & {\footnotesize $( \widetilde{L}^9_4 \,, L^9_4 \,, L^9_3)$} & {\footnotesize S-instantons} \\
 & {\footnotesize $(\Z/8\Z) \oplus (\Z/3\Z)$} & {\footnotesize $(\HP^2\times L^1_4 \,, \HP^2\times L^1_3)$} & {\footnotesize compactifications of $[p,q]$-7-branes }\\
\bottomrule
\end{tabular}
\caption{Defects which source asymptotic Spin-$\Mp(2,\Z)$ bordism classes.}
\label{tab:Mpdefects}
\end{table}

For Spin-$\GL^+ (2, \Z)$ bordism classes, not all the defects can be associated to
generalizations of known string theory objects. This is simply because any putative bulk geometry
would necessarily break supersymmetry, and most objects understood in string theory typically preserve some supersymmetry.
Given this, we find it very remarkable that the inclusion of essentially
one new physical object suffices to generate the remaining bordism classes. This
new object is a 7-brane (codimension two in the 10d spacetime) that
implements reflections in the $\GL^+ (2,\Z)$ duality group, and we refer to it as a reflection 7-brane or ``R7-brane''
for short \cite{Dierigl:2022reg}. The R7-brane trivializes the class of a circle with a reflection bundle,
and compactifications involving wrapped R7-branes are enough to trivialize all
the other classes at lower degrees.  The full results are summarized in Table \ref{tab:GLdefects}.
\begin{table}[h!]
\centering
\begin{tabular}{c c c c}
\toprule
\toprule
$k$ & $\Omega_k^{\Spin\text{-GL}^+(2, \Z)} (\pt)$ & Generators & Defects\\
\midrule
$1$ & {\footnotesize $\Z/2\Z$} & {\footnotesize $L^1_4$} & {\footnotesize $[p,q]$-7-branes} \\
 & {\footnotesize $\Z/2\Z$} & {\footnotesize $S^1_R$} & {\footnotesize R7-brane} \\
$3$ & {\footnotesize $(\Z/2\Z) \oplus (\Z/3\Z)$} & {\footnotesize $(L^3_4, L^3_3)$} & {\footnotesize non-Higgsable clusters} \\
 & {\footnotesize $(\Z/2\Z)^{\oplus 2}$} & {\footnotesize $(\RP^3, \widetilde{\RP}^3)$} & {\footnotesize O5-plane, compactified R7-brane} \\
$5$ & {\footnotesize $\Z/2\Z$} & {\footnotesize $L^5_4$} & {\footnotesize S-folds} \\
 & {\footnotesize $\Z/2\Z$} & {\footnotesize $X_5$} & {\footnotesize twisted compactification of O5 and R7}\\
$7$ & {\footnotesize $\Z/9\Z$} & {\footnotesize $L^7_3$} & {\footnotesize S-string}\\
 & {\footnotesize $\Z/4\Z$} & {\footnotesize $W^7_1$} & {\footnotesize double-twisted compactification of S-fold} \\
 & {\footnotesize $(\Z/2\Z)^{\oplus 2}$} & {\footnotesize $(\RP^7, \widetilde{\RP}^7)$} & {\footnotesize O1-plane, compactified R7-brane} \\
 & {\footnotesize $\Z/2\Z$} & {\footnotesize $\orangeseven$} & {\footnotesize twisted compactification of a D-type Du Val singularity} \\
$9$ & {\footnotesize $\Z/2\Z$} & {\footnotesize $L^9_4$} & {\footnotesize S-instantons} \\
 & {\footnotesize $(\Z/2\Z)^{\oplus 4}$} & {\footnotesize $(W^7_1 \times S^1_p \times S^1_p \,, B \times L_4^1 \,,$} &
{\footnotesize  compactifications of 7-branes} \\
 & & {\footnotesize $\phantom{(} \HP^2 \times L_4^1 \,, \HP^2 \times S_R^1)$} & \\
 & {\footnotesize $(\Z/2\Z)^{\oplus 2}$} & {\footnotesize $(X_9, \widetilde{X}_9)$} & {\footnotesize twisted compactifications of O5  and R7} \\
 & {\footnotesize $\Z/2\Z $} & {\footnotesize $W^9_1$} & {\footnotesize double-twisted compactifications of S-folds}\\
\bottomrule
\end{tabular}
\caption{Defects which source Spin-$\GL^+(2,\Z)$ bordism classes. Here ``R7-brane'' refers to a codimension two defect (i.e., a 7-brane)
in which monodromy around it leads to a reflection on the corresponding F-theory torus.}
\label{tab:GLdefects}
\end{table}

The even-dimensional generators are conceptually different, since they usually descend directly
from the generators of Spin manifolds and thus, with a few exceptions, do not carry a non-trivial duality bundle.
While the particular string backgrounds with these even-dimensional manifolds as a boundary still need to be identified,
they are closely related to the cases which appeared in the analysis of \cite{McNamara:2019rup}.

\section{Conclusions and future directions}
\label{sec:otherapp}

In this section we present our conclusions and future directions. Readers interested in a more complete account
will find it in the subsequent parts of this paper.

The main theme in much of this paper is the interplay between the non-Abelian duality symmetry of
type IIB strings / F-theory and the Cobordism Conjecture. We have seen that for Spin-$\Mp(2,\Z)$
bordisms, the CC passes many non-trivial checks. Conversely, the CC predicts the existence of some genuinely new non-supersymmetric
backgrounds, some of which will be explored in future work. It is important to stress that at no point did our analysis of bordism groups
actually rely on supersymmetry. In this sense it is a rather robust tool which is complementary to many
existing methods. A central component of our analysis has been the computation of all of the
relevant bordism groups via the Adams spectral sequence. The results we have obtained are more broadly applicable in the study of quantum field
and strings, and in the remainder of this section we discuss some further potential avenues for investigation.

We have seen that the CC predicts various new backgrounds and objects for type IIB string theory.
In the case of Spin-$\Mp(2,\Z)$ bordism, these have a lift to F-theory geometries that preserve part of the supersymmetry.
These singular configurations usually contain localized interacting degrees of freedom.
For example, we have found non-Higgsable clusters that realize a simple class of six-dimensional superconformal field theories (SCFTs)
as well as S-folds. It would therefore be interesting to analyze the localized degrees of freedom and their dynamics for the
other backgrounds as well, e.g., the S-strings. The same is true for the more challenging Spin-$\GL^+ (2,\Z)$ setups,
which do not have a conventional F-theory interpretation.

Turning the discussion around, we can also use the Spin-$\GL^+ (2,\Z)$ bordisms to generalize our notion of F-theory itself.
Instead of an $\SL(2,\Z)$ action on the fiber, it is rather clear that we should also permit reflection symmetries generating
$\GL(2,\Z)$ actions on the geometry of the total space. It would be very interesting to see whether one can find
configurations that preserve part of the supersymmetry while still explicitly containing duality transition functions
with determinant $(-1)$. These likely do not contain R7-branes since these objects appear to be
non-supersymmetric \cite{Dierigl:2022reg}.

From this perspective it is natural to ask whether the appearance of reflections
could also be used to construct metastable non-supersymmetric
vacua, a topic of clear phenomenological relevance. From this perspective it is natural
to ask whether R7-branes, in tandem with other string theoretic ingredients, could be used to construct such backgrounds.

Type IIB string theory is of course not isolated and can be related to other string theories as well as M-theory
via various dualities. This motivates a further investigation into the fate of the predicted new backgrounds
under these dualities. Using the usual dictionary, a circle compactification of an F-theory background
is related to M-theory of the same torus-fibered geometry
(without the extra circle). In this way we can track the Spin-$\GL^+ (2,\Z)$ structure to a Pin$^+$ structure of
the M-theory spacetime \cite{Tachikawa:2018njr}. Since the M-theory geometry does not have to preserve the torus fiber, it might be that
several classes can be trivialized without resorting to singular objects. However, at least a subset of the novel IIB
configurations must survive and will thus have a direct M-theory counterpart.

A particularly fascinating implication of our results is that we have uncovered four discrete $\theta$-angles
for IIB string theory in ten dimensions. Some of them were already discussed in the literature long ago;
for instance, the $\theta$-angles that do not vanish on the product of a Bott manifold with a two-torus do not vanish on exotic 10-spheres,
and so were implicitly discussed in \cite{Witten:1985xe}. It remains an important open question to elucidate whether
these and other $\theta$-angles, such as the ones proposed in \cite{Debray:2021vob}
(see also \cite{FH21, Diaconescu:2000wy, Diaconescu:2000wz} as well as \cite{Montero:2022vva}),
are part of the Landscape or the Swampland.

The computation of bordism groups with a prescribed duality bundle structure has many
other potential applications. As an example, in the context of a $(k-1)$-dimensional theory of quantum gravity,
the bordism group $\Omega_{k}^{\mathcal{G}}$ characterizes possible anomalies of a given theory
with non-trivial dualities. In particular, we can apply this for all of the different refinements of IIB duality
considered in this paper, namely $\SL(2,\Z)$, $\Mp(2,\Z)$ and $\GL^{+}(2,\Z)$.
Moreover, this also provides non-trivial constraints on theories
with possibly large duality symmetries. It would be interesting to investigate other supergravity theories
and their U-duality symmetries, perhaps along the lines of \cite{Debray:2021vob, Debray:2022wcd}.

Conversely, it is also natural to consider situations in which the actual duality group is a proper subgroup of the
ones considered in this paper. This arises in various supergravity and supersymmetric quantum field theories (see e.g., \cite{Dierigl:2020wen} for a
recent discussion), and so it is natural to ask about the spectrum of objects predicted by the CC in such cases.

One can also consider limits in which the effects of gravity are switched off, namely
in the study of quantum field theories with a known duality symmetry. In that context, one can also consider
mixed gravitational-duality anomalies, much as in \cite{Seiberg:2018ntt, Hsieh:2019iba, Cordova:2019jnf, Cordova:2019uob}.
Again, the bordism group calculations provided here amount to a characterization of possible anomalies which can arise in a given theory.

Aside from these physical applications, there are also a number of interesting directions to pursue on purely mathematical grounds.
To give one example, here we have focused on the case of bordism groups $\Omega_{k}^{\mathcal{G}}$ for $k \leq 11$. It would be quite
interesting to develop methods which also work at ``large'' values of $k$.

A fortuitous feature of our calculations is that all of the relevant data is concentrated at low primes, in particular primes $p = 2 ,3$. Turning
the question around, one could also ask whether there is a natural sense in which higher primes could arise in this context. A potentially
promising route in this vein would likely to be consider similar calculations for the
congruence subgroups of $\SL(2,\Z)$ such as $\Gamma(p)$, $\Gamma_{0}(p)$ and $\Gamma_{1}(p)$, as well as their metaplectic and Pin$^+$
covers.

Another intriguing direction to develop would be the relation between the bordism theories considered here and
related constructions in mathematics. For example, many of the generators we produced for Spin-$\Mp(2, \Z)$ bordism
can be thought of as the boundary geometries of genus-one fibered Calabi-Yau spaces. Since there is also a purely
algebro-geometric approach to bordism developed by Levine-Morel in \cite{LEVINE2001723, LEVINE2001815, LM02, Lev02,
levine2007algebraic} (see also \cite{levinepandharipande}), it is natural to ask whether there is a more direct connection between Levine-Morel's bordism theory and Spin-$\Mp(2, \Z)$ bordism. In a different direction, there could also be connections with chromatic
homotopy theory: the duality bundles we consider can often be ``geometrized'' in terms of a genus-one curve /
elliptic curve, suggesting connections to elliptic cohomology theories (see, e.g., \cite{morava2007complex});
indeed, there are notions of $\SL(2, \Z)$-equivariant elliptic cohomology or related objects~\cite{HL16} believed
to be related to string theory~\cite{KS05b, KS05a, Sat10}. More directly, work of Tachikawa, Yamashita, and
Yonekura~\cite{TY21, Tac22, Yon22} relates the theory of topological modular forms, a sort of universal elliptic
cohomology theory, to anomaly cancellation in heterotic string theory, and it would be interesting to pursue
analogous questions in type IIB string theory.  We believe that these questions are all ripe for further
exploration, and look forward to exploring some of them in the near future.

\paragraph*{Acknowledgements}

We thank P.\ Giudice for generously donating his time and artistic abilities in crafting the map of IIBordia.

We thank D.S.\ Berman, D.S.\ Freed, I.\ Garc\'{i}a Etxebarria, S.\ Goette, T.\ Johnson-Freyd, C.\ Krulewski, Y.L.\ Liu,
D.\ L\"ust, J.\ McNamara, L.\ Stehouwer, E.\ Torres, A.P.\ Turner, C.\ Vafa, I. Valenzuela, E.\ Witten, M.\ Yu, and H.Y.\ Zhang for helpful discussions. AD
and MM thank the organizers of the conference ``Geometric Aspects of the Swampland'' for kind hospitality during
part of this work. JJH thanks the organizers of the conference ``Geometrization of (S)QFTs in $D \leq 6$'' held at
the Aspen Center for Physics (support from NSF grant PHY-1607611) for kind hospitality during part of this work.
JJH also thanks the organizers of the workshop ``Geometry, Topology and Singular Special Holonomy Spaces'' held at
Freiburg University for kind hospitality. JJH and MM thank the 2021 and 2022 Simons Summer Workshop on Mathematics
and Physics for kind hospitality during part of this work. MM and MD thank the ``Engineering in the Landscape:
Geometry, Symmetries and Anomalies'' workshop held at Uppsala University.  The work of MD is supported by the
German-Israeli Project Cooperation (DIP) on ``Holography and the Swampland''. The work of JJH is supported by DOE
(HEP) Award DE-SC0013528. The work of MM was supported by a grant from the Simons Foundation (602883, CV) for most of the duration of this project, and by a Mar\'{i}a Zambrano Scholarship (UAM) CA3/RSUE/2021-00493 during the last stages. MM acknowledges the support of the Spanish Agencia Estatal de Investigacion through the grant ``IFT Centro de Excelencia Severo Ochoa SEV-2016-0597''.

\newpage\section{Final words of encouragement}

We hope that this brief introduction and summary of our results is sufficient
for those only interested in the final results of our work and their physical implications.

For those brave enough to undertake the journey to understand how these results were obtained:
The road is long and arduous, yet the rewards are more than worth it. You will learn bordism theory,
F-theory, topology, index theory, the intricacies of IIB string theory, and above all, the Dark Art of Spectral Sequences.
This arcane tool will allow you to tackle questions and problems which would be beyond reach otherwise, in String Theory and beyond.
The Map of the Lands of IIBordia, conjured forth through the sorcerous arts of Paco Giudice, is depicted at the beginning of the paper to aid you in your quest. Any road taken from Introduction Camp will lead somewhere interesting, and they can all be taken
largely independently from each other.\footnote{The paths of the physicist warrior and mathemagician are closer than they might first appear; we do not specify a metric on IIBordia (the metric is still unknown).} The legend goes that only the traveler who has explored
all corners of the Land of IIBordia will achieve full enlightenment.
We hope it is of some help should you lose your way. Godspeed.

%%%%%%%%%%%%%%%%%%%%%%%%%%%%%%%%%%%%%%%%%%%%%%%%%%%%
%%%%%%%%%%%%%%% PHYSICS PART %%%%%%%%%%%%%%%%%%%%%%%%%%%%
%%%%%%%%%%%%%%%%%%%%%%%%%%%%%%%%%%%%%%%%%%%%%%%%%%%%
\pagebreak
\part{Type IIB backgrounds and defects}
\label{p:physics}

Having given an overview to the main results of this paper, we now turn to an analysis of the defects predicted by the
Cobordism Conjecture \cite{McNamara:2019rup}. As mentioned in the introduction, the Cobordism Conjecture
states that for a consistent theory of quantum gravity in $D$ dimensions, the lower bordism groups must be trivial
\begin{align}
\Omega_{k}^{\text{QG}} = 0 \,, \quad \text{for} \enspace 0 < k < D \,,
\end{align}
in order to avoid associated global symmetries. Since the bordism groups under investigation are far from trivial (see Table \ref{the_big_table}), we are left with two options:
\begin{itemize}
	\item{There are more ingredients in the full theory of quantum gravity than is captured by the duality bundle and Spin structure. These can often be identified with extended objects in the corresponding string theory. After the inclusion of these objects the spacetime can be deformed to a point and thus becomes null-bordant. This process can be regarded as a breaking of the global symmetry induced by non-trivial bordism classes.}
	\item{The backgrounds associated to non-trivial bordism classes are forbidden in the full theory of quantum gravity. Constraints of this kind might be imposed by subtle tadpole cancellation conditions. This can be interpreted as a gauging of the associated global symmetry.}
\end{itemize}
In the following we will assume that the first of the above possibilities is realized and describe the necessary objects whose inclusion trivializes the bordism groups under consideration. For small $k$ this will lead to the (re)discovery of various type IIB configurations and objects. Moreover, since type IIB contains spacetime fermions we focus on the two variants Spin-$\Mp(2,\Z)$ and Spin-$\GL^+ (2,\Z)$.

Before we start our analysis, we point out that the bordism groups under consideration are only sensitive to the spacetime topology and the discrete duality bundle. In particular our analysis does not include the various higher-form Ramond-Ramond (RR) and Neveu-Schwarz (NS) fluxes of type IIB. A consequence of this is that there can in principle be multiple string backgrounds which have the same asymptotic duality bundle structure but which correspond to physically distinct flux backgrounds. Wherever possible we are interested in the supersymmetric variants of the objects breaking the global symmetry induced by the non-trivial bordism classes which can be related to known string theory objects. Therefore, we allow for induced fluxes, as necessary.

%%%%%%%%%%%%%%%%%%%%%%%%%%%%%%%%%%%%%%%%%%%%%%%%%%%%%%%%

\section{Spin-Mp defects}
\label{sec:Mpdefects}

In this section we focus on the non-trivial bordism classes in $\Omega^{\text{Spin-Mp}(2,\Z)}_k (\pt)$. We will find that these classes are closely related to certain F-theory backgrounds. For these, the value of the axio-dilaton
\begin{align}
\tau = C_0 + i e^{- \phi} \,,
\end{align}
with $C_0$ the RR 0-form and dilaton $\phi$, is encoded in an auxiliary torus $T^2 \simeq \mathbb{C}/\Lambda_{\tau}$ defined by the lattice $\Lambda_{\tau} = \langle 1, \tau \rangle$ fibered over the physical spacetime. Thus, F-theory geometries are described by torus fibrations \cite{Vafa:1996xn, Morrison:1996na, Morrison:1996pp} (see \cite{Heckman:2010bq, Weigand:2018rez} for reviews). If the total space of the fibration is Calabi-Yau, some of the supersymmetry will be preserved and the background is guaranteed to solve the Einstein equations. The $\Mp(2, \Z)$ bundle of the type IIB background is directly encoded in the information associated to this fibration. Moreover, the total space of the elliptic fibrations, as long as they are smooth, allow for a Spin structure in the case of the Spin-$\Mp(2,\Z)$ bordism classes, which already hints at their relevance in the supersymmetric setups.

The specific strategy for odd dimensional Spin-$\Mp(2,\Z)$ bordism generators is as follows. If the generator is given by a lens space it can be described as the asymptotic boundary of the quotient $\mathbb{C}^n / (\Z / k \Z)$. The duality bundle is then specified as the holonomy generated by traversing the non-trivial 1-cycle of the lens space, i.e., to interpolate between two points identified by the $\Z/ k \Z$ action. Using the geometrization of the duality bundle described by F-theory, this can be realized by specifying the action of $\Z /k \Z$ on the $T^2$ fiber, thus obtaining a space of the form
\begin{align}
(\mathbb{C}^n \times T^2) / (\Z / k \Z) \,.
\end{align}
For an appropriate $\Z / k \Z$ action on the fiber and base coordinate the total space is described by a non-compact, singular Calabi-Yau $(n+1)$-fold, indicating that part of the supersymmetry is preserved, with the generator of the bordism group as the asymptotic boundary of the base manifold, see Figure \ref{fig:conicalgeo}. The constructed F-theory geometries with singular central fiber then provide a natural class of backgrounds which have an asymptotic boundary corresponding to a candidate generator in the associated bordism group. For the $\Z / k \Z$ action on the central fiber to be well-defined, the complex structure of the torus might be restricted to special values leading to particular fiber degenerations summarized in Table \ref{tab:singfiber}, see also Appendix \ref{app:groups}.
\begin{table}[h!]
\centering
\begin{tabular}{c c c c}
\toprule
\toprule
Monodromy & Fixed value of $\tau$ & Kodaira type & Gauge algebra \\ \midrule
$S$ & $\tau = i$ & $III^*$ & $\mathfrak{e}_7$ \\
$S^2$ & $\tau$ arbitrary & $I_0^*$ & $\mathfrak{so}(8)$ \\
$S^3$ & $\tau = i$ & $III$ & $\mathfrak{su}(2)$ \\
$U^2$ & $\tau = e^{2 \pi i /3}$ & $IV^*$ & $\mathfrak{e}_6$ \\
$U^4$ & $\tau = e^{2 \pi i /3}$ & $IV$ & $\mathfrak{su}(3)$ \\
\bottomrule
\end{tabular}
\caption{Fixed value of the complex structure $\tau$ of the fiber torus depending on $\SL(2,\Z)$ bundle and associated Kodaira type and gauge algebra of the central singularity.}
\label{tab:singfiber}
\end{table}
At this level we only specified the bosonic version of the duality bundle, i.e., the $\SL(2,\Z)$ bundle. The full Spin-$\Mp(2,\Z)$ structure is specified by choosing a Spin structure on the total space.

We summarize the involved objects associated to $\Omega^{\text{Spin-Mp}(2,\Z)}_k (\pt)$ for odd $k \leq 9$ in Table \ref{tab:Mpdefects} above and will now investigate each dimension individually. As mentioned previously, the even-dimensional asymptotic backgrounds are closely related to backgrounds that appeared in \cite{McNamara:2019rup}.

\subsection{Codimension-two defects}
\label{subsec:Mpcodim2}

Let us start with
\begin{align}
\Omega^{\text{Spin-Mp}(2,\Z)}_1(\pt) = (\Z/8\Z) \oplus (\Z/3\Z) \simeq \Z/{24}\Z \,,
\label{eq:Mp1bord}
\end{align}
generated by circles with non-trivial $\Mp(2,\Z)$ duality bundle. These are classified by the transition functions and thus the non-trivial bordism classes are in one-to-one correspondence with the Abelianization of the duality group
\begin{align}
\text{Ab}\big( \Mp (2,\Z) \big) = \Z/{24}\Z \,.
\end{align}
In terms of the duality group generators presented in Appendix \ref{app:groups} the transition function generating $\Omega^{\text{Spin-Mp}(2,\Z)}(\pt)$ is given by $\hat{T}$, the Spin lift of the $T$ generator in $\SL(2,\Z)$. To break this bordism group we therefore need to introduce a defect in real codimension two, which induces a non-trivial duality monodromy given by $\hat{T}$. Luckily such a defect is well-known in type IIB string theory and F-theory and is given by the D7-brane.\footnote{Note that one can also choose different combinations of 7-branes that break the full bordism group. These can be interpreted as stacks of  more general $[p,q]$-7-branes, see below.} Hence, we see that the requirement of vanishing bordism classes forces us to include defects that can be associated to known stringy objects, namely $[p,q]$-7-branes (see e.g., the review in \cite{Weigand:2018rez}).\footnote{In order to recover the non-Abelian braiding of general $[p,q]$-7-branes \cite{Dierigl:2020lai} suggested an alternative formulation involving (in the supersymmetric case) the moduli spaces of genus-one curves.}

The associated F-theory geometry can be constructed as follows. The non-trivial duality bundle on the circle translates to a monodromy of the F-theory fiber torus around the spacetime circle, leading to a geometry of the type
\begin{align}
( S^1 \times T^2 ) / (\Z/k\Z) \,,
\label{eq:codim2geo}
\end{align}
where the $\Z/k\Z$ acts by a rotation by $\tfrac{2 \pi}{k}$ and as the associated Mp$(2,\Z)$ transformation on the fiber.
Now, filling the base $S^1$ to obtain a disc, one encounters at least one singular fiber in the interior which indicates the presence of $[p,q]$-7-branes. This can be seen by interpreting \eqref{eq:codim2geo} as the boundary of
\begin{align}
( \mathbb{C} \times T^2 ) / (\Z/k\Z) \,,
\end{align}
with the base coordinate $z \mapsto e^{2 \pi i /k} z$, having a singular central fiber at $z = 0$. In this way all the relevant defects can be described as local patches of elliptically-fibered K3 manifolds, demonstrating that half of the supersymmetry is preserved by the associated defects and only BPS-objects are needed to break the global symmetry associated to the bordism group.

Decomposing $\Z/24\Z$ in \eqref{eq:Mp1bord} into the parts $(\Z/8\Z) \oplus (\Z/3\Z)$ we see that instead of a
D7-brane one can use brane stacks generating the monodromies $\hat{S}$ and $\hat{U}^2$, see Appendix
\ref{app:groups}, respectively, to break the individual summands. The brane stack associated to the monodromy
$\hat{S}$ is given in terms of a type $III^*$ F-theory fiber corresponding to a $\mathfrak{e}_7$ stack; similarly, $\hat{U}^2$ corresponds to a type $IV^*$ fiber, i.e., an $\mathfrak{e}_6$ stack, see Table \ref{tab:singfiber} and e.g., \cite{Dasgupta:1996ij,Weigand:2018rez}.

As an aside we now want to give an intuitive explanation for why the monodromy of 24 D7-branes leads to the trivial class of the bordism group. For that we use F-theory intuition and encode the duality bundle in the geometry of an elliptic fibration. The 24 7-branes translate to 24 $I_1$ fibers of the elliptic fibration. Since we know that a K3 surface can be constructed as an elliptic fibration over $\mathbb{CP}^1$ with 24 $I_1$ fibers, we have found the bounding manifold in an F-theory setting. Forgetting about the Spin structure, 12 $I_1$ fibers suffice to obtain a trivial $\SL(2,\Z)$ monodromy and the periodicity of the duality part of $\Omega^{\Spin}_1 \big( B \SL(2,\Z) \big)$, i.e., the $\Z/{12}\Z$, factor can be explained by the existence of an elliptic fibration for $dP_9$ or half-K3, see also \cite{GarciaEtxebarria:2020xsr}.

\subsection{Codimension-three defects}
\label{subsec:Mpcodim3}

We have that $\Omega_2^{\text{Spin-Mp}(2,\Z)}(\pt) = 0$, and so every two-dimensional manifold with a
Spin-Mp$(2,\Z)$ structure is trivial in bordism.  This is to be contrasted with $\Omega_2^{\Spin} (\pt) \cong \Z/2\Z$, but
it is similar to $\Omega^{\SSO}_2 (\pt)=0$ for oriented bordism. One way to understand this vanishing is to look
ahead to \cref{mp2_at_3,mp2_at_2}, where we learn that the $2$-torsion subgroup of $\Omega_2^{\Spin\text{-}\Mp(2,
\Z)}$ is isomorphic to the $2$-torsion subgroup of $\Omega_2^{\Spin\text{-}\Z/8\Z}$, and that the $3$-torsion
subgroup is isomorphic to the $3$-torsion subgroup of $\Omega_2^\Spin(B\Z/3\Z)$. Then, since
$\widetilde{\Omega}_2^{\Spin}(B\Z/3\Z)=0$, $\Omega_2^{\Spin\text{-}\Mp(2,\Z)}(\pt)$ is essentially equivalent to
$\Omega_2^{\Spin\text{-}\Z/8\Z} (\pt)$, corresponding to manifolds with $\text{Spin-}\Z/8\Z$ structure, and this
bordism group also vanishes. We will put this seemingly trivial result to interesting physical use in Section
\ref{subsec:FonK3to}.

\subsection{Codimension-four defects}
\label{subsec:Mpcodim4}

Next, let us discuss the necessary real codimension-four defects associated to
\begin{align}
\Omega_3^{\text{Spin-Mp}(2,\Z)}(\pt) = (\Z/ 2\Z) \oplus (\Z/3\Z) \,,
\end{align}
with generators given by the two lens spaces $L^3_4$ and $L^3_3$ equipped with non-trivial duality bundles. As
above these can be conveniently lifted to F-theory geometries described as the boundaries of
\begin{align}
( \mathbb{C}^2 \times T^2 ) / (\Z/k\Z) \,.
\label{eq:nHcgeo}
\end{align}
The $\Z/k\Z$ action on the base coordinates is given by
\begin{align}
(z_1, z_2) \mapsto (\omega z_1, \omega z_2) \,, \text{ with } \omega = e^{2 \pi i/ k} \,.
\end{align}
The action on the fiber is the natural $\Mp(2,\Z)$ action with an element of order $k$.

For $k = 4$ the $\Mp(2,\Z)$ element is given by $\hat{S}^2$. In $\SL(2,\Z)$ this is associated to the element
\renewcommand{\arraystretch}{1}
\begin{align}
S^2 = \begin{pmatrix} -1 & 0 \\ 0 & -1 \end{pmatrix} \,.
\end{align}
Denoting the complex coordinate on the torus $T^2 \simeq \mathbb{C} / \Lambda_{\tau}$ by $\lambda$ this action can be described as
\begin{align}
\lambda \mapsto - \lambda = \omega^2 \lambda \,.
\end{align}
This geometry
\begin{align}
(\mathbb{C}^2 \times T^2) / (\Z/4\Z): \quad (z_1, z_2, \lambda) \mapsto (\omega z_1, \omega z_2, \omega^2 \lambda) \,,
\label{eq:nHc4}
\end{align}
is also well-known in the F-theory literature, see e.g.,
\cite{Witten:1996qb, Morrison:2012np, Heckman:2013pva, Morrison:2016nrt, DelZotto:2017pti}. It is the superconformal limit of the non-Higgsable cluster over a curve of self-intersection $(-4)$, which has been shrunk to zero size in the geometry \eqref{eq:nHc4}. On the tensor branch, i.e., at finite size of the $(-4)$-curve, one has a gauge theory with gauge algebra $\mathfrak{so}(8)$ localized on a 7-brane stack wrapping this compact curve. This background preserves $\mathcal{N}= (1,0)$ supersymmetry in six dimensions. Observe that the local $(3,0)$-form
\begin{align}
\Omega_{(3,0)} = dz_1 \wedge dz_2 \wedge d \lambda \,,
\label{eq:local30}
\end{align}
is invariant under the $\Z/4\Z$ action, so the quotient is a singular Calabi-Yau.

A very similar approach works for the second generator, whose associated $\Mp(2,\Z)$ element is given by $\hat{U}^4$ with
\begin{align}
U^4 = \begin{pmatrix} 0 & 1 \\ -1 & -1 \end{pmatrix} \in \SL(2,\Z) \,.
\end{align}
As opposed to the background above a well-defined action on the lattice $\Lambda_{\tau}$ fixes the axio-dilaton to
\begin{align}
\tau = e^{2 \pi i / 3} \,,
\end{align}
for the central fiber and the geometry is given by
\begin{align}
(\mathbb{C}^2 \times T^2) / (\Z/3\Z): \quad (z_1, z_2, \lambda) \mapsto (\omega z_1, \omega z_2, \omega \lambda) \,, \text{ with } \omega = e^{2 \pi i / 3} \,.
\end{align}
Again, this is the singular limit of a non-Higgsable cluster, this time the collapsed curve is given by a curve of self-intersection $(-3)$. The induced gauge algebra hosted by 7-branes wrapping the curve is given by $\mathfrak{su}(3)$ associated to a type $IV$ singularity, see Table \ref{tab:singfiber}. The configuration preserves the same amount of supersymmetry as above. Correspondingly, the local $(3,0)$-form of equation \eqref{eq:local30} is invariant under the $\Z/3\Z$ action.

We see that while no new stringy objects need to be introduced in order to break the bordism group, the codimension-four defects are associated to very rich and interesting F-theory backgrounds related to six-dimensional superconformal field theories, see \cite{Heckman:2018jxk, Argyres:2022mnu} for recent reviews.

\subsection{Codimension-five defects}
\label{subsec:Mpcodim5}

The group $\Omega_4^{\text{Spin-Mp}(2,\Z)} (\pt)$ is free and generated by a single generator $E$, the Enriques surface. This can be understood as a refinement of $\Omega^{\Spin}_4 (\pt) = \Z$ generated by K3 to manifolds with Spin-$\Z/8\Z$ structure; the corresponding class is generated, in the F-theory literature, by the Enriques Calabi-Yau (see e.g. \cite{Grimm:2007xm}). Thus we need a type IIB background that allows the collapse of $E$ to a point. Such a background will necessarily be non-supersymmetric. We also expect it will be strongly coupled since the spectrum of IIB on K3 as well as $E$ is chiral. A detailed study of this defect is outside of the scope of this paper.

\subsection{Codimension-six defects}
\label{subsec:Mpcodim6}

Defects of codimension-six become necessary due to the non-vanishing bordism group
\begin{align}
\Omega^{\text{Spin-Mp}(2,\Z)}_5 (\pt) = (\Z/2\Z) \oplus (\Z/{32}\Z) \oplus (\Z/9\Z) \,,
\end{align}
generated by lens spaces. As in codimension four we can lift the configurations to an F-theory background by considering\footnote{The complex structure of the singular central fiber is fixed to be $\tau = e^{2 \pi i/3}$ for $k = 3$ and $\tau = i$ for $k = 4$.}
\begin{align}
(\mathbb{C}^3 \times T^2) / (\Z/k\Z) \,,
\end{align}
which describes so-called S-fold backgrounds \cite{Garcia-Etxebarria:2015wns, Aharony:2016kai}, see also \cite{Ferrara:2018iko, Heckman:2020svr, Giacomelli:2020gee}. The fact that there are two different S-folds with $k = 4$ can be traced back to the fact that there exist two different Spin-$\Z/8\Z$ structures on $L^5_4$. This leads to interesting consequences for the 4d SCFTs localized on the singular point of these backgrounds, which will be discussed in \cite{Spin-Off}. Here, we will focus on their general properties. For that, note that $\eta$-invariants in 5d are bordism invariants and thus can be used to analyze bordism relations between different $\Z/k\Z$ action of the lens spaces.

For $k = 3$ the lens space $L^5_3$ is Spin and can be interpreted as the asymptotic boundary of $\mathbb{C}^3 / (\Z/3\Z)$, where the discrete group action on the complex coordinates of $\mathbb{C}^3$ takes the form
\begin{align}
z_i \mapsto e^{2 \pi i j_i / k} z_i  \text{ with } j_i \in \{ 1, 2\} \,,
\end{align}
see also Appendix \ref{subapp:etalens}. This group action introduces an isolated singularity at $z_i = 0$.
Using the analysis in \cite{10.2307/2044659},
this singularity is terminal if the action on the base $z_i$ takes the form
\begin{equation}
(z_1, z_2, z_3) \mapsto (\zeta z_1, \zeta^a z_2, \zeta^{-a} z_3)
\end{equation}
with $\zeta$ a primitive 3rd root of unity and $a$ co-prime with respect to 3.\footnote{Note also that this interior geometry is not a Calabi-Yau threefold since the determinant is $\zeta$ rather than $1$.}

To get a IIB background, we equip this with a
duality bundle so that the total space in the F-theory model is a Calabi-Yau four-fold. This implicitly involves picking a choice of complex structure. In the local patch near the origin, we thus get a singularity of the for
$\mathbb{C}^4 / (\mathbb{Z} / k \mathbb{Z})$ where the action on the local coordinates is
\begin{equation}
(z_1,z_2,z_3,\lambda) \mapsto (\zeta z_1, \zeta^{a} z_2, \zeta^{-a} z_3, \zeta^{-1} \lambda),
\end{equation}
and it follows from references \cite{10.2307/2044659, AnnoTerminal} that this is an isolated terminal singularity of the local four-fold.\footnote{ Note that this bulk geometry is a Calabi-Yau fourfold since the determinant is $1$.}

In fact, this is just the celebrated $\Z/3\Z$ S-fold in \cite{Garcia-Etxebarria:2015wns}.
Calculating the $\eta$-invariant of a Dirac fermion with charge $q = 1$ in this background (see Appendix \ref{subapp:etalens}), we find
\begin{align}
\eta^{\text{D}}_{\frac{1}{2}} \big( L^5_3 (1,-1,1) \big) = \tfrac{1}{9} \,,
\end{align}
where we explicitly denoted the $\Z/3\Z$ action on the geometry by $L^5_3 (j_1, j_2, j_3)$, in the obvious notation.
Thus, the asymptotic boundary of this
S-fold is a generator for the $\Z/9\Z$  bordism group factor. One can also consider more general lens spaces. Taking
the group action to be uniform on all the $z_i$, we observe that
\begin{align}
\eta^D_{\frac{1}{2}}  \big( L^5_3 (1,1,1) \big) = - \tfrac{1}{9} \,,
\end{align}
and so is related to the generator via:
\begin{align}
8 \, \big[ L^5_3 (1, -1, 1) \big] \sim \big[ L^5_3 (1,1,1) \big] \,,
\end{align}
with $\sim$ indicating that the two sides are bordant.

In the corresponding F-theory description, both the asymptotic lens space $L^5_3 (1, -1, 1)$ and the lens space $L^5_3 (1, 1, 1)$ can preserve supsersymmetry, but this is accomplished in different ways. In the case of the S-fold with asymptotic boundary $L^5_3 (1, -1, 1)$, we require a non-trivial elliptic fibration to retain supersymmetry. On the other hand, $L^5_3 (1, 1, 1)$ is the asymptotic boundary of the local Calabi-Yau threefold $\mathcal{O}(-3)$ over $\mathbb{CP}^2$, and so is compatible with a trivial elliptic fibration.

To see this in more detail, introduce the local $(4,0)$-form:
\begin{align}
\Omega_{(4,0)} = dz_1 \wedge dz_2 \wedge dz_3 \wedge d\lambda \,,
\end{align}
where $\lambda$ is the complex local coordinate on the F-theory torus. For $L^5_3 (1,-1,1)$ we see that $\lambda$ needs to transform as
\begin{align}
\lambda \rightarrow \zeta^{-1} \lambda \,,
\end{align}
in order for $\Omega_{(4,0)}$ to be invariant. This precisely indicates the non-trivial duality bundle, which is translated to a non-trivial torus bundle in the F-theory geometry. Since, the supercharges also transform under the duality, one finds that this background preserves part of the supersymmetry, famously leading to D3-brane probe theories with $\mathcal{N} = 3$ supersymmetry. Summarizing, for the supersymmetric F-theory background the central singularity corresponds to a $\Z/3\Z$ S-fold.

A general comment here is that the above analysis has favored a specific complex structure on the bulk geometry. Observe that we could have complex conjugated, for example $z_2 $ to $\overline{z}_2$. Doing so, it might at first appear that the roles of the lens spaces have switched roles. It is important to note, however, that we have also implicitly fixed a choice of duality bundle. The $\eta$-invariants detect topological data and are insensitive to these holomorphic geometry issues.

A very similar discussion can be performed for $k = 4$ in which case there are two inequivalent Spin-$\Z/8\Z$ structures as we discuss in more detail in \cite{Spin-Off}. For our purposes, it will suffice to consider the lens spaces $L^5_4 (1,1,1)$. Since there are two classes, we also need two bordism invariants in order to establish relations between manifolds. We can choose these to be the $\eta$-invariant for a Dirac fermion as well as a Rarita-Schwinger field, respectively. Distinguishing between the two different Spin-$\Z/8\Z$ structures by a tilde, we find (see also \cite{Hsieh:2020jpj}),
\begin{align}
\begin{tabular}{c | c | c }
 & $\eta^{\text{D}}_{1/2}$ & $\eta^{\text{RS}}_{1/2}$ \\ \hline
 $L^5_4 (1,1,1)$ & $- \tfrac{5}{32}$ & $\tfrac{11}{32}$ \\
 $\widetilde{L}^5_4 (1,1,1)$ & $- \tfrac{3}{32}$ & $- \tfrac{3}{32}$ \\
\end{tabular}
\label{tab:5detas}
\end{align}
As pointed out in \cite{Hsieh:2020jpj}, the following two combinations generate the individual factors in the bordism group
\begin{align}
\Z/{32}\Z: \enspace \big[ L^5_4 (1,1,1) \big] \,, \quad \Z/2\Z: \enspace \big[ \widetilde{L}^5_4 (1,1,1) \big] + 9 \, \big[ L^5_4 (1,1,1) \big] \,.
\end{align}

With respect to this choice of complex structure, we observe that the corresponding bulk geometries do not lead to terminal singularities in $\mathbb{C}^3 / (\Z/4\Z)$, whereas variants of the form $L^5_4 (1,-1,1)$ and $\widetilde{L}^5_4 (1,-1,1)$ would.\footnote{Comparing \cite{Garcia-Etxebarria:2015wns} with \cite{Aharony:2016kai, Hsieh:2020jpj}, one is considering a different choice of complex structure to present the geometry, which in turn influences the presentation of the cyclic group action on the ambient $\mathbb{R}^6 \times T^{2}$. For example, the $\mathbb{Z} / 4 \mathbb{Z}$ S-fold defined by the local group action $(z_1,z_2,z_3,\lambda) \mapsto (\zeta z_1, \zeta^{-1}z_2, \zeta z_3,\zeta^{-1}\lambda)$ of \cite{Garcia-Etxebarria:2015wns} is instead presented as the local group action $(z_1,z_2,z_3,\lambda) \mapsto (\zeta z_1, \zeta z_2, \zeta z_3,\zeta \lambda)$ in \cite{Aharony:2016kai, Hsieh:2020jpj}. The physical content, however, is fixed by the choice of a prescribed duality bundle structure, and this is what is detected by the $\eta$-invariants.}
However, for $\Z/4\Z$ both of the variants can preserve some supersymmetry since
the covering space holomorphic $(4,0)$-form is invariant under both group actions. We give an extensive analysis of these type of backgrounds in \cite{Spin-Off}. In order to make sure that this leads to the known $\Z/4\Z$ S-fold backgrounds, we determine the set of bordism invariants in Table \eqref{tab:5detas} for $(j_1, j_2, j_3) = (1, -1, 1)$
\begin{align}
\begin{tabular}{c | c | c }
 & $\eta^{\text{D}}_{1/2}$ & $\eta^{\text{RS}}_{1/2}$ \\ \hline
 $L^5_4 (1,-1,1)$ & $\tfrac{5}{32}$ & $ - \tfrac{11}{32}$ \\
 $\widetilde{L}^5_4 (1,-1,1)$ & $\tfrac{3}{32}$ & $\tfrac{3}{32}$ \\
\end{tabular}
\label{tab:5detas2}
\end{align}
We see that only the sign of the $\eta$-invariants flip, so indeed the full bordism group is taken care of by including the associated S-fold backgrounds in type IIB.

\subsection{Codimension-seven defects}
\label{subsec:Mpcodim7}

Consider next codimension-seven defects. By inspection of Table \ref{the_big_table}, we see that
the bordism group with $k = 6$ vanishes
for $\SL(2,\Z)$, Spin-Mp$(2,\Z)$ and Spin-$\text{GL}^+(2,\Z)$ structures. Therefore, there
are no singular codimension-seven defects predicted by the Cobordism Conjecture in these cases. To give an explanation for this,
we can compare to the similar group with Spin structure alone, which also
vanishes, $\Omega_6^{\text{Spin}} (\pt) =0$. For instance, for any Calabi-Yau manifold, if the SYZ conjecture is true
\cite{Strominger:1996it}, one expects to have a $T^3$ fibration, and since $T^3$ is itself a boundary, a boundary
for the whole Calabi-Yau can be constructed fiberwise.\footnote{To realize this argument carefully, it would be necessary to take into account degeneracies of the fibration.} In the cases under consideration, this does not work directly, because the corresponding three-dimensional bordism groups do not vanish. But it may be that a more detailed analysis shows that the non-trivial bordism classes can never appear as fibers, recovering the result that the bordism groups vanish. At any rate, it would be interesting to develop a detailed understanding of the smooth geometries that act as null-bordisms for arbitrary 6-manifolds, but this is beyond the scope of the paper.

\subsection{Codimension-eight defects}
\label{subsec:Mpcodim8}

The defects in codimension-eight are associated with the generators of
\begin{align}
\Omega^{\text{Spin-Mp}(2,\Z)}_7 (\pt) = (\Z/4\Z) \oplus (\Z/9\Z) \,.
\end{align}
The asymptotic geometries are given by a seven-dimensional lens space $L^7_3$
and a manifold denoted $Q^7_4$ which can be described as a fibration of $L^5_4$ over $\mathbb{CP}^1$.

Consider $L^7_3$ first. As in the lower codimension cases, this
can be lifted to the boundary of a singular F-theory geometry given by
\begin{align}
(\mathbb{C}^4 \times T^2) / (\Z/3\Z) \,,
\end{align}
fixing the complex structure of the torus on the central fiber to $\tau = e^{2 \pi i / 3}$.
With the variant given by $L^7_3 (1,1,1,1)$, we retain some supersymmetry for the explicit group action:
\begin{equation}
(z_1,z_2,z_3,z_4,\lambda) \mapsto (\zeta z_1, \zeta z_2, \zeta z_3, \zeta z_4, \zeta^{-1} \lambda),
\end{equation}
in the obvious notation. In particular, the existence of a supersymmetric background requires a non-trivial elliptic duality bundle structure.
By contrast, for the asymptotic boundary $L^7_3 (1,-1,1,-1)$, the base of the F-theory model is already a Calabi-Yau fourfold so supersymmetry is compatible with a trivial duality bundle structure.\footnote{Observe that the base of the F-theory model is a terminal singularity in the base.}

We now ask whether these backgrounds provide generators for the corresponding bordism group factors. For that we determine $\eta^{\text{D}}_1 (L^7_3) - \eta^{\text{D}}_0 (L^7_3)$ in both cases. Note that we consider the difference of two $\eta$-invariants with different charges, since a single $\eta$-invariant is not a bordism invariant due to potential contributions of the index density on 8-manifolds.
\begin{align}
\begin{split}
& \eta^{\text{D}}_1 \big( L^7_3 (1,1,1,1) \big) - \eta^{\text{D}}_0 \big( L^7_3 (1,1,1,1) \big) = \tfrac{1}{9} \,,\\
& \eta^{\text{D}}_1 \big( L^7_3 (1,-1,1,-1) \big) - \eta^{\text{D}}_0 \big( L^7_3 (1,-1,1,-1) \big) = \tfrac{1}{9} \,.\\
\end{split}
\end{align}
Indeed, these two lens spaces are bordant. We therefore find new type IIB backgrounds that are described by lower-dimensional versions of the S-folds above, which we call S-strings. Focusing on the variant with $(j_1, j_2, j_3, j_4) = (1,1,1,1)$ we can explicitly determine the transformation properties of the supercharges, see also the analysis in \cite{Spin-Off}. This is compatible with an exotic $\mathcal{N} = (8,2)$ supersymmetry\footnote{As usual we indicate the supersymmetry charges in the smallest spinor representation, which in two dimensions is Majorana-Weyl.} in the two dimensional worldsheet theory of the S-string, but it is possible that in an actual string background with additional fluxes and brane sources switched on that this is reduced further. We will leave a full analysis, as well as potential physical applications for future work.

Consider next the second generator $Q^7_4$. This can be presented as the total space of a lens space bundle over $\mathbb{CP}^1$
\begin{align}
\begin{split}
L^5_4 \enspace \hookrightarrow \enspace & \, Q^7_4 \\
&\, \downarrow \\
& \, \mathbb{CP}^1
\end{split}
\end{align}
as was described in Section \ref{sec:summary}. We therefore see that it is not a genuinely new background but the
compactification of a $\Z/4\Z$ S-fold on a $\mathbb{CP}^1$.

We now  argue that there is a IIB background with asymptotic geometry $Q^7_4$ which
preserves some supersymmetry. In the description of the manifold in Section \ref{sec:summary} we have
focused on the embedding of the lens space in $H^{\pm 2} \oplus \underline{\mathbb{C}} \oplus
\underline{\mathbb{C}}$ over $\CP^1$. This means that the first complex coordinate $z_1$ of the ambient space of
the fiber gets twisted, i.e., transforms as a non-trivial line bundle over $\mathbb{CP}^1$. Since zero modes on a
sphere are only possible for fields transforming as scalars under rotations on the base $\mathbb{CP}^1$, preserved
supercharges need to respect that property. One can view this as specifying a topological twist \cite{Witten:1988ze},
which for the space above happens with respect to a $\Spin(2) \simeq \text{U}(1)$ subgroup of the full R-symmetry associated to rotations in
the internal space acting on $z_1$. Depending on the bundle $H^2$ or $H^{-2}$, as a section of which $z_1$
transforms, this leads to an (anti)-correlation of Spin-components in the respective two-dimensional subspace described by the non-trivial  line bundle of the fiber,
i.e., in the Spin component on $\CP^1$ and $\mathbb{C}_{z_1}$. We see that in general this preserves half of the components of supercharges of $L^5_4$. Since we have seen that these backgrounds preserve supersymmetry, so do the twisted compactification on $Q^7_4$.

It would be interesting to analyze the associated 2d supersymmetric quantum field theories (SQFTs) described by these backgrounds, which potentially realize exotic supersymmetries, see also \cite{Florakis:2017zep}. Thus, we discover new type IIB configurations which we called S-strings, as well as supersymmetry preserving compactifications of the original 4d S-folds.

\subsection{Fibration generators}
\label{subsec:fibgen}

In the previous section we observed that some of the asymptotic geometries are obtained via the process of ``compactification,'' in the sense that they descend from a bordism generator in a lower-dimensional asymptotic geometry.

Indeed, by inspection of the generators in Table \ref{Mp_gens_table} and \ref{gl_gens_table}, several of the higher-dimensional bordism group generators can be understood as fibrations, for which the fiber is related to generators in lower dimension. In general consider the fibration
\begin{align}
\begin{split}
\mathcal{F} \enspace \hookrightarrow \enspace & \, X \\
&\, \downarrow \\
& \, \mathcal{B}
\end{split}
\end{align}
where $X$ generates part of $\Omega^{\mathcal{G}}_k (\pt)$ and the fiber $\mathcal{F}$ is related to an element in $\Omega^{\mathcal{G}}_\ell (\pt)$, with $\ell = k - \text{dim}(\mathcal{B})$. We already know how to describe the fiber as the boundary of a space containing certain defects and we can generalize this to the full fibration by wrapping the corresponding defect on $\mathcal{B}$, see Figure \ref{fig:bord5}.
\begin{figure}
\centering
\includegraphics[width = 0.7 \textwidth]{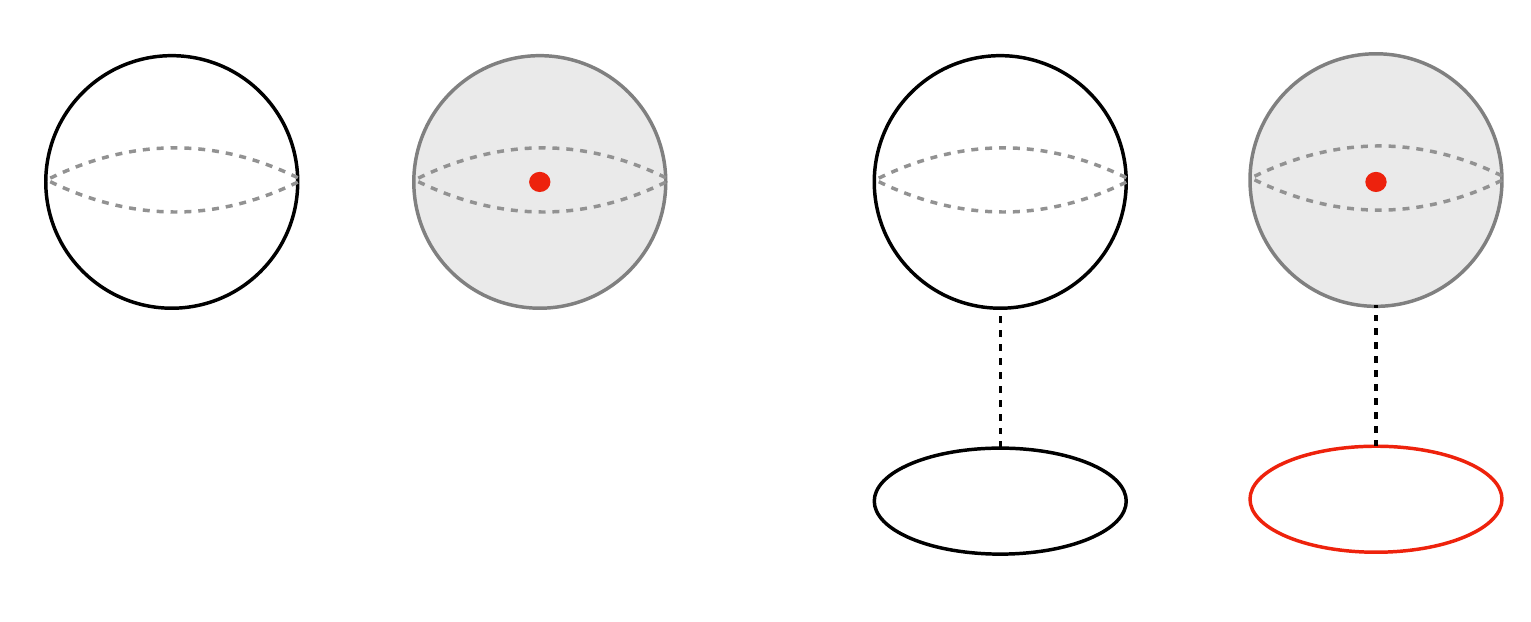}
\caption{Knowing how to bound a certain bordism class allows us to bound a higher-dimensional variant in which the class appears as the fiber of a fibration.}
\label{fig:bord5}
\end{figure}

In the context of an F-theory compactification, the appearance of a non-trivial fibration is actually required to retain supersymmetry in the lower-dimensional vacua. Indeed, given an elliptically fibered $n$-fold $X_{n} \rightarrow \mathcal{B}_{n-1}$ for some base $\mathcal{B}_{n-1}$,
we have demand a fixed duality bundle structure on the asymptotic geometry $\partial \mathcal{B}_{n-1}$. If we compactify further by fibering the original background over a space $\mathcal{C}_{m}$, the condition that the same bundle structure is retained upon further compactification means, in tandem with the Ricci flatness condition, that the resulting Calabi-Yau $(n+m)$-fold $X_{n+m} \rightarrow \mathcal{B}_{n+m-1}$ has a base $\mathcal{B}_{n+m-1}$ which is itself fibered over $\mathcal{C}_{m}$, namely we have $\mathcal{B}_{n-1} \rightarrow \mathcal{B}_{n+m-1} \rightarrow \mathcal{C}_{m}$.

In many cases, this can be interpreted in gauge theory terms as the implementation of a corresponding (partial) topological twist of the sort
considered in reference \cite{Witten:1988ze} (see also \cite{Vafa:1994tf}). To see why, recall that in many F-theory models, a singularity in the bulk can be interpreted as geometrically engineering a gauge theory, as obtained from branes wrapped on cycles of the base of the F-theory model. Now, the condition that we retain supersymmetry in such models amounts to the requirement that the brane worldvolume theory is topologically twisted \cite{Bershadsky:1995qy}. This partial twisting has been carried out for F-theory on Calabi-Yau threefolds (implicitly) in \cite{Bershadsky:1996nh}, for Calabi-Yau fourfolds in \cite{Beasley:2008dc, Donagi:2008ca}, and Calabi-Yau fivefolds in \cite{Apruzzi:2016iac, Schafer-Nameki:2016cfr}. See also \cite{Acharya:2000gb, Pantev:2009de, Heckman:2018mxl} for related analyses in the context of geometric engineering in M-theory backgrounds. From this perspective, it is natural to interpret some of our asymptotic geometries as obtained from a further compactification of a higher-dimensional vacuum.

Of course, we will also be interested in backgrounds where supersymmetry may not be preserved, but even here there is a correlation between the bundle structure in the fiber and base, so we will continue to refer to these as ``twisted compactifications''. In physical terms, these
cases are obtained by starting from a higher-dimensional defect and wrapping it over a cycle in the new base. In
general, the duality bundle will also involve the base of the fibration. This happens for example for the
Spin-GL$^+ (2,\Z)$ generators $\halfQseven$ and $\ninedimgen$, as well as for the Arcana. Moreover, in general the duality bundle will also involve reflections. To construct the defect for the fiber bundle, we would wrap the defect in the fiber along the base, where there is too a duality bundle. In order for these configurations to be well-defined, the worldvolume degrees of freedom of the defects wrapped around $\mathcal{B}$ need to be compatible with these generalized duality bundles, that is, reflections should be a global symmetry of the worldvolume theory of the fiber. Since the defects involve the reflection 7-brane (R7-brane) whose worldvolume theory is hard to identify (see however \cite{Dierigl:2022reg}) we assume that the duality bundle is compatible with the defects and postpone a detailed analysis for future work. An encouraging hint is that the R7-brane is a defect for a circle compactification with a monodromy in the reflection symmetry, a stringy generalization of a vortex, and in general, the symmetry that is protecting a vortex is unbroken at its core (as an example, consider vortex strings in the abelian Higgs model).

\subsection{Codimension-nine defects}
\label{subsec:Mpcodim9}

As in the case of codimension-five defects, these classes correspond to Spin-$\Z/8\Z$ generalizations of the
non-trivial Spin 8-manifolds specified by the value of $p_1^2$ and $p_2$. In the case of the quaternionic
projective space $\HP^2$, we can view this as the boundary geometry of a linear dilaton background for a
gauged WZW model. Indeed, since we have the coset construction $\HP^2  = \Sp(3 )/ \big(\Sp(2) \times \Sp(1)\big)$,
it is enough to begin with the WZW model for the Lie group $\Sp(n+1)$ and gauge the appropriate subgroup. Solving
the string theory equations of motion then requires a dilaton gradient. See e.g., \cite{Gepner:1986wi} for further
discussion. The corresponding geometry will have a singularity as the dilaton runs to strong coupling, but the point of view that we take here is that existence of the boundary (as demanded by the Cobordism Conjecture) in fact requires that this singularity be allowed. The Cobordism Conjecture has been used in the past to argue for the consistency of strong coupling singularities whose fate and consistency are hard to determine otherwise, see \cite{Buratti:2021yia}.

 The other bordism generator we have encountered is the Bott manifold $B$, which we can take to be a
manifold of $\Spin(7)$ metric holonomy. Unlike in the previous case, there is no simple way to describe the bounding geometry of this manifold in the worldsheet, although presumably it would also involve a singularity where the dilaton runs to strong coupling. It would be very interesting to elucidate the physics of the corresponding non-supersymmetric defect, but this lies beyond the scope of this paper.

\subsection{Codimension-ten defects}
\label{subsec:Mpcodim10}

Codimension-ten defects are objects localized in spacetime and can be associated to exotic instanton backgrounds. However, for
\begin{align}
\Omega^{\text{Spin-Mp}(2,\Z)}_9 (\pt) = (\Z/{4}\Z) \oplus (\Z/8\Z) \oplus (\Z/128\Z) \oplus (\Z/3\Z) \oplus (\Z/27\Z) \,,
\end{align}
only the first the three summands $(\Z/4\Z) \oplus (\Z/128\Z) \oplus (\Z/27\Z)$ are associated to localized
defects, which we call S-instantons, whereas $(\Z/8\Z) \oplus (\Z/3\Z)$ can be understood as the 7-brane stacks of
Section \ref{subsec:Mpcodim2} wrapped around $\HP^2$. We therefore focus on the three factors generated by $L^9_k$,
which can be described as the projection to the base of F-theory backgrounds of the form
\begin{align}
(\mathbb{C}^5 \times T^2) / (\Z/k\Z) \,.
\end{align}
% For $L^9_3$ there is one particular way to choose the $j_i$ of the lens space action such that the central singularity in the base is Gorenstein, given by $(1,-1,-1,-1,-1)$ (up to permutations and sign flips thereof).
For $k = 3$ the central singularity is canonical and terminal for all consistent values of the $j_i$. For $L^9_3 (1,1,1,1,1)$ there is the possibility of describing the geometry above in terms of a non-compact singular Calabi-Yau 6-fold, for which the elliptic fiber coordinate needs to transform under $\Z/3\Z$, $\lambda \mapsto e^{2 \pi i/3} \lambda$. It is tempting to speculate about preserved supersymmetries in the effective zero-dimensional system, `living' on the S-instanton defect.

For $k = 4$ the two generators differ once more by the choice of Spin-$\Z/8\Z$ structure and can be detected by the $\eta$-invariants
\begin{align}
\begin{tabular}{c | c | c }
 & $\eta^{\text{D}}_{1/2}$ & $\eta^{\text{RS}}_{1/2}$ \\ \hline
 $L^9_4 (1,1,1,1,1)$ & $\tfrac{9}{128}$ & $- \tfrac{19}{128}$ \\
 $\widetilde{L}^9_4 (1,1,1,1,1)$ & $\tfrac{7}{128}$ & $\tfrac{3}{128}$ \\
\end{tabular}
\label{tab:5detas3}
\end{align}
which suggests the two generators
\begin{align}
\Z/{128}\Z: \enspace \big[ L^9_4(1,1,1,1,1) \big] \,, \quad \Z/4\Z: \enspace 7 \, \big[ L^9_4 (1,1,1,1,1) \big] - 9 \big[ \widetilde{L}^9_4 (1,1,1,1,1) \big] \,.
\end{align}
The above geometries can be realized as boundaries of a local and singular elliptically fibered Calabi-Yau six-fold. All of the central singularities are canonical and terminal, both in the base and total space of the associated F-theory model. Again, the analysis of the bordism groups points towards interesting new backgrounds of type IIB string theory described by S-instantons.

This concludes our discussion of Spin-$\Mp(2,\Z)$ bordism defects.

%%%%%%%%%%%%%%%%%%%%%%%%%%%%%%%%%%%%%%%%%%%%%%%%%%%%%%%%

\section{Spin-$\GL^{\boldsymbol{+}}$ defects}
\label{sec:GLdefects}

We now include the reflections and their Pin$^+$ lift in the duality group and analyze the
defects associated with the generators of the bordism groups $\Omega^{\text{Spin-GL}^+ (2,\Z)}_k (\pt)$. Comparing these bordism groups with the ones for Spin-$\Mp(2,\Z)$ manifolds in Table \ref{the_big_table}, we see that the groups in dimension $k = (4n+1)$ are reduced while the groups in dimension $k = (4n+3)$ are typically increased. The reason is based on two competing effects. First, since every Spin-$\Mp(2, \Z)$ manifold is also a Spin-$\GL^+ (2, \Z)$ manifold, one naively expects to find more generators. However, the Spin-$\GL^+ (2, \Z)$ manifolds also allow a larger class of deformations and thus tend to identify formerly distinct generators.

The fact that this identification predominantly happens in dimensions $k = (4n + 1)$, can be understood as follows. In many cases the Spin-$\Mp(2,\Z)$ generators are associated to characteristic classes of the duality bundle. These characteristic classes switch sign under the action of the reflection operator exactly in dimension $k = (4n +1)$ and thus the associated bundles can be deformed into each other on the level of Spin-$\GL^+ (2,\Z)$ manifolds. In dimension $k = (4n +3)$, however, the characteristic classes are even under reflections, and can not be deformed into each other. We will demonstrate this explicitly for $n = 0$, i.e., $k = 1$, below.

As for Spin-$\Mp(2,\Z)$ bordisms we encounter several interesting new classes and backgrounds for type IIB string theory as well as known objects. We summarize these for the odd-dimensional bordism groups in Table \ref{tab:GLdefects}.

\subsection{Codimension-two defects}
\label{subsec:GLcodim2}

The codimension-two defects predicted by
\begin{align}
\Omega^{\text{Spin-GL}^+(2,\Z)}_1 (\pt) = (\Z/2\Z) \oplus (\Z/2\Z) \,,
\label{eq:GL1bord}
\end{align}
seem to be more sparse than in the Spin-$\Mp(2,\Z)$ case discussed in Section \ref{subsec:Mpcodim2}. The reason, as pointed out above, is that the additional reflections allow for a deformation of several distinct Spin-$\Mp(2,\Z)$ structures into each other, thus, identifying the associated manifolds via Spin-$\GL^+(2,\Z)$ bordisms. We demonstrate this for the subgroup $\Z/3\Z$ of $\Omega^{\text{Spin-Mp}(2,\Z)}_1 (\pt)$.

The $\Z/3\Z$ factor is generated for example by a circle with a transition function associated to $g = \hat{U}^4$,
see Appendix \ref{app:groups}.\footnote{Technically this is twice the generator $L^1_3$ defined above, that also
generates the full $\Z/3\Z$ summand, but is more convenient here because of the amalgam structure.} The disjoint
union of three copies of this generator is null-bordant, see the left-hand side of Figure \ref{Fig:Killingbord}. However, one cannot glue two of the three boundaries because the transition functions do not match. Once we include reflections, however, one can change the transition function according to the group law
\begin{align}
\hat{R} \hat{U}^4 \hat{R}^{-1} = \hat{U}^{-4} = \hat{U}^8 \,.
\end{align}
Including such a reflection transition function in the two-dimensional bulk corresponds to switching on a non-trivial background for the reflection part of the symmetry. This allows us to glue two of the boundary components. This means that $2 \big[L^1_3\big]$ bounds as a Spin-$\GL^+ (2,\Z)$ manifold. Since on $\Z/3 \mathbb{Z}$, both non-trivial elements of the group are twice each other, all bordism classes in this factor are trivialized.
\begin{figure}
\centering
\includegraphics[width = \textwidth]{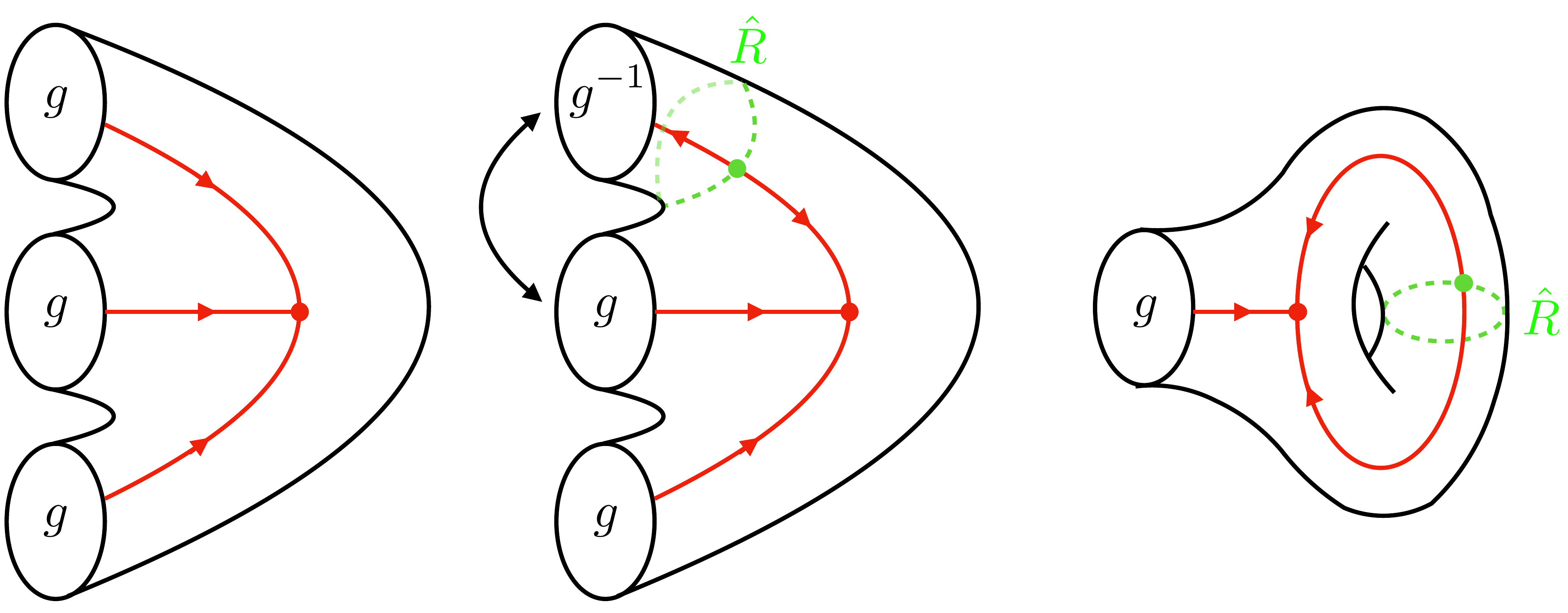}
\caption{Depiction of the null-bordism for $2 \big[ L^1_3 \big]$ in $\Omega^{\text{Spin-GL}^+ (2,\Z)}_1 (\pt)$. The red lines denote $\Z/3\Z$ transition functions with the arrow distinguishing between $g$ and $g^{-1}$, where $g$ generates $\Z/3\Z$. The green dotted line depicts a reflection transition function $\hat{R}$, which effectively flips the arrow.}
\label{Fig:Killingbord}
\end{figure}
This is summarized pictorially in Figure \ref{Fig:Killingbord}.

The same happens for the $\Z/8\Z$ factor generated by $L^1_4$. The reflection acts on the group element encoding the transition function as
\begin{align}
\hat{R} \hat{S} \hat{R} = \hat{S}^{-1} = \hat{S}^7 \,.
\end{align}
This demonstrates that each element in $\Z/8\Z \subset \Omega_1^{\text{Spin-Mp}(2,\Z)} (\pt)$ gets mapped to its inverse under reflection, and by the same argument as above, every class corresponding to an element of $\Z/8\Z$ which is a multiple of two becomes trivial in cobordism. This kills the elements of order 2, and 4, and  the remaining non-trivial class is a circle with a holonomy given by the $\hat{S}$ generator, whose boundary is given by a stack of 7-branes inducing an $\mathfrak{e}_7$ singularity in the fiber, i.e., fiber type $III^*$. This is the generator of one of the $\Z/2\Z$ factors, see also the arguments in \cite{McNamara:2021cuo}.

The other generator in \eqref{eq:GL1bord} is given by a circle with a transition function given by $\hat{R}$,
denoted as $S^1_R$. This requires the addition of a new type of 7-brane in type IIB string theory that induces a
monodromy of determinant $(-1)$, i.e., including reflections, which we refer to as reflection 7-brane or R7-brane. See reference \cite{Dierigl:2022reg} for further analysis on the physics of R7-branes. Two variants of this reflection brane will be particularly important, see also Appendix \ref{app:groups} and Figure \ref{fig:D8emb}. One, the $\Omega$-brane, induces a monodromy given by the worldsheet orientation reversal. The other, the $\mathsf{F}_L$-brane, instead generates fermion parity in the left-moving sector of the string worldsheet. They are S-dual to each other and can therefore be deformed into each other by brane moves involving the usual $[p,q]$-7-branes, similar to the transitions described in \cite{Gaberdiel:1997ud, Gaberdiel:1998mv, DeWolfe:1998zf}.

We want to emphasize that the mechanism that we have discussed here for constructing a geometric boundary of some bordism classes in Spin-$\Mp(2,\Z)$ means that we have geometrized the corresponding 7-branes. Ordinarily, $[p,q]$-7-branes are not described by a smooth geometry, and fields become singular at their core. The geometric defects described above allow us to construct completely smooth objects with the same duality monodromy as e.g., a type $IV^*$ singularity, but that are completely smooth -- just a torus glued to flat space with a particular monodromy. A similar situation arises e.g., in heterotic string theory, where NS5-brane charge can be carried by a singular brane or a K3 geometry (or a gauge bundle). Just as in that case, we should entertain the possibility that the smooth and singular supersymmetric configurations are two different ``phases'' or configurations of a single underlying object, specified by its monodromy charges. The topologically non-trivial, non-supersymmetric smooth 7-brane background we constructed should be viewed as a  highly excited state of the supersymmetric 7-brane, which, being BPS, is the lightest object with the specified monodromy.

Finally, we have seen that an inclusion of reflections allows deformations of many of the different Spin-$\Mp(2,\Z)$ classes into each other. This was possible because they reverse the transition functions, or alternatively the characteristic classes of the duality background. Hence we expect such reductions of the bordism group to happen in dimensions where the typical duality background transforms with a minus sign under reflection. This is the case in dimension $k = (4n +1)$ (see the characterization of the duality bundle in \cite{Debray:2021vob}) and precisely corresponds to the reductions of the bordism groups in Table \ref{the_big_table}. In dimensions $k = (4n +3)$ the duality backgrounds typically are invariant under reflections and one does not encounter this reduction. Instead, one often obtains new backgrounds involving transition functions that contain reflections. In even dimensions the duality structure is usually not what supports the non-trivial bordism classes and we also do not expect a reduction, verifying the observations in Table \ref{the_big_table}.

\subsection{Deforming F-theory on K3 to nothing}
\label{subsec:FonK3to}
There is an interesting application of the bordism groups and branes described so far, which we felt merits its
own subsection. It is related to one of the outstanding questions left open in \cite{McNamara:2019rup}, namely
constructing a boundary for F-theory on K3.  A K3 surface has a limit in which it can be described as an elliptic
fibration over $\CP^1$ with four $I^*_0$ fibers, which is related to the description as a $T^4/(\Z/2\Z)$ orbifold.
This is known as Sen's limit in the F-theory literature \cite{McNamara:2019rup}. F-theory on K3 in Sen's limit is a
perturbative IIB configuration, which can be described as a $T^2/(\Z/2\Z)$ compactification of 10-dimensional IIB where there are four singularities (of Kodaira type $I^*_0$) at the fixed points of the action, as depicted in the left panel of Figure \ref{fig:K3bordism}. Each of these singularities can be described as a brane stack composed of four D7-branes on top of an O7$^-$ orientifold plane.

\begin{figure}
\centering
\includegraphics[width = 0.6 \textwidth]{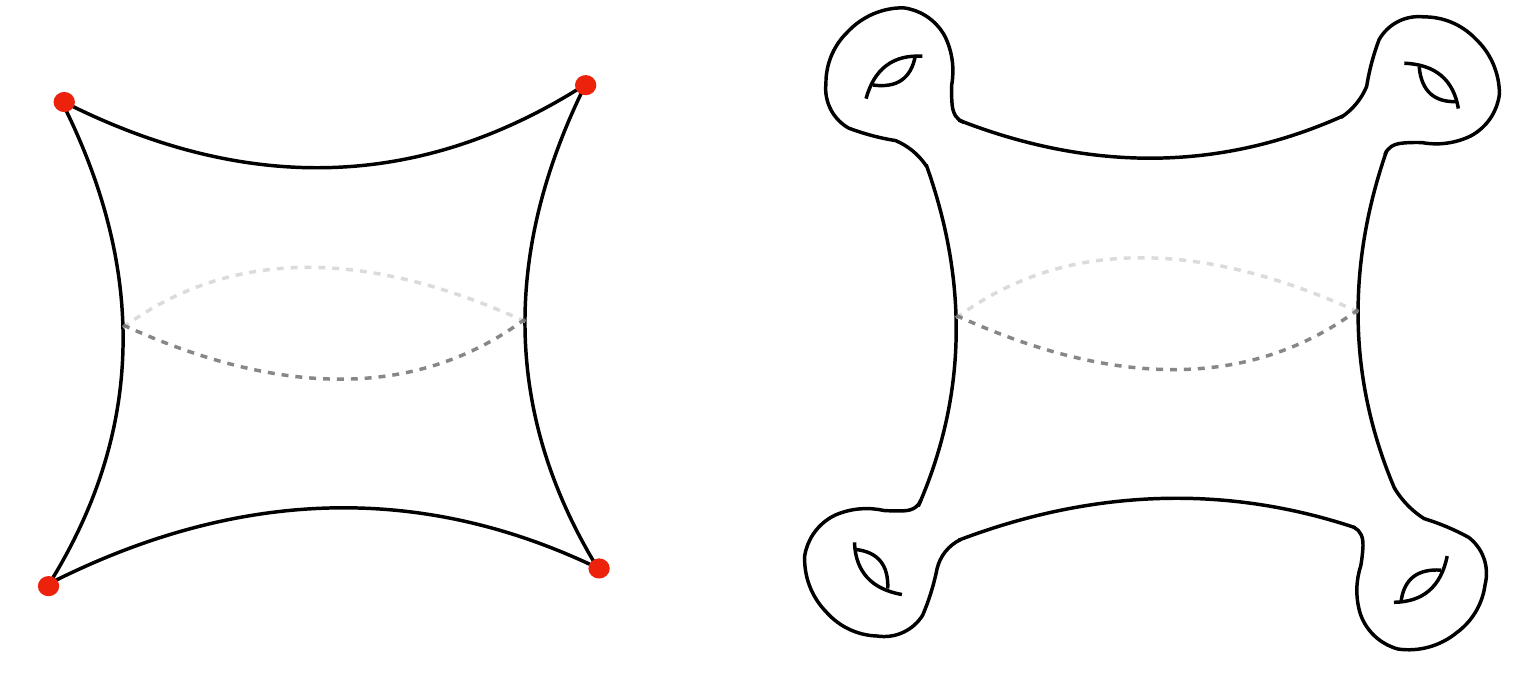}
\caption{Null-bordism for F-theory on K3 via a transition from a brane stack to a smooth Spin-$\GL^+(2,\Z)$ geometry. Here, small discs around the orientifold planes are replaced by 2-tori with a non-trivial duality bundle containing reflections.}
\label{fig:K3bordism}
\end{figure}

An $I^*_0$ singularity is specified by a linking circle with monodromy $\hat{S}^2$, which is one of the classes which we argued is trivial in Spin-$\GL^+(2,\Z)$ bordism in the previous subsection.  As a result, the brane stack that describes its boundary can be replaced, at the level of monodromies, by a completely smooth configuration without any branes, as depicted in Figure \ref{fig:S2bordism}.
\begin{figure}
\centering
\includegraphics[width = 0.6 \textwidth]{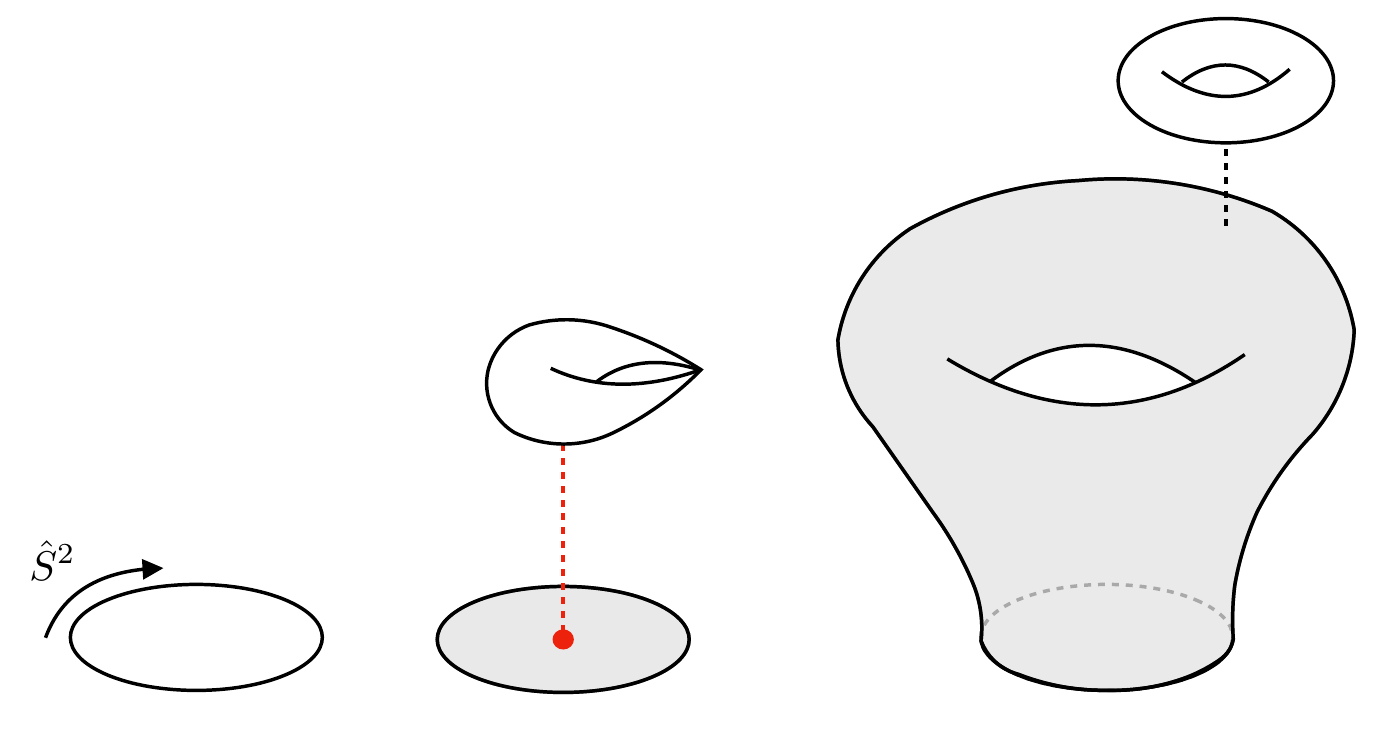}
\caption{Bordism for a circle with transition function $\hat{S}^2$, by singular configuration involving a brane stack or non-singular configuration containing reflections.}
\label{fig:S2bordism}
\end{figure}
As we emphasized near the end of the previous Subsection, we cannot be sure that the $I^*_0$ singularity is the same as the smooth configuration depicted to the right in Figure \ref{fig:S2bordism}. But they have the same asymptotic charge, so we will now assume that the two configurations can be deformed into each other and explore the consequences. What this means is that F-theory on K3 is cobordant to the configuration depicted on the right panel of Figure \ref{fig:K3bordism}, where each of the singularities has been replaced by a glued torus with the appropriate duality monodromy. We have now replaced F-theory on K3 by a completely smooth IIB configuration, which is however non-supersymmetric. The interesting point is that since $\Omega_2^{\text{Spin-}{\GL^+(2,\Z)}} (\pt) =0$, this configuration is trivial in bordism of smooth manifolds with duality bundle.

In other words, we have reduced the global question of finding the boundary of F-theory on K3 to the \emph{local} question of establishing the equivalence between the $I^*_0$ singularity and the smooth bordism we constructed.

\subsection{Codimension-three defects}
\label{subsec:GLcodim3}

There is a single generator of
\begin{align}
\Omega^{\text{Spin-GL}^+(2,\Z)}_2 (\pt) = \Z/2\Z \,,
\end{align}
associated to a codimension-three defect. However, since it is given by a product of the circle with periodic boundary conditions on the fermions without a duality bundle, i.e., the generator of $\Omega^{\Spin}_1 (\pt)$, and the generator of $\Omega_1^{\text{Spin-GL}^+ (2,\Z)} (\pt)$ involving the reflection brane, we do not need to introduce additional defects. Indeed, the inclusion of an R7-brane allows us to fill one of the circles. Thus the associated string theory background corresponds to an R7-brane wrapped on an $S^1$.

\subsection{Codimension-four defects}
\label{subsec:GLcodim4}

For codimension-four objects we have
\begin{align}
\Omega^{\text{Spin-GL}^+(2,\Z)}_3 (\pt) = (\Z/2\Z) \oplus (\Z/3\Z) \oplus (\Z/2\Z) \oplus (\Z/2\Z) \,.
\end{align}
The first two factors are identical to Spin-$\Mp(2,\Z)$ with the exact same generators and we need to include the same non-Higgsable clusters as in Section \ref{subsec:Mpcodim4}. The remaining two $\Z/2\Z$ factors both involve the action of reflections and have an asymptotic $\mathbb{RP}^3$ geometry. It is not difficult to identify these objects as O5-planes, or their S-duals \cite{Hanany:2000fq} if we are interested in the background with $\mathsf{F}_L$-branes. Once more we stress that since our analysis is not sensitive to RR- and NS-fluxes in the boundary we cannot further specify the discrete torsion allowed for these objects, nor can we match their $[p,q]$-5-brane charges. While the background with reflection branes likely breaks supersymmetry completely, the orientifold backgrounds can preserve part of the supersymmetry. It would be interesting to explore whether also the second generator can be bounded by a supersymmetric background.

We can also provide an alternative construction for a related defect, as a background of the form
\begin{align}
(\mathbb{C}^2 \times T^2)/ (\Z/2\Z) \,,
\end{align}
where the $\Z/2\Z$ action of the fiber now also involves a reflection $R$. The two different variants are associated to the two possible embeddings of $D_4$ or $D_8$ into $D_{16}$, see also \cite{Debray:2021vob} and Appendix \ref{app:embdihedral}, illustrated in Figure \ref{fig:D8emb}.
\begin{figure}
\centering
\includegraphics[width = 0.65 \textwidth]{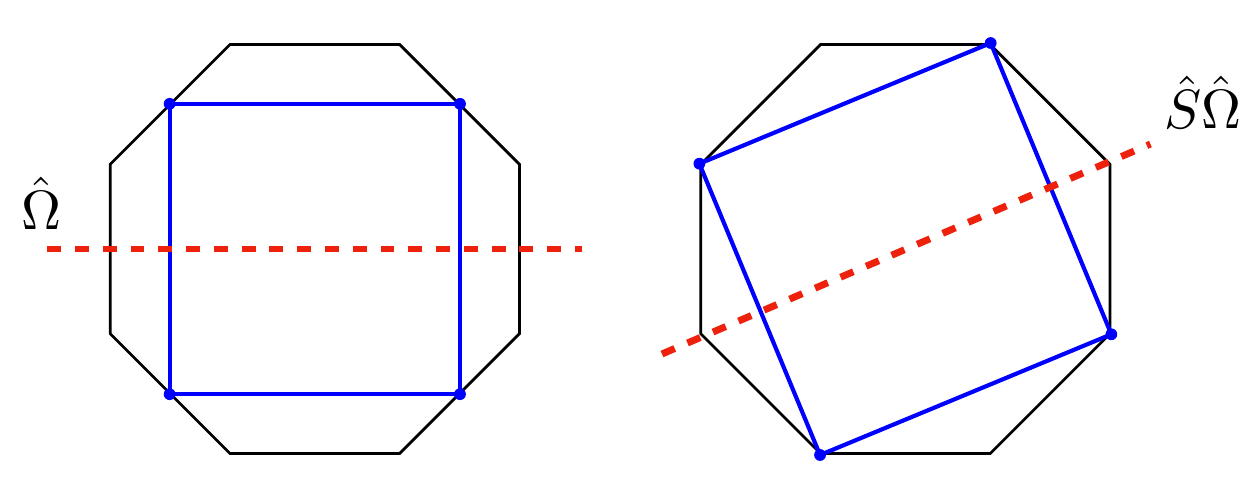}
\caption{Two embeddings of $D_4$ and $D_8$ into $D_{16}$ with the action of $\hat{\Omega}$ and $\hat{S} \hat{\Omega}$ indicated.}
\label{fig:D8emb}
\end{figure}
We see that for the two variants the reflections can be described by $\hat{\Omega}$ and $\hat{S} \hat{\Omega}$, respectively. Note that instead of $\hat{\Omega}$ one could also use its S-dual given by $(-1)^{\hat{F}_L}$. Therefore, the base geometry is given by the asymptotic boundary of
\begin{align}
\mathbb{C}^2 / (\Z/ 2 \mathbb{Z})_{\hat{\Omega}} \quad \text{or} \quad \mathbb{C}^2 / (\Z/ 2 \mathbb{Z})_{\hat{S} \hat{\Omega}} \,,
\label{eq:O5backg}
\end{align}
where the subscript on the $\Z/2\Z$ indicates which symmetries are acting on the background. Picturing the real projective spaces as asymptotic boundaries of the complex line bundle $\mathcal{O}(-2)$ over $\mathbb{CP}^1$, we see that this configuration can be obtained by wrapping an $\Omega$-brane and a combination of an $\Omega$-brane with a stack of $[p,q]$-7-branes generating $\hat{S}$ on the base $\mathbb{CP}^1$. This is parallel to the construction of the non-Higgsable clusters above, with the difference that the curve of self-intersection $(-2)$ is now additionally wrapped by an R7-brane. Thus, the defect killing this bordism class can be realized by either ordinary O-planes or by a wrapped R7-brane. This suggests that the compactification of an R7-brane on $\mathbb{RP}^2$ is a non-supersymmetric object with the same charges as an O5-plane.

\subsection{Codimension-five defects}
\label{subsec:GLcodim5}

As for Spin-$\Mp(2,\Z)$ the bordism group in dimension four is also a free $\Z$ factor generated by the Enriques surface $E$. As above we do not know how to properly bound this manifold in type IIB string theory. However, since it is the same generator as in Section \ref{subsec:Mpcodim5}, we do not expect that new ingredients or backgrounds involving R7-branes will be necessary.

\subsection{Codimension-six defects}
\label{subsec:GLcodim6}

Much as in the case of codimension-two defects considered in Section \ref{subsec:GLcodim2}, in codimension-six (an increase in dimension of four), we inherit bordism generators from the related Spin-$\Mp(2,\Z)$ backgrounds
\begin{align}
\Omega^{\text{Spin-GL}^+(2,\Z)}_5 (\pt) = (\Z/2\Z) \oplus (\Z/2\Z) \,.
\end{align}
The only surviving configuration is $L^5_4$ generating $\Z/2\Z$ which can be bounded by including the $\Z/4\Z$ S-fold background discussed in Section \ref{subsec:Mpcodim6}, with either Spin structure.

The other generator $X_5$, Arcanum V, can be described as an $\RP^3$-fibration over $\RP^2$
\begin{align}
\begin{split}
\RP^3 \enspace \hookrightarrow \enspace & \, X_5 \\
\enspace & \, \downarrow \\
 \enspace & \RP^2
\end{split}
\end{align}
The fiber $\RP^3$ can then be identified with the asymptotic boundary of $\mathbb{C}^2 / (\Z/2\Z)$, and can
be bounded by the configuration in Section \ref{subsec:GLcodim4} -- an O5-plane. Similar to the twisted
compactification involving $Q^7_4$ in Section \ref{subsec:Mpcodim8}, we can describe the boundary of the Arcanum~V
class by wrapping an O5-plane on the base $\RP^2$. Although the O5-plane does not have localized degrees of
freedom, one could in principle consider a O5-D5 brane stack, which would also be a boundary. For such
configurations, there is a potential question as to whether it is indeed consistent to compactify the theory on
$\RP^2$, which is not a Spin manifold, similar to other classes we discussed before. We will now see that
this is the case more explicitly. As described in Section \ref{arcana}, the Arcanum V manifold is a quotient of
$S^2\times S^3$ by two $\Z/2\Z$ actions, one of which is just the antipodal mapping on $S^3$ (with a certain
duality bundle). So the boundary of this is precisely an $\R^4/(\Z/2\Z)$ singularity with an action of
$\Omega$ (an O5-plane compactified on $S^2/(\Z/2\Z)=\RP^2$, where the $(\Z/2\Z)$ action is a
composition of the antipodal mapping on $S^2$ with an action of $(-1)^{F_L}$ and the reflection of three coordinates of
the normal bundle of the O5. These additional three reflections ensure that the action commutes with the one
defining the O5-plane, and that it is a symmetry of the worldvolume theory. In terms of worldvolume fields, one is
quotienting by a $\Z/2\Z$ involving the antipodal mapping together with a combination of $(-1)^{F_L}$ and a $\Z/2\Z$
on the $\SSO(4)$ normal bundle directions, which is a symmetry of the IIB background. Furthermore, because it involves only perturbative symmetries, this defect should be amenable to a IIB worldsheet description too, as a particular orbifold of an orientifold.

To sum up, once more, the defects necessary to generate the bordism groups can all be obtained by (non-trivially fibered)
compactifications involving the reflections of real coordinates in the ambient space of the fiber.

\subsection{Codimension-seven defects}
\label{subsec:GLcodim7}

Just as in Section \ref{subsec:Mpcodim7}, the bordism group with $k = 6$ vanishes, and all classes can be killed geometrically.

\subsection{Codimension-eight defects}
\label{subsec:GLcodim8}

Defects in codimension eight are associated to elements in
\begin{align}
\Omega^{\text{Spin-GL}^+(2,\Z)}_7 (\pt) = (\Z/2\Z)^{\oplus 3} \oplus (\Z/4\Z) \oplus (\Z/9\Z) \,.
\end{align}
As predicted in dimension $(4n + 3)$ we discover the full set of Spin-$\Mp(2,\Z)$ classes bounded by S-strings and
twisted compactifications of the S-folds discussed in Section \ref{subsec:Mpcodim8}. We further find two real projective spaces, for which the bordism defects are simply the O1-plane and its S-dual. These singularities can be described as $\mathbb{C}^4 / (\mathbb{Z} / 2 \mathbb{Z})$,
where the $\mathbb{Z} / 2 \mathbb{Z}$ group action is generated by an appropriate combination of duality generators:
\begin{align}
\mathbb{C}^4 / ( \Z/ 2 \mathbb{Z})_{\hat{\Omega}}  \quad \text{and} \quad \mathbb{C}^4 / ( \Z/ 2 \mathbb{Z})_{\hat{S} \hat{\Omega}} \,,
\end{align}
associated to the two $D_8$ embeddings depicted in Figure \ref{fig:D8emb}. The subscript on the $\mathbb{Z} / 2 \mathbb{Z}$ indicates the generator for this group. Much as in the case of the defects in Section \ref{subsec:GLcodim4}, objects with the same bordism charge can be constructed by wrapping reflection branes or a combination of reflection and conventional $[p,q]$-7-branes on the base $\mathbb{CP}^3$.  Once more it would be interesting to add fluxes and study the associated string charges as well as supersymmetry properties.

The defects for the two remaining generators are slightly more subtle. The defect for $W^7_1$ can be obtained from the Spin-$\Mp(2,\Z)$ generator by an additional $\Z/2\Z$ action involving reflections. As for $Q^7_4$, we can interpret its defect as a compactification (with non-trivial fibration) of the S-fold in codimension-six, with the subtle difference that the fibration does not only involve discrete rotations encoded in a Spin-$\Z/8\Z$ structure but also involves reflections captured by the more general Spin-$D_{16}$ structure of $W^7_1$. We therefore refer to these theory as double-twisted compactifications in Table \ref{tab:GLdefects} and once more point out the assumption discussed in Section \ref{subsec:fibgen} concerning the consistency of the duality action on the worldvolume fields.

The last generator, $\orangeseven$, has a slightly different interpretation. It originates from a manifold with
Spin-GL$^- (2,\Z)$ structure via a Smith homomorphism. It can be understood as a quotient of $S^3 / Q_{16} \times
T^2 \times S^2$, with $Q_{16}$ the binary dihedral group of order $16$. The discrete quotient introduces a
non-trivial fibration of $S^3 / Q_{16}$ over $\RP^2$, whose total space is in turn fibered over a Klein bottle, see
Section \ref{7d_orange} for further details. As discussed in \ref{subsec:fibgen} we can interpret the corresponding defect as the twisted compactification of an object with $S^3 / Q_{16}$ as an asymptotic boundary, i.e., a D-type Du Val singularity in type IIB string theory. Therefore, this defect does not require the introduction of any further string theory object, but once more points towards interesting spacetime configurations (here, twisted compactifications of ADE singularities). As in other cases above it would be very interesting to determine the behavior of the localized degrees of freedom of this background  under the twist introduced by the fibration structure.

We see that these backgrounds are natural generalizations of those found in \ref{subsec:Mpcodim8}, which induces a different class of string like defects and twisted compactifications.

\subsection{Codimension-nine defects}
\label{subsec:GLcodim9}

The codimension-nine defects are induced by
\begin{align}
\Omega_8^{\text{Spin-GL}^+(2,\Z)} (\pt) = \Z \oplus \Z \oplus (\Z/2\Z) \,.
\end{align}
Let us first focus on the $\Z/2\Z$ generator given by $W^7_1 \times S^1_p$. We see that a generator of the bordism group in dimension seven appears and one can interpret this as a further circle compactification of the twisted theory discussed in the previous subsection. Therefore, we also do not need to include more defects to break the symmetries associated to this class. Similarly, we also know how to describe $S^1_p$ as a boundary and one can complete the second factor to a disc with several 7-branes wrapping around $W^7_1$.

The remaining two generators associated to the free $\Z$ subgroups are closely related to the generators of Spin bordisms. We find the quaternionic projective space $\HP^2$ that also appeared for the Spin-$\Mp(2,\Z)$ groups, see Section \ref{subsec:Mpcodim9}. The second generator $W_{1,8}$ is a refinement of the Bott manifold $B$ to Spin-$D_{16}$ manifolds and due to the fact that its Pontrjagin numbers are half of those of the Bott manifold could also be called half-Bott, which might also be realized in terms of a linear dilaton background.

\subsection{Codimension-ten defects}
\label{subsec:GLcodim10}

Last but not least we also analyze the defects in codimension ten associated to
\begin{align}
\Omega_9^{\text{Spin-GL}^+(2,\Z)} (\pt) = (\Z/2\Z)^{\oplus 8} \,.
\end{align}
Once more many of the Spin-$\Mp(2,\Z)$ classes are trivial due to the action of reflections on the various backgrounds in dimension $k = (4n+1)$.

Four of the generators, namely
\begin{align}
W^7_1 \times S^1_p \times S^1_p \,, \enspace B \times L^1_4 \,, \enspace \HP^2 \times L^1_4 \,, \enspace \HP^2 \times S^1_R \,,
\end{align}
are related to the compactification of classes above and thus are already taken care of by the higher-dimensional defects. Another generator $L^9_4$ describes the surviving Spin-$\Mp(2,\Z)$ configuration associated to an S-instanton discussed in Section \ref{subsec:Mpcodim10}.

This leaves three genuinely new backgrounds. Two of them are described by nine-dimensional versions of the Arcana
$X_9$ and $\widetilde{X}_9$. These backgrounds have the same underlying manifold, which is an $\RP^3$-bundle over
$\RP^6$, with two different Spin-$D_{8}$ structures associated to the two distinct embeddings into $D_{16}$ in
Figure \ref{fig:D8emb}. Since the fibers are given by the generators in $k = 3$, we can generate the corresponding
bordism classes by wrapping R7-branes, without the need for any additional new objects. More concretely, since the defect of  $\mathbb{RP}^3$ is just a  $D_{4}$ ALE singularity in IIB string theory, the additional background can be obtained as a $\mathbb{Z}/2\Z$ quotient of the resulting worldvolume geometry.

This leaves the final generator $W^9_1$ described by a $L^5_4$ lens space bundle over $\RP^4$. Since we know the defect that has $L^5_4$ as an asymptotic boundary, the most direct way to realize the generator for this $\Z / 2 \Z$ factor is via  wrapping the corresponding S-fold defect over $\RP^4$ with a twist involving reflections. Again, we emphasize that we have  not checked the consistency of this wrapping; although reflections are indeed a symmetry of the S-fold worldvolume theory at the bosonic level (since $\tau$ is fixed to $i$, and reflections flip the real, but not the imaginary part of $\tau$), it could be that this compactification is inconsistent at the quantum level. To study this more systematically would also require a detailed analysis of anomalies in the worldvolume theory. We defer this issue to future work.

This concludes our analysis of Spin-$\GL^+ (2,\Z)$ bordism groups and their associated defects. By the application of the Cobordism Conjecture we indeed find a new string theory defect that was not discussed before, the R7-brane. The inclusion of this codimension-two object then also allows the higher-dimensional backgrounds to appear as boundaries by wrapping the R7-branes on compact curves in the bulk. While the presence of R7-branes is expected to break supersymmetry, it is interesting that in some situations there are supersymmetry preserving orientifold planes which also trivialize the corresponding bordism groups. Including the asymptotic Ramond-Ramond fluxes in the discussion is a promising future direction to distinguish these backgrounds.

%%%%%%%%%%%%%%%%%%%%%%%%%%%%%%%%%%%%%%%%%%%%%%%%%%%%%%%%

\section{$\mathbf{\Omega}_{10}$ and discrete $\boldsymbol{\theta}$-angles}
\label{sec:thetaang}

So far we have discussed bordism groups in dimensions $k \leq 9$, which are associated to defects. The bordism groups
\begin{align}
\Omega^{\text{Spin-Mp}(2,\Z)}_{10} (\pt) = \Z/2\Z \,, \quad \Omega^{\text{Spin-GL}^+(2,\Z)}_{10} (\pt) = (\Z/2\Z)^{\oplus 4} \,,
\end{align}
describe potential discrete $\theta$-angles, topological couplings that can be added to the 10d action without changing their low-energy physics.

For example the generator in Spin-$\Mp(2,\Z)$ bordisms is inherited from a generator of $\Omega^{\Spin}_{10} (\pt)$ and is described by a manifold $X_{10}$ with non-trivial $w_4 w_6$ composed of Stiefel-Whitney classes. This opens up the possibility of adding a term of the form
\begin{equation}
- \Delta S = i \pi \int w_4 w_6 \,, \label{topterm}
\end{equation}
to the action. This implies that all spacetimes that are in the bordism class of $X_{10}$ are weighted with an additional minus sign in the path integral, whereas manifolds that are trivial in bordism enter with a plus sign. See \cite{Bergman:2013ala} for examples of discrete $\theta$-angles in field theory and \cite{FH21,Montero:2022vva} for examples in string and M-theory.

Once one can write down a topological term such as \eqref{topterm} to the supergravity action, two possibilities
arise: either the value of the $\theta$-angle is frozen (to either 0 or $\pi$) in the case of \eqref{topterm} by some
consistency condition, or it is dynamical. In the latter case, it leads to two physically distinct versions of type
IIB string theory, such as is the case in nine dimensions \cite{Montero:2022vva}, or the proposed alternative
anomaly cancellation mechanisms in \cite{Debray:2021vob}. As an example of the former, consider equation \eqref{topterm} in
type IIA string theory. Since IIA is secretly the small radius limit of M-theory on $S^1$, and $\Omega^{\Spin}_{11} (\pt) =0$, the discrete $\theta$-angle is absent. It is likely that this fact can be used to argue that \eqref{topterm} is also absent from the IIB supergravity action via T-duality, but this argument must be done carefully. We leave it to future work.

Moving on to Spin-$\GL^+ (2,\Z)$ $\theta$-angles, the manifold $X_{10}$ also generates one of the $\Z/2\Z$ factors of the ten-dimensional Spin-$\GL^+ (2,\Z)$ bordism group. The remaining generators are all of the form $M \times S^1$, with $M$ a generator of one of the nine-dimensional bordism classes discussed above, similarly to the M-theory discrete $\theta$-angle proposed in \cite{FH21}. Just like the above, these discrete $\theta$-angles are captured by $\eta$-invariants of fermions, and cannot be given local expressions in terms of cohomology classes.  In fact, the background  $B \times S^1_R \times S^1_p$ has the F-theory description\footnote{There $KB$ denotes the Klein bottle and $B$ denotes a Bott manifold.} of $B\times KB\times T^2$, which under M- / F-theory duality corresponds to $B\times KB \times S^1$, which is the background charged under the $\theta$-angle proposed in \cite{FH21}. Thus, we see the $\theta$-angle described in \cite{FH21}, plus additional contributions. It is likely that the M- / F-theory duality map can give additional consistency conditions that freeze some of these $\theta$-angles: for instance, another one of our $\theta$-angles detects $M=\HP^2\times S^1_R$. While this is a non-trivial bordism class,  $\HP^2\times KB \times S^1$ is trivial in $\mathfrak{m}^c$ bordism, and there is no corresponding $\theta$-angle. This probably means that there is a valid F-theory bordism, that we do not detect with our analysis, that kills the corresponding IIB $\theta$-angle. It would be interesting to study this in more detail.

%%%%%%%%%%%%%%%%%%%%%%%%%%%%%%%%%%%%%%%%%%%%%%%%%%%%%%%%

\section{Duality anomalies}
\label{sec:anomalies}

We are left with the bordism groups in dimension eleven. They parametrize the potential non-perturbative anomalies of the duality as well as duality mixing with gravity of a ten-dimensional theory. In order to explore these anomalies one has to define an invertible topological field theory in eleven dimensions \cite{FH16}, the anomaly theory $\mathcal{A}$, that reproduces the phase of the partition function of the 10d theory when placed on a manifold with boundary.

For type IIB the anomaly theory $\mathcal{A}_{\text{IIB}}$ receives contributions from the dilatini, the gravitni,
as well as the chiral 4-form (with self-dual five-form field strength), and a detailed derivation thereof was presented in \cite{Debray:2021vob}. Quite surprisingly, the duality anomaly does not vanish. While the Spin-$\GL^+ (2,\Z)$ manifolds with reflections are not problematic from the anomaly viewpoint, certain backgrounds inherited from Spin-$\Mp(2,\Z)$ are. For example one has
\begin{align}
\text{exp} \big( 2 \pi i \mathcal{A}[L^{11}_3] \big) = e^{2 \pi i / 3} \,,
\end{align}
indicating a mod $3$ anomaly for the duality. In order to restore the duality one therefore is forced to make modifications to the theory.

One possibility is to cancel the anomaly using a topological version of the Green-Schwarz mechanism, \cite{Garcia-Etxebarria:2017crf}. This introduces new discrete fields of various form degrees whose higher-form gauge transformations depend on the duality background. While the discrete, topological nature of these fields does not introduce any new local degrees of freedom, the completeness hypothesis \cite{Polchinski:2003bq} suggests the presence of defects coupling electrically and magnetically to them. Also these defects will be discrete, in that they are able to decay to the vacuum once one stacks a sufficient number of them and therefore do not appear in a perturbative description. The implementation of this mechanism is not unique and in principle several versions are imaginable. This would lead to theories that differ in the spectrum of non-local objects and opens up the possibility of a discrete landscape in ten dimensions. It might however also be the case that unexplored quantum gravity constraints affect this anomaly cancellation mechanism.

Alternatively, there is the (less radical, and thus more plausible) possibility of cancelling the duality anomalies without the introduction of new topological sectors. For this to work one couples the chiral 4-form field to the duality and gravitational background in a certain way, leading to a new topological term of the type IIB action. Once more this term only depends on discrete/torsional information and will not affect the local low-energy dynamics. However, this coupling leads to an additional contribution to the anomaly theory, via the quadratic refinement $\widetilde{\mathcal{Q}}$ necessary for the description of the self-dual 4-form, that is capable of compensating the anomalous phase. It is very intriguing that the anomalies we find are precisely of the type that can be cancelled by this mechanism, see also \cite{Dierigl:2022zll} for a six-dimensional realization of such a discrete Green-Schwarz mechanism. Moreover, up to a small number of prefactors, that we believe to be fixed by the analysis of certain special type IIB backgrounds this cancellation is unique.

To illustrate this discrete anomaly cancellation ``miracle'' let us come back to the anomaly on $[L^{11}_3]$.
We need
\begin{align}
\mathcal{A}_{\text{IIB}} [L^{11}_3] - \widetilde{\mathcal{Q}}[\check{c}] = 0 \text{ mod } \Z \,,
\end{align}
where $\check{c}$ indicates the differential cohomology element describing the duality background. It turns out that on $L^{11}_3$ the only possibility is
\begin{align}
\widetilde{\mathcal{Q}}[\check{c}] = \tfrac{1}{3} \,,
\end{align}
which shows that the anomaly can only be cancelled in such a way if $\mathcal{A}_{\text{IIB}} [L^{11}_3] = 1/3$ which is precisely what we find. Since $[L^{11}_3]$ describes a $\Z / 27\Z$ factor of the bordism group, there is a $1$ in $26$ chance of being able to cancel the anomaly (with uniform prior). A similar discussion holds for the other anomalous generators.

We see that the investigation of discrete anomalies heavily uses the determination of bordism group and their cancellation demand the modification of the low energy theory in subtle but interesting ways.

This concludes the physical interpretation of the defects predicted by the Cobordism Conjecture. We now turn to the computation of IIB duality bordism groups.

%%%%%%%%%%%%%%%%%%%%%%%%%%%%%%%%%%%%%%%%%%%%%%%%%%%%
%%%%%%%%%%%%%%% MATHS PART %%%%%%%%%%%%%%%%%%%%%%%%%%%%%
%%%%%%%%%%%%%%%%%%%%%%%%%%%%%%%%%%%%%%%%%%%%%%%%%%%%

\pagebreak
\part{Bordism classes for type IIB}
\label{p:maths}

In this part of the paper we compute $\Omega_k^{\Spin}\big(B\SL(2, \Z)\big)$, $\Omega_k^{\Spin\text{-}\Mp(2, \Z)} (\pt)$, and
$\Omega_k^{\Spin\text{-}\GL^+(2, \Z)} (\pt)$ for $k \le 11$. For $k \le 5$, $\Omega_k^\Spin\big(B\SL(2, \Z)\big)$ and
$\Omega_k^{\Spin\text{-}\Mp(2, \Z)} (\pt)$ were known before, but the rest of our computations are new.
The results of our computations are presented in \cref{the_big_table,,Mp_gens_table,,gl_gens_table}.

Our primary tool for making these computations is the Adams spectral sequence. Though we are far from the first in
the physics literature to use it, this method has a reputation for being difficult to understand. Therefore we
start in Section \ref{math_background} by introducing some definitions and concepts that will be helpful for
understanding the calculations we make in the next few sections. In Section \ref{the_Adams_SS} we summarize the essential
facts needed to follow our Adams spectral sequence arguments: how to compute the $\cA(1)$-module structure on
cohomology, how to determine the $E_2$-page from this data, and a few standard tricks for resolving differentials
and extension questions. Then we march through the three main computations.
\begin{itemize}
	\item In Section \ref{sl2_spin}, we determine $\Omega_k^{\Spin}\big(B\SL(2, \Z)\big)$. We show how to determine this in terms
	of $\Omega_k^\Spin(B\Z/4\Z)$ and $\Omega_k^\Spin(B\Z/3\Z)$, both of which can quickly be reduced to things already
	in the mathematics literature.
	\item In Section \ref{mp_spin}, we compute $\Omega_k^{\Spin\text{-}\Mp(2, \Z)} (\pt)$. The result is the same, except that
	instead of $\Omega_k^\Spin(B\Z/4\Z)$, we get $\Omega_k^{\Spin\text{-}\Z/8\Z} (\pt)$. This bordism theory has been studied
	in dimensions $5$ and below by Campbell~\cite{Cam17}, Hsieh~\cite{Hsi18}, and Davighi-Lohitsiri~\cite{DL20b};
	we extend their computations to dimension $11$.
	\item In Section \ref{gl_spin}, we address $\Omega_k^{\Spin\text{-}\GL^+(2, \Z)} (\pt)$. In this case,
	$\Omega_k^\Spin(B\Z/3\Z)$ is replaced with $\Omega_k^\Spin(BD_6)$, which we compute with the Atiyah-Hirzebruch
	spectral sequence; and $\Omega_k^{\Spin\text{-}\Z/8\Z} (\pt)$ is replaced with a bordism theory we call
	$\Omega_k^{\Spin\text{-}D_{16}} (\pt)$, which we compute with the Adams spectral sequence.
\end{itemize}

With this very brief overview of the upcoming sections we will now go into the details of the derivation and the
construction of the generating manifolds. Since this part contains many mathematical results that are of interest
beyond their applications in physics, we will also adopt a more mathematical style in this part. Wherever possible,
we try to provide a paragraph that is intended to add some physical intuition to the more technical procedures and
which is indicated by a heading ``\textbf{Physics intuition}''.\footnote{Of course this should neither dissuade mathematicians from reading
these paragraphs, nor physicists from delving into the full details of the construction.}

\section{Preliminaries on computations}
	\label{math_background}

The bordism computations we undertake require some tools and concepts which are standard in algebraic topology, but
are not as well-known in physics. In this subsection, we briefly review these ideas and point the
interested reader to more complete references.

%\textcolor{red}{Things to cover in this section: possibly how to compute $\eta$-invariants of lens spaces?}
\subsection{Classifying spaces}
% quick section standaridizing notation for tautological bundles, classifying spaces, etc.
\label{ss:taut_bundle}
In this subsection we go over some basics on classifying spaces, their tautological bundles, and Lashof's approach
to different kinds of bordism. We do this both to provide background to the reader and to standardize notation.

Let $G$ be a topological group. Then there is a space $BG$, called the \term{classifying space} of $G$, which has
the following properties.
\begin{itemize}
	\item $BG$ is the quotient of a contractible space $EG$ by a free $G$-action. Said differently, $EG\to BG$ is a
	principal $G$-bundle with $EG$ contractible. This principal $G$-bundle is called the \term{universal} or
	\term{tautological} principal $G$-bundle.
	\item For any space $X$, isomorphism classes of principal $G$-bundles $P\to X$ are equivalent to homotopy
	classes of maps $f\colon X\to BG$; in one direction, this equivalence is achieved by setting $P\coloneqq
	f^*(EG)$.
\end{itemize}
Cohomology classes for $BG$ define characteristic classes for principal $G$-bundles: given $c\in H^*(BG)$, set
$c(P)\coloneqq f^*(c)$ if $f$ is the classifying map for $P\to X$.

$BG$ is only well-defined up to homotopy type. Generally classifying spaces are large: for example, the homotopy
type $B\Z/2\Z$ can be realized by infinite-dimensional real projective space, but not by any finite-dimensional CW
complex.

When $G$ is a matrix group, it has a canonical representation on a vector space $V$, so we can define a
\term{tautological vector bundle} $\mathcal T_G\to BG$ by the following ``mixing construction:''
\begin{equation}
	\mathcal T_G\coloneqq EG\times_G V\coloneqq EG\times G/((x\cdot g, v)\simeq (x, g\cdot v)).
\end{equation}
For example, $\mathcal T_{\O(n)}\to B\O(n)$ is a rank-$n$ real vector bundle; $\mathcal T_{\mathrm U(1)}\to B\mathrm
U(1)$ is a complex line bundle; and so on. We use the notation $\sigma \to\RP^n$ to denote the tautological line
bundle on $\RP^n$, which is the pullback of $\mathcal T_{B\Z/2\Z}\to B\Z/2\Z$ by the inclusion
$\RP^n\hookrightarrow\RP^\infty = B\Z/2\Z$.
%In both of these cases, when the base space is clear, we may just refer
%to this bundle as $\mathcal T$.

Let $\O\coloneqq\varinjlim_n\O(n)$. Concretely, one can think of this as the union of all orthogonal groups of all
sizes, where we include $\O(n)$ in $\O_{n+1}$ by sending $A\mapsto \begin{pmatrix}A & 0\\0 & 1\end{pmatrix}$. This
is a topological group, and its classifying space $B\O$ classifies \term{rank-zero virtual vector bundles}, i.e., 
pairs of vector bundles $V$ and $W$ such that $\mathrm{rank}(V) = \mathrm{rank}(W)$. We think of this data as ``$V
- W$,'' in much the same way that the integers are defined as equivalence classes of pairs of natural numbers $(p,
q)$ and thought of as $p - q$. That is, a map $X\to B\O$ is equivalent data to a rank-zero virtual vector bundle $V
- W\to X$. As before, this correspondence arises as the pullback of a tautological bundle $\mathcal T_{B\O}\to
  B\O$. Analogues of this story hold for other families of Lie groups, defining spaces $B\SSO$, $B\Spin$, etc.

Any actual vector bundle $V\to X$ of rank $n$ defines a rank-zero virtual vector bundle $V - \underline\R^n$, hence
a map to $B\O$. This construction is \term{stable} in the sense that $V$ and $V\oplus\underline\R$ define
isomorphic virtual vector bundles.
\begin{defn}
\label{tang_str}
A \term{tangential structure} is a map $\xi\colon B\to B\O$. A \term{$\xi$-structure} on a vector bundle $V\to X$
is data of a lift of the classifying map $f\colon X\to B\O$ to a map $\widetilde f\colon X\to B$ such that
$\xi\circ\widetilde f = f$:
\begin{equation}
	\begin{tikzcd}
	& B \\
	X & B\O
	\arrow["f"', from=2-1, to=2-2]
	\arrow["\xi", from=1-2, to=2-2]
	\arrow["{\widetilde f}", dashed, from=2-1, to=1-2]
\end{tikzcd}
\end{equation}
If $M$ is a manifold, a $\xi$-structure on $M$ means a $\xi$-structure on $TM$.
\end{defn}
For example, if $\xi$ is $B\SSO\to B\O$, a $\xi$-structure is equivalent to an orientation; for $B\Spin\to B\O$, we
obtain a Spin structure; and for $B\O\times BG\to B\O$, we obtain a principal $G$-bundle.

Suppose $M$ is a manifold with boundary $\partial M$, and let $\nu\to \partial M$ be the normal bundle to the
inclusion $\partial M\hookrightarrow M$. Then $TM|_{\partial M}\cong T(\partial M)\oplus\nu$.\footnote{To obtain
such a splitting, one needs to choose a Riemannian metric on $\partial M$, but the space of such metrics is
contractible. Therefore from the perspective of homotopy theory, which is all that is needed in this section, the
specific choice does not matter.} Moreover, the outward unit normal vector field trivializes $\nu$, so $T(\partial
M)\oplus\underline\R\cong TM|_{\partial M}$. Thus the virtual vector bundles defined by $T(\partial M)$ and
$TM|_{\partial M}$ are isomorphic, so a $\xi$-structure on $M$ induces a $\xi$-structure on $\partial M$.

Because $\xi$-structures restrict to boundaries, we can define bordism of manifolds equipped with $\xi$-structures:
a closed manifold $M$ with $\xi$-structure is trivial in $\xi$-bordism if $M$ bounds a compact manifold $W$ with
$\xi$-structure, such that the identification $M\cong\partial W$ can be made compatible with the $\xi$-structures
on both sides. This general formalism is due to Lashof~\cite{Las63}; the bordism group of $n$-manifolds with
$\xi$-structure is denoted $\Omega_n^\xi$. It is often the case that $\xi$ is a map $BG\to B\O$, e.g., oriented and
Spin bordism; in that case, the $\xi$-bordism groups are typically denoted $\Omega_*^G$. For example,
$\Omega_*^\SSO$ refers to the oriented bordism groups.

\begin{rem}[Multiplicative structure]
\label{mult_str_rem}
The product of two oriented manifolds has an induced orientation, and this is compatible with the bordism
equivalence relation, making $\Omega_\ast^\SSO$ into a $\Z$-graded commutative ring. This is likewise true for
many other bordism theories, including $G$-bordism for $G = \mathrm O$, $\Spin$, $\Spin^c$, $\mathrm{String}$, and
$\mathrm U$, though it is not always true, e.g., for $G = \mathrm{Pin}^\pm$. When this is true, the multiplication
lifts from bordism rings to ring structures on the corresponding Thom spectra, which can be thought of as
expressing the naturality of this multiplication with respect to $\xi$-bordism groups of spaces. Determining the ring
structure, when present, was classically an important part of bordism theory.

When $\xi$-bordism is not a ring, it is often a module over some other kind of bordism which is a ring. For
example, the product of a Spin manifold and a $\Pin^+$ manifold has a canonical $\Pin^+$ structure, making
$\Omega_*^{\Pin^+}$ into an $\Omega_*^\Spin$-module.

From a physics point of view, these ring and module structures provide information about compactifications. In the
framework of the cobordism conjecture, if a class $x\in\Omega_k^\xi$ can be represented by a product $[M\times N]$
of $\xi$-manifolds,\footnote{Here we are being a little imprecise about which tangential structures are placed on
$M$, $N$, and $M\times N$. In this paper $M$ has a Spin structure, and $N$ as well as $M\times N$ have $\xi$-structures in
a setting where $\xi$-bordism is a module over Spin bordism. The full details of the relationship between
compactifications and tangential structures are subtle; see Schommer-Pries~\cite[\S 9]{SP18} for a careful
analysis.} then a natural candidate for the defect that the cobordism conjecture predicts for the class $[M\times N]$ is the
compactification on $M$ of the defect predicted for $N$, although there may be obstructions to wrapping the defect on $M$, as described briefly in Section \ref{subsec:fibgen} and various other places of the Chronicles.

There is a similar interpretation of ring and module structures in the bordism classification of invertible field
theories~\cite{FHT10, FH16}: a unitary invertible $n$-dimensional topological field theory on $\xi$-manifolds is
determined by its partition function $Z$, which is a bordism invariant $Z\colon \Omega_n^\xi\to\C^\times$.
Compactifying this theory on a $k$-dimensional $\xi$-manifold $M$ yields another unitary invertible TFT, classified
by the bordism invariant
\begin{equation}
	\Omega_{n-k}^\xi\overset{\text{--}\times [M]}{\longrightarrow} \Omega_n^\xi\overset{Z}{\longrightarrow}
	\C^\times.
\end{equation}
Here we assumed $\xi$-bordism is a ring, but the story generalizes to modules too.
\end{rem}

\vspace{0.3cm}

\noindent $\triangleright$ {\bf Physics intuition:} Classifying spaces $BG$ are also well-known in physics where they define the various physically distinct gauge backgrounds, the isomorphism classes of principal $G$-bundles, of a gauge theory with gauge group $G$. A simple example is a theory with gauge group U$(1)$ in which case the classifying space is given by
\begin{align}
B \text{U} (1) = \mathbb{CP}^{\infty} = K(\Z,2) \,,
\label{eq:physclassi}
\end{align}
where $K(2,\mathbb{Z})$ denotes the Eilenberg-Mac Lane space with $\pi_2 \big(K(\Z,2) \big) = \Z$ and all other
homotopy groups $\pi_i$ with $i \neq 2$ vanishing. The cohomology of $B \text{U}(1)$ is generated by a single
element in degree $2$ which pulls back on the spacetime to the first Chern class, i.e., the field strength,
\begin{align}
c_1 = \tfrac{1}{2 \pi} F \in H^2 (M;\Z) \,,
\end{align}
of the gauge bundle specified by the classifying map $f$ from spacetime $M$ into $B \text{U}(1)$. One can also specify the background of higher-form gauge fields by maps into higher classifying spaces. For example, a 1-form symmetry gauge theory with gauge group U$(1)$ (as specified by a 2-form gauge potential) is described by maps into
\begin{align}
B^2 \text{U}(1) = B K(\Z,2) = K(\Z,3) \,,
\end{align}
which once more is described by an Eilenberg-Mac Lane space containing information about fluxes in $H^3 (M; \Z)$.
 $\triangleleft$

\subsection{Working one prime at a time}
\label{ss:one_prime}

Fix a prime $p$, and let $\Z_{(p)}$ denote the ring of rational numbers whose denominators are not divisible by $p$.
Tensoring an Abelian group with $\Z_{(p)}$ throws out information that is prime to $p$. For example, if $A$ is a
finitely generated Abelian group, there is an isomorphism
\begin{equation}
	A\cong \Z^r\oplus (\Z/p_1^{e_1}\Z) \oplus \dots \oplus (\Z/p_n^{e_n}\Z)
\end{equation}
for some primes $p_1,\dotsc,p_n$, and nonnegative integers $r, e_1,\dotsc,e_n$. Then $A\otimes\Z_{(p)}$ is a direct
sum of $\Z_{(p)}^r$ and the $\Z/p_i^{e_i}\Z$ summands for which $p_i = p$. In particular, given
$A\otimes\Z_{(p)}$ for all primes $p$, $A$ is determined up to isomorphism.

Let $E$ and $F$ be generalized homology theories. A map $E\to F$ is a \term{$p$-local equivalence} if for all
spaces $X$, $E_*(X)\otimes\Z_{(p)}\to F_*(X)\otimes\Z_{(p)}$ is an isomorphism. In homotopy theory, it is very
common to calculate $E$-homology by finding $p$-local equivalences to simpler theories at different primes, then
putting everything back together afterwards. We will use this approach, both directly in Section \ref{subsec:splitodd}
and \ref{subsec:spliteven}, and also indirectly in analyzing the $p$-torsion in bordism groups for different
primes $p$ separately. When we say ``$p$-locally'' or ``the $p$-primary part,'' we mean working up to $p$-local
equivalences, or equivalently throwing out torsion that is prime to $p$.\footnote{``Equivalently'' is not quite
true, but it is true for all spaces or spectra whose homology groups are finitely generated Abelian groups.
Essentially all situations one could reasonably encounter in these sorts of computations meet this criterion.}

\subsection{Bluff your way through spectra}
A few steps in our computations require saying the word ``spectrum,'' in the sense of homotopy theory. The purpose
of this subsection is not to give a definition or proper introduction, but to provide just enough information so
that a reader without a background in algebraic topology can follow those steps in our computations. For more
in-depth references, see Freed-Hopkins~\cite[Section 6.1]{FH16} or Beaudry-Campbell~\cite[Section 2]{BC18}, or
Schwede~\cite{Schwede} for a more homotopical and in-depth perspective.

\vspace{0.3cm}

\noindent $\triangleright$ {\bf Physics intuition:} Spectra are essential in the study of generalized (co)homology theories such as bordism.
Generalized (co)homology classes of $X$ are described by constructing maps to the associated spectrum, then
taking homotopy classes.

For generalized cohomology theories the cohomology groups of a space $X$ are given by homotopy classes of maps from
$X$ to the respective spectrum. This is very similar to what classifying spaces do to define the various
inequivalent principal $G$-bundles as we discussed above. However, spectra know more than classifying spaces. For
example they can have homotopy groups in negative degree, which means that they are not necessarily are given by
topological spaces. Since they are naturally very big (remember that already $B\mathrm U(1) = \mathbb{CP}^{\infty}$
is infinite-dimensional), it is often easier to split them into pieces. That can be done by working with other
spectra that are $p$-locally equivalent to the original spectrum and hence contains the same information at the prime $p$.

% In a sense the spectrum classifies the structure of $X$, similar to classifying spaces determining the various inequivalent principal bundles on $X$. Notably, in general spectra are not topological spaces, since, for example, they can have homotopy groups in negative degree. With the spectrum containing all the necessary information about the generalized (co)homology theory, this also implies that to work $p$-locally when calculating generalized (co)homology as
%described above, we can consider a decomposition of the corresponding spectrum into easier parts which are
%$p$-locally equivalent to the original spectrum.

A relatively simple example is the Eilenberg-Mac Lane spectrum (\cref{EMspec})
$HA$ with $A$ an Abelian group, for which one has
\begin{align}
H^k (X; A) = [X, \Sigma^k HA] \,,
\end{align}
i.e., the cohomology with coefficients in $A$. Here $\Sigma^k$ is the $k$-fold suspension operation.
Interpreting the classes in $H^k (X;A)$ as field strengths of various (higher-form) gauge fields,\footnote{See the
``Physics intuition'' part on classifying spaces around equation \eqref{eq:physclassi}.} we see that, in a sense,
the spectrum $H A$ does not only have information about $B A$, the classifying space of $A$, but simultaneously
about all higher classifying spaces $B^n A$ as well as their relations. Likewise, more complicated spectra contain
information about classifying spaces of Abelian $n$-groups, and in fact there is an equivalence between the category
of spectra whose homotopy groups are only nonzero in degrees $0\le k\le n-1$ and a category of Abelian
$n$-groups~\cite{MOPSV22}. See
%Intuitively, this means that (at least some) spectra act as classifying spaces for Abelian higher-group symmetries, see also
\cite{Freed:2018cec, Liu:2021hhy} for some applications of this idea in the mathematical physics literature.

If one can decompose $A$ into simpler Abelian groups, one can also decompose $HA$: for example, if $A = \Z/6\Z =
(\Z/2\Z) \oplus (\Z/3\Z)$, $H\Z/6\Z\simeq (H\Z/2\Z )\vee (H\Z/3\Z)$, inducing a direct sum on cohomology groups:
\begin{align}
\begin{split}
H^k (X; \Z/6\Z) &= [X, \Sigma^k H\Z/6\Z] = [X, \Sigma^k H\Z/2\Z \vee \Sigma^k H\Z/3\Z] \\ &= H^k (X; \Z/2\Z) \oplus
H^k (X; \Z/3\Z).
\end{split}
\end{align}
Therefore at $p = 2$, we would only have to worry about $H\Z/2\Z$, and at $p = 3$ we would only have to worry about
$H\Z/3\Z$. This is a somewhat simple example, but we will take advantage of analogous $p$-local simplifications of
bordism spectra to reduce the evaluation of bordism groups to either known or more tractable calculations, as we discuss below in
Section \ref{subsec:splitodd} and \ref{subsec:spliteven}. We are, however, not only interested in Spin bordisms but
more general tangential structures such as Spin-Mp$(2,\Z)$ and Spin-GL$^+ (2,\Z)$. Again there is a trick to relate
these to Spin bordism called ``shearing,'' which we discuss in Section \ref{subsec:shearing}. $\triangleleft$

\vspace{0.3cm}

Spectra in the sense of homotopy theory were invented and named by Lima~\cite{Lim59, Lim60}; they are etymologically
unrelated to spectra in algebraic geometry, functional analysis, quantum mechanics, etc. Here are two key facts
about spectra:
\begin{itemize}
	\item Spectra behave very much like topological spaces: they have homotopy, homology, and cohomology groups, and
	one can perform operations such as the suspension of a spectrum, or the spectrum of maps between a space and a
	spectrum, or two spectra; there is also a notion of homotopy equivalence of spectra.
	\item Spectra represent to generalized (co)homology theories. For every spectrum $E$, there is a generalized
	homology theory $h_*(X)\coloneqq \pi_*(X\wedge E)$ and a generalized cohomology theory $h^{*}(X) =
	\pi_{-*}(\mathrm{Map}(X, E))$,\footnote{The minus sign here is an artifact of homological versus cohomological
	grading: stable homotopy groups have long exact sequences whose differentials lower degree, but in a
	generalized cohomology theory, we want the differential in a long exact sequence to raise degree. This nuance
	is not crucial for understanding the mathematics in this paper.} and every generalized homology or cohomology
	theory arises from a spectrum in this way.
\end{itemize}
There are a few important differences between spaces and spectra: spectra can have homotopy groups in negative
degrees, and the suspension operation is invertible on spectra. That is, taking $\Sigma^{-1}$ of a spectrum is
a sensible operation even though taking $\Sigma^{-1}$ of a space does not make sense. Also, there is not generally a cup product on
the cohomology of a spectrum.

It is possible to take direct sums, etc., of generalized cohomology theories, and this lends spectra an algebraic
flavor. It can be useful to think of spectra as akin to Abelian groups.
\begin{exm}[Eilenberg-Mac Lane spectra]
\label{EMspec}
Given an Abelian group $A$, there is a spectrum $HA$ whose corresponding homology theory is $H_*(\text{--}; A)$ and
whose cohomology theory is $H^*(\text{--}; A)$. $HA$ is called an \term{Eilenberg-Mac Lane spectrum}.
\end{exm}
\begin{exm}[Suspension spectra]
\label{suspension_spectra}
Given a topological space $X$ with chosen basepoint in $X$, there is a spectrum $\Sigma^\infty X$, called the
\term{suspension spectrum} of $X$, whose corresponding homology theory is $\lim_{k\to\infty}
\pi_{*+k}(\text{--}\wedge \Sigma^k X)$. When $X = S^0$, its suspension spectrum is denoted $\mathbb S$ and called
the \term{sphere spectrum}; its homology theory evaluated on a space $X$ are the stable homotopy groups of $X$.

The sphere spectrum is an important object in stable homotopy theory, playing a role akin to the integers in the
land of Abelian groups. Given an Abelian group $A$, there is a natural isomorphism $\Z\otimes A\overset\cong\to A$,
making Abelian groups naturally into $\Z$-modules. Similarly, for every spectrum $E$, there is a natural homotopy
equivalence $\mathbb S\wedge E\overset\simeq\to E$, and there is a sense in which this makes spectra naturally into
modules over $\mathbb S$. This has the concrete consequence that $\pi_*(\mathbb S)$, called \term{the stable
homotopy groups of the spheres}, is a graded ring, and for every spectrum $E$, $\pi_*(E)$ is a $\pi_*(\mathbb
S)$-module. We use this fact a few times, e.g., in \cref{Z4_5_extn}.
\end{exm}
\begin{exm}[Thom spectra]
\label{Thom_exp}
%\TODO: where did we discuss symmetry types/tangential structures? Cite Lashof.
%
Let $X$ be a space and $V\to X$ be a virtual vector bundle, and let $f_V\colon X\to B\O$ be the classifying map of
$V$. Then there is a spectrum $X^V$ called a \term{Thom spectrum} with two key properties: its cohomology resembles
that of $X$, and its homotopy groups are bordism groups. Specifically:
\begin{itemize}
	\item The \term{Thom isomorphism theorem} produces a natural isomorphism
	\begin{equation}
		\label{Thom_iso}
		U\colon H^*(X;\Z_{w_1(V)}) \overset\cong \longrightarrow \widetilde H^{*+\mathrm{rank}(V)}(X^V;\Z),
	\end{equation}
	where $\Z_{w_1(V)}$ denotes the local system given by the orientation bundle of $V$. Often one sees the
	technically ambiguous notation $U\coloneqq U(1)$; this class is called the \term{Thom class}. The analogous
	isomorphism is true for other coefficient groups, and for $\Z/2\Z$ cohomology the local coefficient system
	$(\Z/2\Z)_{w_1(V)}$ is trivial, so the Thom isomorphism uses untwisted cohomology on both sides.
	\item The \term{Pontrjagin-Thom theorem}~\cite{ThomThesis, Pon38, Pon50, Pon55} identifies
	$\pi_{*+\mathrm{rank}(V)}(X^V)$ with the bordism groups of manifolds with a \emph{normal} $f_V$-structure,
	i.e.\ a lift of the classifying map $f_\nu\colon M\to B\O$ of the stable normal bundle to a map $\widetilde
	f_\nu\colon M\to X$ such that $f_\nu = f_V\circ\widetilde f_\nu$.\footnote{Pontrjagin and Thom focused on a few
	specific normal structures; the idea to consider them in general is due to Lashof~\cite{Las63}.}
	%\textsuperscript{,}\footnote{In this paper, and in the mathematical physics literature in
	%general, it is much more common to consider $f_V$-structures on the tangent bundle, rather than the normal
	%bundle. One can exchange the notions of tangential and normal structures by composing $f_V$ with the map
	%$-1\colon B\O\to B\O$.}
\end{itemize}
Because of the shifts by $\mathrm{rank}(V)$, it is common to replace $V$ with the rank-zero virtual vector bundle
$V - \underline\R^{\oplus \mathrm{rank}(V)}$, often denoted $V - \mathrm{rank}(V)$ for short, which allows one to
drop the shifts in the descriptions above. We often do this in this paper.

When $V$ is trivial and rank-$0$, $X^V$ is the suspension spectrum $\Sigma^\infty X$ from
\cref{suspension_spectra}. Thus given a space $X$ with basepoint $i\colon \pt\hookrightarrow X$ and a vector bundle
$V\to X$, we obtain a map $\pt^{i^*V}\to X^V$, i.e.\ a map $\mathbb S\to X^V$. We will use this fact a few times.

If $G$ is a topological group with a map to the infinite orthogonal group $\O$, we would like to obtain a Thom
spectrum whose homotopy groups are bordism groups of manifolds with a $G$-structure on $TM$, rather than on the
stable normal bundle. To do this, compose with the map $-1\colon B\O\to B\O$ before taking the Thom spectrum\footnote{The intuition being that the normal and tangent bundles add up to a trivial bundle, so their stable equivalence classes add up to zero.}. The
resulting spectrum is called a \term{Madsen-Tillman spectrum} and is denoted $\mathit{MTG}$. We will most
frequently use $\MTSO$, whose homotopy groups are the bordism groups of oriented manifolds; $\MTSpin$, whose
homotopy groups are Spin bordism groups; and variants of $\MTSpin$.
\end{exm}

\subsection{Shearing}
\label{subsec:shearing}

It is not always obvious how to write the symmetry type of interest in terms of Spin structures. For example, if
one studies a theory with fermions and a $\mathrm U(1)$ symmetry with an additional selection rule that bosonic fields have even charges and fermionic fields have odd charges, the symmetry type is $\Spin^c$. Our Adams
spectral sequence techniques compute Spin bordism groups, so we would like to express $\Spin^c$ structures in terms
of Spin structures. ``Shearing'' is a general approach to describing a tangential structure $\xi$ as equivalent to
some sort of ``twisted $\xi'$-structure,'' where $\xi'$ is easier to understand; for us, $\xi'$ will always be
Spin. Once we find this equivalence, \cref{shearing_PT} describes the $\xi$-bordism groups as $\xi'$-bordism groups
of a Thom spectrum, giving us access to the spectral sequence techniques we discuss in the next section.

For example, a $\Spin^c$ structure on a vector bundle $E\to M$ is equivalent to a complex line bundle $L\to M$ and
a Spin structure on $E\oplus L$; we will see how to pass this through \cref{shearing_PT} to deduce that
%For example, a Pin$^+$
%structure on a vector bundle $E$ is equivalent data to a spin structure on $E\oplus 3\Det(E)$, and a Pin$^-$
%structure is equivalent to a spin structure on $E\oplus\Det(E)$. This will allow us to write down an isomorphism \textcolor{red}{What is the tilde?}
\begin{equation}
\label{spinc_isom}
	\Omega_*^{\Spin^c} (\pt) \cong \Omega_*^\Spin \big((B\mathrm U(1))^{\mathcal T_{\mathrm U(1)}-2}\big) \,.
\end{equation}
There are different perspectives on shearing, phrased more or less abstractly. Though some amount of homotopy
theory is needed for the proofs, we have tried to express the statements of
\cref{twisted_defn,shearing_PT,w2_twisted_cor} in terms of groups, representations, and vector bundles, so that
they are hopefully easier to use.

%\footnote{The fact that $X_\pm$ are spectra, rather than spaces, is not an important point
%for one's first pass through this material. As far as we need to know to understand and make these
%computations, spectra are things which behave like topological spaces, and we can calculate their cohomology and
%bordism groups in similar ways as for spaces.}
%Specifically, the notion of ``a spin structure on $E$ plus some
%other bundle'' can be recast as twisted spin bordism, and from there one can write down $X_\pm$ using a formula.

\vspace{0.3cm}

\noindent $\triangleright$ {\bf  Physics intuition:} Shearing is also familiar in physical constructions in which fermions are charged under an internal symmetry. While uncharged fermions might be forbidden, charged fermions are allowed if the internal bundle compensates for the missing Spin structure. One of the most familiar examples is that of a Spin$^c$ structure, for which one only allows fermions of odd charge under an internal U$(1)$ symmetry. The obstruction to a Spin structure, given by the second Stiefel-Whitney class $w_2$, is compensated by a non-trivial field strength of the U$(1)$ gauge field. Extending the tangent bundle of spacetime with a complex line bundle associated to the U$(1)$ principal bundle one can define a Spin structure on this extended bundle, the Spin$^c$ structure on spacetime. From this it is also clear that the twisted Spin structures will lead to a correlation between $w_2$ and the characteristic classes of the bundles describing the internal symmetries.

For example, $\CP^2$ has $w_2\ne 0$, so it has no Spin structure, but it has a Spin$^c$ structure in which the
$\mathrm U(1)$-bundle is the unit circle bundle inside the canonical bundle.

For the cases of interest for our later discussion of type IIB string theory, the twisted Spin structures have a
nice interpretation in terms of F-/M-theory. On the one hand, a Spin-Mp$(2,\Z)$ structure on a manifold corresponds
to a Spin structure on an associated 2-torus fibration over this manifold. It thus induces a Spin structure on the
total space in F- and M-theory. A Spin-GL$^+ (2,\Z)$ structure, on the other hand, allows for orientation reversal
and thus can be interpreted as a Pin$^+$ structure on a torus fibration, where the orientation reversal only
happens in the fiber. This can be further lifted by the arguments above to a Spin structure on a 5-torus
fibration,\footnote{Note that a Pin$^+$ structure on the tangent bundle of a manifold $TM$ is equivalent to a Spin
structure on $TM \oplus \text{Det} (TM)^{\oplus 3}$.} where the additional three circles are invariant under SL$(2,\Z)$ transformations
but flip orientation under application of the reflection operator. $\triangleleft$

\vspace{0.3cm}

\begin{defn}
\label{twisted_defn}
Let $X$ be a space and $V\to X$ be a virtual vector bundle. Then an \term{$(X, V)$-twisted Spin structure} on a
vector bundle $E\to M$ is a map $f\colon M\to X$ and a Spin structure on $E\oplus f^*V\to M$.

With $X$ and $V$ fixed, $(X, V)$-twisted Spin structures define a symmetry type and thus a notion of bordism. We
let $\Omega_{k+V}^{\Spin}(X)$ denote the Abelian group of bordism classes of $k$-manifolds with an $(X, V)$-twisted
Spin structure.
\end{defn}
\begin{lem}[Shearing]
\label{twisted_classifier}
Let $\mathcal T_{\Spin}\to B\Spin$ denote the tautological stable vector bundle, defined in Section \ref{ss:taut_bundle}.
The symmetry type for $(X, V)$-twisted Spin structures is $\phi\colon B\Spin\times X\to B\O$ with the map given by
the vector bundle $-\mathcal T_{\Spin}\oplus V\to B\Spin\times X$.
\end{lem}
\begin{proof}
Given an $(X, V)$-twisted Spin structure on $E\to M$, namely a map $f\colon M\to X$ and a Spin structure on
$E\oplus f^*V$, we obtain a map $M\to B\Spin\times X$ given by $(E\oplus f^*V, f)$, and using that homotopy classes
of maps to $B\Spin$ are naturally identified with stable virtual Spin vector bundles. Conversely, given a map
$(\psi, f)\colon M\to B\Spin\times X$, take $\phi\circ\psi$ to obtain a Spin vector bundle $E'$, and let $E\coloneqq
E' - f^*V$; then we have a canonical $(X, V)$-twisted Spin structure on $E$. These two operations are inverses up
to isomorphism, so every $(X, V)$-twisted Spin structure on any vector bundle on any space pulls back from
$B\Spin\times X$ in a unique way up to homotopy, which is what we wanted to prove.
\end{proof}
The Pontrjagin-Thom theorem then implies
\begin{cor}
\label{shearing_PT}
There is a natural isomorphism $\Omega_{*+V}^\Spin(X)\cong\Omega_*^\Spin(X^{V-\mathrm{rank}(V)})$, where
$X^{V-\mathrm{rank}(V)}$ is the Thom spectrum of the virtual bundle $V - \underline\R^{\mathrm{rank}(V)}\to X$.
\end{cor}
That is, twisted Spin bordism is Spin bordism of something --- and as we discussed in \cref{Thom_exp}, that
``something'' has relatively easy-to-understand cohomology. This suggests that tools such as the Atiyah-Hirzebruch
and Adams spectral sequences, which take that cohomology as input,
can be used to compute twisted Spin bordism groups.

%\textcolor{red}{TODO: make sure this is readable to physicists. Yes, this is a bit hard (at least for me) but I can try to come up with some `translation'}
%
%So our goal is to express our symmetry types in terms of $(X, V)$-twisted spin structures. We will do this in
%\cref{shearing_spin_Z8} for spin-$\Z/n$-bordism and in (\TODO) for spin-$D_{2n}$ bordism. This uses little homotopy
%theory, instead studying maps between Lie groups.
One way to produce vector bundles is to choose a space with a principal $G$-bundle $P\to M$ and a
$G$-representation $\rho$ on a vector space $V$. Then the \term{associated bundle}
\begin{equation}
	P_\rho \coloneqq P\times V/((p\cdot g, v)\simeq (p, \rho(g)\cdot v))
\end{equation}
is a vector bundle with rank equal to $\dim(V)$. Often $P\to M$ is the universal bundle $EG\to BG$; in this case we
sometimes write $\rho$ as shorthand for $(EG)_\rho$.
\begin{lem}
\label{twisted_spin_pullback_lem}
Let $G$ and $H$ be topological groups and $\rho\colon G\to\SSO(d)$ be a representation. If there exists a pullback
diagram
% https://q.uiver.app/?q=WzAsNSxbMCwwLCJIIl0sWzAsMSwiXFxTU09cXHRpbWVzIEciXSxbMSwxLCJcXFNTT1xcdGltZXNcXFNPX2QiXSxbMiwxLCJcXFNPIl0sWzIsMCwiXFxTcGluIl0sWzAsNF0sWzAsMV0sWzEsMiwiKFxcaWQsIFxccmhvKSJdLFsyLDMsIlxcb3BsdXMiXSxbNCwzXV0=
\begin{equation}
\label{twisted_spin_pullback}
\begin{tikzcd}
	H && \Spin \\
	{\SSO\times G} & {\SSO\times\SSO(d)} & \SSO,
	\arrow[from=1-1, to=1-3]
	\arrow["p_1", from=1-1, to=2-1]
	\arrow["{(\id, \rho)}", from=2-1, to=2-2]
	\arrow["\oplus", from=2-2, to=2-3]
	\arrow["p_2", from=1-3, to=2-3]
\end{tikzcd}
\end{equation}
then $H$-structures are naturally equivalent to $(BG, (EG)_\rho)$-twisted Spin structures.
%\textcolor{red}{Above $\rho$ is a representation here it is a vector bundle associated to the representation, right? Or is this the same?}.
\end{lem}
\begin{proof}
That~\eqref{twisted_spin_pullback} is a pullback square implies that
a lift of a map $f\colon M\to \SSO\times G$ across $p_1$ is equivalent to a lift of ${\oplus}\circ(\id, \rho)\circ
f$ across $p_2$. The former is the data on transition functions of an oriented vector bundle $E\to M$ and a principal
$G$-bundle $P\to M$ to define an $H$-structure, and the latter is the data on transition functions
to define a Spin structure on $E\oplus P_\rho\to M$, where $P_\rho\to M$ is the vector bundle associated to $P$ and
$\rho$.
\end{proof}
\begin{cor}
\label{w2_twisted_cor}
Let
\begin{equation}
\label{shear_cext}
	\shortexact*{\Z/2\Z}{\widetilde G}{G}{}
\end{equation}
be a central extension corresponding to the cohomology class $\omega\in H^2(BG;\Z/2\Z)$.
% [AD] resolved \textcolor{red}{G or BG, i.e., cohomology as topological space or group cohomology, I assume the latter}.
Let $\rho\colon
G\to\SSO(d)$ be a representation with $w_2(\rho) = \omega$. Then $\Spin\times_{\Z/2\Z}\widetilde G$-structures are
naturally equivalent to $(BG, (EG)_\rho)$-twisted Spin structures.
\end{cor}
\begin{proof}
By \cref{twisted_spin_pullback_lem}, we just have to produce a diagram~\eqref{twisted_spin_pullback}, where $H =
\Spin\times_{\Z/2\Z}\widetilde G$, and show it is a pullback diagram. The first step is to produce a map
$\phi\colon\Spin\times_{\Z/2\Z}\widetilde G\to H$. Let $g\in\widetilde G$ denote the nonzero element of the
$\Z/2\Z\subset\widetilde G$ defined by the central extension~\eqref{shear_cext}. By taking Spin covers, the map
\begin{equation}
	(\id, \rho)\circ p_1\colon\Spin\times_{\Z/2\Z}\widetilde G\longrightarrow\SSO
\end{equation}
lifts to a map
\begin{equation}
	\widetilde\phi\colon \Spin\times\widetilde G\longrightarrow\Spin;
\end{equation}
we will show it descends to a map $\phi$ with domain $\Spin\times_{\Z/2\Z}\widetilde G$ by showing that $(-1,
g)\in\Spin\times\widetilde G$ is sent to the identity; therefore $\widetilde\phi$ descends to the quotient
$\Spin\times_{\Z/2\Z}\widetilde G$.

Now to show $\widetilde\phi(-1, g) = 1$. Restricted to $\Spin\subset\Spin\times\widetilde G$, $\widetilde\phi$ is
the identity,\footnote{One way to think about this is that, downstairs on the map $\SSO\to\SSO$, this map takes the
direct sum with an identity matrix. The group $\SSO$ consists of equivalence classes of special orthogonal matrices
of any size, with $A$ and $A\oplus I_n$ identified. So this map, and its Spin cover, are the identity.} so
$\widetilde\phi(-1, 1) = -1$. And restricted to $\widetilde G$, this map is the Spin cover of $\rho\colon
G\to\SSO$, so by definition it sends $g\mapsto -1$. Therefore $\widetilde\phi(-1, g) = (-1)^2 = 1$ and we can
descend to $\Spin\times_{\Z/2\Z}\widetilde G$.

Finally, we have to check that~\eqref{twisted_spin_pullback} is a pullback diagram, i.e., that the pullback of
the two maps to $\SSO$ is $\Spin\times_{\Z/2\Z}\widetilde G$. We can again check on the two factors: it suffices to
pull back further to $\SSO$ and $G$ and check that we obtain $\Spin\to\SSO$ and $\widetilde G\to G$. For
$\Spin\to\SSO$, this is tautological: pull back $\Spin\to\SSO$ by the identity $\SSO\to\SSO$ and obtain
$\Spin\to\SSO$ again. For pulling back to $G$, and asking whether we get $\widetilde G\to G$, this is asking
precisely that $\widetilde G\to G$ is the Spin cover of $G$ for the representation $\rho$. This occurs precisely
when $w_2(\rho)$ is the class of the central extension~\eqref{shear_cext}.
\end{proof}
\begin{rem}
There are a few other approaches to proving shearing theorems. Freed-Hopkins~\cite[Section 10]{FH16} identify
classifying spaces of twisted Spin bordism groups with $B\Spin\times X$ by showing both spaces are homotopy
pullbacks of the same diagram; their approach is also used in~\cite{Cam17, WWZ19, DL20, WW20a, WW20b, Deb21,
Ste21}. Another approach uses a result of Beardsley~\cite{Bea17}, who works with Thom spectra directly.
\end{rem}

\vspace{0.3cm}

\noindent $\triangleright$ {\bf Physics intuition:} Let us apply the general discussion above to the simple specific example of Spin-$\Z/8\Z$, that will turn out to be useful in the following applications. First, we can define the central extension
\begin{equation}
1 \longrightarrow \Z/2\Z \longrightarrow \Z/8\Z \longrightarrow \Z/4\Z \longrightarrow 1\,,
\end{equation}
defined by $\text{Ext}(\Z/4\Z, \Z/2\Z) \simeq H^2 (B\Z/4\Z; \Z/2\Z) = \Z/2\Z$, whose non-trivial element we will denote by $\omega$. Now we choose a complex one-dimensional representation of $\Z/4\Z$, i.e., a map $\rho: \Z/4\Z \rightarrow \text{U}(1)$, such that its associated line bundle over $B\Z/4\Z$, which we also denote by $\rho$, has
\begin{align}
w_2 (\rho) = \omega \in H^2 (B\Z/4\Z; \Z/2\Z) \,.
\end{align}
These are precisely the representations with odd charge $q \in \{1,3\}$ in the sense that the map $\rho$ is defined as
\begin{align}
\rho\colon a \longmapsto e^{2 \pi i q a / 4} \,, \quad a \in \Z/4\Z \,.
\end{align}
We thus find that Spin-$\Z/8\Z$ structures can be described as twisted $(B\Z/4\Z, \rho)$-Spin structures. Using the Pontrjagin-Thom theorem we have
\begin{align}
\Omega^{\text{Spin-}\Z/8\Z}_{\ast} (\pt) \simeq \widetilde{\Omega}_{\ast}^{\text{Spin}} \big((B\Z/4\Z)^{\rho-2}\big) \,,
\end{align}
and we successfully reduced the computation to the theory of Spin bordisms to which we can apply the Adams spectral sequence. To further simplify the evaluation we can split the bordism groups to their prime parts, as we will discuss next. $\triangleleft$

\subsection{Splitting Spin bordism at odd primes}
\label{subsec:splitodd}

The homotopy theorists' philosophy of working one prime at a time makes for an effective way to compute Spin
bordism: at each prime $p$, Spin bordism decomposes as a direct sum of easier-to-understand homology theories, and
this information can be calculated separately and then put together into the final answer. For all odd primes $p$,
the $p$-local story is similar; the $2$-local description is more complicated, and we will go over it in the next
subsection.

The map $\Omega_*^\Spin\to\Omega_*^\SSO$ is a $p$-local equivalence, ultimately because $\Spin(n)\surj\SSO(n)$ is
a double cover.\footnote{In a little more detail: the short exact sequence $1\to\Z/2\Z\to\Spin\to\SSO\to 1$ gives us
a fiber bundle $\pi\colon B\Spin\surj B\SSO$ with fiber $B\Z/2\Z$. If $p$ is an odd prime, $H^*(B\Z/2\Z;\Z/p\Z)$ is
trivial, meaning the map $H^*(B\SSO;\Z/p\Z)\to H^*(B\Spin;\Z/p\Z)$ is an isomorphism, e.g.\ using the Serre spectral
sequence. The Thom isomorphism shows that the induced map of Thom spectra is an isomorphism in $\Z/p\Z$ cohomology,
and the $p$-local stable Whitehead theorem~\cite[Chapitre III, Théorème 3]{Ser53} turns this into the desired
$p$-local equivalence.} Brown-Peterson~\cite{BP66} show that for any odd prime $p$, $\Omega_*^\SSO$ splits
$p$-locally as a sum of copies of \term{Brown-Peterson homology}, denoted $\BP$.\footnote{There is one kind of
$\BP$ for each prime $p$.  Unfortunately, standard notation is for $p$ to be implicit.} The $\BP$-homology of a
point is the polynomial ring
\begin{equation}
	\BP_*(\pt)\cong \Z_{(p)}[v_1, v_2, v_3, \dotsc],
\end{equation}
where $v_i$ has degree $2(p^i-1)$. Work of Thom~\cite{ThomThesis}, Averbuh~\cite{Ave59}, and Milnor~\cite{Mil60}
showed that
\begin{equation}
	\Omega_*^\SSO\otimes\Z_{(p)}\cong \Z_{(p)}[x_4, x_8, x_{12}, \dotsc],
\end{equation}
where $x_{4i}$ is in degree $4i$, i.e.\ it is a bordism class of $4i$-dimensional manifolds.
So to determine the precise way in which oriented bordism splits into copies of
$\BP$-homology at $p$, one just has to compare these  two graded polynomial rings: begin with $\BP$ in degree
$0$, and then move upward, adding another copy of $\BP$ for each monomial not already accounted for.

For example, if $p = 3$, $v_1$  is in degree $4$, so the
$3$-local coefficient groups of both $\BP_*$ and $\Omega_*^\SSO$ have a $\Z_{(3)}$ in degrees $0$ and $4$, and nothing
else below degree $8$. That is, for $k\le 7$, $\Omega_k^\SSO(X)\otimes\Z_{(3)}\cong \BP_k(X)$.

Continuing in this way, we find that for $k\le 15$,
%(\TODO: say as homology theories or as spectra? or both?)
%\begin{equation}
%       \MTSO_{(3)}\overset\simeq\to \BP\vee \Sigma^8 \BP \vee \Sigma^{12}\BP \vee \dots
%\end{equation}
%and the three specified summands suffice for degrees $15$ and below. That is, if $k\le 15$,
\begin{equation}
\label{BP_SO_decomp}
	\Omega_k^\SSO(X)\otimes\Z_{(3)} \cong \BP_k(X) \oplus \BP_{k-8}(X) \oplus \BP_{k-12}(X).
\end{equation}
However, for $p\ge 5$, $v_1$ is in degree at least $8$, so $\BP_4(\pt) = 0$. Therefore we need an additional
$\BP_{k-4}$ summand to capture $\Omega_4^\SSO\cong\Z$:
\begin{equation}
	\Omega_k^\SSO(X)\otimes\Z_{(5)} \cong \BP_k(X) \oplus \BP_{k-4}(X) \oplus \BP_{k-8}(X) \oplus\dotsm
\end{equation}
In this paper, we only need $p = 3$.

\subsection{Splitting Spin bordism 2-locally}
\label{subsec:spliteven}

When $p = 2$, $\BP$ does not appear as a summand in Spin bordism. Instead, we obtain ordinary mod $2$ homology as
well as two generalized homology theories related to $K$-theory.
\begin{itemize}
	\item \term{Connective real $K$-theory}, denoted $\ko_*(X)$. This is a variant of real $K$-theory which
	vanishes in negative degrees. Its homotopy groups follow the usual Bott periodicity $\Z$, $\Z/2\Z$, $\Z/2\Z$,
	$0$, $\Z$, $0$, $0$, $0$, and then repeating, but starting in degree $0$ and continuing upwards only: if $k <
	0$, $\ko_k = 0$.  \item Another theory denoted $\ko\ang 2_*(X)$. This is built from $\KO$-theory in a similar
	way to $\ko$, but instead of asking for everything in degrees $k < 0$ to vanish, we ask for everything in
	degrees $k < 2$ to vanish, then shift down by $2$ so that the lowest nonzero homotopy group is in degree $0$.
	Thus the homotopy groups begin with $\Z/2\Z$ in degree $0$, then $0$, $\Z$, $0$, $0$, $0$, $\Z$, $\Z/2\Z$,
	$\Z/2\Z$, $0$, $\Z$, \dots.
	%Note that we shift the grading so the lowest nonzero homotopy group is in degree $0$.
	%\textcolor{red}{(TODO: technically we should say $\Sigma^{-2}\ko\ang 2$, I think. I'll make sure our notation is good.)}
\end{itemize}
\begin{thm}[Anderson-Brown-Peterson~\cite{ABP67}]
\label{ABPthm}
There is a $2$-primary isomorphism from Spin bordism to a sum of shifts of $\ko$-theory, $\ko\ang 2$-theory, and
mod $2$ homology. For $k < 16$, this isomorphism is of the form
\begin{equation}
\label{ABPdecomp}
	\Omega_k^\Spin(X)\otimes\Z_{(2)} \cong \paren{\ko_k(X) \oplus \ko_{k-8}(X) \oplus \ko\ang
	2_{k-10}(X)}\otimes\Z_{(2)}.
\end{equation}
\end{thm}
\begin{rem}
In principle, one can extract the explicit decomposition generalizing~\eqref{ABPdecomp} to $k\ge 16$ from
Anderson-Brown-Peterson's paper, but they do not give a closed form. Freed-Hopkins~\cite[Figure 7]{FH16} draw a
picture which may be helpful for seeing the next several summands.
\end{rem}
We therefore want to compute $\ko$-homology and $\ko\ang 2$-homology at the prime $2$. For $\ko$, there is a
convenient and well-established trick: the Adams spectral sequence is especially simple. We will discuss this in
Section \ref{the_Adams_SS}.

For $\ko\ang 2$-homology, we only need to know it in degrees $0$ and $1$, which is used in $10$- and
$11$-dimensional Spin bordism, and there it is especially simple.
\begin{lem}
\label{ko2_is_homology}
For $k\le 1$, and $X$ a space or connective spectrum, $\ko\ang 2_k(X)\cong H_k(X;\Z/2\Z)$.
\end{lem}
\begin{proof}
The Postnikov truncation map $\ko\ang 2\to \tau_{\le 1}(\ko\ang 2)$ is $1$-connected by definition,\footnote{A map
of spaces or spectra $X\to Y$ is \term{$n$-connected} if it induces an isomorphism on all homotopy groups in
degrees $n$ and below. This Postnikov truncation map is the universal example of such a $1$-connected map.}
and $\tau_{\le
1}(\ko\ang 2)$ has only one nonzero homotopy group, $\Z/2\Z$ in degree $0$, so it is equivalent to $H\Z/2\Z$. So
$\ko\ang 2$ has a $1$-connected map to $H\Z/2\Z$, which implies the theorem statement.
\end{proof}
In higher degrees, one can compute $\ko\ang 2$-homology similar to $\ko$-homology. See Freed-Hopkins~\cite[Section
D.1]{FH16} for how to do this with the Adams spectral sequence.
\begin{rem}
\label{ku_remark}
There is a complex analogue of $\ko$-theory, called $\ku$; it can be built from $\KU$ in an analogous way,
displaying Bott periodicity only in nonnegative degrees: the homotopy groups of $\ku$ are $\Z$, $0$, $\Z$, $0$,
$\Z$, \dots. Anderson-Brown-Peterson's decomposition of Spin bordism in terms of $\ko$-homology has an analogue:
Spin$^c$ bordism decomposes as a sum of shifts of $\ku$-homology and mod $2$ homology~\cite{ABP67, BG87a, BG87b}.

There are several useful maps involving $\ko$ and $\ku$, and the ways in which they interact occasionally help
address differentials or extension questions in spectral sequences computing $\ko$- or $\ku$-homology. We use four
maps.
\begin{itemize}
	\item The \term{complexification} homomorphism $c\colon\ko_n(X)\to\ku_n(X)$ is the connective cover of the map
	$\KO\to\KU$ which complexifies a real vector bundle.
	\item \term{Realification}, or forgetting from a complex vector bundle to a real one, induces a map $R\colon
	\ku_n(X)\to\ko_{n-2}(X)$. The degree shift is related to complex Bott periodicity.
	\item A map $\eta\colon \ko_n(X)\to\ko_{n+1}(X)$, given by multiplying by the nonzero class in
	$\ko_1(\pt)\cong\Z/2\Z$; this lifts in Spin bordism to taking the product with $S_p^1$.
	\item A map $b\colon\ku_n(X)\to\ku_{n+2}(X)$ which is the connective cover of the Bott periodicity isomorphism.
\end{itemize}
These maps interact quite nicely. We will use the following two facts.
\begin{itemize}
	\item The composition $R\circ b\circ c\colon\ko_n(X)\to\ko_n(X)$ is multiplication by $2$ (see~\cite[Theorem
	1]{Bru12}).
	\item The three maps $\eta$, $c$, and $R$ fit together into a long exact sequence~\cite[Section 12]{Bot69}\footnote{Bott only
considers the case of periodic $K$-theory, though his result implies the connective version we use. See
Bruner-Greenlees~\cite[Section 4.1.B]{bruner2010connective}.}
\begin{equation}
\label{etaCR}
% https://q.uiver.app/?q=WzAsNixbMCwwLCJcXGRvdHNiIl0sWzEsMCwiXFxrb19rKFgpIl0sWzIsMCwiXFxrb197aysxfShYKSJdLFszLDAsIlxca3Vfe2srMX0oWCkiXSxbNCwwLCJcXGtvX3trLTF9KFgpIl0sWzUsMCwiXFxkb3RzYiJdLFswLDEsIlIiXSxbMSwyLCJcXGV0YSJdLFsyLDMsImMiXSxbMyw0LCJSIl0sWzQsNSwiXFxldGEiXV0=
\begin{tikzcd}
        \dotsb & {\ko_n(X)} & {\ko_{n+1}(X)} & {\ku_{n+1}(X)} & {\ko_{n-1}(X)} & \dotsb
        \arrow["R", from=1-1, to=1-2]
        \arrow["\eta", from=1-2, to=1-3]
        \arrow["c", from=1-3, to=1-4]
        \arrow["R", from=1-4, to=1-5]
        \arrow["\eta", from=1-5, to=1-6]
\end{tikzcd}
\end{equation}
\end{itemize}
\end{rem}
%Things to quickly summarize. Including but not limited to
%\begin{itemize}
%	\item What do the pictures of Steenrod modules mean?
%	\item What does the $E_2$-page of the Adams spectral sequence mean? For example, the action of $\Ext(\Z/2)$
%	encoded via lines.
%	\item The $E_\infty$-page and the filtration, and the lift of the $\Ext(\Z/2)$-action to the $\ko_*$-action.
%\end{itemize}
%\subsection{$\eta$-invariants}
%
%\textcolor{red}{Not sure if we want to add something on $\eta$-invariants here. Maybe the Appendix, or an extension thereof suffices.}
%
%\TODO: possibly a subsection summarizing/stating and citing the theorems:
%\begin{enumerate}
%	\item Given such-and-such data, we can define an index in this dimension and an $\eta$-invariant in that
%	dimension
%	\item The index is a bordism invariant; when is the $\eta$-invariant a bordism invariant?
%	\item When (reality condition?), the index is in $2\Z$, so the $\eta$-invariant is in $\mathbb R/2\Z$
%	\item The eta invariant of a product factors as index times eta invariant
%	\item Eta invariants of lens spaces (maybe this would be later, when we discuss generators, so we can also
%	discuss eta invariants of lens space bundles)
%\end{enumerate}
%%%\dots but I'm not sure the correct level of generality/references for some of these statements

\section{What you need to know about the Adams spectral sequence}
	\label{the_Adams_SS}

Some of our computations use the Adams spectral sequence. This is a standard tool in homotopy theory, but not (yet) in
theoretical physics, so we summarize a few important facts here to allow the reader to follow along at a high
level. For a more complete introduction including some things we gloss over, see Beaudry-Campbell's excellent
paper~\cite{BC18}.

\vspace{0.3cm}

\noindent $\triangleright$ {\bf Physics intuition:}  Before we describe the Adams spectral sequence explicitly, we give a brief introduction to
spectral sequences themselves. Spectral sequences are used in order to determine certain generalized (co)homology
groups. They do this via increasingly precise approximations described by the pages of a spectral sequence. The
starting point is the second page, which usually can be obtained by ordinary (co)homology or something like it,
which often can be found in the literature or determined by other means (as we will describe for the Adams
spectral sequence for $ko$-homology below). Each page is bigraded, parametrized by two non-negative integers $s$
and $t$.  Moreover, on each page there are maps, the differentials $d_r$, that  map from the $(s,t)$
entry on the $r^{\text{th}}$ page to the $(s+r, t+r-1)$ entry.\footnote{Different spectral sequences have different
grading conventions; this is how differentials in the Adams spectral sequence work. The same is true for extension
questions.} These differentials fit together into chain
complexes, and the next page of the spectral sequence, i.e., the next better approximation,
is the homology of these chain complexes. Iterating this procedure one obtains the $\infty$ page of
the spectral sequence $E^{s,t}_{\infty}$. Now the desired (co)homology groups $h_n$ arise as an extension of the entries $E^{s,t}_{\infty}$ with $s + t = n$. In other words these entries form a filtration of the homology
groups of interest. We therefore see that the spectral sequences can be seen as an iterative procedure to
approximate and finally to determine generalized (co)homology groups. The specification of the action of
the differentials as well as the evaluation of the extension problems are highly non-trivial problems, for which
one can utilize several tricks, some of which will appear in the following. The Adams spectral sequence, for which
it is more challenging to determine the second page, has the advantage that many of the differentials and extension
questions can be
accessed more straightforwardly. For the Atiyah-Hirzebruch spectral sequence, for which the second page reduces to
ordinary (co)homology, it can be harder to formulate the differentials and extension questions. $\triangleleft$

\vspace{0.3cm}

Though there are more general versions of the Adams spectral sequence, we focus on the one for computing
$\ko$-homology, which takes the form
\begin{equation}
	E_2^{s,t} = \Ext_{\cA(1)}^{s,t}(H^*(X;\Z/2\Z), \Z/2\Z) \Longrightarrow \ko_*(X)_2^\wedge.
\end{equation} The notation $(\text{--})_2^\wedge$ means that this spectral sequence computes the $2$-completion of
the $\ko$-homology of $X$. ``$2$-completed'' and ``$2$-localized'' are not exactly the same thing, but for finitely
generated Abelian groups the $2$-completion and $2$-localization determine each other. All bordism groups we need
to worry about in our computations are finitely generated, so this distinction will not cause any problems for us.

%This is not the same thing as the
%$2$-local homotopy theory, but in all reasonable cases they are equivalent data.\footnote{A sufficient condition
%for $\ko_*(X)_2^\wedge$ to determine $\ko_*(X)\otimes\Z_{(2)}$ is that $X$ admits a CW structure with finitely many
%cells in each dimension.}

The $E_2$-page is a little more complicated: for any space or spectrum $X$, $H^*(X;\Z/2\Z)$ is a module over the
\term{Steenrod algebra} $\cA$ of stable cohomology operations. $\cA(1)$ is the subalgebra generated by $\Sq^1$ and
$\Sq^2$. $\Ext$ is a functor classifying extensions of $\cA(1)$-modules. The purpose of this section is to go over
how to define and work with these algebraic objects. In the cases we consider in this paper, the $\cA(1)$-action on
$H^*(X;\Z/2\Z)$ and the corresponding Ext groups can more or less be looked up, but the reader who wants to learn
more on how to compute them should consult Beaudry-Campbell~\cite{BC18}.

\vspace{0.3cm}

\noindent $\triangleright$ {\bf Physics intuition:}  As discussed above, the Adams spectral sequence determines groups which appear in the full
bordism group at prime 2. The starting point of the Adams spectral sequence, its second page, relies on extensions
of the mod 2 cohomology of the space or spectrum under investigation. Moreover it is sensitive to an algebra
structure of this extension generated by the Steenrod operators Sq$^1$ and Sq$^2$. The big advantage of the Adams
spectral sequence compared to other approaches, such as for example the Atiyah-Hirzebruch spectral sequence, is
that the extra structure present determines a large number of differentials and extension questions for free.
Analogous computations are notoriously difficult in many of the alternative spectral sequences.

While some of the pieces in the Adams spectral sequence as well as other spectral sequences have a physical interpretations, see e.g., \cite{Diaconescu:2000wy, Maldacena:2001xj, McNamara:2022lrw}, at this point we do not have a good intuition concerning the full picture. In the following we thus mainly restrict to a mathematical description of the involved techniques and their application for the cases under consideration and leave a physical analysis of the Adams spectral sequence for future work. See however \cite{McNamara:2022lrw} for a physical interpretation of the $E_1$ page of the Adams spectral sequence. $\triangleleft$

\subsection{The Steenrod algebra}
\label{ss:steenrod}

The goal of this subsection is to explain what the Steenrod algebra $\cA$ and its subalgebra $\cA(1)$ are, and to
provide some techniques for computing $\cA(1)$-module structures on mod $2$ cohomology. See~\cite[Section 3]{BC18} for
more information.

A \term{stable cohomology operation} (in mod $2$ cohomology) is a natural transformation of Abelian groups
$H^*(\bl;\Z/2\Z)\to H^{*+k}(\bl;\Z/2\Z)$ which commutes with the suspension isomorphism. Under composition, the set of stable cohomology operations is a graded $\Z/2\Z$-algebra, denoted $\cA$ and called
the \term{Steenrod algebra}; its action on cohomology makes the mod $2$ cohomology of any space or spectrum into an
$\cA$-module. Steenrod showed that $\cA$ is generated by operators called \term{Steenrod squares} $\Sq^k\colon
H^*(\bl;\Z/2\Z)\to H^{*+k}(\bl;\Z/2\Z)$ subject to some relations. In this paper, we only need $\cA(1)$, the subalgebra
of $\cA$ generated by $\Sq^1$ and $\Sq^2$.

\vspace{0.3cm}

\noindent $\triangleright$ {\bf Physics intuition:} We offer a potential physical interpretation of $\Sq^1$ and $\Sq^2$. As discussed in \cite{Diaconescu:2000wy}, the
Wu formula characterizes many Steenrod squares in the cohomology of a manifold by describing certain information
associated to submanifolds and their normal bundles, with $\Sq^i$ roughly given by $w_i (N) \cup\text{--}$. Since
we are only concerned with Spin and related structures, we only need to worry about $i = 1,2$.
%It should be pretty clear that our discussion only needs $i = 1,2$ since we are only sensitive on generalized Spin structures, i.e., things involving $w_1$ and $w_2$ and not above. Going up the Whitehead tower to, e.g., String structure also $\mathcal{A}(1)$ would not suffice.
This suggests that the $\cA(1)$-module structure on the mod $2$ cohomology of a space knows about the realization
of fermions in the space $X$ as well as its subspaces, and perhaps there is even a direct interpretation in terms
of possible brane worldvolumes. $\triangleleft$

\vspace{0.3cm}

%\textcolor{red}{I think there might be a nice physical interpretation of what is going on here. As discussed in \cite{Diaconescu:2000wy} the Steenrod squares can be interpreted as describing certain information associated to submanifolds as well as their normal bundle, with Sq$^i$ roughly given by $w_i (N) \cup$. It should be pretty clear that our discussion only needs $i = 1,2$ since we are only sensitive on generalized Spin structures, i.e., things involving $w_1$ and $w_2$ and not above. Going up the Whitehead tower to, e.g., String structure also $\mathcal{A}(1)$ would not suffice. From this I would say that the $\mathcal{A}(1)$ module structure of a space knows about the realization of fermions in the space $X$ as well as its subspaces. There might even be a direct interpretation in terms of possible brane worldvolumes. What I do not understand is the implications of the extensions...}

The definitions of Steenrod squares are fairly abstract, making it difficult to compute the $\cA$- or
$\cA(1)$-module structures on cohomology from the definition. Instead, one can generally determine these module
structures using a few key properties. Let $X$ be a space.
\begin{enumerate}
	\item\label{when_cup} If $x\in H^k(X;\Z/2\Z)$, then $\Sq^k(x) = x^2$.
	\item\label{only_spaces} If $x\in H^n(X;\Z/2)$ and $n < k$, then $\Sq^k(x) = 0$.
	\item (Naturality) given a map $f\colon X\to Y$ and $x\in H^n(Y;\Z/2\Z)$, $\Sq^k(f^*x) = f^*\Sq^k(x)$.
	\item (Cartan formula) If $X$ is a space and $x,y\in H^*(X;\Z/2\Z)$, then
	\begin{equation}
	\label{cartfor}
		\Sq^k(xy) = \sum_{i+j=k} \Sq^i(x)\Sq^j(y).
	\end{equation}
	Here we allow $i,j = 0$: $\Sq^0 = \id$.
\end{enumerate}
These properties actually uniquely characterize the Steenrod squares.
\begin{rem}
A few caveats: for spectra, most of these are false, in part because the cohomology groups of spectra do not carry
a cup product. Steenrod squares are still natural in the above sense, though.

Another unrelated caveat is that the Cartan formula shows that the $\cA$-module structure on the cohomology of a
space is related to the cup
product, but the cup product does not make cohomology into an $\cA$-algebra.
\end{rem}
Here are a few other useful facts about Steenrod squares. First, a formula for the Steenrod squares of Thom
spectra (recall the introduction of Thom spectra and the Thom isomorphism $U$ from \cref{Thom_exp}).
%\TODO: where did we say the Thom isomorphism \textcolor{red}{I mention them in the physics part of Section 7.2, but I am not sure if we define them.}
\begin{prop}
\label{Steenrod_Thom}
Let $V\to X$ be a vector bundle and $Ux\in \tH^*(X^V;\Z/2\Z)$. Then
\begin{equation}
	\Sq^k(Ux) = \sum_{i+j=k} Uw_i(V)\Sq^j(x).
\end{equation}
\end{prop}
That is, if $U$ denotes the Thom class, i.e.\ the image of $1\in H^0(X;\Z/2\Z)$ under the Thom isomorphism,
then $\Sq^k(U) = Uw_k(V)$, and so one can compute Steenrod squares
of $Ux$ using the Cartan formula. This property is specific to Thom spectra.

Second, the Wu formula computes Steenrod squares of Stiefel-Whitney classes of vector bundles:
\begin{equation}
\label{wu_formula}
	\Sq^k(w_j(V)) = \sum_{i=0}^k \binom{j-k + i-1}{i} w_{k-i}(V)w_{j+i}(V).
\end{equation}
So computing Steenrod squares is relatively routine:
\begin{itemize}
	\item \Cref{Steenrod_Thom} computes Steenrod squares for a Thom spectrum $X^V$ in terms of the Steenrod squares
	for $X$ and the Stiefel-Whitney classes of $V$.
	\item The Cartan formula~\eqref{cartfor} determines the Steenrod squares of all of $H^*(X;\Z/2\Z)$ in terms of
	the Steenrod squares of the generators of the cohomology ring.
	\item For a specific generator $x$, one can deduce some Steenrod squares automatically from
	items~\eqref{when_cup} and~\eqref{only_spaces} above. Otherwise, one can try to realize $x = w_n(V)$ for some
	vector bundle $V$ and use the Wu formula, or use naturality: find a space $Y$ that is easier to understand,
	together with a map $f\colon X\to Y$ such that $x = f^*(y)$, where we understand $\Sq^k(y)$.
\end{itemize}
Often one can gain a complete understanding of the $\cA(1)$-module structure on the mod $2$ cohomology of a space
or Thom spectrum by combining these techniques with computations already in the literature. For example, this
suffices for all computations we need in this paper.

\begin{exm}
\label{cyclic_steenrod}
Let $k > 1$; then $H^*(B\Z/2^k\Z;\Z/2\Z)\cong\Z/2\Z[x, y]/(x^2)$,
% resolved [AD] :\textcolor{red}{What is $\Z/2^k$? is it $\mathbb{Z}/ 2^k$?},
where the degrees of $x$ and $y$ are respectively $\abs x = 1$ and $\abs y = 2$~\cite[Proposition 4.5.1]{CTVZ03}. 
We will determine the $\cA(1)$-module 
structure on this cohomology ring. Because $x$ is degree $1$,~\eqref{only_spaces} implies $\Sq^i(x) = 0$ for $i\ge
2$. $\Sq^0$ is the identity, so $\Sq^0(x) = x$, and since $\abs x = 1$, $\Sq^1(x) = x^2 = 0$ by~\eqref{when_cup}.

For $y$, we know $\Sq^0(y) = y$, $\Sq^2(y) = y^2$ by~\eqref{when_cup}, and $\Sq^i(y) = 0$ for $i \ge 3$
by~\eqref{only_spaces}, so only $\Sq^1$ is left. In \cref{Z4_coh}, we prove
that if $\rho\colon \Z/2^k\Z\to\UU (1)$ denotes the
standard one-dimensional complex representation of $\Z/2^k\Z$ by rotations and $V_\rho\coloneqq
E\Z/2^k\Z \times_{\Z/2^k\Z} \C{} \rightarrow B\Z/2^k\Z$ is the associated complex line bundle, then $y = w_2(V_\rho)$.
%It
%suffices to show that $V_\rho$ is not Spin: that implies $w_2(V_\rho)\ne 0$, and $y$ is the only nonzero element of
%$H^2(B\Z/2^k\Z;\Z/2\Z)$. This bundle is Spin iff $\rho\colon \Z/2^k\Z \to \UU (1)$ lifts across $\Spin (2) \to \SSO (2) = \UU (1)$,
%but one can check that $\rho$ does not lift.
Because $V_\rho$ is a complex line bundle, it is the pullback of the
tautological bundle $\mathcal T_{\UU(1)} \to B\UU (1)$ by some map $f\colon B\Z/2^k\Z \to B\UU (1)$.

Because both Steenrod squares and Stiefel-Whitney classes are natural under pullback, $\Sq^1(y) =
f^*\Sq^1 \big(w_2(\mathcal T_{\UU(1)})\big)$, and $\Sq^1\big(w_2(\mathcal T_{\UU(1)})\big)\in H^3(B\UU (1);\Z/2\Z) = 0$, so
$\Sq^1(y) = f^*(0) = 0$, finishing our calculation of Steenrod squares of $y$.

For a general element of $H^*(B\Z/2^k\Z;\Z/2\Z)$ we can use the Cartan formula. For example, $\Sq^1(xy) = \Sq^1(x)y +
x\Sq^1(y) = 0$, and
\begin{equation}
	\Sq^2(xy) = \Sq^2(x)y + \Sq^1(x)\Sq^1(y) + x\Sq^2(y) = 0 + 0 + xy^2.
\end{equation}
We will use this example in the proof of \cref{ko_spin_z8}.
\end{exm}
See Beaudry-Campbell~\cite[Section 3.4]{BC18} for some more worked examples.

It is traditional to describe $\cA(1)$-modules pictorially, rather than just algebraically: one uses a straight
vertical line from $x$ to $y$ to indicate $\Sq^1(x) = y$, and a curve from $z$ to $w$ to indicate $\Sq^2(z) = w$.
This convention goes back at least to Mahowald-Milgram~\cite[Section 4]{MM76} and is standard in the field. We give some
examples in \cref{cell_diagrams}.
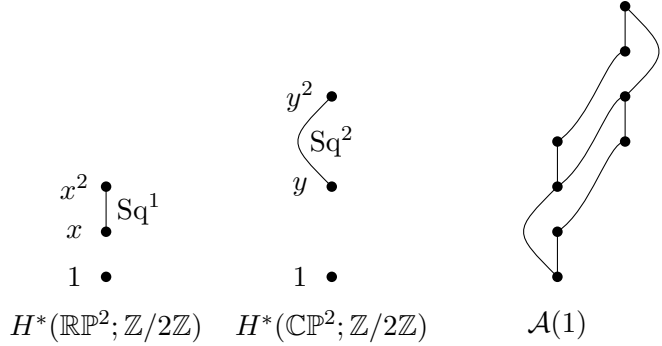
\begin{figure}[h!]
\centering
\begin{tikzpicture}[scale=0.6, every node/.style = {font=\small}]
% TODO: degree scale
\tikzpt{0}{0}{$1$}{};
\tikzpt{0}{1}{$x$}{};
\tikzpt{0}{2}{$x^2$}{};
\node[right] at (0, 1.5) {$\Sq^1$};
\sqone(0, 1);
\node[below] at (0, -0.5) {$H^*(\RP^2;\Z/2\Z)$};

\tikzpt{5}{0}{$1$}{};
\tikzpt{5}{2}{$y$}{};
\tikzpt{5}{4}{$y^2$}{};
\node at (5, 3) {$\Sq^2$};
\sqtwoL(5, 2);
\node[below] at (5, -0.5) {$H^*(\CP^2;\Z/2\Z)$};

\Aone{10}{0}{};
\node[below] at (10, -0.5) {$\cA(1)$};
\end{tikzpicture}
\caption{Drawing the $\cA(1)$-module structures on mod $2$ cohomology. Left:
$H^*(\RP^2;\Z/2\Z)\cong\Z/2\Z[x]/(x^3)$ with $\abs x = 1$; the vertical line indicates that $\Sq^1(x) = x^2$.
Center: $H^*(\CP^2;\Z/2\Z)\cong\Z/2\Z[y]/(y^3)$ with $\abs y = 2$; the curve indicates $\Sq^2(y) = y^2$.
Right: $\cA(1)$ as a module over itself. All straight lines denote $\Sq^1$-actions and all curves denote
$\Sq^2$-actions.
}
\label{cell_diagrams}
\end{figure}

%\textcolor{red}{TODO: drawing Steenrod squares is important for understanding the figures, so I want to say it in a way that
%doesn't get lost in the details.}

\subsection{Determining Ext groups}
\label{ss:ext}

The input to the Adams spectral sequence is the Ext groups of the $\cA(1)$-module structure on cohomology.
If $M$ is an $\cA(1)$-module, $\Ext_{\cA(1)}^{s,t}(M, \Z/2\Z)$ has a lot of additional structure, and the most
important takeaway of this section is what that structure is and how we display it in Ext charts; it will be
important in determining differentials and extensions in the Adams spectral sequence.

These
Ext groups are Abelian groups of equivalence classes of certain extensions of $\cA(1)$-modules, but one rarely if
ever needs to work directly with the definition. In examples, the Ext groups needed for an Adams spectral
sequence calculation of twisted Spin bordism can usually be looked up or calculated using a computer program, and
if one needs to compute by hand, it is possible to use a long exact sequence to quickly reduce to computations in the
literature. As always, Beaudry-Campbell~\cite[Section 4]{BC18} is a good reference for additional information that we
do not provide here.

Since $\cA(1)$ is a graded algebra, we will work with graded $\cA(1)$-modules. If $M$ is an $\cA(1)$-module and
$k\in\Z$, we let $\Sigma^kM$ denote the $\cA(1)$-module which is the same underlying ungraded module, but with the
gradings of all elements increased by $k$. It is common to think of $\Sigma^k$ as $k$ applications of a ``shift
operator'' $\Sigma$, and as such one writes $\Sigma M\coloneqq\Sigma^1 M$.

If $M$ and $N$ are $\cA(1)$-modules, $\Ext_{\cA(1)}^{s,t}(M, N)$ is an Abelian group of equivalence classes of
extensions of the form
% [AD] resolved \textcolor{red}{Did we define what $\Sigma$ is?}
\begin{equation}
\label{yoneda_extension}
% https://q.uiver.app/?q=WzAsNyxbMCwwLCIwIl0sWzEsMCwiXFxTaWdtYV50IE4iXSxbMiwwLCJQXzEiXSxbMywwLCJcXGRvdHNiIl0sWzQsMCwiUF9zIl0sWzUsMCwiTSJdLFs2LDAsIjAiXSxbMCwxXSxbMSwyXSxbMiwzXSxbMyw0XSxbNCw1XSxbNSw2XV0=
\begin{tikzcd}
	0 & {\Sigma^t N} & {P_1} & \dotsb & {P_s} & M & 0,
	\arrow[from=1-1, to=1-2]
	\arrow[from=1-2, to=1-3]
	\arrow[from=1-3, to=1-4]
	\arrow[from=1-4, to=1-5]
	\arrow[from=1-5, to=1-6]
	\arrow[from=1-6, to=1-7]
\end{tikzcd}
\end{equation}
where $P_1,\dotsc,P_s$ are $\cA(1)$-modules and all maps are $\cA(1)$-module maps. Two of these sequences $(P_i)$,
$(P_i')$ are equivalent if there is a third sequence $(Q_i)$ and $\cA(1)$-module maps $(P_i)\gets (Q_i)\to (P_i')$
commuting with the maps in $(P_i)$, $(P_i')$, and $(Q_i)$ such that the maps on $M$ and on $\Sigma^t N$ are the
identity (see~\cite[\href{https://stacks.math.columbia.edu/tag/06XP}{Tag 06XP}, Definition 06XT]{stacks-project}).
Defining the Abelian group structure on $\Ext_{\cA(1)}^{s,t}(M, N)$, called \term{Baer sum}, is a little involved,
and we will not need to know the definition in this paper; we point the interested reader to Mac Lane~\cite[\S
III.5]{ML63}. This perspective on Ext is due to Yoneda~\cite[Section 3.4]{Yon54}. For $s = 0$, there is an easier
description of Ext: $\Ext_{\cA(1)}^{0,t}(M, N) = \Hom_{\cA(1)}(M, \Sigma^t N)$.
\begin{rem}
Given an extension of the form~\eqref{yoneda_extension}, which represents a class in $\Ext^{s,t}(M, N)$, together
with a map $L\to M$, we can pull back~\eqref{yoneda_extension} by taking the fiber products of the maps $P_i\to M$
with $L\to M$, and thereby obtain an extension representing a class in $\Ext^{s,t}(L, N)$. That is, when the second
argument is held constant, Ext is a contravariant functor.
\label{ext_is_contravariant}
\end{rem}

There is a product on Ext groups called the \term{Yoneda product}~\cite[Section 4]{Yon54} with signature
\begin{equation}
\label{Yoneda_product}
	\Ext_{\cA(1)}^{s_1, t_1}(L, N)\times \Ext_{\cA(1)}^{s_2, t_2}(N, M)\longrightarrow \Ext_{\cA(1)}^{s_1+s_2,
	t_1+t_2}(L, M).
\end{equation}
The idea of this product is that, given two extensions of the form~\eqref{yoneda_extension} with matching
ends, one can glue them together. See~\cite[Section 4.2]{Yon54} or~\cite[Section 4.2]{BC18} for the full story. Letting $L =
N = M = \Z/2\Z$, the Yoneda product makes $\Ext_{\cA(1)}^{*,*}(\Z/2\Z, \Z/2\Z)$ into a $(\Z\times\Z)$-graded commutative
$\Z/2\Z$-algebra,
%\footnote{This algebra is also denoted $H^{*,*}(\cA(1))$, thought of as the cohomology of the
%algebra $\cA(1)$. This is by analogy with the fact that the group cohomology of a group $G$ with coefficients in a
%commutative ring $A$ is naturally isomorphic to $\Ext_{A[G]}^*(A, A)$.}
and letting $L = N = \Z/2\Z$, the Yoneda
product makes $\Ext_{\cA(1)}^{*,*}(M, \Z/2\Z)$ into a $(\Z\times\Z)$-graded $\Ext_{\cA(1)}^{*,*}(\Z/2\Z, \Z/2\Z)$-module.
We will refer to these as $\Ext(\Z/2\Z)$ and $\Ext(M)$ respectively. The $\Ext(\Z/2\Z)$-module structure on $\Ext(M)$
provides quite a bit of information in the Adams spectral sequence; in cases where the Adams spectral sequence is
more powerful than the Atiyah-Hirzebruch spectral sequence, this module structure is often the reason.

Here is what $\Ext(\Z/2\Z)$ looks like.
\begin{thm}[{Liulevicius~\cite[Theorem 3]{Liu62}}]
\label{ext_Z2}
\[\Ext(\Z/2\Z)\cong\Z/2\Z[h_0, h_1, v, w]/(h_0h_1, h_1^3, vh_1, h_0^2w - v^2),\]
where $\deg(h_0) = (1, 1)$, $\deg(h_1) = (1, 2)$,
% [AD] resolved \textcolor{red}{This does not match the picture below, is it $(1,2)$?},
$\deg(v) = (3, 7)$, and $\deg(w) = (4, 12)$.
\end{thm}
Ext groups are traditionally displayed in a diagram, and this will make the algebra structure in \cref{ext_Z2}
easier to visualize when we draw it in \cref{pic_ext_Z2}. The standard conventions for drawing Ext of $\cA(1)$-modules are as follows.  Each $\Z/2\Z$
summand in $\Ext^{s,t}$ is displayed as a dot at coordinates $(t-s, s)$; $t-s$ is called the \term{topological
degree} and $s$ the \term{filtration}. With this convention, the action of $h_0$ does not change the topological degree and
increases the filtration by $1$, so we use a vertical line connecting two dots to indicate that $h_0$ carries the
lower-filtration $\Z/2\Z$ isomorphically to the higher-filtration $\Z/2\Z$.  Likewise, the action of $h_1$ increases
both the topological degree and the filtration by $1$, so we depict it with a diagonal line. It is less common to
draw the actions of $v$ and $w$ in Ext charts. With these conventions, we can draw $\Ext(\Z/2\Z)$ in
\cref{pic_ext_Z2}.
\begin{figure}[h!]
\centering
\begin{sseqdata}[name=BottPer, classes={fill, show name={left=0.1em}}, xrange={0}{10}, yrange={0}{6}, scale=0.6,
x label = {$\displaystyle{s\uparrow \atop t-s\rightarrow}$},
y axis gap=2em,
x label style = {font = \small, xshift = -24.7ex, yshift=7ex}]
\class[name = 1](0, 0)
\class[name = h_0](0, 1)
\class[name= h_0^2](0, 2)\AdamsTower{}
\structline(0, 0)(0, 1)
\structline(0, 1)(0, 2)
\class[name=h_1, show name={right=0.1em}](1, 1)
	\structline(0, 0)(1, 1)
\class[name=h_1^2, show name={right=0.1em}](2, 2)
	\structline(1, 1)(2, 2)

\class[name=v](4, 3)
\class[name=h_0v](4, 4)
\class[name=h_0^2v](4, 5)\AdamsTower{}
\structline(4, 3)(4, 4)
\structline(4, 4)(4, 5)
\class[name=w](8, 4)
\class[name=h_0w](8, 5)\AdamsTower{}
\structline(8, 4)(8, 5)
\class[name=h_1w, show name={right=0.1em}](9, 5)
\structline(8, 4)(9, 5)
\class(10, 6)\structline(9, 5)(10, 6)
%\classoptions["h_0"](0, 1)
%\classoptions["h_1"](1, 1)
%\classoptions["v"](4, 3)
%\classoptions["w"](8, 4)
\end{sseqdata}
\printpage[name=BottPer, page=2]
\caption{$\Ext_{\cA(1)}^{*,*}(\Z/2\Z, \Z/2\Z)$, with some elements labeled. Multiplication by $w$ is injective, and
the diagram continues in a similar way to the right. One can see some of the relations in \cref{ext_Z2} using this
diagram: for example, $h_1^2\ne 0$ and $h_1^3 = 0$, because traveling two steps from $1$ on a diagonal line leads
to a nonzero element, but traveling one more step leads to an empty square, corresponding to the zero group. The
relations $h_0h_1 = 0$ and $vh_1 = 0$ follow from the absence of a diagonal line ($h_1$-action) out of $h_0$ and
$v$. The relation $h_0^2w = v^2$ is not depicted in this diagram.}
\label{pic_ext_Z2}
\end{figure}

\begin{rem}[Using preexisting calculations of Ext groups]
There are many calculations of Ext groups of $\cA(1)$-modules in the literature, and often they suffice for
whatever calculations one might need; this is the case for this paper. Beaudry-Campbell~\cite{BC18} provide
calculations of $\Ext$ of many of the most common $\cA(1)$-modules; some more calculations can be found
in~\cite{GMM68, Gia73, Dav74, AP76, Gia76, BG97, Cam17, WW19, Bak20, DL20, Deb21}. There are also computer programs
for computing Ext written by Bruner~\cite{Bru18} and Chatham-Chua~\cite{CC21}.
\end{rem}
The last thing we should mention here is that when it \emph{is} necessary to know the Ext groups of a module and
you cannot find them in the literature, one good trick is to use that a short exact sequence of $\cA(1)$-modules
$0\to L\to M\to N\to 0$ induces a long exact sequence in Ext of the form
\begin{equation}
\label{ExtLES}
% https://q.uiver.app/?q=WzAsNixbMCwwLCJcXGRvdHNiIl0sWzEsMCwiXFxFeHRfe1xcY0EoMSl9XntzLHR9KE4sIFxcWi8yKSJdLFsyLDAsIlxcRXh0X3tcXGNBKDEpfV57cyx0fShNLCBcXFovMikiXSxbMywwLCJcXEV4dF97XFxjQSgxKX1ee3MsdH0oTCwgXFxaLzIpIl0sWzQsMCwiXFxFeHRfe1xcY0EoMSl9XntzKzEsdH0oTiwgXFxaLzIpIl0sWzUsMCwiXFxkb3RzYiJdLFswLDFdLFsxLDJdLFsyLDNdLFszLDQsIlxcZGVsdGEiXV0=
\begin{tikzcd}
	\dotsb & {\Ext^{s,t}(N)} & {\Ext^{s,t}(M)} & {\Ext^{s,t}(L)} & {\Ext^{s+1,t}(N)} & \dotsb
	\arrow[from=1-1, to=1-2]
	\arrow[from=1-2, to=1-3]
	\arrow[from=1-3, to=1-4]
	\arrow["\delta", from=1-4, to=1-5]
	\arrow[from=1-5, to=1-6]
\end{tikzcd}
\end{equation}
The proof is purely formal, relying on a different but equivalent
characterization of Ext as the derived functors of Hom. It is common to compute with this long exact sequence by
drawing $\Ext(L)$ and $\Ext(N)$ both on the same chart; the boundary map increases the filtration by $1$ and lowers
the topological degree by $1$. All maps in~\eqref{ExtLES} commute with the $\Ext(\Z/2\Z)$-action, which typically
reduces the computation of the boundary map to a small number of calculations. We will give an example below;
there are many more examples in~\cite{BC18}, and a few more in~\cite{WW19, Deb21}.
\begin{exm}
\label{LES_ext_exm}
Let $C\eta\coloneqq\Sigma^{-2}\widetilde H^*(\CP^2;\Z/2\Z)$: it consists of two $\Z/2\Z$ summands in degrees $0$ and
$2$, joined by a nonzero $\Sq^2$-action. One could equivalently define $C\eta = \cA(1)/(\Sq^1, \Sq^3)$. There is a
short exact sequence of $\cA(1)$-modules
\begin{subequations}
\label{Ceta_ext}
\begin{equation}
	% https://q.uiver.app/?q=WzAsNSxbMCwwLCIwIl0sWzEsMCwiXFxTaWdtYVxcWi8yIl0sWzIsMCwiQ1xcZXRhIl0sWzMsMCwiXFxaLzIiXSxbNCwwLCIwIl0sWzAsMV0sWzEsMl0sWzIsM10sWzMsNF1d
\begin{tikzcd}
	0 & {\textcolor{BrickRed}{\Sigma^2\Z/2\Z}} & C\eta & {\textcolor{MidnightBlue}{\Z/2\Z}} & 0,
	\arrow[from=1-1, to=1-2]
	\arrow[from=1-2, to=1-3]
	\arrow[from=1-3, to=1-4]
	\arrow[from=1-4, to=1-5]
\end{tikzcd}
\end{equation}
which looks like this:
\begin{equation}
\begin{gathered}
\begin{tikzpicture}[scale=0.6]
\sqtwoR(0, 0);
\begin{scope}[BrickRed]
	\tikzpt{-4}{2}{}{};
	\tikzpt{0}{2}{}{};
	\draw[thick, ->] (-3.5, 2) -- (-0.5, 2);
\end{scope}
\begin{scope}[MidnightBlue]
	\tikzpt{0}{0}{}{rectangle, minimum size=3.5pt};
	\tikzpt{4}{0}{}{rectangle, minimum size=3.5pt};
	\draw[thick, ->] (0.5, 0) -- (3.5, 0);
\end{scope}
\node[below] at (-4, -0.3) {$\Sigma^2\Z/2\Z$};
\node[below] at (0, -0.3) {$C\eta$};
\node[below] at (4, -0.3) {$\Z/2\Z$};
\end{tikzpicture}
\end{gathered}
\end{equation}
\end{subequations}
We will use the induced long exact sequence in Ext to compute $\Ext(C\eta)$. This can be looked up (e.g.\
\cite[Figure 22]{BC18}), but in the proof of \cref{MM_ceta}, we will need to know the action of
$v\in\Ext(\Z/2\Z)$, which is not discussed in the cited reference.

To run the long exact sequence, we need to know $\Ext(\textcolor{BrickRed}{\Sigma^2\Z/2\Z})$ and
$\Ext(\textcolor{MidnightBlue}{\Z/2\Z})$. The latter is \cref{ext_Z2}; for the former, shifting an $\cA(1)$-module
simply shifts the $t$-grading on its Ext by the same amount. Therefore we can draw the long exact sequence in Ext
associated to~\eqref{Ceta_ext} in \cref{Ceta_LES_pic}, top.
\begin{figure}[h!]
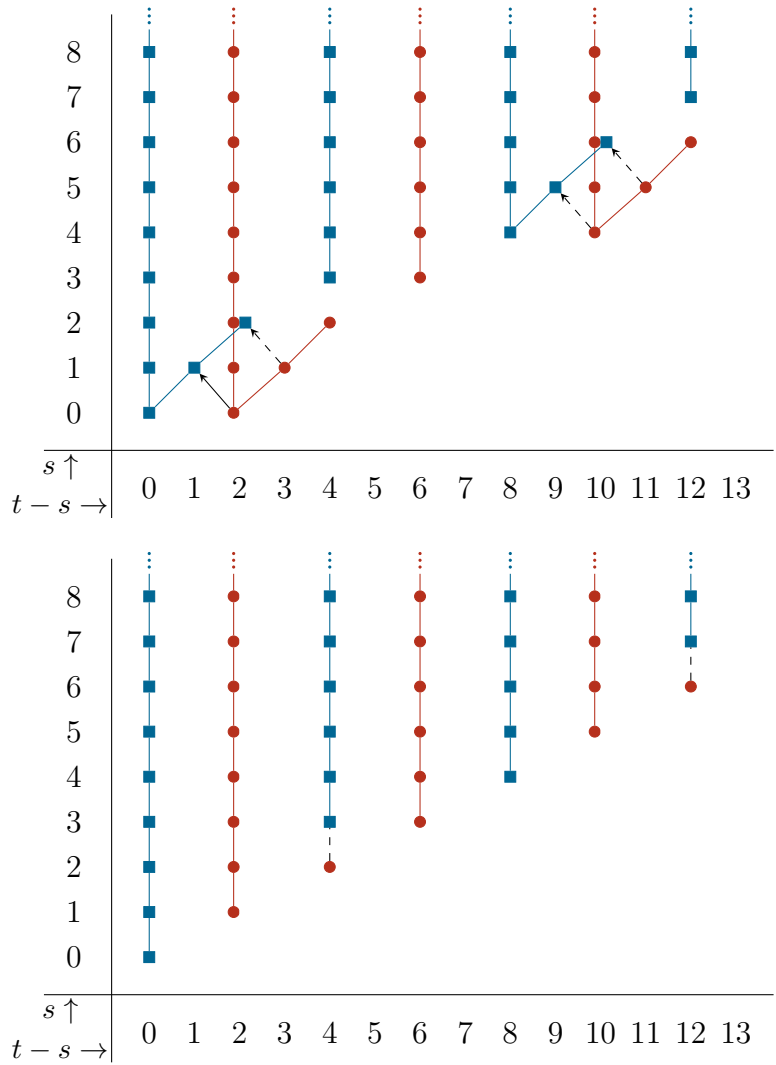

\centering
\begin{subfigure}[c]{0.6\textwidth}
\begin{sseqdata}[name=Cetaext, classes=fill, Adams grading, >=stealth, xrange={0}{13}, yrange={0}{8}, scale=0.6, x
label = {$\displaystyle{s\uparrow \atop t-s\rightarrow}$}, x label style = {font = \small, xshift = -28ex,
yshift=5.5ex}]
\begin{scope}[BrickRed]
	\class(2, 0)\AdamsTower{}
	\class(3, 1)\structline(2, 0)(3, 1)
	\class(4, 2)\structline(3, 1)(4, 2)
	\class(6, 3)\AdamsTower{}
	\class(10, 4)\AdamsTower{}
	\class(11, 5)\structline(10, 4)(11, 5)
	\class(12, 6)\structline(11, 5)(12, 6)
\end{scope}
\begin{scope}[white]
	\class(2, 0)
	\class(2, 1)
	\class(2, 3)\AdamsTower{}
	\class(10, 4)
	\class(10, 5)
	\class(10, 7)\AdamsTower{}
\end{scope}
\begin{scope}[MidnightBlue, rectangle]
	\class(0, 0)\AdamsTower{}
	\class(1, 1)\structline(0, 0)(1, 1)
	\class(2, 2)\structline(1, 1)(2, 2, -1)
	\class(4, 3)\AdamsTower{}
	\class(8, 4)\AdamsTower{}
	\class(9, 5)\structline(8, 4)(9, 5)
	\class(10, 6)\structline(9, 5)(10, 6, -1)
	\class(12, 7)\AdamsTower{}
\end{scope}
\d1(2, 0)
\d[dashed]1(3, 1)(2, 2, -1)
\d[dashed]1(10, 4)
\d[dashed]1(11, 5)(10, 6, -1)
\end{sseqdata}
\printpage[name=Cetaext, page=1]
\end{subfigure}\\\vspace{0.3cm}
\begin{subfigure}[c]{0.6\textwidth}
\begin{sseqdata}[name=Cetaext, update existing]
	\structline[dashed](4, 2)(4, 3)
	\structline[dashed](12, 6)(12, 7)
\end{sseqdata}
\printpage[name=Cetaext, page=2]
\end{subfigure}
\caption{Top: the long exact sequence in Ext associated to the short exact sequence of
$\cA(1)$-modules~\eqref{Ceta_ext}. The solid and dashed arrows are boundary maps.
Bottom: $\Ext(C\eta)$ as calculated by this long exact sequence. The dashed lines indicate
$h_0$-actions not visible to the long exact sequence, which must be calculated another way; see \cref{LES_ext_exm}
for more details.}
\label{Ceta_LES_pic}
\label{final_Ceta}
\end{figure}

For most of the boundary maps, either their source or their target vanishes; in the range visible in
\cref{Ceta_LES_pic}, top, only the four boundary maps pictured (the solid and dashed arrows) could be nonzero. Moreover,
because the long exact sequence commutes with the action of $\Ext(\Z/2\Z)$, the three dashed boundary
maps are determined by the solid boundary map, via the actions of $h_1$, $w$, and $h_1w$. We will show the solid
boundary map is an isomorphism, so that the dashed boundary maps are too.

To show the solid boundary map is an isomorphism, it suffices by exactness to show that $\Ext^{0,2}(C\eta) = 0$,
since we already know $\Ext^{1,2}(\Sigma^2\Z/2\Z) = 0$.
% [AD] resolved \textcolor{red}{This together with Ext$^{1,2} (\Sigma^2 \mathbb{Z}_2) = 0$?}.
To show $\Ext^{0,2}(C\eta)$ vanishes, use that it is identified with $\Hom_{\cA(1)}(C\eta, \Sigma^2\Z/2\Z)$, which
vanishes: since $C\eta$ is a cyclic $\cA(1)$-module, maps out of $C\eta$ are determined by their values on the
generator, which is in degree $0$; since $\Sigma^2\Z/2\Z$ has no nonzero elements in degree $0$, a map $C\eta\to
\Sigma^2\Z/2\Z$ must vanish. This finishes the calculation; we draw the final answer in \cref{Ceta_LES_pic}, bottom.
%\begin{figure}[h!]
%\centering
%\caption{$\Ext(C\eta)$ as calculated using the long exact sequence in Ext. }
%\label{final_Ceta}
%\end{figure}
There are two things we want to mention about the result of the computation.
\begin{itemize}
	\item For all $n$, the $v$-action on $\Ext(C\eta)$ carries the blue tower in topological degree $4n$
	injectively into the blue tower in topological degree $4n+4$, and likewise carries the red tower in topological
	degree $4n+2$ injectively into the red tower in topological degree $4n+6$. This is because the maps in the long
	exact sequence respect the $\Ext(\Z/2\Z)$-action, so it suffices to know this fact for
	$\Ext(\textcolor{BrickRed}{\Sigma^2\Z/2\Z})$ and $\Ext(\textcolor{MidnightBlue}{\Z/2\Z})$; here, the
	$\Ext(\Z/2\Z)$-action is multiplication.
	\item In a short exact sequence of $\cA(1)$-modules, the $\Ext(\Z/2\Z)$-action on Ext of the middle term cannot
	always be determined from the long exact sequence: there may be ``missing actions.'' For example, acting by
	$h_0$ is an isomorphism $\Ext^{2,6}(C\eta)\overset\cong\to \Ext^{3, 7}(C\eta)$ (and similarly in topological
	degree $12$), but looking at \cref{Ceta_LES_pic}, our method cannot see this. To check for missing
	actions, one has to compute another way, such as fitting the module into another long exact sequence.
	Beaudry-Campbell~\cite[Example 4.5.6]{BC18} calculate $\Ext(C\eta)$ in a different way, and can see these
	hidden actions.
\end{itemize}
\end{exm}
%\begin{rem}
%We said above that the maps in the long exact sequence are $\Ext(\Z/2)$-linear, but that does not mean they fully
%determine the $\Ext(\Z/2)$-action on $\Ext(B)$. For example, compare (TODO: example in literature). This subtlety
%will not arise in this paper.
%\end{rem}
%Things to highlight:
%\begin{itemize}
%	\item You likely don't have to actually compute any Ext groups; you can look them up in, e.g., \cite{BC18} or
%	Ext calculation programs
%	\item Important structure: the $\Ext(\Z/2)$-module structure on $\Ext(M)$, and what this looks like
%	\item Ok, but if you \emph{do} need to compute Ext groups, you can use the long exact sequence (example for
%	$C\eta$)
%\end{itemize}
%

\subsection{Differentials}
\label{ss:differentials}

The differentials on the $E_r$-page of the Adams spectral sequence have signature $d_r\colon E_r^{s,t}\to E_r^{s+r,
t+r-1}$. Thus when the Adams $E_r$-page is displayed with $(t-s, s)$ coordinates as is standard, this differential
goes one tick to the left and $r$ ticks upwards.

In general, Adams spectral sequence differentials are very difficult: for example, in the Adams spectral sequence
computing the stable homotopy groups of the spheres, entire papers are dedicated to computing single differentials!
See~\cite{Bru84, IX15, WX17} for some examples. The ``Mahowald uncertainty principle'' is the belief that any single calculation to
compute Adams differentials will leave infinitely many unaddressed.

However, Adams differentials are also extremely constrained: there is a lot of structure on the $E_r$-page of the
Adams spectral sequence, and differentials are compatible with this structure in various ways. This has the
consequence that determining a differential in the Adams spectral sequence tends to either be very easy or very
hard.

In our situation, things are good: there is not enough room for differentials to be too difficult in degrees $12$
and below, and the simpler Adams spectral sequence for $\ko$-theory tends to have easier differentials as well. We
will go over a few standard tools for computing differentials; they will suffice to determine all Adams
differentials in this paper.

\subsubsection{Differentials are equivariant for the $\Ext(\Z/2\Z)$-action}

The first, and most useful, fact is that differentials are equivariant for the $\Ext(\Z/2\Z)$-action on the
$E_r$-page. That is, if $a\in\Ext(\Z/2\Z)$ and $x\in E_r^{s,t}$, $d_r(a\cdot x) = a\cdot d_r(x)$. For example, this
can be used to show that all differentials in the Adams spectral sequence for $\ko_*(\pt)$ vanish: the $E_2$-page
is $\Ext(\Z/2\Z)$. Looking at \cref{pic_ext_Z2}, it looks like there could be a $d_r$ differential from $h_1$ to
$h_0^{r+1}$, but $h_0h_1 = 0$ and $h_0(h_0^{r+1})\ne 0$. If $d_r(h_1)\ne 0$, then $d_r(h_0h_1) = h_0(h_0^{r+1})$,
so $d_r$ sends $0$ to something nonzero, which is a contradiction.

In summary, the only way for a differential to kill elements in an $h_0$-tower is with another $h_0$-tower. In a
similar way, if $h_1x = 0$ and $h_1y\ne 0$, then $d_r(x)$ cannot equal $y$. We will use this fact often.

\subsubsection{Margolis' theorem}
Margolis' theorem is another very useful tool.
\begin{thm}[Margolis~\cite{Mar74}]
\label{margolis}
A splitting $H^*(X;\Z/2\Z)\cong \Sigma^k\cA(1)\oplus M$ as $\cA(1)$-modules lifts to
a splitting $\ko\wedge X\simeq \Sigma^k H\Z/2\Z\vee X'$ of spectra.
\end{thm}
This is a bit abstract, so we go over some more explicit consequences. $\Ext \big(\Sigma^k\cA(1)\big)$ consists of a
single $\Z/2\Z$ summand in bidegree $s = 0$, $t = k$. This is a direct summand of the $E_2$-page of the Adams
spectral sequence computing $\ko_*(X)$; Margolis' theorem says this splitting comes from a splitting of spectra, so
the entire Adams spectral sequence, including all differentials and extension problems, splits into the single
$\Z/2\Z$ coming from $\Sigma^k\cA(1)$ and everything else. That is:
\begin{cor}
\label{Margolis_kills_differentials}
Let $S\subset \Ext \big(H^*(X;\Z/2\Z)\big)$ be a $\Z/2\Z$ summand coming from a $\Sigma^k\cA(1)$ summand in $H^*(X;\Z/2\Z)$. Then
all differentials to and from $S$ vanish, and no element of $S$ participates in a nonsplit extension on the
$E_\infty$-page.
\end{cor}
This will zero out a lot of differentials in the Adams spectral sequence for Spin-$\GL^+(2, \Z)$ bordism for us.

\subsubsection{Miscellaneous tricks}

The Adams spectral sequence is natural, which is abstract nonsense, but helps us deduce differentials. In a little
more detail, given a map of spectra $f\colon X\to Y$, there is a map from the $E_r$-page of the Adams spectral
sequence of $X$ to the Adams $E_r$-page of $Y$, and this map commutes with differentials. When $r = 2$, this map
has another description: cohomology is contravariantly functorial, giving us a map $f^*\colon H^*(Y;\Z/2\Z)\to
H^*(X;\Z/2\Z)$, and $\Ext(\text{--},\Z/2\Z)$ is contravariantly functorial (as we discussed in
\cref{ext_is_contravariant}), giving us a map $f_\sharp\colon
\Ext\big(H^*(X;\Z/2\Z)\big)\to\Ext\big(H^*(Y;\Z/2\Z)\big)$, and $f_\sharp$ is the map on Adams $E_2$-pages.

To apply this, one lets $X$ be a spectrum whose Adams spectral sequence is well-understood, and
$Y$ one whose Adams spectral sequence is not as understood --- or sometimes vice versa. In either case,
differentials that one already knows might map to differentials one wants to understand. This is particularly useful
for us for the map from Spin-$\Mp(2, \Z)$ bordism to Spin-$\GL^+(2, \Z)$ bordism: we determine all differentials in
the Adams spectral sequence for $\Omega_*^{\Spin\text{-}\Mp(2, \Z)} (\pt)$ using the May-Milgram theorem (see
below), then use the map $\Omega_*^{\Spin\text{-}\Mp(2, \Z)} (\pt) \to\Omega_*^{\Spin\text{-}\GL^+(2, \Z)} (\pt)$
to deduce many differentials in the Adams spectral sequence for Spin-$\GL^+(2, \Z)$ bordism. See
\cref{image_from_spin_Z8,d4_equals_zero}. This is a common technique when computing with spectral sequences: see,
e.g.,~\cite{Deb21, MT67, Mah68, DM81, LM98, Hil09, Fra11, WX17, Isa19, BR21a, BR21b, Cul21, FH21} for other Adams
examples.

The May-Milgram theorem~\cite{MM81} is another helpful tool for computing differentials: it expresses differentials
between towers of elements linked by $h_0$-actions in terms of Bockstein homomorphisms on cohomology, which are
typically much easier to compute. This is an essential piece of our calculation of Spin-$\Z/8\Z$ bordism in
Section \ref{ss:spin_Z8}, and we say more about the May-Milgram theorem there.

There are various other techniques for computing Adams differentials in relatively systematic ways. A lot of
information can be gained from secondary operations called Massey products and Toda brackets~\cite{Tod62, Mos70}.
Recently, Gheorghe, Isaksen, Wang, and Xu have used methods originating in motivic stable homotopy theory to make
great progress on the Adams spectral sequence for the sphere spectrum~\cite{Isa19, IWX20, GWX21}.

Of course, bordism groups have geometric meaning, and it is sometimes possible to deduce differentials in an Adams
spectral sequence calculation of a bordism group by finding a generator for that bordism group. This technique is
used in, e.g.,~\cite[Appendix F]{KPMT20}.

\vspace{0.3cm}

\noindent $\triangleright$ {\bf Physics intuition:} Some aspects of the Adams spectral sequence admit a physical interpretation.
The core object which enters in the $E_2$ page of the Adams spectral sequence involves objects
of the form $E_{2}^{s,t} = \text{Ext}^{s,t}_{\mathcal{A}(1)}(M,N)$, where $M$ and $N$ are $\mathcal{A}(1)$
modules. To unpack the definitions, the extension problem of interest
involves the long exact sequence:
\begin{equation}
\begin{tikzcd}
	0 & {\Sigma^t N} & {P_1} & \dotsb & {P_s} & M & 0,
	\arrow[from=1-1, to=1-2]
	\arrow[from=1-2, to=1-3]
	\arrow[from=1-3, to=1-4]
	\arrow[from=1-4, to=1-5]
	\arrow[from=1-5, to=1-6]
	\arrow[from=1-6, to=1-7]
\end{tikzcd}
\end{equation}
where $P_1,\dotsc,P_s$ are $\cA(1)$-modules and all maps are $\cA(1)$-module maps, namely we look for sequences involving $s$ terms in the resolution of $M$, and $\Sigma^{t} N $ is the (reduced) suspension of $N$. Interpreting $N$ as a CW complex, recall that the
reduced suspension $\Sigma N$ involves sweeping out $N$ from a ``starting point'' to a ``final point''. Said differently, one can introduce an
auxiliary timelike direction and track the evolution of $N$ along this auxiliary time. This is a rather natural operation in the context of a path integral for a topological quantum field theory where one specifies an initial and final time evolution step, and provides a physical starting point for understanding the $t = 0$ and $t = 1$ terms of the $E_2$ page. Treating $\Sigma N = \Sigma^{1} N$ itself as another CW complex, observe that we can repeat the entire process, arriving at $\Sigma^{2} N$. Consider next the grading by the integer $s$. These sorts of resolutions figure prominently in the study of B-branes on Calabi-Yau threefolds (see e.g., 
\cite{Kontsevich:1994dn, Douglas:2000gi, Aspinwall:2001pu} and \cite{Aspinwall:2004jr} for a review). In that context, B-branes are objects of the (bounded) derived category of coherent sheaves, and a long exact sequence can be interpreted as specifying a bound state of branes undergoing tachyon condensation (via morphisms between coherent sheaves). This suggests that in the term with $P_1,...,P_s$, each one of these objects should be viewed as specifying as an intermediate object appearing in the bound state fusing $\Sigma^{t}N$ to $M$. Lastly, let us turn to the differentials of the spectral sequence. The case which is somewhat ``easier'' to understand from a physical perspective just involves differentials $d_1$ which act on the bigrading as $(s,t) \rightarrow (s+1,t)$, namely we simply add on another term in the long exact sequence. In related contexts such as the derived category of coherent sheaves, such a differential can be interpreted as a supercharge action, and taking its cohomology returns the supersymmetric bound states (up to stability). The interpretation of $d_2$ and the higher differentials in principle follows from this more primitive $d_1$ object, but we leave the details of how this works and how to interpret the physical systems described above in the context of the Adams spectral sequence in detail for future investigation. $\triangleleft$

\subsection{Extensions and the $E_\infty$-page}
% NOTE: I (Arun) am factoring this into a separate file so that it's easier for everyone to edit simultaneously
% without causing version conflicts. Once I've finished drafting this, it doesn't matter to me whether we merge it
%into its parent file, or keep it separate; you should feel free to do whichever makes life easier for you.
\label{ss:extensions}

Once we have computed all the differentials in the Adams spectral sequence, can we read off the $\ko$-homology
groups? Almost, but not quite; the purpose of this section is to finish the job.

Many spectral sequences arise from a \term{filtration} on the algebraic object $A$ we want to compute. This means
something like data of subobjects $A_1\subset A_2\subset\dotsb\subset A$ such that $\bigcup_i A_i =
A$.\footnote{In general, one uses a colimit rather than a union, but the idea is the same. Additionally, sometimes
one sees variants of this notion of filtration, e.g.\ we might encounter $A_i\supset A_{i+1}$, etc.} What the
spectral sequence computes is the \term{associated graded} of the filtration, i.e.\ the quotients $A_{i+1}/A_i$.
This is a good approximation of $A$, but it is not always enough: for example, $A = \Z/4\Z$ and $A= (\Z/2\Z)\oplus
(\Z/2\Z)$ both admit two-term filtrations with $A_1 = \Z/2\Z$ and $A_2 = A$, and in both cases the associated
gradeds are two copies of $\Z/2\Z$. Therefore if a spectral sequence outputs $(\Z/2\Z)\oplus (\Z/2\Z)$, one has to
determine whether $A$ is in fact $(\Z/2\Z)\oplus (\Z/2\Z)$, or whether $A\cong\Z/4\Z$. This is known as an
\term{extension question}.

Extensions in the Adams spectral sequence behave a lot like differentials: there is a lot of structure that handles
most of them, but when the structure does not help, extensions can be quite difficult, and we have a
variety of ad hoc techniques to address them.
\subsubsection{General structure of extension questions in the Adams spectral sequence}
The $(t-s, s)$-grading convention for the Adams spectral sequence has the advantage that extension questions are
vertical. Specifically, suppose that for topological degree $k$ fixed, the $E_\infty$-page for the Adams spectral
sequence computing $\ko_*(Y)$ is nonzero only at $s = \textcolor{BrickRed}{0}$, $\textcolor{Green}{1}$,
$\textcolor{MidnightBlue}{3}$, and $\textcolor{Fuchsia}{4}$; these are the groups
$\textcolor{BrickRed}{E_\infty^{0, k}}$, $\textcolor{Green}{E_\infty^{1, k+1}}$,
$\textcolor{MidnightBlue}{E_\infty^{3, k+3}}$, and $\textcolor{Fuchsia}{E_\infty^{4, k+4}}$. We draw a picture of
this in \cref{rainbow_xtn}.
\begin{figure}[h!]
\centering
\begin{sseqdata}[name=rainxtn, classes={draw=none}, xrange={0}{2}, yrange={0}{4}, scale=0.7, Adams grading,
>=stealth]
\class["\textcolor{BrickRed}{E_\infty^{0,k}}"](1, 0)
\class["\textcolor{Green}{E_\infty^{1,k+1}}"](1, 1)
\class["\textcolor{MidnightBlue}{E_\infty^{3,k+3}}"](1, 3)
\class["\textcolor{Fuchsia}{E_\infty^{4,k+4}}"](1, 4)
\class["0"](1, 2)
\end{sseqdata}
\printpage[name=rainxtn, page=2]
\caption{A piece of an $E_\infty$-page in fixed topological degree. This scenario leads to the three extension
questions in~\eqref{the_xtn_qns}.}
\label{rainbow_xtn}
\end{figure}

Given this data, there are abelian groups $A_1$ and $A_2$ and extensions
\begin{subequations}
\label{the_xtn_qns}
\begin{gather}
	\shortexact{\textcolor{Green}{E_\infty^{1, k+1}}}{A_1}{\textcolor{BrickRed}{E_\infty^{0, k}}}{}\\
	\shortexact{\textcolor{MidnightBlue}{E_\infty^{3, k+3}}}{A_2}{A_1}{}\\
	\shortexact{\textcolor{Fuchsia}{E_\infty^{4, k+4}}}{\ko_k(Y)}{A_2}.
\end{gather}
\end{subequations}
That is, we work upwards, extending what we have already computed by the next group in the $E_\infty$-page.

If there are only finitely many nonzero groups on the $E_\infty$-page in fixed topological degree, then this
process stops once we have incorporated the last one, but it is possible to have infinitely many nonzero groups,
and therefore infinitely many steps to take in this part of the computation. This is quite common: it occurs in
both the computations for Spin-$\Mp(2, \Z)$ bordism and Spin-$\GL^+(2, \Z)$ bordism, and in most of the examples
seen in~\cite{FH16, Cam17, BC18, Deb21}.

Fortunately, this infinite-seeming question has a finite answer --- the Atiyah-Hirzebruch spectral sequence can be
used to show that unless $X$ is overwhelmingly large,\footnote{Specifically, we require that $X$ is
homotopy equivalent to a CW complex with finitely many cells in each dimension. This is true not just for every
spectral sequence computation in this paper, but also all of the spectral sequence computations in mathematical
physics that we cite.} the $(X, V)$-twisted Spin bordism groups are finitely generated Abelian groups. Therefore in
an extension question with infinitely many steps, all but finitely many of the groups must combine into a free
Abelian group.\footnote{Under the hood, what is happening here is that in an infinite-steps extension problem, the
final answer is the limit of the extensions at each finite step. This is where $2$-completion assists us: a $\Z$
summand in $\ko_k(Y)$ becomes a $\Z_2$-summand in $\ko_k(Y))_2^\wedge$, where $\Z_2$ denotes the $2$-adic integers,
and $\ko_k(Y)_2^\wedge$ is what the Adams spectral sequence actually computes. Since $\Z_2 = \varprojlim_n
\Z/2^n\Z$, then infinite-step extension questions can produce copies of $\Z_2$; since we know $\ko_k(Y)$ is
finitely generated, any infinite-step extension question must combine all but finitely many pieces into
$\Z_2^{\oplus r}$, which arises from $\Z^r$ by $2$-completing.}
\begin{rem}
Similar extension questions arise in other spectral sequences; in particular, we have to confront extension
questions in an Atiyah-Hirzebruch spectral sequence in the proof of \cref{the_thm}. The broad shape of the question
is the same: the final answer we want to compute is an iterated extension of all elements on the $E_\infty$-page in
a given total degree. However, this time the grading is different: rather than working vertically, all elements of
total degree $k$ lie on the diagonal line $p+q=k$ on the $E^\infty$-page. Our proof of \cref{the_thm} does not need
to get into the details of this kind of extension question; see~\cite[Section 2.2.2, Section 3]{Garcia-Etxebarria:2018ajm}
for more information and more examples of extension problems in Atiyah-Hirzebruch spectral sequences.
\end{rem}
\subsubsection{What we can infer from the $\Ext(\Z/2\Z)$-action}
\label{sss:non_hidden}
Because the $\Ext(\Z/2\Z)$-action on the $E_2$-page of the Adams spectral sequence that we introduced in
Section \ref{ss:ext} commutes with differentials, it passes to an $\Ext(\Z/2\Z)$-action on each page of the spectral
sequence, including the $E_\infty$-page. This extra data automatically resolves a large quantity of extension
problems, especially when computing $\ko$-homology in physically relevant degrees; this fact is probably the
biggest competitive advantage of the Adams spectral sequence over the Atiyah-Hirzebruch spectral
sequence.\footnote{As an example, Pin$^-$ structures are equivalent to $(B\Z/2\Z, \sigma)$-twisted Spin structures,
where $\sigma\to B\Z/2\Z$ is the tautological line bundle. Therefore one can compute Pin$^-$ bordism by computing
$\Omega_*^\Spin \big( (B\Z/2\Z)^{\sigma-1} \big)$, which one can study with either the Adams or Atiyah-Hirzebruch spectral
sequences. In the Adams spectral sequence, the $\Ext(\Z/2\Z)$-action solves all extension problems~\cite[Example
6.3]{Cam17}, but in the Atiyah-Hirzebruch spectral sequence, we have almost no information about extensions without
additional work.}

The key fact is:
\begin{lem}
\label{h0_xtn_lemma}
In the Adams spectral sequence computing $\ko_*(Y)$,\footnote{This fact, along with most of the facts in this
subsection, generalize to Adams spectral sequences over any subalgebra of the Steenrod algebra computing other
generalized homology groups.} suppose there are classes $\overline x\in E_\infty^{s, t}$ and $\overline y\in
E_\infty^{s+1, t+1}$ such that $h_0\overline x = \overline y$. Then there are elements
$x,y\in\ko_{t-s}(Y)$ whose images in the $E_\infty$-page are $\overline x$, resp.\ $\overline y$, such that $2x =
y$.
\end{lem}
That is, an $h_0$-action forces the extension $0\to \Z/2\Z\to A\to\Z/2\Z\to 0$ to have $A = \Z/4\Z$, not
$A = (\Z/2\Z)\oplus (\Z/2\Z)$! Likewise, longer sequences of elements on the $E_\infty$-page linked by
$h_0$-actions correspond to copies of $\Z/4\Z$, $\Z/8\Z$, and so on in $\ko$-homology.
See \cref{h0_xtn_example} for a picture.
The converse is false: not every non-trivial extension is detected by $h_0$, as we
discuss in Section \ref{sss:hidden_xtn}.
\begin{figure}[h!]
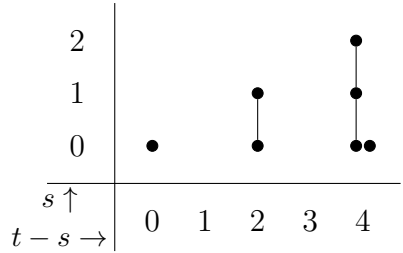

\centering
\begin{sseqdata}[name=pinc, classes=fill, xrange={0}{4}, yrange={0}{2}, scale=0.7, Adams grading, >=stealth,
x label = {$\displaystyle{s\uparrow \atop t-s\rightarrow}$},
x label style = {font = \small, xshift = -14.5ex, yshift=5.5ex}]
\class(0, 0)
\class(2, 0)
\class(2, 1)\structline
\class(4, 0)
\class(4, 1)\structline
\class(4, 2)\structline
\class(4, 0)
\begin{scope}[draw=none, fill=none]
	\class(4, 1)
	\class(4, 2)
\end{scope}
\end{sseqdata}
\printpage[name=pinc, page=2]
\caption{The $E_\infty$-page for the Adams spectral sequence computing $\Pin^c$ bordism. There are extension
questions in degrees $2$ and $4$; the action of $h_0$ in degree $2$ forces $\Omega_2^{\Pin^c}\cong\Z/4\Z$, rather
than $(\Z/2\Z)\oplus (\Z/2\Z)$, and the $h_0$-actions in degree $4$ mean both extensions are nonsplit, forcing
$\Omega_4^{\Pin^c}\cong (\Z/8\Z) \oplus (\Z/2\Z)$. Figure adapted from Beaudry-Campbell~\cite[Figure 42]{BC18}; see
there for more information.}
\label{h0_xtn_example}
\end{figure}

The other generators of $\Ext(\Z/2\Z)$ correspond to actions by other elements of $\ko_*$.
\begin{itemize}
	\item The analogue of \cref{h0_xtn_lemma} for $h_1$ is that if $\overline x\in E_\infty^{s,t}$, $\overline y\in
	E_\infty^{s+1, t+2}$, and $h_1x = y$, then there are classes $x\in \ko_{t-s}(X)$ and $y\in\ko_{t-s+1}(X)$
	corresponding to $\overline x$, resp.\ $\overline y$, such that $\eta\cdot x = y$, where $\eta$ is the nonzero
	element of $\ko_1$. When we pass to Spin bordism, $\eta$ is the class of $S_p^1$, so if $x$ is represented by a
	manifold $M$, $y = [S_p^1\times M]$. We use this fact quite a bit in our search for generators in
	Section \ref{ss:gl2_gens}.
	\item The action by $v\in\Ext_{\cA(1)}^{3,7}(\Z/2\Z)$ lifts in Spin bordism to taking the product with the K3
	surface.
	\item The action by $w\in\Ext_{\cA(1)}^{4,12}(\Z/2\Z)$ lifts in Spin bordism to taking the product with the
	Bott manifold.
\end{itemize}

\subsubsection{Hidden extensions and some tricks}
\label{sss:hidden_xtn}
Though the presence of an $h_0$-action on the $E_\infty$-page indicates multiplication by $2$ in $\ko$-homology,
the converse is not true: there are multiplications by $2$ in some $\ko$-homology groups that do not correspond to
$h_0$-actions on the $E_\infty$-page. The same thing can happen with other elements of $\ko_*$ and their images in
the Adams $E_\infty$-page, notably $h_1$ and $\eta$. Extensions that are not detected by the $\Ext(\Z/2\Z)$-action
are called \term{hidden extensions}.

Hidden extensions are relatively uncommon, especially in physically relevant dimensions: for example, there are
none present in the Adams spectral sequence computations done by Freed-Hopkins~\cite{FH16} and
Beaudry-Campbell~\cite{BC18}. However, hidden extensions are not so uncommon that we can ignore them, and in fact
we will discover non-trivial hidden extensions in Spin-$\Z/8\Z$ and Spin-$D_{16}$ bordism.

One nice technique to address extension questions is to map to or from some other group we already understand. From
the induced map on $E_\infty$-pages of Adams spectral sequences we can often deduce that the original map was
injective or surjective, which can suffice: for example, if we are trying to differentiate $(\Z/2\Z)\oplus
(\Z/2\Z)$ from $\Z/4\Z$, it suffices to produce an injective map from $\Z/4\Z$. We use this technique to resolve an
Adams hidden extension in \cref{Z4_5_extn} and some Atiyah-Hirzebruch extensions in \cref{the_thm}.

% computing bordism groups
Another useful technique when applying the Adams spectral sequence to bordism questions is to compute subgroups of
bordism groups in different ways. For example, in Section \ref{ss:spin_Z8} we see that the Adams spectral sequence implies
$\Omega_5^{\Spin\text{-}\Z/8\Z} (\pt)$ is an Abelian group of order $64$, but does not fully determine which one we
obtain. In Appendix \ref{subsec:exteta}, we show that a combination of $\eta$-invariants is a bordism invariant
$\Omega_5^{\Spin\text{-}\Z/8\Z} (\pt) \to \R/\Z$ and on a lens space takes on the value $-5/32$, so
$\Omega_5^{\Spin\text{-}\Z/8\Z} (\pt)$ must be isomorphic to either $(\Z/32\Z)\oplus(\Z/2\Z)$ or $\Z/64\Z$; another
$\eta$-invariant computation shows we obtain the former.
%We found $\eta$-invariant computations to be the most
%effective at resolving extension questions in spin-$\Z/8\Z$ bordism, though that is in part because of readily
%available computations of $\eta$-invariants of lens spaces and lens space bundles due to [\TODO: cite].

% 2eta = 0
In the Adams spectral sequence for Spin-$D_{16}$ bordism, we will make frequent use of the following result. Recall
that $\eta$ is the nonzero element of $\ko_1(\pt)$, and in twisted Spin bordism $\eta$ represents the
circle with non-bounding boundary conditions for fermions $S^1_p$.
\begin{lem}
\label{2eta_lemma}
Let $x\in\ko_n(X)$ be such that $\eta x\ne 0$. Then there is no $y\in \ko_n(X)$ such that $2y = x$.
\end{lem}
\begin{proof}
Suppose such a $y$ exists. Then $\eta(2y)\ne 0$, but this is also equal to $(2\eta)y$ and $2\eta = 0$ in
$\ko_*(\pt)$.
\end{proof}
Typically when we use this, such as in \cref{ko8_xtn,ko9_xtn}, we deduce $\eta x\ne 0$ from an $h_1$-action on the
$E_\infty$-page.
\begin{figure}[h!]
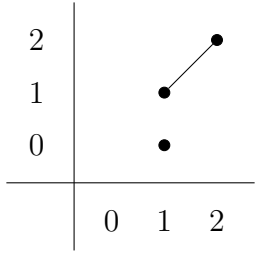

\centering
\begin{sseqdata}[name=twoeta, classes=fill, xrange={0}{2}, yrange={0}{2}, scale=0.7, Adams grading, >=stealth]
\class(1, 0)
\class(1, 1)
\class(2, 2)\structline
\end{sseqdata}
\printpage[name=twoeta, page=2]
\caption{A picture of the $E_\infty$-page in the scenario in \cref{2eta_lemma}. The $h_1$-action in this figure
implies that $\ko_1(X)\cong (\Z/2\Z)\oplus (\Z/2\Z)$, rather than $\Z/4\Z$.}
\end{figure}

The final tool we rely on heavily in this paper is Margolis' theorem (\cref{Margolis_kills_differentials}): we
previously used it to kill differentials on elements in the $E_r$-page coming from free $\cA(1)$-module summands in
cohomology, but this theorem also prevents those elements from participating in hidden extensions.
\begin{rem}
There are many other tricks people use to determine extensions in Adams spectral sequences. One that might be
useful in related questions is that the multiplication-by-$2$ map on $\ko$-theory factors as \begin{equation}
\ko\overset{c}{\longrightarrow} \ku \overset{b}{\longrightarrow} \Sigma^2 \ku \overset{R}{\longrightarrow} \ko,
\end{equation} where $c$ is the complexification map, $b$ is the Bott map,
and $R$ is a ``realification'' map obtained from forgetting the complex structure on a vector bundle (see Section
\ref{subsec:spliteven}). Occasionally
this is helpful for splitting extensions, e.g., if $b$ acts by $0$ on $\ku_n(X)$, one can conclude that
multiplication by $2$ is the zero map on $\ko_n(X)$. We use this in \cref{quater_hidden}.

Like with differentials, recently developed techniques involving synthetic spectra have been used on extension
questions in~\cite{IWX20, Bur21, Chu22, Mar22}.
\end{rem} 

%\textcolor{red}{
%Things to mention:
%\begin{itemize}
%	\item General structure of extension question
%	\item Learning some extensions from the $\Ext(\Z/2)$-action on the $E_\infty$-page
%	\item Hidden extensions and tricks to solve them (Margolis' theorem, $2\eta = 0$)
%\end{itemize}}

\section{Computation of $\Omega_*^\Spin\big(B\SL(2, \Z)\big)$}
	\label{sl2_spin}
In this section we undertake the first of the three bordism computations in this paper: $\Omega_*^\Spin(B\SL(2,
\Z))$ in dimensions $11$ and below. This is the easiest of the three: we can assemble everything we need from
results already in the literature. The technique to do so is a generalization of a better-known technique to
calculate group cohomology: if a group $G$ factors as an amalgamated product $G\cong H_1 *_K H_2$, then there is a
Mayer-Vietoris sequence computing the cohomology of $G$ in terms of that of $H_1$, $H_2$, and $K$:
\begin{equation}
\label{coh_MV}
\begin{tikzcd}
\phantom{.} & {H^k(BK)} & {H^k(BH_1)\oplus H^k(BH_2)} & {H^k(BG)} & {H^{k+1}(BK)} & \phantom{.}
	\arrow[from=1-1, to=1-2]
	\arrow[from=1-2, to=1-3]
	\arrow[from=1-3, to=1-4]
	\arrow[from=1-4, to=1-5]
	\arrow[from=1-5, to=1-6]
\end{tikzcd}
\end{equation}
There is an isomorphism $\SL(2, \Z)\cong (\Z/4\Z) *_{\Z/2\Z} (\Z/6\Z)$, as we review in Appendix~\ref{app:groups},
and~\eqref{coh_MV} allows for a complete calculation of $H^*\big(B\SL(2, \Z)\big)$. Note also that the two factors in the amalgamated product correspond to the two finite subgroups, generated by $S$ and $U$, fixing the axio-dilaton $\tau$ to the special values indicated in Table \ref{tab:singfiber}.

An analogous Mayer-Vietoris sequence exists for any generalized homology or cohomology theory. This reduces the
computation of $\Omega_*^\Spin\big(B\SL(2, \Z)\big)$ to the computation of $\Omega_*^\Spin(B\Z/\ell \Z)$ for $\ell = 2, 4, 6$.
By localizing at $p = 2$ and $p = 3$, as described in Section \ref{ss:one_prime}, we will be able to simplify further, so
that we only need to know $\Omega_*^\Spin(B\Z/\ell \Z)$ for $\ell = 3$ and $\ell = 4$. We then compute these groups
using closely related computations that are already in the literature. In particular, we do not need any spectral
sequences in this section.

In Section \ref{ss:sl_2_at_2}, we compute the $2$-local spin bordism groups of $B\SL(2, \Z)$ in dimensions $11$ and
below; in Section \ref{sl2_odd_primes} we work at odd primes. Then, in Section \ref{sl2_gens}, we find a set of generators for
$\Omega_k^\Spin\big(B\SL(2, \Z)\big)$ for $k\le 11$.

\subsection{Working at $p = 2$}
\label{ss:sl_2_at_2}

First, we simplify $B\SL(2, \Z)$.
\begin{lem}
\label{sl2_at_2}
The inclusion $\Z/4\Z \inj \SL(2, \Z)$ induces a map
\begin{align}
\Omega_k^\Spin(B\Z/4\Z) \longrightarrow \Omega_k^\Spin \big(B\SL(2, \Z)\big) \,,
\end{align}
which is an
isomorphism after tensoring both sides with $\Z_{(2)}$.
\end{lem}
The trick we use to prove this comes up a few more times later, so we extract it as a lemma. First, recall that for
an amalgamated product, $G = H_1 *_K H_2$, there is a Mayer-Vietoris sequence for any generalized homology theory
$E$:
\begin{equation}
\label{MV_bar}
\begin{tikzcd}[column sep=3ex]
	\dotsb & {E_n(BK)} & {E_n(BH_1)\oplus E_n(BH_2)} & {E_n(BG)} & {E_{n-1}(BK)} & \dotsb
	\arrow[from=1-1, to=1-2]
	\arrow[from=1-2, to=1-3]
	\arrow[from=1-3, to=1-4]
	\arrow[from=1-4, to=1-5]
	\arrow[from=1-5, to=1-6]
\end{tikzcd}
\end{equation}
\begin{lem}
\label{natural_MV}
The Mayer-Vietoris sequence~\eqref{MV_bar} is natural in amalgamated products of groups. That is, if $f\colon
H_1*_K H_2\to H_1' *_{K'} H_2'$ satisfies $f(H_i)\subset H_i'$ and $f(K)\subset K'$, then the maps it induces on
$E$-homology commute with~\eqref{MV_bar}.
\end{lem}
\begin{proof}
We use the following model for the classifying space $BG$ of a discrete group $G$: let $C_G$ denote the category
with a single object $\pt$ and with $\mathrm{End}_{C_G}(\pt) = G$. Then let $BG$ be the geometric realization of
the nerve of $C_G$. This is a simplicial complex whose $n$-simplices are indexed by $n$-tuples of elements of $G$.

This model has the advantage that if $G_0\subset G$, $BG_0\subset BG$ as the union of the simplices of $n$-tuples
of elements of $G_0$. In particular, for an amalgamated product $G\coloneqq H_1*_K H_2$, $BH_1\cap BH_2 = BK$ inside
$BG$, yielding~\eqref{MV_bar}. And with this model, the induced map $f\colon B(H_1 *_K H_2)\to B(H_1' *_{K'} H_2')$
satisfies $f(BH_i)\subset BH_i'$ and $f(BK)\subset BK'$, which is the condition needed to commute with a
Mayer-Vietoris sequence.
\end{proof}
\begin{proof}[Proof of \cref{sl2_at_2}]
We can write $\SL(2, \Z) = (\Z/4\Z)*_{\Z/2\Z}(\Z/6\Z)$ and $\Z/4\Z = (\Z/4\Z)*_{\Z/2\Z}(\Z/2\Z)$, and with these descriptions the
inclusion $f\colon \Z/4\Z \inj \SL(2, \Z)$ is compatible with these descriptions as amalgamated products: it is the
identity on the first factor, inclusion $\Z/2\Z \inj \Z/6\Z$ on the second factor, and the identity on the common $\Z/2\Z$
subgroup.

By \cref{natural_MV}, we obtain a commutative diagram of Mayer-Vietoris sequences for Spin bordism:
\begin{equation}
\label{spin_MV_commut}
\begin{gathered}
% https://q.uiver.app/?q=WzAsMTAsWzAsMCwiXFxkb3RzYiJdLFswLDEsIlxcZG90c2IiXSxbMSwwLCJcXE9tZWdhX25eXFxTcGluKEJcXFovMlxcWikiXSxbMSwxLCJcXE9tZWdhX25eXFxTcGluKEJcXFovMlxcWikiXSxbMiwxLCJcXE9tZWdhX25eXFxTcGluKEJcXFovNFxcWilcXG9wbHVzXFxPbWVnYV9uXlxcU3BpbihCXFxaLzZcXFopIl0sWzIsMCwiXFxPbWVnYV9uXlxcU3BpbihCXFxaLzRcXFopXFxvcGx1cyBcXE9tZWdhX25eXFxTcGluKEJcXFovNlxcWikiXSxbMywwLCJcXE9tZWdhX25eXFxTcGluKEJcXFovNFxcWikiXSxbMywxLCJcXE9tZWdhX25eXFxTcGluKEJcXFNMKDIsIFxcWikpIl0sWzQsMCwiXFxkb3RzYiJdLFs0LDEsIlxcZG90c2IiXSxbMiwzLCJhIl0sWzUsNCwiKGIsIGMpIl0sWzYsNywiZCJdLFswLDJdLFsxLDNdLFs2LDhdLFs3LDldLFszLDRdLFsyLDVdLFs1LDZdLFs0LDddXQ==
\begin{tikzcd}[column sep=3ex]
	\dotsb & {\Omega_n^\Spin(B\Z/2\Z)} & {\Omega_n^\Spin(B\Z/4\Z)\oplus \Omega_n^\Spin(B\Z/6\Z)} & {\Omega_n^\Spin(B\Z/4\Z)} & \dotsb \\
	\dotsb & {\Omega_n^\Spin(B\Z/2\Z)} & {\Omega_n^\Spin(B\Z/4\Z)\oplus\Omega_n^\Spin(B\Z/6\Z)} & {\Omega_n^\Spin\big(B\SL(2, \Z)\big)} & \dotsb
	\arrow["a", from=1-2, to=2-2]
	\arrow["{(b, c)}", from=1-3, to=2-3]
	\arrow["d", from=1-4, to=2-4]
	\arrow[from=1-1, to=1-2]
	\arrow[from=2-1, to=2-2]
	\arrow[from=1-4, to=1-5]
	\arrow[from=2-4, to=2-5]
	\arrow[from=2-2, to=2-3]
	\arrow[from=1-2, to=1-3]
	\arrow[from=1-3, to=1-4]
	\arrow[from=2-3, to=2-4]
\end{tikzcd}
\end{gathered}
\end{equation}
As $f$ is the identity on the amalgamating $\Z/2\Z$, the vertical arrow $a$ is the identity map; likewise for the
first factor $\Z/4\Z$ and the map $b$. For $c$, we use the fact that the map $H_*(B\Z/2\Z)\to H_*(B\Z/6\Z)$ is an
isomorphism at the prime $2$, which implies (e.g.\ by naturality of the Atiyah-Hirzebruch spectral sequence) that
$c\colon \Omega_n^\Spin(B\Z/2\Z)\to\Omega_n^\Spin(B\Z/6\Z)$ is also a $2$-local isomorphism. By the five lemma, the
remaining arrow, $d\colon \Omega_n^\Spin(B\Z/4\Z)\to\Omega_n^\Spin \big(B\SL(2, \Z)\big)$, is also a $2$-local
isomorphism.\footnote{In a little more detail, we want to prove that after tensoring everything
in~\eqref{spin_MV_commut} with $\Z_{(2)}$, the rightmost vertical arrow is an isomorphism. To use the five lemma,
the rows must be exact, which is true because $\Z_{(p)}$ is a flat $\Z$-module for any prime $p$.}
\end{proof}
So we need to determine $\Omega_k^\Spin(B\Z/4\Z)$ for $k\le 11$, and by \cref{ABPthm}, we need $\ko_k(B\Z/4\Z)$ for
$k\le 11$ and $\ko\ang 2_k(B\Z/4\Z)$ for $k\le 1$. The groups $\ko_*(B\Z/4\Z)$ are computed by
Bruner-Greenlees~\cite[Example 7.3.3]{bruner2010connective},\footnote{Beware that there is a typo in
Bruner-Greenlees' calculations in the table on p.\ 135 in the row $n = 10$: $\ko_{10}(B\Z/4\Z)\cong
(\Z/2\Z)^{\oplus 2}$, not $\Z/2\Z$.} and for $\ko\ang 2_*(B\Z/4\Z)$, use~\cref{ko2_is_homology}.

%\textcolor{red}{Should we give a table with the results?}
We summarize the results in Table \ref{tab:SLprime2data},
\begin{table}[h!]
\centering
\begin{tabular}{c c c}
\toprule
$k$ & $\ko_k (B\Z/4\Z)$ & $\ko\ang 2_k (B\Z/4\Z)$ \\
\midrule
$0$ & {\footnotesize $\Z$} & {\footnotesize $\Z/2\Z$} \\
$1$ & {\footnotesize $(\Z/2\Z) \oplus (\Z/4\Z)$} & {\footnotesize $\Z/2\Z$} \\
$2$ & {\footnotesize $(\Z/2\Z)^{\oplus 2}$} & \\
$3$ & {\footnotesize $(\Z/2\Z) \oplus(\Z/8\Z)$} & \\
$4$ & {\footnotesize $\Z$} & \\
$5$ & {\footnotesize $\Z/4\Z$} & \\
$6$ & {\footnotesize $0$} & \\
$7$  & {\footnotesize $(\Z/2\Z) \oplus (\Z/32\Z)$} & \\
$8$ & {\footnotesize $\Z$} & \\
$9$ & {\footnotesize $(\Z/2\Z)^{\oplus 2} \oplus (\Z/8\Z)$} & \\
$10$ & {\footnotesize $(\Z/2\Z)^{\oplus 2}$} & \\
$11$ & {\footnotesize $(\Z/8\Z) \oplus(\Z/128\Z)$} & \\
\bottomrule
\end{tabular}
\caption{The groups $\ko_k (B\Z/4\Z)$ for $k \leq 11$ and $\ko\ang 2_k (B\Z/4\Z)$ for $k\leq1$.}
\label{tab:SLprime2data}
\end{table}
which specifies $\ko_k \big( B \SL (2,\Z) \big)$ at the prime $2$, which can be assembled together with
$\ko_{k-8}(B\Z/4\Z)$ and $\ko\ang 2_{k-10}(B\Z/4\Z)$ into $\Omega_k\Spin(B\Z/4\Z)$ as in \cref{ABPthm}.

\subsection{Working at odd primes}
\label{sl2_odd_primes}

\begin{lem}
\label{sl2_at_3_lem}
The inclusion $\Z/3\Z \inj \SL(2, \Z)$ induces a map
\begin{align}
\Omega_*^\Spin(B\Z/3\Z) \longrightarrow \Omega_*^\Spin \big(B\SL(2, \Z)\big)\,,
\end{align}
which is an
isomorphism after tensoring with $\Z/p\Z$ for any odd prime $p$.
\end{lem}
\begin{proof}
The proof is essentially the same as for \cref{sl2_at_2}. The inclusion $\Z/6\Z \inj \SL(2, \Z)$ can be restated as
$(\Z/2\Z)*_{\Z/2\Z}(\Z/6\Z)\inj (\Z/4\Z) *_{\Z/2\Z}(\Z/6\Z)$, inducing maps commuting with Mayer-Vietoris
sequences. The map $H_*(B\Z/2\Z)\to H_*(B\Z/4\Z)$ is a $p$-local equivalence at any odd prime $p$, so the maps on
Spin bordism are also $p$-local equivalences, allowing us to conclude that
$\Omega_*^\Spin(B\Z/6\Z)\to\Omega_*^\Spin(B\SL(2, \Z))$ is a $p$-local equivalence. Finally, because $\Z/6\Z\cong
\Z/3\Z\times\Z/2\Z$, the map $B\Z/3\Z\to B\Z/6\Z$ induces a $p$-local isomorphism on homology, hence also on
spin bordism, so we can replace $B\Z/6\Z$ with $B\Z/3\Z$ to obtain the theorem statement.
\end{proof}
For $p > 3$, $H^\ell(B\Z/3\Z;\Z/p\Z)$ vanishes for $\ell > 0$~\cite[Corollary 10.2]{Bro82}, so
the map $B\Z/3\Z\to\pt$ induces an isomorphism on mod $p$ cohomology, hence by a version of Whitehead's
theorem~\cite[Chapitre III, Théorème 3]{Ser53}, is a stable $p$-local equivalence.
%is stably a $p$-local equivalence whenever $p\ne 3$,\footnote{The
%idea of this argument is that an isomorphism on $\Z/p\Z$ cohomology between two spaces with the homotopy type of CW
%complexes (this includes $B\Z/3\Z$ and $\pt$) induces a $p$-local stable equivalence, by a version of Whitehead's
%theorem~\cite[Chapitre III, Théorème 3]{Ser53}. Thus it suffices to show that when $p\ne 3$,
%$H^\ell(B\Z/3\Z;\Z/p\Z)$ vanishes for $\ell > 0$. This follows because
%
%$H^\ell(B\Z/3\Z;\Z/p\Z)$ can be identified with group cohomology
%$H^\ell(\Z/3\Z;\Z/p\Z)\coloneqq \Ext_{k[\Z/3\Z]}(k, k)$, where $k = \Z/p\Z$ is the ground field. Now Maschke's
%theorem says that $k[\Z/3\Z]$ is semisimple when the characteristic of $k$ does not divide the order of $\Z/3\Z$,
%and semisimplicity means all positive-degree Ext groups vanish. \label{masch_foot}} so
Thus it suffices to understand
the $3$-local story. Recall from~\eqref{BP_SO_decomp} the decomposition of $\Omega_*^\SSO\otimes\Z_{(3)}$ into
copies of $\BP$-homology: when $k\le 15$, this decomposition takes the form

\begin{equation}
%\label{BP_SO_decomp}
	\Omega_k^\SSO(X)\otimes\Z_{(3)} \cong \BP_k(X) \oplus \BP_{k-8}(X) \oplus \BP_{k-12}(X) \,.
\end{equation}
Bahri-Bendersky-Davis-Gilkey~\cite[Theorem 1.2(a)]{BBDG89} compute $\BP_*(B\Z/3\Z)$ in terms of the $\ell$-homology
of $B\Z/3\Z$, where $\ell$ is a generalized homology theory called the $3$-local \term{Adams summand} of $\ku$ (so
that $\ell_*(X)$ refers to the $\ell$-homology groups of $X$).
\begin{thm}[{Bahri-Bendersky-Davis-Gilkey~\cite[Theorem 1.2(a)]{BBDG89}}]
\label{BBDG}
\begin{equation}
	\BP_*(B\Z/3\Z)\cong \ell_*(B\Z/3\Z)\otimes_{\Z_{(3)}} \Z_{(3)}[v_2, v_3, \dotsc] \,.
\end{equation}
\end{thm}
So we must determine $\ell_*(B\Z/3\Z)$, which we do using $\ku$-homology. There is an isomorphism\footnote{The
$p$-local decomposition of $\ku$ into a sum of shifts of $\ell$ was a folk theorem. Adams~\cite[Lecture 4,
Corollary 8]{Ada69} proved the analogous decomposition for $\mathit{KU}$; for $\ku$ the earliest statement we could
find is~\cite[Proposition 2.7]{JW73}.}
\begin{equation}
\label{Adams_decomp}
	\ku_k(X)\otimes\Z_{(3)}\cong\ell_k(X)\oplus\ell_{k-2}(X) \,.
\end{equation}
Here $\ku$ is \term{connective complex $K$-theory} (see \cref{ku_remark}).
\begin{thm}[{Hashimoto~\cite[Theorem 3.1]{Has83}}]
%Bruner-Geenlees~\cite[Remark 3.4.6]{BG03}}]
There is an isomorphism
\begin{align}
\widetilde{\ku}_n(B\Z/3\Z) \cong (\Z/3^{j+1}\Z)^{\oplus s} \oplus (\Z/3^j\Z)^{\oplus 2-s} \,,
\end{align}
if $n = 4j +2s-1$ for $0 < s\le 2$. Otherwise, $\widetilde{\ku}_n(B\Z/3\Z) = 0$.
\end{thm}
Explicitly, in low degrees,
\begin{equation}
\begin{aligned}
	\widetilde{\ku}_1(B\Z/3\Z) &\cong \Z/3\Z \,,\\
	\widetilde{\ku}_3(B\Z/3\Z) &\cong (\Z/3\Z) \oplus (\Z/3\Z) \,,\\
	\widetilde{\ku}_5(B\Z/3\Z) &\cong (\Z/9\Z) \oplus(\Z/3\Z) \,, \\
	\widetilde{\ku}_7(B\Z/3\Z) &\cong (\Z/9\Z) \oplus (\Z/9\Z) \,,\\
	\widetilde{\ku}_9(B\Z/3\Z) &\cong (\Z/27\Z) \oplus (\Z/9\Z) \,, \\
	\widetilde{\ku}_{11}(B\Z/3\Z) &\cong (\Z/27\Z) \oplus (\Z/27\Z) \,,
\end{aligned}
\end{equation}
and the pattern continues. Feeding this to~\eqref{Adams_decomp},
\begin{equation}
\begin{aligned}
	\widetilde\ell_1(B\Z/3\Z) &\cong \Z/3\Z \,, \\
	\widetilde\ell_3(B\Z/3\Z) &\cong \Z/3\Z \,, \\
	\widetilde\ell_5(B\Z/3\Z) &\cong \Z/9\Z \,, \\
	\widetilde\ell_7(B\Z/3\Z) &\cong \Z/9\Z \,, \\
	\widetilde\ell_9(B\Z/3\Z) &\cong \Z/27\Z \,, \\
	\widetilde\ell_{11}(B\Z/3\Z) &\cong \Z/27\Z \,, \\
%	\widetilde\ell_{13}(B\Z/3\Z) &\cong \Z/81\Z \,,
\end{aligned}
\end{equation}
and so on. Even-degree reduced $\ell$-homology of $B\Z/3\Z$ vanishes.

Now feed this to \cref{BBDG}. Because $\abs{v_2} = 2(3^2-1) = 16$, in degrees $15$ and below,
\cref{BBDG} tells us there is no difference between the $\BP$-homology and the $\ell$-homology of $B\Z/3\Z$.

Finally,~\eqref{BP_SO_decomp} gets us back to (reduced) oriented bordism.
\begin{equation}
\label{SO_BZ3}
\begin{aligned}
	\widetilde\Omega_1^\SSO(B\Z/3\Z)\otimes\Z_{(3)} &\cong \Z/3\Z \,, \\
	\widetilde\Omega_3^\SSO(B\Z/3\Z) \otimes\Z_{(3)}&\cong \Z/3\Z \,, \\
	\widetilde\Omega_5^\SSO(B\Z/3\Z) \otimes\Z_{(3)}&\cong \Z/9\Z \,, \\
	\widetilde\Omega_7^\SSO(B\Z/3\Z) \otimes\Z_{(3)}&\cong \Z/9\Z \,, \\
	\widetilde\Omega_9^\SSO(B\Z/3\Z) \otimes\Z_{(3)}&\cong (\Z/27\Z) \oplus (\Z/3\Z) \,, \\
	\widetilde\Omega_{11}^\SSO(B\Z/3\Z)\otimes\Z_{(3)} &\cong (\Z/27\Z) \oplus (\Z/3\Z) \,.
%	\widetilde\Omega_{13}^\SSO(B\Z/3) \otimes\Z_{(3)}&\cong \Z/81\oplus\Z/9\oplus\Z/3.
\end{aligned}
\end{equation}
With this we have all the ingredients we need to determine $\Omega^{\Spin}_{\ast} \big( B \SL(2,\Z) \big)$, which we summarize in Table \ref{tab:BSL_bordism}.

\renewcommand{\arraystretch}{1.5}
\begin{table}[h!]
\centering
\begin{tabular}{c c c}
\toprule
$k$ & $\Omega^{\Spin}_k \big(B\SL(2,\Z)\big)$ & Generators \\
\midrule
$0$ & {\footnotesize $\Z$} & {\footnotesize $\pt$} \\ %\hline
$1$ & {\footnotesize $(\Z/2\Z) \oplus (\Z/4\Z) \oplus (\Z/3\Z)$} & {\footnotesize $(S^1_p \,, L^1_4 \,, L^1_3)$} \\ %\hline
$2$ & {\footnotesize $(\Z/2\Z)^{\oplus 2}$} & {\footnotesize $(S^1_p \times S^1_p \,, S_p^1\times L_4^1)$} \\ %\hline
$3$ & {\footnotesize $(\Z/2\Z) \oplus (\Z/8\Z) \oplus (\Z/3\Z)$} & {\footnotesize $(S_p^1\times S_p^1\times L_4^1 \,, L_4^3 \,, L^3_3)$} \\ %\hline
$4$ & {\footnotesize $\Z$} & {\footnotesize K3} \\ %\hline
$5$ & {\footnotesize $(\Z/4\Z) \oplus (\Z/9\Z)$} & {\footnotesize $(Q_4^5, L^5_3)$} \\ %\hline
$6$ & {\footnotesize $0$} & \\ %\hline
$7$ & {\footnotesize $(\Z/2\Z) \oplus (\Z/{32}\Z) \oplus (\Z/9\Z)$} & {\footnotesize $(\widetilde{L}^7_4 \,, L^7_4 \,, L^7_3)$} \\ %\hline
$8$ & {\footnotesize $\Z \oplus \Z$} & {\footnotesize $(B \,, \HP^2)$} \\ %\hline
$9$ & {\footnotesize $(\Z/2\Z)^{\oplus 3} \oplus (\Z/4\Z)$} &
{\footnotesize $(B\times L_4^1 \,, B\times S_p^1 \,, \HP^2\times S_p^1 \,, \HP^2\times L_4^1 \,,$} \\ %\hline 	
 & {\footnotesize $\oplus (\Z/8\Z)  \oplus (\Z/3\Z) \oplus (\Z/27\Z)$} &
{\footnotesize $Q_4^9 \,, \HP^2 \times L^1_3 \,, L^9_3)$} \\
$10$ & {\footnotesize $(\Z/2\Z)^{\oplus 5}$} & {\footnotesize $(B\times S_p^1\times S_p^1 \,, \HP^2\times
S_p^1\times S_p^1 \,, B\times S^1_p\times L^1_4 \,, $} \\ %\hline
 & & {\footnotesize $ \HP^2\times S^1_p\times L^1_4 \,, X_{10})$} \\ %\hline
$11$ & {\footnotesize $(\Z/2\Z)^{\oplus 2}$} &
	{\footnotesize $(X_{10}\times L_4^1 \,,  \HP^2\times L_4^1\times S_p^1\times S_p^1
	 \,,
	 $} \\ %\hline
& {\footnotesize $ \oplus (\Z/8\Z)^{\oplus 2} \oplus (\Z/128\Z)$} &
	{\footnotesize $
	 \HP^2\times L_4^3\,,
	 \widetilde{L}^{11}_4 \,,
	 L^{11}_4\,,
	$} \\ %\hline
& {\footnotesize $ \oplus (\Z/3\Z) \oplus
(\Z/27\Z)$} &
	{\footnotesize $
	 \HP^2\times L_3^3\,,
	 L_3^{11})$} \\ %\hline
\bottomrule
\end{tabular}
\caption{Bordism groups $\Omega^{\Spin}_k \big( B\SL(2,\Z)\big)$ and their generators (in the same order as the group summands or described by linear combinations of them) for $k \leq 11$.}
\label{tab:BSL_bordism}
\end{table}
%This concludes our discussion of Spin bordisms with $\SL(2,\Z)$ duality bundle.
%
%\textcolor{red}{I included the table here as a conclusion of the section. Include missing generators}

\subsection{Finding generators}
\label{sl2_gens}

We focus on $\tOmega_*^\Spin \big(B\SL(2, \Z)\big)$, i.e., the subgroups of the bordism groups for which the
principal $\SL(2, \Z)$-bundle is non-trivial --- when the bundle is trivial, this is just Spin bordism, whose
low-degree generators are standard. We discuss these generators in Appendix~\ref{app:collection}.
%Nonetheless, we include them here to fix notation.
%\textcolor{red}{(TODO: unless we don't need this here.)}
%\begin{itemize}
%	\item $\Omega_0^\Spin (\pt) \cong\Z$ is generated by the point.
%	\item $\Omega_1^\Spin (\pt) \cong\Z/2\Z$ is generated by $S_p^1$ with Spin structure induced from the Lie group
%	framing, i.e., periodic boundary condition for fermions.
%	\item $\Omega_2^\Spin (\pt) \cong\Z/2\Z$ is generated by $S^1_p \times S^1_p$.
%	\item $\Omega_4^\Spin (\pt) \cong\Z$ is generated by the K3 surface.
%	\item $\Omega_8^\Spin (\pt) \cong\Z\oplus\Z$ is generated by a Bott manifold $B$ and $\HP^2$, quaternionic projective
%	space. The name ``Bott manifold'' does not fix $B$ up to diffeomorphism; we never need to pick a specific Bott
%	manifold, but explicit examples are constructed in~\cite[\S 7]{Lau04} and~\cite[\S 5.3]{FH21}.
%	\item $\Omega_9^\Spin (\pt) \cong (\Z/2\Z)^{\oplus 2}$ is generated by $B\times S^1_p$ and $\HP^2 \times S^1_p$.
%	\item $\Omega_{10}^\Spin (\pt) \cong (\Z/2\Z)^{\oplus 3}$ is generated by $B\times S^1_p\times S_p^1$,
%	$\HP^2\times S^1_p\times S_p^1$,
%	and $X_{10}$. $X_{10}$ is a Milnor hypersurface.
%\end{itemize}
%For $k = 3, 5, 6, 7, 11$, $\Omega_k^\Spin (\pt)$ vanishes.

Now the more interesting part. As a consequence of \cref{sl2_at_2,sl2_at_3_lem}, the inclusions $\Z/3\Z \to \SL(2,
\Z)$ and $\Z/4\Z \to \SL(2, \Z)$ induce an isomorphism
\begin{equation}
	\tOmega_k^\Spin (B\Z/4\Z)\oplus \tOmega_k^\Spin(B\Z/3\Z)\overset\cong\longrightarrow \tOmega_k^\Spin \big(B\SL(2, \Z) \big) \,.
\end{equation}
So we will describe the generators of $\tOmega_k^\Spin\big(B\SL(2, \Z)\big)$ in terms of Spin manifolds with
principal $\Z/3\Z$- and $\Z/4\Z$-bundles. We are not the first to study generators of $\tOmega_*^\Spin(B\Z/k\Z)$,
and previous work suggests that we
should try lens spaces and lens space bundles for generators, and use $\eta$-invariants to detect them.
%Specifically, \textcolor{red}{(TODO: notation for lens spaces and lens space bundles)}.

We start with the $\Z/3\Z$ case: it is easier, because the lens spaces $L_3^{2k-1}$ all have unique Spin
structures. A theorem of Rosenberg~\cite[Proof of Theorem 2.12]{Ros86} shows that the group $\Omega_*^\Spin(B\Z/3\Z)$ is
generated by lens spaces and products of lens spaces with Spin manifolds. We use $\eta$-invariants to find specific
generators of specific groups, as formulas for $\eta$-invariants of twisted Dirac operators on lens spaces are
readily available. We go over these formulas in Appendix \ref{subapp:etalens}, and in \cref{tab:L3bordinv}, we compute
some $\eta$-invariants which are bordism invariants on $L_3^{2k-1}$ for $2\le k\le 6$. This gives us the following
generators; in all cases the duality bundle is induced from the principal $\Z/3\Z$-bundle $S^{2k-1}\to L_3^{2k-1}$
by the inclusion $\Z/3\Z\inj \SL(2, \Z)$.
\begin{itemize}
	\item $L_3^1\cong S^1$ generates the $\Z/3\Z$ summand in $\Omega_1^\Spin\big(B\SL(2, \Z)\big)$.
	\item $L_3^3$ generates the $\Z/3\Z$ summand in $\Omega_3^\Spin\big(B\SL(2, \Z)\big)$.
	\item $L_3^5$ generates the $\Z/9\Z$ summand in $\Omega_5^\Spin\big(B\SL(2, \Z)\big)$.
	\item $L_3^7$ generates the $\Z/9\Z$ summand in $\Omega_7^\Spin\big(B\SL(2, \Z)\big)$.
	\item $L_3^9$ generates the $\Z/27\Z$ summand in $\Omega_9^\Spin\big(B\SL(2, \Z)\big)$.
	\item $L_3^{11}$ generates the $\Z/27\Z$ summand in $\Omega_{11}^\Spin\big(B\SL(2, \Z)\big)$.
\end{itemize}
If $M$ is a Spin manifold of dimension $4m$, $\eta^D(M\times N) = \mathrm{Index}^D(M)\eta^D(N)$, where $D$ refers
to the Dirac operator; similar results hold for twisted Dirac operators. This gives us the last two $3$-torsion
generators we need.
\begin{itemize}
	\item $\HP^2\times L_3^1\cong \HP^2\times S^1$ generates the $\Z/3\Z$ summand in $\Omega_9^\Spin\big(B\SL(2 \Z)\big)$.
	\item $\HP^2\times L_3^3$ generates the $\Z/3\Z$ summand in $\Omega_{11}^\Spin\big(B\SL(2 \Z)\big)$.
\end{itemize}
In both cases, the duality bundle is trivial on $\HP^2$ and induced from $S^{2k-1}\to L_3^{2k-1}$, as above, on the
lens component of the product.

Now for the $2$-torsion. An analogue of Rosenberg's result above is still true, but now the generators are a little
more complicated: lens spaces and lens space bundles over $S^2$. Specifically, regard $S^2$ and $\CP^1$ and
consider a vector bundle which is a sum of complex line bundles $\mathcal L_1\oplus\dotsb\oplus \mathcal
L_{m-1}\to\CP^1$. Act by $\Z/2^\ell\Z$ as multiplication by the $2^\ell$th roots of unity on the unit sphere bundle of
this vector bundle; this is a free action and the quotient is an $L_{2^\ell}^{2m-3}$-bundle over $S^2$. We will
specifically let $Q_{2^\ell}^{2m-1}$ denote the case that $\mathcal L_1 = \mathcal O(2)$ and the remaining line
bundles are trivial.  See Appendix~\ref{ss:eta_Q} for more information.
\begin{thm}[{Botvinnik-Gilkey-Stolz~\cite[Section 5]{BGS97}}]
For values $\ell\ge 2$, the subgroup $\ko_{2m-1}(B\Z/2^\ell \Z)$ of $\Omega_{2m-1}^\Spin(B\Z/\ell\Z)$ is generated
by lens spaces and lens space bundles over $S^2$ of the form $Q_4^{2m-1}$ described above, and is detected by
$\eta$-invariants.
\end{thm}
This will be good enough --- in \cref{tab:SLprime2data}, we saw that if $k$ is even and less than $11$,
$\widetilde{\ko}_k(B\Z/4\Z)$ vanishes except in degrees $2$ and $10$, where it is isomorphic to $\Z/2\Z$. These two
$\Z/2\Z$ summands are the products of lens spaces with Spin manifolds, specifically $L_1^4\times S_p^1$
and $L_1^4\times S_p^1\times\HP^2$. Thus we know how to find the entire $\ko_k$ summand of
$\tOmega_k^\Spin(B\Z/4\Z)$.\footnote{Alternatively, Botvinnik-Gilkey~\cite[Theorem 4.5]{BG97} tell us how to detect
all of $\ko_k(B\Z/2^\ell\Z)$.} The generators of the $\ko_{k-8}$ summands are the direct products of the generators
of the $\ko_k$ summands with $\HP^2$. And for $\ko\ang 2_{k-10}$, \cref{ko2_is_homology} (identifying $\ko\ang
2_k(X)\cong H_k(X;\Z/2\Z)$ for $k\le 1$ and $X$ connective) implies that we can detect generators using mod $2$
cohomology characteristic classes, giving us $X_{10}\times L_1^4$ generating $\ko\ang 2_1(B\Z/4\Z)\inj
\Omega_{11}^\Spin(B\Z/4\Z)$.

Once again we refer to Appendix~\ref{app:eta} for formulas and calculations of $\eta$-invariants of lens spaces and
lens space bundles. We learn from \cref{tab:SL2inv} the following generators; in all cases the duality bundles are
induced from $S^{2k-1}\to L_4^{2k-1}$ and the inclusion $\Z/4\Z\inj\SL(2, \Z)$ like in the $3$-torsion case. This
time there is the subtlety that $L_4^{2k-1}$ has two Spin structures for $k$ even and zero Spin structures for $k >
1$ and odd.
\begin{itemize}
	\item $L_4^1\cong S^1$ with either of its Spin structures generates the remaining $\Z/4\Z$ summand of
	$\Omega_1^\Spin\big(B\SL(2, \Z)\big)$.
	\item $L_4^3$ with either of its Spin structures generates the $\Z/8\Z$ summand of $\Omega_3^\Spin\big(B\SL(2,
	\Z)\big)$.
	\item The lens space bundle $Q_4^5$ with either of its Spin structures generates the $\Z/4\Z$ summand of
	$\Omega_5^\Spin \big(B\SL(2, \Z)\big)$.
	\item $L_4^7$ with its two Spin structures generates the remaining $(\Z/32\Z) \oplus (\Z/2\Z)$ summands of
	$\Omega_7^\Spin \big(B\SL(2, \Z)\big)$.
	\item The lens space bundle $Q_4^9$ with either of its Spin structures generates the $\Z/8\Z$ summand of
	$\Omega_9^\Spin\big(B\SL(2, \Z)\big)$. Similarly to the $3$-torsion case, $\HP^2\times L_4^1$ generates the $\Z/4\Z$
	summand in $\Omega_9^\Spin\big(B\SL(2, \Z)\big)$.
	\item $L_4^{11}$ with its two Spin structures generates the two summands $(\Z/128\Z) \oplus (\Z/8\Z)$ of
	$\Omega_{11}^\Spin(B\SL(2, \Z))$. Similarly to the $3$-torsion case, $\HP^2\times L_4^3$ generates the remaining $\Z/8\Z$ summand in $\Omega_{11}^\Spin\big(B\SL(2, \Z)\big)$.
\end{itemize}
%To compute $\eta$-invariants of lens spaces and lens space bundles, we use the formulas of~\cite[Theorem
%4.7]{BGS97} and \textcolor{red}{(TODO: APS2)\dots also reference Appendix}
%
%\textcolor{red}{TODO: table! Probably also somewhere we need to fix notation for $\eta$-invariants.}
%
%
%
%
%
%
%\textcolor{red}{Right now this subsection is just to keep track of references.
%\begin{itemize}
%	\item If $\ell \ge 2$, $\ko_*(B\Z/2^\ell)\subset\Omega_*^\Spin(B\Z/2^\ell)$ is detected by the $\widehat
%	A$-genus and $\eta$-invariants. This is in Botvinnik-Gilkey~\cite[Theorem 4.5]{BG97}. Earlier, weaker result:
%	$\ko_{2k+1}(B\Z/2^\ell)$ is detected by eta invariants, and generated by lens spaces and lens space bundles
%	over $S^2$~\cite[\S 5]{BGS97}.
%	\item $\Z/p$ for $p$ an odd prime: Rosenberg~\cite[Proof of Theorem 2.12]{Ros86} shows that
%	$\Omega_*^\Spin(B\Z/p)$ is generated by lens spaces and products of lens spaces with spin manifolds.
%	\item Formula for eta invariants of lens spaces: $\pi_1 = \Z/2^k$: \cite[Theorem 4.7]{BGS97}.
%\end{itemize}
%Sketch of the structure of this subsection:
%\begin{itemize}
%	\item We know these invariants completely detect the bordism groups ($\eta$-invariants, $\widehat A$-genus, and
%	then taking products with $\HP^2$ and $X_{10}$)
%	\item We have a formula for the $\eta$-invariants of lens spaces
%	\item Here's a table of $\eta$-invariants: QED, these are our generators.
%\end{itemize}}

\section{Computation of $\Omega_*^{\Spin\text{-}\Mp(2, \Z)} (\pt)$}
	\label{mp_spin}

The defining representation of $\SL(2, \Z)$ is not Spin; its Spin cover is called the \term{metaplectic group}
$\Mp(2, \Z)$. Taking the quotient by the central $\Z/2\Z$ subgroup of the Spin group defines a short exact sequence
\begin{equation}
\label{SLMp}
	\shortexact*{\Z/2\Z}{\Mp(2, \Z)}{\SL(2, \Z)} \,.
\end{equation}
The S-duality symmetry of type IIB string theory, a priori an $\SL(2, \Z)$ symmetry, mixes with fermion parity to
form the extension~\eqref{SLMp}, so to study duality defects in type IIB string theory, as we do in
Section \ref{sec:Mpdefects}, we should compute the bordism groups of manifolds with structure group
\begin{equation}
	\Spin\text{-}\Mp(2, \Z) \coloneqq \Spin\times_{\Z/2\Z}\Mp(2, \Z) \,.
\end{equation}
In this section, we compute Spin-$\Mp(2, \Z)$ bordism groups in dimensions $11$ and below. Our approach is similar
to how we determined $\Omega_*^\Spin \big(B\SL(2, \Z)\big)$ in Section \ref{sl2_spin}: we use the amalgamated product
decomposition~\eqref{mp_amalg}
\begin{equation}
	\Mp(2, \Z)\overset\cong\longrightarrow (\Z/12\Z) *_{(\Z/4\Z)} (\Z/8\Z)
\end{equation}
to express Spin-$\Mp(2, \Z)$ in simpler terms at $p = 2$ and at odd primes. We tackle odd primes in
Section \ref{ss:mp2_at_3}, where we find that for any odd prime $p$, Spin-$\Mp(2, \Z)$ bordism is $p$-locally isomorphic
to $\Omega_*^\Spin(B\Z/3\Z)$ (\cref{mp2_at_3}), so we can reuse our computations from Section \ref{sl2_odd_primes}. In
Section \ref{ss:spin_Z8}, we localize at $2$ and learn that Spin-$\Mp(2, \Z)$ bordism is $2$-locally isomorphic to
bordism for the group $\Spin\text{-}\Z/8\Z\coloneqq \Spin\times_{\Z/2\Z}\Z/8\Z$. For $k > 5$,
$\Omega_k^{\Spin\text{-}\Z/8\Z} (\pt)$ is not in the literature. We compute it up to dimension $11$ using the
Adams spectral sequence, following a strategy of Campbell~\cite[Section 7.9]{Cam17} and Davighi-Lohitsiri~\cite[Appendix
A.4]{DL20b}.

In Section \ref{mp2_gens}, we produce generators for the Spin-$\Mp(2, \Z)$ bordism groups we computed, following the same
strategy that we did in Section \ref{sl2_gens} for $\Omega_*^\Spin \big(B\SL(2, \Z)\big)$. Then in Section \ref{mult_Z4}, we discuss the
ring structure on $\Omega_*^{\Spin\text{-}\Z/8\Z} (\pt)$.

\subsection{Working at odd primes}
\label{ss:mp2_at_3}

\begin{lem}
\label{mp2_at_3}
Let $p$ be an odd prime.
The inclusion $\Z/12 \Z \inj \Mp(2, \Z)$ induces a $p$-local equivalence
\begin{align}
\Omega_k^\Spin(B\Z/3\Z)\longrightarrow \Omega_k^{\Spin\text{-}\Mp(2, \Z)} (\pt) \,.
\end{align}
\end{lem}
\begin{proof}
Just as in the proofs of \cref{sl2_at_2,sl2_at_3_lem}, the inclusion $\Z/12\Z\hookrightarrow \Mp(2, \Z)$ can be written as
\begin{equation}
	\Z/12\Z = (\Z/4\Z) *_{\Z/4\Z}(\Z/12\Z) \longrightarrow (\Z/8\Z) *_{\Z/4\Z} (\Z/12\Z) = \Mp(2, \Z) \,,
\end{equation}
so the induced map on $\Z_{(p)}$-homology is an isomorphism. Therefore the map
\begin{equation}
	B(\Spin\times_{\Z/2\Z}\Z/12\Z)\longrightarrow B\big(\Spin\times_{\Z/2\Z}\Mp(2, \Z)\big)
\end{equation}
also induces an isomorphism on $\Z_{(p)}$-homology; this map also intertwines the maps down to $B\mathrm O$ given by
forgetting $\Z/12\Z$ or $\Mp(2, \Z)$, so it induces a map of Thom spectra
\begin{equation}
\label{first_MT_coh}
	\mathit{MT}(\Spin\times_{\Z/2\Z}\Z/12\Z)\longrightarrow \mathit{MT}\big(\Spin\times_{\Z/2\Z}\Mp(2, \Z)\big) \,,
\end{equation}
which by the Thom isomorphism is also an isomorphism on $\Z_{(p)}$-homology, hence by the stable Whitehead
theorem~\cite[Chapitre III, Théorème 3]{Ser53} is also an isomorphism on $p$-local homotopy groups.

Next, the map $i\colon \Z/6\Z\hookrightarrow\Z/12\Z$ is an isomorphism on $p$-local homology, because under the
isomorphisms $\Z/6\Z\cong\Z/2\Z\times\Z/3\Z$ and $\Z/12\Z\cong\Z/4\Z\times\Z/3\Z$, $i$ can be identified with the
inclusion $\Z/2\Z\hookrightarrow \Z/4\Z$ and the identity on $\Z/3\Z$. Therefore the induced map of Thom spectra
\begin{equation}
	\mathit{MT}(\Spin\times_{\Z/2\Z}\Z/6\Z)\longrightarrow \mathit{MT}\big(\Spin\times_{\Z/2\Z}\Z/12\Z\big) \,,
\end{equation}
induces an isomorphism on $\Z_{(p)}$-homology, hence also on $p$-local homotopy groups, just
like~\eqref{first_MT_coh}.

Lastly, because $\Z/6\Z\cong\Z/2\Z\times\Z/3\Z$, $\Spin\times_{\Z/2\Z}\Z/6\Z\cong\Spin\times\Z/3\Z \,$.
\end{proof}
%The inclusion $\mathbb{Z} / 3 \mathbb{Z} \rightarrow \mathbb{Z} / 6 \mathbb{Z}$
%induces an isomorphism on $\mathbb{Z}_{(p)}$-cohomology, so we may replace $B \mathbb{Z}/6 \mathbb{Z}$ with $B \mathbb{Z}/3 \mathbb{Z}$.
%
So at odd primes we just need $\Omega_k^\Spin(B\Z/3\Z)$, which we determined in Section \ref{sl2_odd_primes}.
The generators are the same as the $3$-torsion generators we found in Section \ref{sl2_gens}.

\subsection{Working at at $p = 2$}
\label{ss:spin_Z8}

In this section, we run the Adams spectral sequence for the $2$-primary part of $\Omega_*^{\Spin\text{-}\Mp(2,
\Z)} (\pt)$. There are several hidden extensions we have to resolve, and we resolve them by finding generators for the
bordism groups. We find several non-split hidden extensions.

We let Spin-$\Z/8\Z$ denote the symmetry type $\Spin \times_{\Z/2\Z}\Z/8\Z$ with the map to $\mathrm O$ projection
onto the first factor. Spin-$\Z/8\Z$ bordism has been studied in~\cite{Cam17, Hsi18, Garcia-Etxebarria:2018ajm,
DL20b, Hsieh:2020jpj, Deb21}, but only in dimensions $5$ and below. Our computations in dimensions 6-11 are new.
\begin{lem}
\label{mp2_at_2}
The inclusion $\Z/8\Z \inj \Mp(2, \Z)$ induces a $2$-local equivalence
\begin{align}
\Omega_*^{\Spin\text{-}\Z/8\Z} (\pt) \longrightarrow \Omega_*^{\Spin\text{-}\Mp(2, \Z)} (\pt) \,.
\end{align}
\end{lem}
The proof follows the same line of reasoning as \cref{mp2_at_3}, except that we cannot untwist like in the last
line of that proof.

Next, we shear Spin-$\Z/8\Z$ bordism.
\begin{lem}
\label{Z4_coh}
$H^*(B\Z/4\Z;\Z/2\Z)\cong\Z/2\Z[x, y]/(x^2)$, where $\abs x = 1$ and $\abs y = 2$. The class $y$ can be
characterized as follows:
\begin{enumerate}
	\item Let $\rho\colon \Z/4\Z\to\SSO(2)$ denote the rotation representation and $(E\Z/4\Z)_\rho\to B\Z/4\Z$
	denote the associated rank-$2$ vector bundle; then $w_2 \big( (E\Z/4\Z)_\rho \big) = y$.
	\item The cohomology class of the central extension
	\begin{equation}
		\label{Z8_cext}
		\shortexact{\Z/2\Z}{\Z/8\Z}{\Z/4\Z}{}
	\end{equation}
	is equal to $y$.
\end{enumerate}
\end{lem}
\begin{proof}
The cohomology ring is standard; see, e.g.\ \cite[Proposition 4.5.1]{CTVZ03}. Since the group
$H^2(B\Z/4\Z;\Z/2\Z)\cong\Z/2\Z$, there is only one nonzero element, so for the rest of the lemma it suffices to
show that $w_2 \big( (E\Z/4\Z)_\rho \big)$ and~\eqref{Z8_cext} are non-trivial. For the former, this follows because $\rho$
does not lift to $\Spin (2)$, and for the latter, this follows because~\eqref{Z8_cext} is not split.
\end{proof}
%For $k$ even, let $\rho\colon\Z/k\Z\to\SSO (2)$ denote the real rotation
%representation, and given a principal $(\Z/k\Z)$-bundle $P\to M$, let $P_\rho\to M$ denote the associated rank-two
%vector bundle. Since  when $k$ is even,
%$w_2(\rho)$ is nonzero.  $H^2(B\Z/k\Z;\Z/2\Z)\cong\Z/2\Z$, so this determines $w_2(\rho)$. Likewise, the central extension
%is not split for even $k$, so its cohomology class is nonzero, and therefore equal to $w_2(\rho)$. So
Thus \cref{w2_twisted_cor} tells us that Spin-$\Z/8\Z$ structures are naturally equivalent to $(B\Z/4\Z,
(E\Z/4\Z)_\rho)$-twisted Spin structures. For ease of reading, we will use $\rho$ to denote both the
representation and the associated vector bundle that we have been calling $(E\Z/4\Z)_\rho$; the specific meaning
will be clear from context. Thus Spin-$\Z/8\Z$ structures are naturally equivalent to $(B\Z/4\Z, \rho)$-twisted
Spin structures, and \cref{twisted_classifier} implies
\begin{equation}
	\Omega_*^{\Spin\text{-}\Z/8\Z} (\pt) \cong \Omega_*^\Spin \big((B\Z/4\Z)^{\rho-2} \big) \,.
\end{equation}
%
%(\TODO: output of shearing) says we must consider spin-$\Z/8$ bordism, i.e.\ the spin bordism of $(B\Z/4)^{V-2}$,
%where $V$ is the associated vector bundle to the real two-dimensional rotation representation of $\Z/4$.
%\textcolor{red}{TODO: cite not just the papers which did the Adams computation, but other researchers who considered spin-$\Z/2^k$
%bordism.}
By~\eqref{ABPdecomp}, we need to compute $\ko_* \big((B\Z/4\Z)^{\rho-2} \big)$ in degrees $11$ and below as well as $\ko\ang
2_* \big((B\Z/4\Z)^{\rho-2}\big)$ in degrees $1$ and below. \Cref{ko2_is_homology} and the mod $2$ Thom isomorphism take
care of the latter: we learn $\ko\ang 2_k((B\Z/4\Z)^{\rho-2})$ is $\Z/2\Z$ for $k = 0$ and $k = 1$.
\begin{thm}
\label{ko_spin_z8}
The first several $\ko$-homology groups of $(B\Z/4\Z)^{\rho-2}$ are:
%\begin{table}[h!]
%\centering
%\begin{tabular}{c c}
%\toprule
%$k$ & $\ko_k\big((B\Z/4\Z)^{\rho-2}\big)$ \\
%\midrule
%$0$ & {\footnotesize $\Z$} \\
%$1$ & {\footnotesize $\Z/8\Z$} \\
%$2$ & {\footnotesize $0$} \\
%$3$ & {\footnotesize $\Z/2\Z$} \\
%$4$ & {\footnotesize $\Z$} \\
%$5$ & {\footnotesize $(\Z/32\Z) \oplus (\Z/2\Z)$} \\
%$6$ & {\footnotesize $0$} \\
%$7$  & {\footnotesize $\Z/4\Z$} \\
%$8$ & {\footnotesize $\Z$} \\
%$9$ & {\footnotesize $(\Z/128\Z) \oplus (\Z/4\Z)$} \\
%$10$ & {\footnotesize $0$} \\
%$11$ & {\footnotesize $\Z/8\Z$} \\
%\bottomrule
%\end{tabular}
%\caption{The groups $\ko_k\big((B\Z/4\Z)^{\rho-2}\big)$ for $k \leq 11$.}
%\end{table}
%
%\textcolor{red}{Maybe the table is a bit much... just having a single column}

\begin{equation}
\begin{alignedat}{2}
	\ko_0 \big((B\Z/4\Z)^{\rho-2}\big) &\cong \Z\qquad\qquad\qquad\quad \quad & \ko_6 \big((B\Z/4\Z)^{\rho-2}\big) &\cong 0\\
	\ko_1\big((B\Z/4\Z)^{\rho-2}\big) &\cong \Z/8\Z & \ko_7 \big((B\Z/4\Z)^{\rho-2}\big) &\cong \Z/4\Z\\
	\ko_2\big((B\Z/4\Z)^{\rho-2}\big) &\cong 0  & \ko_8\big((B\Z/4\Z)^{\rho-2}\big) &\cong \Z\\
	\ko_3\big((B\Z/4\Z)^{\rho-2}\big) &\cong \Z/2\Z & \ko_9\big((B\Z/4\Z)^{\rho-2}\big) &\cong (\Z/128\Z) \oplus (\Z/4\Z) \\
	\ko_4\big((B\Z/4\Z)^{\rho-2}\big) &\cong \Z  & \ko_{10}\big((B\Z/4\Z)^{\rho-2}\big) &\cong 0\\
	\ko_5\big((B\Z/4\Z)^{\rho-2}\big) &\cong (\Z/32\Z) \oplus (\Z/2\Z) & \ko_{11}\big((B\Z/4\Z)^{\rho-2}\big) &\cong \Z/8\Z \,.
\end{alignedat}
\end{equation}
%\textcolor{red}{TODO: add in $\ko_*((B\Z/4)^{V-2})$ in degrees $12$ and below.}
\end{thm}
Barrera-Yanez~\cite[Theorem 3.1]{BY99} computes $\ko_{2k+1} \big((B\Z/4\Z)^{\rho-2} \big)$ using analytic methods. We
nonetheless work through the spectral sequence computation because we will need it when we study Spin-$\GL^+(2,
\Z)$ bordism in Section \ref{gl_spin}.
\begin{proof}
For ease of reading, let $X\coloneqq (B\Z/4\Z)^{\rho-2}$. We use the Adams spectral sequence, as we reviewed in
Section \ref{the_Adams_SS}. The steps of the problem are:
\begin{enumerate}
	\item Determine the $\cA(1)$-module structure on $H^*(X;\Z/2\Z)$ using the techniques we
	reviewed in Section \ref{ss:steenrod}.
	\item Use this to compute $\Ext_{\cA(1)}^{s,t} \big(H^*(X ;\Z/2\Z), \Z/2\Z \big)$, which is the
	$E_2$-page of the Adams spectral sequence. We discussed Ext and how to determine it in Section \ref{ss:ext}.
	\item Compute the differentials in the spectral sequence. We discussed differentials in general Adams spectral
	sequences in Section \ref{ss:differentials}; in this example, (a variant of) the May-Milgram theorem computes all
	the differentials for us. We discuss this tool in more detail in Appendix~\ref{s:MM_appendix}.
	\item Finally, we need to resolve some extension questions; we discussed the generalities of extension
	questions in Section \ref{ss:extensions}. To resolve the extension questions in the calculation of $\ko_*(X)$, we use
	calculations of $\eta$-invariants of lens spaces and lens space bundles from Appendix \ref{subsec:exteta}.
\end{enumerate}
First, the $\cA(1)$-module structure on cohomology. The Thom isomorphism~\eqref{Thom_iso} provides an isomorphism
$U\colon H^*(B\Z/4\Z;\Z/2\Z)\to H^*(X;\Z/2\Z)$; the class $U\coloneqq U(1)$ is called the \term{Thom class}.

\Cref{Steenrod_Thom} computes $\Sq(U\alpha)$ in terms of $\Sq(\alpha)$ and $w(\rho)$. We just need
$\Sq^1$ and $\Sq^2$, so we just need $w_1$ and $w_2$. Since $\rho$ has image contained in $\SSO(2)$, $w_1$
vanishes, and \cref{Z4_coh} tells us $w_2 = y$. We calculated the Steenrod squares of classes in
$H^*(B\Z/4\Z;\Z/2\Z)$ in \cref{cyclic_steenrod}. We learn $\Sq^1(x)= 0$, $\Sq^2(x) = 0$, $\Sq^1(y) = 0$, and
$\Sq^2(y) = y^2$. Using \Cref{Steenrod_Thom}, which tells us $\Sq^1(U\alpha) = \Sq^1(U)\alpha + U\Sq^1(\alpha)$ and
$\Sq^2(U\alpha) = \Sq^2(U)\alpha + \Sq^1(U)\Sq^1(\alpha) + U\Sq^2(\alpha)$, we learn that $\Sq^1$ vanishes on all
classes in $H^*(X;\Z/2\Z)$, and that
\begin{equation}
	\Sq^2(U\alpha) = U\alpha y + U\Sq^2(\alpha).
\end{equation}
Using this, we can completely describe the $\cA(1)$-module structure on $H^*(X;\Z/2\Z)$. This module splits as a
direct sum of modules which are two-dimensional vector spaces over $\Z/2\Z$. Specifically, if $C\eta$ denotes the
$\cA(1)$-module consisting of $\Z/2\Z$ summands in degrees $0$ and $2$ connected by a $\Sq^2$, then
\begin{equation}
\label{bunch_of_Cetas}
	H^*\big((B\Z/4\Z)^{\rho-2};\Z/2\Z\big) \cong
		\textcolor{BrickRed}{C\eta} \oplus
		\textcolor{RedOrange}{\Sigma C\eta} \oplus
		\textcolor{Goldenrod!67!black}{\Sigma^4 C\eta} \oplus
		\textcolor{Green}{\Sigma^5 C\eta} \oplus
		\textcolor{PineGreen}{\Sigma^8 C\eta} \oplus
		\textcolor{MidnightBlue}{\Sigma^9 C\eta} \oplus
		\textcolor{Fuchsia}{\Sigma^{12} C\eta} \oplus P,
\end{equation}
where $P$ is \term{$12$-connected}, i.e.\ has no nonzero classes in degrees $12$ or below. We only care about the
Adams spectral sequence in degrees $12$ and below, which means we can (and will) ignore $P$.
We draw the $\cA(1)$-module structure on $H^*(X;\Z/2\Z)$ in \cref{Z4Thomcoh}.

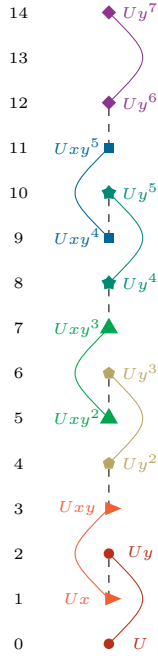
\begin{figure}[h!]
\centering
\begin{tikzpicture}[scale=0.6, every node/.style = {font=\tiny}]
\foreach \y in {0, ..., 14} {
	\node at (-2, \y) {$\y$};
}
\foreach \y in {1, 3, ..., 11} {
	\drawbock(0, \y);
}
\begin{scope}[BrickRed]
	\tikzptR{0}{0}{$U$}{};
	\tikzptR{0}{2}{$Uy$}{};
	\sqtwoR(0, 0);
\end{scope}
\begin{scope}[RedOrange]
	\tikzpt{0}{1}{$Ux$}{isosceles triangle};
	\tikzpt{0}{3}{$Uxy$}{isosceles triangle};
	\sqtwoL(0, 1);
\end{scope}
\begin{scope}[Goldenrod!67!black]
	\tikzptR{0}{4}{$Uy^2$}{regular polygon,regular polygon sides=5};
	\tikzptR{0}{6}{$Uy^3$}{regular polygon,regular polygon sides=5};
	\sqtwoR(0, 4);
\end{scope}
\begin{scope}[Green]
	\tikzpt{0}{5}{$Uxy^2$}{regular polygon,regular polygon sides=3};
	\tikzpt{0}{7}{$Uxy^3$}{regular polygon,regular polygon sides=3};
	\sqtwoL(0, 5);
\end{scope}
\begin{scope}[PineGreen]
	\tikzptR{0}{8}{$Uy^4$}{star};
	\tikzptR{0}{10}{$Uy^5$}{star};
	\sqtwoR(0, 8);
\end{scope}
\begin{scope}[MidnightBlue]
	\tikzpt{0}{9}{$Uxy^4$}{rectangle, minimum size=3.5pt};
	\tikzpt{0}{11}{$Uxy^5$}{rectangle, minimum size=3.5pt};
	\sqtwoL(0, 9);
\end{scope}
\begin{scope}[Fuchsia]
	\tikzptR{0}{12}{$Uy^6$}{diamond};
	\tikzptR{0}{14}{$Uy^7$}{diamond};
	\sqtwoR(0, 12);
\end{scope}
\end{tikzpicture}
\caption{The mod $2$ cohomology of $H^*\big((B\Z/4\Z)^{\rho-2};\Z/2\Z\big)$ in low degrees. Different colors/shapes
represent different summands, and the curved lines represent $\Sq^2$-actions as described in Section
\ref{ss:steenrod}. The dashed lines indicate the presence of a non-trivial Bockstein $\beta\colon H^k
\big((B\Z/4\Z)^{\rho-2}; \Z/4\Z \big)\to H^{k+1}\big((B\Z/4\Z)^{\rho-2};\Z/4\Z\big)$ from a preimage of the lower
class to a preimage of the higher class. This figure is complete in degrees $12$ and below.}
\label{Z4Thomcoh}
\end{figure}

Now to determine the Adams $E_2$-page.
%This does not require any Bockstein knowledge, so looking at
%\cref{Z4Thomcoh}, the $\cA(1)$-module $H^*((B\Z/2^n\Z)^{V-2};\Z/2\Z)$ splits as a direct sum of
%modules which are two-dimensional vector spaces over $\Z/2\Z$. Specifically, if $C\eta$ denotes the $\cA(1)$-module
The $E_2$-page of the Adams spectral sequence splits as the direct sum of the Ext groups of all of the summands we
identified in~\eqref{bunch_of_Cetas}. For a single $C\eta$ summand, one can look this up in, e.g.,
Beaudry-Campbell~\cite[Figure 22]{BC18}, so we can draw the $E_2$-page in \cref{the_E2_page}, top.
\begin{figure}
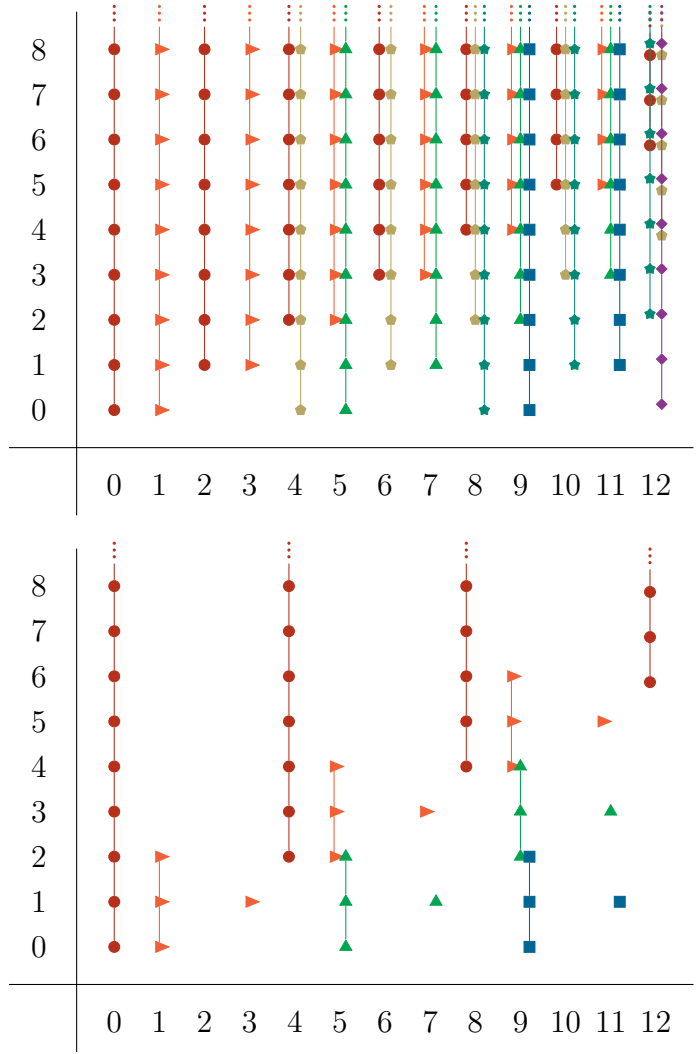

\centering
\begin{subfigure}[c]{0.6\textwidth}
\begin{sseqdata}[name=ThomZ4E2, classes=fill, xrange={0}{12}, yrange={0}{8}, scale=0.6, Adams grading, >=stealth,
class labels = {below = 0.05em, font=\tiny}]
\begin{scope}[BrickRed]
	\foreach \y in {0, ..., 6} {
		\class(2*\y, \y)\AdamsTower{}
	}
	%\classoptions["U"](0, 0)
\end{scope}

\begin{scope}[RedOrange, isosceles triangle]
	\foreach \y in {0, ..., 5} {
		\class(2*\y+1, \y)\AdamsTower{}
	}
	%\classoptions["Ux"](1, 0)
\end{scope}
\begin{scope}[draw=none, fill=none] % makes display prettier
	\class(4, 0)
	\class(4, 1)
	\class(5, 0)\class(5, 1)
	\class(6, 1)
	\class(6, 2)
	\class(8, 0)\class(8, 0)
	\class(8, 1)\class(8, 1)
	\class(8, 2)
	\class(8, 3)
	\class(9, 0)\class(9, 0)
	\class(9, 1)\class(9, 1)
	\class(7, 1)
	\class(7, 2)
	\class(9, 2)\class(9, 3)
	\class(10, 1)\class(10, 1)
	\class(10, 2)\class(10, 2)
	\class(10, 3)\class(10, 4)
	\class(11, 1)\class(11, 1)
	\class(11, 2)\class(11, 2)
	\class(11, 3)\class(11, 4)
	\class(12, 0)\class(12, 0)\class(12, 0)
	\class(12, 1)\class(12, 1)\class(12, 1)
	\class(12, 2)\class(12, 2)
	\class(12, 3)\class(12, 3)
	\class(12, 4)\class(12, 5)
\end{scope}

\begin{scope}[Goldenrod!67!black, regular polygon, regular polygon sides=5]
	\foreach \y in {0, ..., 4} {
		\class(2*\y+4, \y)\AdamsTower{}
	}
	%\classoptions["Uy^2"](4, 0, 2)
\end{scope}

\begin{scope}[Green, regular polygon, regular polygon sides=3, minimum width=1ex]
	\foreach \y in {0, ..., 3} {
		\class(2*\y+5, \y)\AdamsTower{}
	}
	%\classoptions["Uxy^2"](5, 0, -1)
\end{scope}

\begin{scope}[PineGreen, star]
	\foreach \y in {0, ..., 2} {
		\class(2*\y+8, \y)\AdamsTower{}
	}
	%\classoptions["Uy^4"](8, 0, -1)
\end{scope}

\begin{scope}[MidnightBlue, rectangle]
	\foreach \y in {0, 1} {
		\class(2*\y+9, \y)\AdamsTower{}
	}
	%\classoptions["Uxy^4"](9, 0, -1)
\end{scope}

\begin{scope}[Fuchsia, diamond]
	\class(12, 0)\AdamsTower{}
	%\classoptions["Uy^6"](12, 0, -1)
\end{scope}
% classes we need to display differentials, though they're out of range
\begin{scope}[draw=none, fill=none]
	\class(1, 10)
	\class(3, 10)
	\class(5, 10)\class(5, 10)
	\class(7, 10)\class(7, 10)
	\class(9, 10)\class(9, 10)\class(9, 10)
	\class(11, 10)\class(11, 10)\class(11, 10)
\end{scope}
% ok, now for some differentials
\foreach \y in {1, ..., 8} {
	\d2(2, \y)
	\d2(6, \y, 2)(5, \y+2, 2)
	\d2(10, \y, 3)(9, \y+2, 3)
}
\foreach \y in {0, ..., 8} {
	\d2(4, \y, 2)
	\d2(8, \y, 3)(7, \y+2, 2)
	\d2(12, \y, 4)(11, \y+2, 3)
}
\foreach \y in {3, ..., 8} {
	\d2(6, \y, 1)(5, \y+2, 1)
	\d2(10, \y, 2)(9, \y+2, 2)
}
\foreach \y in {2, ..., 8} {
	\d2(8, \y, 2)(7, \y+2, 1)
	\d2(12, \y, 3)(11, \y+2, 2)
}

\foreach \y in {5, ..., 8} {
	\d2(10, \y, 1)(9, \y+2, 1)
}

\foreach \y in {4, ..., 8} {
	\d2(12, \y, 2)(11, \y+2, 1)
}

\end{sseqdata}
\printpage[name=ThomZ4E2, page=1]
\end{subfigure}
\\\vspace{0.3cm}
\begin{subfigure}[c]{0.6\textwidth}
	\printpage[name=ThomZ4E2, page=3]
\end{subfigure}
\caption{The Adams spectral sequence computing
$\ko_*\big((B\Z/4\Z)^{\rho-2};\Z/2\Z\big) \,$. Top: the $E_2$-page. Bottom: the $E_\infty$-page.}
\label{the_E2_page}
\label{ThomZ4Einfty}
\end{figure}

There are plenty of differentials that we cannot rule out using the $h_i$-action. One can quickly check that many
of them are nonzero, e.g.\ by computing $\ko_*(X)\otimes\Q$ using the fact that, after tensoring with $\Q$, the
Atiyah-Hirzebruch spectral sequence always collapses~\cite[Bemerkung 14.18]{Dol66}.

The good news is, the $E_2$-page has enough structure to allow us to determine the differentials without too much
work.
\begin{defn}
An \term{$h_0$-tower} in the $E_r$-page of the Adams spectral sequence is an infinite sequence $x_i\in E_r^{s+i,
t-s-i}$ such that $h_0x_i = x_{i+1}$.
\end{defn}
Looking at the $E_2$-page in \cref{the_E2_page}, top, we see lots of $h_0$-towers in our spectral sequence. Because differentials commute
with $h_0$-actions, if we know $d_r$ on one element of an $h_0$-tower, we know $d_r$ on the entire tower; for this
reason it is common to refer to differentials between $h_0$-towers, instead of just elements.

May-Milgram~\cite{MM81} relate differentials between $h_0$-towers to Bockstein homomorphisms in cohomology. The
latter are typically not so hard to compute, suggesting an approach to solving our Adams spectral sequence: compute
Bocksteins in $H^*(X)$ and use them to determine Adams differentials. We discuss the May-Milgram theorem in greater
detail in Appendix~\ref{s:MM_appendix}.

We will actually use a variant of the May-Milgram theorem, \cref{MM_ceta}, tailored for spaces/spectra whose
cohomology is a direct sum of copies of $C\eta$. Let $\beta_r\colon H^n(\bl;\Z/2^r\Z)\to H^{n+1}(\bl;\Z/2^r\Z)$ be the
Bockstein associated to the short exact sequence
\begin{equation}
	\shortexact{\Z/2^r\Z}{\Z/2^{2r}\Z}{\Z/2^r\Z} \,;
\end{equation}
then \Cref{MM_ceta} says that if $\alpha\in H^n(X;\Z/2^r\Z)$ and $\beta_r(\alpha)\in H^{n+1}(X;\Z/2^r\Z)$ both
generate $\Z/2^r$ summands and the mod $2$ reductions of $\alpha$ and $\beta_r(\alpha)$ belong to $C\eta$
$\cA(1)$-module summands $M_0$, resp.\ $M_1$ of $H^*(X;\Z/2\Z)$, then there is a $d_r$ differential between the
$h_0$-towers in $\Ext_{\cA(1)}(X)$ in degrees $n+4k$ and $n+4k+1$ coming from $M_0$ and $M_1$, and conversely
differentials imply such Bockstein relations. We will use \cref{MM_ceta} as a black box in this section and direct
the curious reader to read the proof and context in Appendix~\ref{s:MM_appendix}.

The Thom isomorphism commutes with the Bockstein $\beta_r$ when $r > 1$,\footnote{For $r = 1$, one must be careful
with twisted versus untwisted coefficients.} so we can compute the action of $\beta_r$ in $H^*(X;\Z/2^r\Z)$ using
the action in $H^*(B\Z/4\Z;\Z/2^r\Z)$. It is a standard fact\footnote{To prove this standard fact, recall from
Appendix~\ref{s:MM_appendix} (specifically, around~\eqref{LES_diag}) that $\beta_r$ acts non-trivially on
$H^{n-1}(X;\Z/2^r\Z)$ if and only if $H^n(X;\Z)$ has $2^r$-torsion; then compute using the identification
$H^*(B\Z/4\Z;\Z)\cong\Z[y]/4y$, with $\abs y = 2$.} that there is a class $\alpha$ in degree $n$ satisfying the
conditions of \cref{MM_ceta} outlined in the previous paragraph exactly when $r = 2$ and $n$ is odd, and that the
mod $2$ reductions of $\alpha$ and $\beta_r(\alpha)$ are $Uxy^k$ and $Uy^{k+1}$ where $k = (n-1)/2$; thus, in
\cref{Z4Thomcoh}, we represent these Bocksteins with dotted lines.

Invoking \cref{MM_ceta} on these classes, we learn the following differentials.
\begin{enumerate}
	\item Suppose $x\in E_2^{s,t}$ is an element of a red $h_0$-tower, i.e.\ coming from the copy of
	$\textcolor{BrickRed}{C\eta}\subset H^*(X;\Z/2\Z)$. If $t-s \equiv 0\bmod 4$, $d_r(x) = 0$ for all $r$. If
	$t-s\equiv 2\bmod 4$, $d_2(x)$ is nonzero, and is an element of the $h_0$-tower in degree $t-s-1$ coming from
	$\Ext(\textcolor{RedOrange}{\Sigma C\eta})$ (right-pointing triangles). These differentials follow from the absence of any Bockstein in
	degree $0$ and the presence of a Bockstein on $\Z/4\Z$-cohomology lifts of $Ux$ and $Uy$.
	\item All differentials emerging from the $h_0$-towers coming from $\Ext(\textcolor{RedOrange}{\Sigma C\eta})$
	(right-pointing triangles)
	vanish, because $\Z/4\Z$-cohomology preimages of nonzero elements of $\textcolor{RedOrange}{\Sigma C\eta}$ are
	not in the image of $\beta_2$. An analogous statement is true for $\textcolor{Green}{\Sigma^5 C\eta}$
	(upward-pointing triangles) and
	$\textcolor{MidnightBlue}{\Sigma^{9}C\eta}$ (squares).
	\item If $x\in E_2^{s,t}$ is an element of an $h_0$-tower coming from
	$\Ext(\textcolor{Goldenrod!67!black}{\Sigma^4C\eta})$ (pentagons), $d_2(x)\ne 0$. If $t-s\equiv 0\bmod 4$, $d_2(x)\in
	\Ext(\textcolor{RedOrange}{\Sigma C\eta})$ (right-pointing triangles), and if $t-s\equiv 2\bmod 4$,
	$d_2(x)\in\Ext(\textcolor{Green}{\Sigma^5 C\eta})$ (upward-pointing triangles). These differentials follow from Bocksteins in
	$\Z/4\Z$-valued cohomology whose
	mod $2$ reductions are $Uxy\mapsto Uy^2$, resp.\ $Uxy^2\mapsto Uy^3$.
	\item If $x\in E_2^{s,t}$ is an element of an $h_0$-tower coming from
	$\Ext(\textcolor{PineGreen}{\Sigma^8C\eta})$ (stars), $d_2(x)\ne 0$. If $t-s\equiv 0\bmod 4$, $d_2(x)\in
	\Ext(\textcolor{Green}{\Sigma^5 C\eta})$ (upward-pointing triangles), and if $t-s\equiv 2\bmod 4$,
	$d_2(x)\in\Ext(\textcolor{MidnightBlue}{\Sigma^9 C\eta})$ (squares). These differentials follow from Bocksteins in
	$\Z/4\Z$-valued cohomology whose
	mod $2$ reductions are $Uxy^3\mapsto Uy^4$, resp.\ $Uxy^4\mapsto Uy^5$.
	\item If $x$ is an element of the $h_0$-tower in topological degree $12$ coming from
	$\Ext(\textcolor{Fuchsia}{\Sigma^{12}C\eta})$ (diamonds), then $d_2(x)\in \Ext(\textcolor{MidnightBlue}{\Sigma^9C\eta})$
	(squares),
	which follows from a $\beta_2$ whose mod $2$ reduction sends $Uxy^5\mapsto Uy^6$.
\end{enumerate}
One could continue, but we do not need to.
This is all of the differentials in the range
we care about, and we display the $E_\infty$-page in \cref{ThomZ4Einfty}, bottom. From this we can immediately deduce
all of the $\ko$-homology groups in range except in degrees $5$, $7$, $9$, and $11$, where there could be extension
problems. Specifically:
%\begin{figure}[h!]
%\centering
%\caption{The $E_\infty$-page of the Adams spectral sequence computing $\ko_*\big((B\Z/4\Z)^{\rho-2}\big)$.}
%\end{figure}
%
%Summary of extension problems:
\begin{itemize}
	\item In degree $5$, we could have $(\Z/8\Z) \oplus (\Z/8\Z)$, $(\Z/16\Z) \oplus (\Z/4\Z)$, or $(\Z/32\Z) \oplus (\Z/2\Z)$, depending on how
	the green pieces (upward-pointing triangles) interact with the orange pieces (rightward-pointing triangles) in the Adams filtration. %Hsieh~\cite[\S 2.2]{Hsi18} and
	%Hsieh-Tachikawa-Yonekura~\cite[\S 8.4]{Hsieh:2020jpj} show that we in fact get $(\Z/32\Z) \oplus (\Z/2\Z)$.
	\item In degree $7$, we could have $(\Z/2\Z) \oplus (\Z/2\Z)$ or $\Z/4\Z$.%, and \textcolor{red}{Miguel shows that we get $\Z/4\Z$}.
	\item In degree $9$, there are a lot of options. The answer must be an Abelian group of order $2^9$ and
	must have either two or three cyclic summands. No element has order greater than $2^7$. For example, we could
	have $(\Z/8\Z)^{\oplus 3}$, $(\Z/128\Z)\oplus (\Z/2\Z)^{\oplus2}$, $(\Z/128\Z) \oplus (\Z/4\Z)$, or other
	options. %\textcolor{red}{Miguel's
	%computations} show that we get $(\Z/128\Z) \oplus (\Z/4\Z)$.
	\item In degree $11$, we could have $(\Z/2\Z)^{\oplus 3}$, $(\Z/4\Z) \oplus (\Z/2\Z)$, or $\Z/8\Z$.
	%\textcolor{red}{Miguel's computations} give us $\Z/8\Z$.
\end{itemize}
These extension questions have been studied by Hsieh~\cite[Section 2.2]{Hsi18} and Hsieh-Tachikawa-Yonekura~\cite[Section
8.4]{Hsieh:2020jpj} in degree $5$ and by Barrera-Yanez~\cite[Theorem 3.1]{BY99} in all odd degrees. Their proofs
compute bordism invariants built from $\eta$-invariants. For example, they show that a particular $\eta$-invariant
evaluated on $L_4^5$ is equal to $1/32\bmod \Z$, and that the value of this $\eta$-invariant in $\R/\Z$ is a
bordism invariant of Spin-$\Z/8\Z$ manifolds, implying that $\Omega_5^{\Spin\text{-}\Z/8\Z} (\pt)$ has a $\Z/32\Z$
subgroup, hence must be $(\Z/32\Z) \oplus (\Z/2\Z)$. We resolve the extension problems in a similar way to
Barrera-Yanez, using formulas for $\eta$-invariants calculated by Atiyah-Patodi-Singer~\cite[Proposition
2.12]{APS2}, Donnelly~\cite[Section 4]{Don78}, Botvinnik-Gilkey-Stolz~\cite{BGS97}, and
Hsieh-Tachikawa-Yonekura~\cite[Appendix C]{Hsieh:2020jpj}. Specifically, in Section \ref{subsec:exteta}, if $L_4^{2k-1}$
denotes the lens space $S^{2k-1}/(\Z/4\Z)$ and $Q_4^{2k-1}$ is the $L_4^{2k-3}$-bundle over $S^2$ defined in
Section~\ref{sl2_gens},
%denotes the $L_4^{2k-3}$-bundle obtained by the
%fiberwise ``lensification'' quotient of $\mathbb P(\mathcal O(2)\oplus\underline\C^{k-2})\to\CP^1$ by $\Z/4\Z$,
then
\begin{equation}
\label{etas_spZ8}
\begin{alignedat}{2}
	\eta_{1/2}^{\text{D}} (L_4^5) &= -\tfrac{5}{32} \qquad\qquad & \eta_{1/2}^{\text{D}} (L_4^9) &= \phantom{-}\tfrac{9}{128}\\
	\eta_{3/2}^{\text{D}} (\widetilde L_4^5) &= -\tfrac{3}{32} & \eta_{3/2}^{\text{D}} (\widetilde L_4^9) &= -\tfrac{7}{128}\\
	(\eta_{3/2}^{\text{D}} - \eta_{1/2}^{\text{D}}) (Q_4^7) &= \phantom{-}\tfrac 14 & (\eta_{3/2}^{\text{D}} - \eta_{1/2}^{\text{D}}) (Q_4^{11}) &=
	\phantom{-}\tfrac 18 \,.
\end{alignedat}
\end{equation}
In order to say this we must specify the Spin-$\Z/8\Z$ structures on these lens space bundles; we do this in
Appendix~\ref{app:eta}, so see there for details.

Since the values of these $\eta$-invariants mod $\Z$ are bordism invariants, these imply that
$\ko_5(X)\cong (\Z/32\Z) \oplus (\Z/2\Z)$, $\ko_7(X)\cong\Z/4$, $\ko_9(X)\cong (\Z/128\Z) \oplus (\Z/4\Z)$, and
$\ko_{11}(X)\cong\Z/8\Z$.

Though all previous calculations of these groups used analytic methods, we can address the extension problem in
degree $5$ purely with algebra, providing a check on the $\eta$-invariant computations. It would be interesting to
try to generalize this to the extension questions in dimensions $7$, $9$, and $11$.
\begin{lem}
\label{Z4_5_extn}
$\ko_5 \big((B\Z/4\Z)^{\rho-2} \big)\cong (\Z/32\Z) \oplus (\Z/2\Z)$.
\end{lem}
\begin{proof}
Recall from \cref{ku_remark} the $\eta$, $c$, $R$ long exact sequence~\eqref{etaCR};
we will study this long exact sequence for the $\ko$- and $\ku$-homology of $X = (B\Z/4\Z)^{\rho-2}$. Since $V\to B\Z/4\Z$ is a complex vector bundle, it is in particular
Spin\textsuperscript{$c$}. The generalized cohomology theory $\ku$ is oriented for Spin\textsuperscript{$c$} vector
bundles~\cite{ABS, Joa04},\footnote{Alternatively, one can use the $\ku$-orientation for complex vector bundles
constructed by Conner-Floyd~\cite[Section 5]{CF66}.} so there is a Thom isomorphism
$\ku_*\big((B\Z/4\Z)^{\rho-2}\big)\cong\ku_*(B\Z/4\Z)$. Hashimoto~\cite{Has83} shows $\ku_5(B\Z/4\Z)\cong (\Z/16\Z) \oplus (\Z/2\Z)^{\oplus 2}$.

Our three options for $\ko_5\big((B\Z/4\Z)^{\rho-2}\big)$ are $(\Z/8\Z)^{\oplus 2}$, $(\Z/16\Z) \oplus (\Z/4\Z)$, and $(\Z/32\Z) \oplus (\Z/2\Z)$. None
of these Abelian groups can map injectively into $(\Z/16\Z) \oplus (\Z/2\Z)^{\oplus 2}$, so
$c\colon\ko_5\big((B\Z/4\Z)^{\rho-2}\big)\to\ku_5(B\Z/4\Z)$ is not injective. Exactness of~\eqref{etaCR} implies that the map
$\eta\colon\ko_4\big((B\Z/2\Z)^{\rho-2}\big)\to\ko_5((B\Z/2\Z)^{\rho-2})$ is nonzero;\footnote{One expects to see the action of
$\eta$ in the $E_\infty$-page of the Adams spectral sequence, as it is detected by $h_1\in\Ext(\Z/2\Z)$. This
$\eta$-action is not seen by the $E_\infty$-page, which makes it an example of a ``hidden multiplicative
extension.'' See Section \ref{mult_Z4}.} since $2\eta = 0$ in $\pi_1(\mathbb S)$, $2\eta(x) = 0$ for any
$x\in\ko_4\big((B\Z/4\Z)^{\rho-2}\big)$. Because $\ko_4\big((B\Z/4\Z)^{\rho-2}\big)$ has a single generator, the image of $\eta$ is a $\Z/2\Z$
subgroup of $\ko_5\big((B\Z/4\Z)^{\rho-2}\big)$, and exactness of~\eqref{etaCR} then implies
\begin{equation}
\label{injects}
	c/\ker(c)\colon \ko_5\big((B\Z/4\Z)^{\rho-2}\big)/(\Z/2\Z) \longrightarrow \ku_5\big((B\Z/4\Z)^{\rho-2}\big)\cong (\Z/16\Z) \oplus (\Z/2\Z)^{\oplus 2} \,,
\end{equation}
is injective. If $\ko_5\big((B\Z/4\Z)^{\rho-2}\big)$ were isomorphic to $(\Z/8\Z)^{\oplus 2}$, this could not
happen: after quotienting out by $\Z/2\Z$, there would be too many elements $x$ with $2x\ne 0$.

With a little more work, we can also eliminate the possibility $(\Z/16\Z) \oplus (\Z/4\Z)$: since $\eta$ has Adams
filtration $1$ in the Adams spectral sequence for the sphere spectrum, multiplication by $\eta$ must raise Adams
filtration by at least $1$. Looking at \cref{ThomZ4Einfty}, bottom, the generator $x$ of
$\ko_4\big((B\Z/4\Z)^{\rho-2}\big)$ has Adams filtration at least $2$, because its image in the $E_\infty$-page is
in filtration $2$. Thus $\eta x$ has filtration at least $3$ and $2\eta = 0$; looking at the $5$-line of the
$E_\infty$-page in \cref{ThomZ4Einfty}, bottom, this can only occur if the image of $\eta x$ in the $E_\infty$-page is the
generator of $E_\infty^{4,9}\cong\textcolor{RedOrange}{\Z/2\Z}$. In particular, the $h_0$-actions on the $5$-line
imply there is some class $y\in\ko_5\big((B\Z/4\Z)^{\rho-2}\big)$ such that $4y = \eta$. Thus if there is an
isomorphism $\ko_5\big((B\Z/4\Z)^{\rho-2}\big)\cong (\Z/16\Z) \oplus (\Z/4\Z)$, $\eta x$ is identified with either
$(8, 0)$ or $(8, 2)$. Thus the quotient by $\eta x$ is $(\Z/8\Z)\oplus (\Z/4\Z)$, which does not inject into
$(\Z/16\Z) \oplus (\Z/2\Z)^{\oplus 2}$.

The only remaining option is $\ko_5\big((B\Z/4\Z)^{\rho-2}\big)\cong (\Z/32\Z) \oplus (\Z/2\Z)$.
\end{proof}
This finishes the calculation of Spin-$\Z/8\Z$ bordism groups in dimensions $11$ and below.
\end{proof}

Collecting all the pieces we find the Spin-Mp$(2,\Z)$ bordism groups summarized in Table~\ref{tab:Mp_bordism}.
\renewcommand{\arraystretch}{1.5}
\begin{table}[h!]
\centering
\begin{tabular}{c c c}
\toprule
$k$ & $\Omega^{\text{Spin-Mp}(2,\Z)}_k (\pt)$ & Generators \\
\midrule
$0$ & {\footnotesize $\Z$} & {\footnotesize $\pt_+$} \\ %\hline
$1$ & {\footnotesize $(\Z/8\Z) \oplus (\Z/3\Z)$} & {\footnotesize $(L_1^4 \,, L^1_3)$} \\ %\hline
$2$ & {\footnotesize $0$} & \\ %\hline
$3$ & {\footnotesize $(\Z/2\Z) \oplus (\Z/3\Z)$} & {\footnotesize $(L^3_4 \,, L^3_3)$} \\ %\hline
$4$ & {\footnotesize $\Z$} & {\footnotesize $E$} \\ %\hline
$5$ & {\footnotesize $(\Z/2\Z) \oplus (\Z/32\Z) \oplus (\Z/9\Z)$} & {\footnotesize $(\widetilde{L}^5_4 \,, L^5_4 \,, L^5_3)$} \\ %\hline
$6$ & {\footnotesize $0$} & \\ %\hline
$7$ & {\footnotesize $(\Z/4\Z) \oplus (\Z/9\Z)$} & {\footnotesize $(Q^7_4 \,, L^7_3)$} \\ %\hline
$8$ & {\footnotesize $\Z \oplus \Z$} & {\footnotesize $(B, \HP^2)$} \\ %\hline
$9$ & {\footnotesize $(\Z/4\Z) \oplus (\Z/8\Z) \oplus (\Z/128\Z) \oplus (\Z/3\Z) \oplus (\Z/27\Z)$} & {\footnotesize $(\widetilde{L}^9_4 \,, \HP^2 \times L^1_4 \,, L^9_4 \,, \HP^2 \times L^1_3 \,, L^9_3 )$} \\ %\hline 	
$10$ & {\footnotesize $(\Z/2\Z)$} & {\footnotesize $X_{10}$} \\ %\hline
$11$ & {\footnotesize $(\Z/2\Z)^{\oplus 2} \oplus (\Z/8\Z) \oplus (\Z/3\Z) \oplus (\Z/27\Z)$} & {\footnotesize $(\HP^2 \times L^3_4 \,, X_{10} \times L^1_4 \,, Q^{11}_4 \,, \HP^2 \times L^3_3 \,, L^{11}_3)$} \\ %\hline
\bottomrule
\end{tabular}
\caption{Bordism groups $\Omega^{\text{Spin-Mp}(2,\Z)}_k (\pt)$ and their generators (in the same order as group summands) for $k \leq 11$.}
\label{tab:Mp_bordism}
\end{table}
For convenience we already include the various generating manifolds that we will explore next.

\subsection{Finding generators}
\label{mp2_gens}

At odd primes, there is no difference between $\Omega_*^\Spin\big(B\SL(2, \Z)\big)$ and
$\Omega_*^{\Spin\text{-}\Mp(2, \Z)} (\pt)$, so we can recycle the generators from Section \ref{sl2_gens}. And at prime
$2$ we do not have to modify our strategy much: Botvinnik-Gilkey~\cite[Theorem 4.5]{BG97} show that the $\hat
A$-genus and $\eta$-invariants completely detect the $\ko_*$ summand of Spin-$\Z/2^k\Z$ bordism when $k\ge 3$. We
already determined the specific lens spaces and $\eta$-invariants for this in and around~\eqref{etas_spZ8}, so we
have found the generators coming from the $2$-torsion in $\ko_*$. For $\ko_{*-8}$, much like we did for
$\Omega_*^\Spin(B\SL(2, \Z))$, take the product of $\HP^2$ with the corresponding generator of
$\ko_*\big((B\Z/4)^{\rho-2}\big)$; and for $\ko\ang 2_{*-10}$, \cref{ko2_is_homology} means we can detect generators with mod
$2$ cohomology classes in the degrees we need, giving us the generators $X_{10}$ in degree $10$ and $X_{10}\times
L_4^1$ in degree $11$.
%\textcolor{red}{Then: we make a table, and voilà, our generators!}
%
%
\subsection{Multiplicative structure}
\label{mult_Z4}
Because of the use of products in studying compactifications, as we discussed in \cref{mult_str_rem}, we discuss
how one can compute the product structure on Spin-$\Z/8\Z$ bordism in this subsection.

%The product of two Spin manifolds has an induced Spin structure, and this is compatible with the bordism
%equivalence relation, making $\Omega_\ast^\Spin$ into a $\Z$-graded commutative ring. This is likewise true for
%many other bordism theories, including $G$-bordism for $G = \mathrm O$, $\SSO$, $\Spin^c$, $\mathrm{String}$, and
%$\mathrm U$, though it is not always true, e.g.\ for $G = \mathrm{Pin}^\pm$. When this is true, the multiplication
%lifts from bordism rings to ring structures on the corresponding Thom spectra, which can be thought of as
%expressing the naturality of this multiplication with respect to $G$-bordism groups of spaces. Determining the ring
%structure, when present, was classically an important part of bordism theory, though it is less directly useful in
%physics applications.

The product of two Spin-$\Z/2k\Z$ manifolds $(M, P)$ and $(N, Q)$ (here $P\to M$ and $Q\to N$ are the associated
$\Z/k\Z$-bundles) has an induced Spin-$\Z/2k\Z$ structure, given in informal terms by adding the characteristic
classes of $P$ and $Q$. In a little more detail, let $\otimes$ denote the tensor product of $\Z/k\Z$-bundles, which
lifts the addition operation of their characteristic classes in $H^1(\text{--};\Z/k\Z)$. If $\pi_1\colon M\times
N\to M$ and $\pi_2\colon M\times N\to N$ are the projections onto the first, respectively second factor, then there
are isomorphisms
\begin{equation}
\begin{aligned}
	T(M\times N)\oplus (\pi_1^\ast P\otimes\pi_2^\ast Q)_\rho  &\overset\cong\longrightarrow \pi_1^\ast TM \oplus
	\pi_2^\ast TN\oplus \pi_1^\ast P_\rho \oplus \pi_2^\ast Q_\rho\\
	 &\cong \pi_1^\ast(TM\oplus P_\rho) \oplus \pi_2^\ast(TN\oplus Q_\rho),
\end{aligned}
\end{equation}
so Spin structures on $TM\oplus P_\rho$ and $TN\oplus Q_\rho$, i.e.\ Spin-$\Z/k\Z$ structures on $(M, P)$ and $(N,
Q)$, induce a Spin structure on $T(M\times N)\oplus (\pi_1^\ast P\otimes\pi_2^\ast Q)_\rho$, i.e.\ a Spin-$\Z/k\Z$
structure on $(M\times N, \pi_1^\ast P\otimes\pi_2^\ast Q)$. Allowing $P$ and/or $Q$ to be trivial shows this
product is compatible with the product on Spin bordism, implying that Spin-$\Z/2k\Z$ bordism is not just a ring, but
also an $\Omega_\ast^\Spin$-algebra.

We briefly describe how to make the ring structure on Spin-$\Z/8\Z$ bordism explicit in low dimensions. In high
dimensions, determining the complete ring structure on $\Omega_*^\Spin$ is still open (see, e.g.,~\cite{Lau03}), and the ring structure
on Spin-$\Z/8\Z$ bordism is likely to be yet more complicated, so we focus on dimensions $11$ and below. Most
products we might care about vanish for degree reasons, or because torsion cannot map non-trivially to free summands.
The only interesting products are of the form
\begin{subequations}
\begin{equation}
	\Omega_{4k}^{\Spin\text{-}\Z/8\Z}\times\Omega_{4\ell}^{\Spin\text{-}\Z/8\Z} \longrightarrow
	\Omega_{4(k+\ell)}^{\Spin\text{-}\Z/8\Z}
\end{equation}
and
\begin{equation}
\label{second_spinZe_prod}
\Omega_{4k}^{\Spin\text{-}\Z/8\Z}\times\Omega_{2\ell+1}^{\Spin\text{-}\Z/8\Z}\longrightarrow
\Omega_{4k+2\ell+1}^{\Spin\text{-}\Z/8\Z}.
\end{equation}
\end{subequations}
For products of the first type, both bordism groups are free Abelian, and so one can compute products by tensoring
with $\mathbb Q$ and using rational characteristic classes including Pontrjagin classes.\footnote{In high enough
dimensions, one has to also use analogues of Pontrjagin classes in $\mathit{KO}$-theory, which is a theorem of
Anderson-Brown-Peterson~\cite{ABP66}. See Freed-Hopkins~\cite[Appendix B]{FH21} for an example computation with
these characteristic classes.}

For products of type~\eqref{second_spinZe_prod}, Botvinnik-Gilkey~\cite{BG97} show that
$\Omega_{4k+2\ell+1}^{\Spin\text{-}\Z/8\Z}$ is detected by $\eta$-invariants, so one can compute the product
structure in Spin-$\Z/8\Z$ bordism by computing $\eta$-invariants of product manifolds, aided by decomposition
results of the form $\eta(M\times N) = \mathrm{Index}(M)\eta(N)$.

Most of the $\Omega_\ast^\Spin$-algebra structure on Spin-$\Z/8\Z$ bordism is determined by the action of
$\Ext(\Z/2\Z)$ on the $E_\infty$-page corresponding to Spin-$\Z/8\Z$ bordism, together with the extensions we found.
However, there can also be extensions involving not just multiplication by $2^m$, but also products with various
elements of $\Omega_\ast^\Spin$. For example, $\eta = [S^1_p]$ corresponds to
$h_1\in\Ext(\Z/2\Z)$, but there can be products of the form $y = \eta x$ such that, if $\overline x$ and $\overline
y$ are the images of $x$ and $y$ in the $E_\infty$-page, $h_1\overline x = 0$.

This happens in Spin-$\Z/8\Z$ bordism, and in the lowest degree possible! Looking at \cref{ThomZ4Einfty}, bottom, $h_1$ acts
trivially on the $0$-line, suggesting that $S^1_p$ bounds as a Spin-$\Z/8\Z$ manifold, but this is false!
Recall from the proof of \cref{Z4_5_extn} that
$\eta\colon\Omega_4^{\Spin\text{-}\Z/8\Z}\to\Omega_5^{\Spin\text{-}\Z/8\Z}$ is nonzero, and sends the generator to
$(16, 0)\in \Omega_5^{\Spin\text{-}\Z/8\Z}\cong (\Z/32\Z) \oplus (\Z/2\Z)$. The only way for this to be compatible with the
multiplication of classes in degrees $0$, $4$, and $8$ is for $\eta\colon\Omega_{4k}^{\Spin\text{-}\Z/8\Z}\to
\Omega_{4k+1}^{\Spin\text{-}\Z/8\Z}$ to always be nonzero. With $k = 0$, we learn that the Spin-$\Z/8\Z$ bordism class
of $S^1_p$ is non-trivial and is divisible by $4$.\footnote{There is also a more direct proof by studying
the image of $[S^1_p]\in\Omega_1^\Spin$ under the map of Atiyah-Hirzebruch spectral sequences induced
by the unit map $\mathbb S \to (B\Z/4\Z)^{\rho-2}$.} We display all of the hidden extensions that we have unearthed
in \cref{spin_Z8_with_hidden_extensions}.
\begin{figure}[h!]
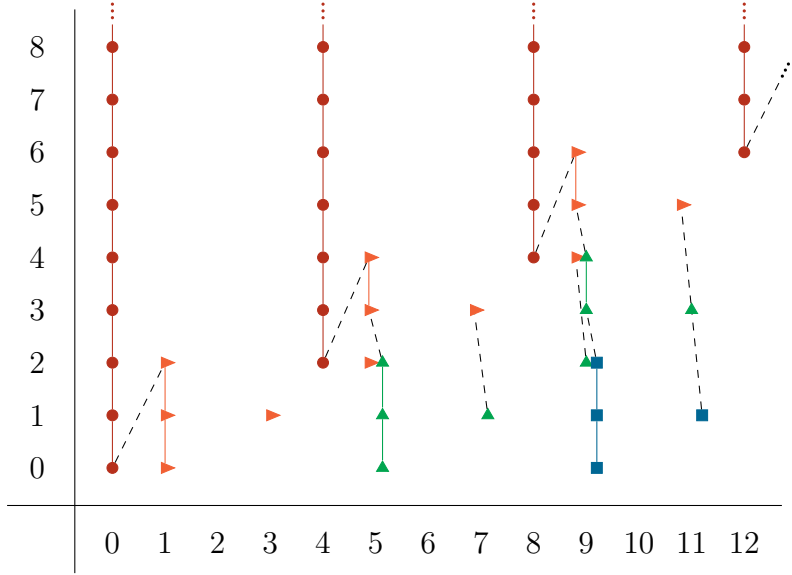

\centering
\begin{sseqdata}[name=hiddenXtns, classes=fill, xrange={0}{12}, yrange={0}{8}, scale=0.7, Adams grading]
\begin{scope}[BrickRed]
	\class(0, 0)\AdamsTower{}
	\class(4, 2)\AdamsTower{}
	\class(8, 4)\AdamsTower{}
	\class(12, 6)\AdamsTower{}
\end{scope}
\begin{scope}[RedOrange, isosceles triangle]
	\class(1, 0)
	\class(1, 1)\structline
	\class(1, 2)\structline
	\class(3, 1)
	\class(5, 2)
	\class(5, 3)
	\class(5, 4)\structline
	\class(7, 3)
	\class(9, 4)
	\class(9, 5)
	\class(9, 6)\structline
	\class(11, 5)
\end{scope}
\begin{scope}[fill=none, draw=none]
	\class(5, 0)
	\class(5, 1)
	\class(5, 3)
	\class(5, 4)
	\class(7, 1)
	\class(7, 3)
	\class(13, 8)
	\class(9, 0)
	\class(9, 0)
	\class(9, 1)
	\class(9, 1)
	\class(9, 2)
	\class(9, 3)
	\class(11, 1)
	\class(11, 1)
	\class(11, 3)
	\class(11, 5)
	\class(11, 5)
	\class(9, 5)
	\class(9, 5)
	\class(9, 6)
	\class(9, 6)
\end{scope}
\begin{scope}[Green, regular polygon,regular polygon sides=3, minimum width=1ex]
	\class(5, 0)
	\class(5, 1)\structline
	\class(5, 2)\structline
	\class(7, 1)
	\class(9, 2)
	\class(9, 3)
	\class(9, 4)\structline
	\class(11, 3)
\end{scope}
\begin{scope}[draw=none, fill=none]
	\class(9, 3)
	\class(9, 4)
	\class(11, 3)
\end{scope}
\begin{scope}[MidnightBlue, rectangle]
	\class(9, 0)
	\class(9, 1)\structline
	\class(9, 2)\structline
	\class(11, 1)
\end{scope}
\begin{scope}[dashed]
	\structline(0, 0)(1, 2)
	\structline(4, 2)(5, 4)
	\structline(5, 2, -1)(5, 3, 1)
	\structline(7, 1, -1)(7, 3, 1)
	\structline(12, 6)(13, 8)
	\structline(8, 4)(9, 6, 1)
	\structline(9, 2, 2)(9, 4, 1)
	\structline(9, 2, -1)(9, 3, 2)
	\structline(9, 4, 2)(9, 5, 1)
	\structline(11, 1, -1)(11, 3, 2)
	\structline(11, 3, 2)(11, 5, 1)
\end{scope}
\end{sseqdata}
\printpage[name=hiddenXtns, page=2]
\caption{Adding hidden extensions to the $E_\infty$-page for $\ko_*\big((B\Z/4\Z)^{\rho-2}\big)$; compare
\cref{ThomZ4Einfty}, bottom. In this figure, all lines that do not shift the topological degree ($x$-coordinate) represent
multiplication by $2$ in Spin-$\Z/8\Z$ bordism, and all lines that increase the topological degree by $1$ represent
multiplication by $[S_p^1]$. The dashed lines are hidden extensions that we cannot see from the Adams spectral
sequence.}
\label{spin_Z8_with_hidden_extensions}
\end{figure}

A similar story holds for Spin-$\Z/2k\Z$ bordism for all even $k$: the class of $S^1_p$ is nonzero, and is
divisible by $k$. When $k\equiv 0\bmod 4$, this is a hidden multiplicative extension; when $k\equiv 2\bmod 4$, this
extension comes from an $h_1$-action on the $E_\infty$-page.\footnote{This $h_1$-action appears in~\cite[Section
2]{GMM68} but was left out by~\cite[Figure 7.4]{Cam17} and \cite[Figures 25 and 26]{BC18}.  To see it, compute the
long exact sequence in Ext associated to the short exact sequence of $\cA(1)$-modules $0\to \Sigma R_0\to R_1\to
\Z/2\Z\to 0$, where $R_0$ and $R_1$ are as defined in~\cite[Section 4.6]{BC18}. Campbell~\cite[Section 7.8]{Cam17} shows that
the cohomology of the Thom spectrum corresponding to Spin-$\Z/4\Z$ bordism is isomorphic to $R_1$.}

\section{Computation of $\Omega_*^{\Spin\text{-}\GL^+(2, \Z)} (\pt)$}
	\label{gl_spin}

The full duality group of type IIB string theory is the Pin\textsuperscript{$+$} cover of $\GL(2, \Z)$, which we denote by $\GL^+(2, \Z)$. Thus
$\GL^+(2, \Z)$ contains a canonical $\Z/2\Z$ subgroup, the kernel of the quotient $\GL^+(2, \Z)\to\GL(2, \Z)$. The
nonzero element of this $\Z/2\Z$ subgroup acts by fermion parity, so the spacetime symmetry type is
$\Spin\text{-}\GL^+(2, \Z)\coloneqq \Spin\times_{\Z/2\Z}\GL^+(2, \Z)$. In this section, we determine the bordism
groups and their generators for this symmetry type in dimensions $11$ and below; just like for $\Spin\times \SL(2,
\Z)$ and Spin-$\Mp(2, \Z)$, we make heavy use of an amalgamated product decomposition for $\GL^+(2, \Z)$, which we
prove in Appendix~\ref{app:groups}:
\begin{equation}
\label{glp_amalg}
	\GL^+(2, \Z)\overset\cong\longrightarrow D_{16} *_{D_8} D_{24}.
\end{equation}
We compute at odd primes in Section \ref{odd_primary_glplus} using the Atiyah-Hirzebruch spectral sequence and at $p = 2$
in Section \ref{gl_p_at_2} using the Adams spectral sequence. In Section \ref{ss:gl2_gens} we determine generators of the
bordism groups we calculate; this is the lengthiest part of the calculation. We identify families of generators
coming from Spin-$\Mp(2, \Z)$ bordism, Spin bordism of $B\Z/2\Z$, and Spin\text{-}$D_{8}$ bordism in
Sections \ref{sss:mp2}, \ref{Z2_to_D16}, and \ref{arcana} respectively; after that we still have four generators left to
find, and we construct each one in turn in Sections \ref{8d_red}, \ref{9d_yellow}, \ref{7d_green}, and
\ref{7d_orange}.  Finally, in Section \ref{mult_str_D16}, we discuss multiplicative structures on Spin-$D_{16}$ bordism.

% first the 3-primary part
\subsection{Working at odd primes}
\label{odd_primary_glplus}

Because the map $\Spin\times_{\Z/2\Z}\GL^+(2, \Z)\to\SSO\times\GL(2, \Z)$ is a double cover, the map
$B\big(\Spin\text{-}\GL^+(2, \Z)\big)\to B\SSO\times B\GL(2, \Z)$ is a fibration with fiber $B\Z/2\Z$. This implies
that if $p$ is an odd prime, the forgetful map $\Omega_*^{\Spin\text{-}\GL^+(2,
\Z)}\to\Omega_*^\SSO\big(B\GL(2, \Z)\big)$ is a $p$-local isomorphism; the proof is analogous to the argument we
gave in Section \ref{subsec:splitodd} that $\Omega_*^\Spin\to\Omega_*^\SSO$ is an odd-primary equivalence.

We will therefore compute $\Omega_*^\SSO\big(B\GL(2,
\Z)\big)$ at the prime $3$; at other odd primes, the map $B\GL(2, \Z)\to\pt$ induces an isomorphism on $\Z_{(p)}$ cohomology,
hence also on $p$-local bordism.\footnote{To see this claim about cohomology, use the Mayer-Vietoris sequence
associated to the amalgamated product description of $\GL(2, \Z)$ given in~\eqref{GL_amalgamation}, together with
the fact that if $G$ is a finite group and $p$ does not divide the order of $G$, then $H^*(BG;\Z/p\Z)$ vanishes in
positive degrees~\cite[Corollary 10.2]{Bro82}. Once we know the map is an isomorphism on
$\Z_{(p)}$ cohomology, the Atiyah-Hirzebruch spectral sequence also allows us to conclude it for $p$-local oriented
bordism.}
In~\eqref{GL_amalgamation}, we factored $\GL(2, \Z)$ as an amalgamated product. Using this and~\eqref{MV_bar}, we
can show that the map
$\Omega_*^\SSO(BD_{12}) \to \Omega_*^\SSO\big(B\GL(2, \Z)\big)$ is a $3$-local isomorphism, with essentially the
same proof as \cref{sl2_at_2}. There is an isomorphism
$D_{12}\cong\Z/2\Z\times D_6$;\footnote{Inscribe the Star of David in a hexagon.}
since the mod $3$ cohomology of $B\Z/2\Z$ vanishes in positive degrees, $D_6\inj D_{12}$ induces an odd-primary equivalence on
classifying spaces. Thus we should compute $\Omega_*^\SSO(BD_6)$. Since $\abs{D_6}$ is only divisible by $2$
and $3$, this is trivial at primes $5$ and above, so we focus on $3$ as usual.

As we discussed in Section \ref{sl2_odd_primes} (see especially~\eqref{BP_SO_decomp}), for $k\le 15$ there is a $3$-local
isomorphism
\begin{equation}
	\Omega_k^\SSO(BD_6)\otimes\Z_{(3)} \overset\cong\longrightarrow \BP_k(BD_6)\oplus\BP_{k-8}(BD_6) \oplus
	\BP_{k-12}(BD_6) \,.
\end{equation}
\begin{thm}
\label{the_thm}
\begin{equation}
        \widetilde\BP_k(BD_6) \cong \begin{cases}
                \Z/3\Z, & k = 3\\
                \Z/9\Z, & k = 7\\
                \Z/27\Z, & k = 11\\
                0, & \text{\rm all other $k\le 11$.}
        \end{cases}\end{equation}
\end{thm}
\begin{proof}
We compute $\widetilde{\BP}_*(BD_6)$ using the Atiyah-Hirzebruch spectral sequence (AHSS)\footnote{Compared with the Adams spectral sequence, we indicate the page by a superscript $E^{2}_{\ast,\ast}$ to emphasize that this sequence arises from homology.}
\begin{equation}
        E^2_{p,q} = H_p(BD_6; \BP_q) \Longrightarrow \BP_{p+q}(BD_6).
\end{equation}
Handel~\cite{Han93} computes $H^*(BD_{2n};\Z)$, and his result together with the universal coefficient theorem
tells us $H_*(BD_6;\Z_{(3)})$:
\begin{equation}
        H_k(BD_6;\Z_{(3)}) \cong \begin{cases}
                \Z_{(3)}, & k = 0\\
                \Z/3\Z, & k\equiv 3\bmod 4\\
                0, &\text{otherwise.}
        \end{cases}
\end{equation}
Thus we can draw the $E^2$-page of the Atiyah-Hirzebruch spectral sequence in \cref{AHSS_D6}. All differentials
vanish by degree reasons: in the homological AHSS, differentials go up and to the left, and decrease total degree
by $1$, so all such differentials are either to or from a zero group.

\begin{figure}[h!]
\centering
\begin{sseqdata}[name=justAHSS, page=2, homological Serre grading, xrange={0}{12}, yrange={0}{12},
classes={draw=none}, xscale=0.5, yscale=0.5,
x label = {$\displaystyle{p\uparrow \atop q\rightarrow}$},
x label style = {font = \small, xshift = -22.5ex, yshift=5.5ex}
]
        \foreach \y in {0, 4, 8, 12} {
                \class["\square"](0, \y)
                \foreach \x in {3, 7, 11} {
                        \class["\bullet"](\x, \y)
                }
        }
\end{sseqdata}
\printpage[name=justAHSS, page=2]
\caption{The $E^2$-page of the Atiyah-Hirzebruch spectral sequence computing $\Omega_*^\SSO(BD_6)\otimes\Z_{(3)}$.
Here $\square\coloneqq \Z_{(3)}$ and $\bullet\coloneqq\Z/3\Z \,$.}
\label{AHSS_D6}
\end{figure}

For reduced $\BP$-homology, we quotient out by the $\Z_{(3)}$s along the vertical line $p = 0$. So we have a
$\Z/3\Z$ in degree $3$ and two extension questions in degrees $7$ and $11$, and the remaining reduced
$\BP$-homology groups in the range we care about vanish.

%\textcolor{red}{Is $\Z/3$ the same as $\Z/3\Z$ in the following. I think yes, just want to make sure.}
The inclusion $\Z/3\Z\inj D_6$ induces a map of Atiyah-Hirzebruch spectral sequences for
$\BP_*(B\Z/3\Z)\to\BP_*(BD_6)$.  The homology of $B\Z/3\Z$ is $\Z$ in degree $0$, $\Z/3\Z$ in all odd degrees, and
$0$ in positive even degrees, and the map of spectral sequences is an isomorphism on the $E^2$-page in total
degrees $3$, $7$, and $11$.  For the same reason as for $BD_6$, this spectral sequence collapses for $B\Z/3\Z$;
therefore the map on the $E^\infty$-pages of the AHSSes is an isomorphism in degrees $p+q = 7$ and $p+q = 11$,
which implies the map $\BP_k(B\Z/3\Z)\to\BP_k(BD_6)$ is an isomorphism for $k = 7,11$.

In Section \ref{sl2_odd_primes}, we saw that $\BP_7(B\Z/3\Z)\cong\Z/9\Z$ and $\BP_{11}(B\Z/3\Z)\cong\Z/27\Z$, resolving
the extension question.
\end{proof}
\begin{rem}
Another way to make this calculation is to use the ``Leray-Serre-Atiyah-Hirzebruch spectral sequence''
\begin{equation}
\label{LSAH}
    E^2_{p,q} = H_p \big(B\Z/2\Z; \BP_q(B\Z/3\Z)\big) \Longrightarrow \BP_{p+q}(BD_6) \,.
\end{equation}
We found working through this variant to be a nice exercise --- the spectral sequence collapses on the $E^2$-page
and yields an isomorphism $\BP_*(B\Z/n\Z)_{\Z/2\Z}\cong\BP_*(BD_6)$.\footnote{There are at
least two additional, and quite different, proofs of this fact due to Kamata-Minami~\cite[Theorem 1.3]{KM73} and
Kamata~\cite[Theorem 1.3]{Kam74}.} Then we directly compute this action and see what the coinvariants are.
\end{rem}
Surjectivity of the map $\Omega_*^\SSO(B\Z/3\Z)\to\Omega_*^\SSO(BD_6)$ in the range we care about implies that the
generators of the $3$-torsion subgroup of $\Omega_*^{\Spin\text{-}\GL^+(2, \Z)}$ in this range can be
chosen to be the generators we found for $\Omega_k^\SSO(B\Z/3\Z)$ for $k = 3$, $7$, and $11$. We produced these
generators in Section \ref{sl2_gens}: $L_3^3$ in dimension $3$, $L_3^7$ in dimension $7$, and $L_3^{11}$ and $\HP^2\times
L_3^3$ in dimension $11$.

\subsection{Adams spectral sequence at $p = 2$}
\label{gl_p_at_2}
Using the amalgamated product descriptions of $\Mp(2, \Z)$ and $\GL^+(2, \Z)$ in~\eqref{mp_amalg}
and~\eqref{eq:GL+group}, there is a commutative diagram
\begin{equation}
\label{2_prim_GL_comm}
\begin{gathered}
% https://q.uiver.app/?q=WzAsNCxbMCwwLCJcXFovOFxcWiJdLFsxLDAsIlxcTXAoMiwgXFxaKSJdLFswLDEsIkRfezE2fSJdLFsxLDEsIlxcR0xeKygyLCBcXFopIl0sWzEsM10sWzAsMl0sWzAsMSwiYSJdLFsyLDMsImIiXV0=
\begin{tikzcd}
	{\Z/8\Z} & {\Mp(2, \Z)} \\
	{D_{16}} & {\GL^+(2, \Z),}
	\arrow[from=1-2, to=2-2]
	\arrow[from=1-1, to=2-1]
	\arrow["a", from=1-1, to=1-2]
	\arrow["b", from=2-1, to=2-2]
\end{tikzcd}
\end{gathered}
\end{equation}
where the vertical arrows are the standard inclusions, $a$ is the Spin cover of the standard inclusion
$\Z/4\Z\hookrightarrow \SL(2, \Z)$, and $b$ is the $\Pin^+$ cover of the standard inclusion
$D_8\hookrightarrow\GL(2, \Z)$. Now apply $B(\Spin\times_{\Z/2\Z}\bl)$ to~\eqref{2_prim_GL_comm}; then the
horizontal arrows induce equivalences on mod $2$ cohomology, which follows from a Mayer-Vietoris argument similar
to \cref{sl2_at_2,sl2_at_3_lem,mp2_at_3,mp2_at_2}, using the amalgamated product descriptions of $\Mp(2, \Z)$ and
$\GL^+(2, \Z)$ referenced above. Thus, in particular:
\begin{cor}
The inclusion $b$ in~\eqref{2_prim_GL_comm} induces a map $\Omega_*^{\Spin\times_{\Z/2\Z} D_{16}}\to
\Omega_*^{\Spin\text{-}\GL^+(2, \Z)}$, and this map is a $2$-local equivalence, where $\Z/2\Z\subset D_{16}$ is
generated by a half turn.
\end{cor}
Another consequence of~\eqref{2_prim_GL_comm} is that this equivalence is compatible with the
$2$-primary equivalences from Spin-$\Z/8\Z$ bordism
to Spin-$\Mp(2, \Z)$ bordism from \cref{mp2_at_2}. We will use this to compute the $2$-torsion subgroup of the
image of the map $\Omega_k^{\Spin\text{-}\Mp(2, \Z)}\to\Omega_k^{\Spin\text{-}\GL^+(2, \Z)}$ for $k\le 11$ in
Section \ref{sss:mp2}.

First, we need to  shear the Spin-$D_{16}$ symmetry type. To do so, we will need to know the mod $2$ cohomology of
$BD_8$.
\begin{lem}[{Handel~\cite[Theorem 5.5]{Han93}}]
\label{dih_mod_2_coh}
When $k\equiv 0\bmod 4$,
\begin{align}
H^*(BD_{2k};\Z/2\Z)\cong \Z/2\Z[x, y, w]/(xy+y^2)
\end{align} with $\abs x = \abs y = 1$ and $\abs w =
2$. If $V$ is the standard two-dimensional real representation of $D_{2k}$ and $\chi_y$ is the
one-dimensional representation which is $1$ on reflections and $-1$ on a generating rotation, then $x = w_1(V)$, $y
= w_1(\chi_y)$, and $w = w_2(V)$.
\end{lem}
We also need the Steenrod squares: $\Sq(x) = x+x^2$ and $\Sq(y) = y+y^2$ are forced by the axioms, and $\Sq(w) = w
+ wx + w^2$ follows from the Wu formula~\eqref{wu_formula}.
\begin{lem}
\label{whichcoh}
The cohomology class of the central extension
\begin{equation}
\label{d16_d8_ext}
	\shortexact*{\Z/2\Z}{D_{16}}{D_8}{}
\end{equation}
is equal to $w_2$ of the standard representation $V\colon D_8\to\mathrm{O} (2)$.
\end{lem}
Here and in the proof of the lemma, when we say ``$w_k$ of a representation,'' we mean $w_k$ of its associated
vector bundle.
\begin{proof}
First, observe that the extension~\eqref{d16_d8_ext} is split when pulled back by the inclusion of any reflection
$\Z/2\Z\hookrightarrow D_8$, but non-trivial when pulled back to the rotation subgroup $\Z/4\Z\hookrightarrow D_8$.
This is because the inclusion $\Z/2\Z\hookrightarrow D_8$ defining any reflection lifts to an inclusion
$\Z/2\Z\hookrightarrow D_{16}$ defining another reflection, but the preimage of $\Z/4\Z$ in $D_{16}$ is the
rotation subgroup $\Z/8\Z$, and the quotient $\Z/8\Z\to\Z/4\Z$ does not split.

We will finish by showing $w_2(V)$ is the unique element of $H^2(BD_8;\Z/2\Z)$ which is trivial when pulled back by
the inclusion of a reflection subgroup, but non-trivial when pulled back to $B\Z/4\Z$. Since the classification of
extensions by elements in $H^2$ is natural, this will force the class of~\eqref{d16_d8_ext} to be $w_2(V)$.

To compute pullback of $w_2(V)$ along the inclusion of a subgroup $H\inj D_8$, it suffices to compute $w_2$ of the
vector bundle induced by the restriction of the standard representation of $D_8$ to $H$.
\begin{enumerate}
	\item If $i\colon \Z/2\Z\inj D_8$ is a subgroup generated by a reflection, then the reflection leaves a line
	invariant, so $i^*V\cong\R\oplus \chi$ for some one-dimensional representation $\chi$. Stiefel-Whitney classes
	do not change when one adds a trivial summand, so
	\begin{equation}
		i^*(w_2(V)) = w_2(\R\oplus\chi) = w_2(\chi) = 0,
	\end{equation}
	because $w_2$ vanishes on one-dimensional representations.
	\item The restriction of $V$ to $\Z/4\Z$ is the standard rotation representation of $\Z/4\Z$, and as we
	discussed in \cref{Z4_coh}, the standard rotation representation is not Spin. Therefore $w_2(V)$ is
	nonzero when pulled back to $\Z/4\Z \,$.
\end{enumerate}
Using \cref{dih_mod_2_coh}, the other seven elements of $H^2(BD_8;\Z/2\Z)$ are either $0$ or of the form $aw + bx^2
+cy^2$ with $b\ne 0$ or $c\ne 0$. The zero class is trivial when pulled back to $\Z/4\Z$; if we can show that $x^2$
and $y^2$ are non-trivial when pulled back to some reflection subgroup, then we will have shown that $w$ is the only
class which is trivial when pulled back to reflection subgroups and non-trivial when pulled back to the rotation
subgroup, which is what we wanted to show.

The Whitney sum formula implies $x^2 = w_2 \big(2\Det(V)\big)$, and $\Det(V)$ can be identified with the
representation $D_8\overset{q}{\to}\Z/2\Z = \mathrm O(1)$ given by quotienting by the subgroup of rotations. Thus
$2\Det(V)$ and its characteristic classes pull back from $\Z/2\Z$ by $q$. For any reflection $T$ in $D_8$, the
composition $\Z/2\Z\cong \ang T\inj D_8\overset{q}{\to}\Z/2\Z$ is the identity, so induces the identity map on
characteristic classes, so $w_2(2\Det(V))$ must pull back to a nonzero class.

For $y^2$, the story is subtler. Recall from Appendix \ref{app:embdihedral} that there are two classes of
reflection subgroups in $D_8$; the two embeddings $i_4$ and $\tilde\imath_4$ of $\Z/2\Z$ into $D_8$ (see the
appendix for their definitions) represent the two classes, and do not induce the same map on cohomology.
Specifically, the one-dimensional representation $L$ from \cref{dih_mod_2_coh} is non-trivial when pulled back along
$\tilde\imath_4$, so $w_1$ of the associated line bundle must be nonzero, hence it is the nonzero element of
$H^1(B\Z/2\Z;\Z/2\Z)\cong\Z/2\Z$. Thus the square of that element, which is the pullback of $y^2$ along
$\tilde\imath_4$, is also nonzero.
%
%We will prove this by showing that both $w_2(V)$ and the class of~\eqref{d16_d8_ext} are the unique class in
%
%
%\Cref{dih_mod_2_coh} tells us that $H^2(BD_8;\Z/2\Z)$ is spanned by $x^2$, $y^2$, and $w$.
%
%has eight elements, and $w_2(V)$ is uniquely picked out by the
%requirement that it is non-trivial when restricted to $\Z/4\Z\subset D_8$, but trivial when restricted to each
%reflection $\Z/2\Z$ subgroup. And, sure enough, if $\Z/2\Z\inj D_8$ is a reflection, it lifts to a map to $D_{16}$; and
%the restriction of~\eqref{d16_d8_ext} to $\Z/4\Z\subset D_8$ is the rotation subgroup of $D_{16}$, which is $\Z/8\Z$.
\end{proof}
$V$ is not orientable, so we cannot invoke \cref{w2_twisted_cor} to shear just yet --- instead, consider the
representation $V\oplus 3\Det(V)$. One can check with the Whitney sum formula that $w_1\big(V\oplus 3\Det(V)\big) = 0$ and
$w_2\big(V\oplus 3\Det(V)\big) = w_2(V)$, so using \cref{w2_twisted_cor} we discover:
\begin{cor}
\label{spin_D16_shear}
Spin-$D_{16}$ structures are naturally
equivalent to $\big(BD_8, V\oplus 3\Det(V)\big)$-twisted Spin structures.
\end{cor}
Therefore, by \cref{shearing_PT}, if $X\coloneqq (BD_8)^{V + 3\Det(V) - 5}$,
the Spin-$D_{16}$ bordism groups are
the Spin bordism groups of $X$.\footnote{The $-5$ is so that
there is no degree shift in the Thom isomorphism and Pontrjagin-Thom theorem; see \cref{Thom_exp}.} In degree $4$, this bordism group and its generator were studied by
Pedrotti~\cite[Theorem 9.0.11]{Ped17}, and bordism for a closely related structure, where $D_8$ is replaced with
$\O (2)$, is studied by Guillou-Marin~\cite{GM80}, Kirby-Taylor~\cite[Section 6]{KT90a}, and Stehouwer~\cite[Section
4]{Ste21}.\footnote{Spin-$\O(2)$ structures, but not their bordism groups, also appear in work of
Han-Huang-Liu-Zhang~\cite[Section 6]{HHLZ22} and Lazaroiu-Shahbazi~\cite{LS22}.}
\begin{lem}
$\ko\ang 2_0(X)\cong\Z/2\Z$ and $\ko\ang 2_1(X)\cong (\Z/2\Z)\oplus (\Z/2\Z)$.
\end{lem}
\begin{proof}
\Cref{ko2_is_homology} reduces this to computing the mod $2$ homology of $X$; then the result follows by passing
\cref{dih_mod_2_coh} through the universal coefficient theorem and the Thom isomorphism.
\end{proof}
\begin{thm}
\label{the_ko_part}
The first several $\ko$-homology groups of $(BD_8)^{V + 3\Det(V)-5}$ are:
\begin{equation}
\begin{alignedat}{2}
	\ko_0\big((BD_8)^{V + 3\Det(V)-5}\big) &\cong \Z\qquad\qquad & \ko_6\big((BD_8)^{V + 3\Det(V)-5}\big) &\cong 0 \\
	\ko_1\big((BD_8)^{V + 3\Det(V)-5}\big) &\cong (\Z/2\Z)^{\oplus 2} & \ko_7\big((BD_8)^{V + 3\Det(V)-5}\big) &\cong
		(\Z/4\Z)\oplus (\Z/2\Z)^{\oplus 3} \\
	\ko_2\big((BD_8)^{V + 3\Det(V)-5}\big) &\cong \Z/2\Z & \ko_8\big((BD_8)^{V + 3\Det(V)-5}\big) &\cong \Z\oplus (\Z/2\Z) \\
	\ko_3\big((BD_8)^{V + 3\Det(V)-5}\big) &\cong (\Z/2\Z)^{\oplus 3} & \ko_9\big((BD_8)^{V + 3\Det(V)-5}\big) &\cong
	(\Z/2\Z)^{\oplus 6} \\
	\ko_4\big((BD_8)^{V + 3\Det(V)-5}\big) &\cong \Z & \ko_{10}\big((BD_8)^{V + 3\Det(V)-5}\big) &\cong (\Z/2\Z)^{\oplus 2} \\
	\ko_5\big((BD_8)^{V + 3\Det(V)-5}\big) &\cong (\Z/2\Z)^{\oplus 2} & \ko_{11}\big((BD_8)^{V + 3\Det(V)-5}\big) &\cong (\Z/8\Z)
	\oplus (\Z/2\Z)^{\oplus 4} \,.
\end{alignedat}
\end{equation}
\end{thm}
\begin{proof}
We will use the Adams spectral sequence. The first step is to compute the $E_2$-page, which is the Ext over
$\cA(1)$ of $H^*(X;\Z/2\Z)$. Therefore we need to determine how $\cA(1)$ acts on $H^*(X;\Z/2\Z)$.
\Cref{Steenrod_Thom} implies that for $a\in H^*(BD_8;\Z/2\Z)$,
\begin{subequations}
\begin{align}
	\Sq^1(Ua) &= U \big(w_1(V\oplus 3\Det(V))a + \Sq^1(a)\big)\\
	\Sq^2(Ua) &= U \big(w_2(V\oplus 3\Det(V))a + w_1(V\oplus 3\Det(V))\Sq^1(a) + \Sq^2(a)\big) \,.
\end{align}
\end{subequations}
In \cref{dih_mod_2_coh}, we learned $w_1(V) = x$ and $w_2(V) = w$, so using the Whitney sum formula, $w_1(V\oplus
3\Det(V)) = 0$ and $w_2(V\oplus 3\Det(V)) = w$. Thus we have determined the $\cA(1)$-action on
$H^*(X;\Z/2\Z)$ in terms of the $\cA(1)$-action on $H^*(BD_8;\Z/2\Z)$, at least in principle. To make it explicit,
first use the Cartan formula~\eqref{cartfor} to reduce to the computation of Steenrod squares of the generators
$x$, $y$, and $w$. We gave the Steenrod squares of these elements right before the statement of \cref{whichcoh}.
%The axiomatic definition of Steenrod squares implies that if $\abs a = 1$, $\Sq(a) = a + a^2$, which gives us
%%\begin{subequations}
%%\begin{align}
%	\Sq(x) &= x + x^2\\
%	\Sq(y) &= y + y^2.
%\end{align}
%For $w$, the axioms tell us $\Sq(w) = w + \Sq^1(w) + w^2$; since $w = w_2(V)$, we can use the Wu
%formula~\eqref{wu_formula} to compute $\Sq^1(w)$. Specifically, we learn $\Sq^1(w_2) = w_1w_2 + w_3$; since $V$ is
%rank-$2$, $w_3(V)$ vanishes, and we conclude
%\begin{equation}
%	\Sq(w) = w + wx + w^2.
%\end{equation}
%\end{subequations}
Now we have a completely explicit description of the $\cA(1)$-module structure on the cohomology of $X$; for
example, $\Sq^1(Ux) = Ux^2$ and $\Sq^2(U) = Uwx$. One can continue in this vein to determine the
$\cA(1)$-action for all classes in degrees $12$ and below; we used a computer program to verify the following
computation.
\begin{prop}
\label{dih_thom_coh}
There is an isomorphism of $\cA(1)$-modules
\begin{equation}
\begin{aligned}
	H^*\big((BD_8)^{V + 3\Det(V) - 5};\Z/2\Z\big) &\cong
		\textcolor{BrickRed}{R_2} \oplus
		\textcolor{RedOrange}{\Sigma J} \oplus
		\textcolor{Goldenrod!67!black}{\Sigma^4 \uQ} \oplus
		\textcolor{Green}{\Sigma^7\Z/2\Z} \oplus
		\textcolor{PineGreen}{\Sigma^8 R_2} \oplus
		\textcolor{MidnightBlue}{\Sigma^9 J} \oplus
		\textcolor{Fuchsia}{\Sigma^{12}\uQ}\\
	&\phantom{\cong} \oplus \Sigma^3\cA(1) \oplus\Sigma^3\cA(1) \oplus \Sigma^5\cA(1) \oplus \Sigma^5\cA(1) \oplus
	\Sigma^7\cA(1)\\
	&\phantom{\cong} \oplus\Sigma^7\cA(1) \oplus \Sigma^9\cA(1) \oplus \Sigma^9\cA(1)
	\oplus \Sigma^{11}\cA(1) \oplus\Sigma^{11}\cA(1)\\
	&\phantom{\cong} \oplus \Sigma^{11}\cA(1) \oplus \Sigma^{11}\cA(1) \oplus P
\end{aligned}
\end{equation}
where $P$ is $13$-connected. The free $\cA(1)$ summands are generated by the classes $Ux^3$, $Uy^3$, $Uw^2x$,
$Uw^2y$, $Ux^7$, $Uy^7$, $Uw^2x^5$, $Uw^2y^5$, $Ux^{11}$, $Uy^{11}$, $Uw^4x^3$, and $Uw^4y^3$. We draw this in
\cref{pin_plus_cover_drawn}.
\end{prop}
These component $\cA(1)$-modules are defined as follows.
\begin{itemize}
	\item $R_2$ is the kernel of the unique nonzero map $\Sigma^{-1}\cA(1)\to \Sigma^{-1}\Z/2\Z$.
	\item $J\coloneqq\cA(1)/\Sq^3$. This module is called the \term{Joker}.
	\item $\uQ\coloneqq\cA(1)/(\Sq^1, \Sq^2\Sq^3)$. This is called the \term{upside-down question mark}.
	\item $\Z/2\Z$ is simply a one-dimensional $\Z/2\Z$-vector space with a trivial $\cA(1)$-action.
\end{itemize}
We will not worry about $P$: its Ext is concentrated in degrees $13$ and above, and hence will not affect the
computation. For the rest of these $\cA(1)$-modules $M$, $\Ext_{\cA(1)}(M)$ is known: see \cref{ext_Z2}
for $\Z/2\Z$ and~\cite[Figure 29]{BC18} for $R_2$, $J$, and $\uQ$.

With this in hand, we draw the $E_2$-page of the Adams spectral sequence for $\ko_*(X)$ in
\cref{dihedral_E2_page}, top. The black dots are the portion coming from the $\Sigma^k\cA(1)$ summands in
$H^*(X;\Z/2\Z)$; Margolis' theorem (in the form of \cref{Margolis_kills_differentials}) implies they do not
participate in nonzero differentials or non-trivial extensions.
\begin{figure}[h!]
\centering
\begin{tikzpicture}[scale=0.6, every node/.style = {font=\tiny}]
	\foreach \y in {0, ..., 15} {
		\node at (-2, \y) {$\y$};
	}
	\begin{scope}[BrickRed]
		\Rtwo{0}{0}{$U$}{$Ux$}{};
	\end{scope}
	\begin{scope}[RedOrange]
		\Joker{3}{1}{$Uy$}{isosceles triangle};
	\end{scope}
	\begin{scope}[Goldenrod!67!black]
		\SpanishQnMark{5}{4}{$Uw^2$}{regular polygon,regular polygon sides=5};
	\end{scope}
	\tikzptB{6}{7}{$Uw^3y$}{Green, regular polygon,regular polygon sides=3};
	\begin{scope}[PineGreen]
		\Rtwo{7}{8}{$Uw^4$}{$Uw^4x$}{star};
	\end{scope}
	\begin{scope}[MidnightBlue]
		\Joker{10}{9}{$U w^4y$}{rectangle, minimum size=3.5pt};
	\end{scope}
	\begin{scope}[Fuchsia]
		\SpanishQnMark{12}{12}{$Uw^6$}{diamond};
	\end{scope}
\end{tikzpicture}
\caption{Part of the mod 2 cohomology of $(BD_8)^{V + 3\Det(V) - 5}$ in low degrees, as stated in
\cref{dih_thom_coh}. The pictured summand contains all elements in degrees $12$ and below except for those
contained in $12$ $\Sigma^k\cA(1)$ summands which are not drawn to avoid cluttering the diagram.
\label{pin_plus_cover_drawn}}
\end{figure}
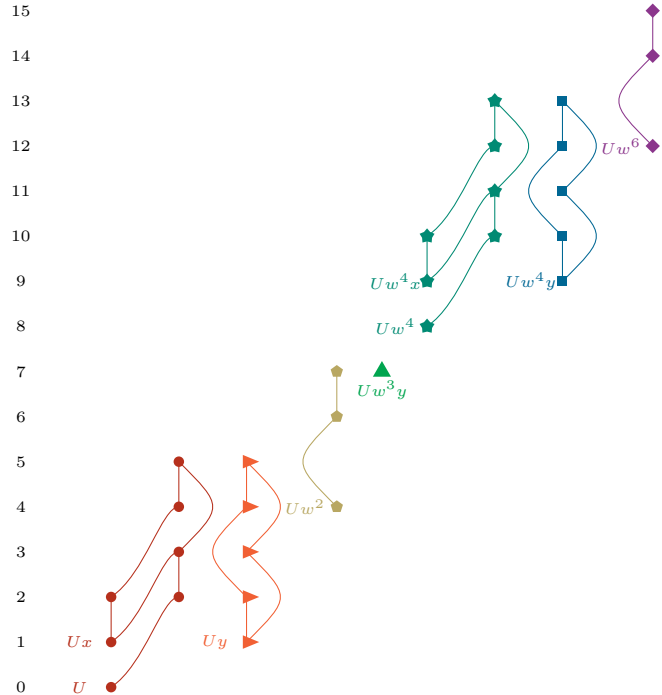

\begin{figure}
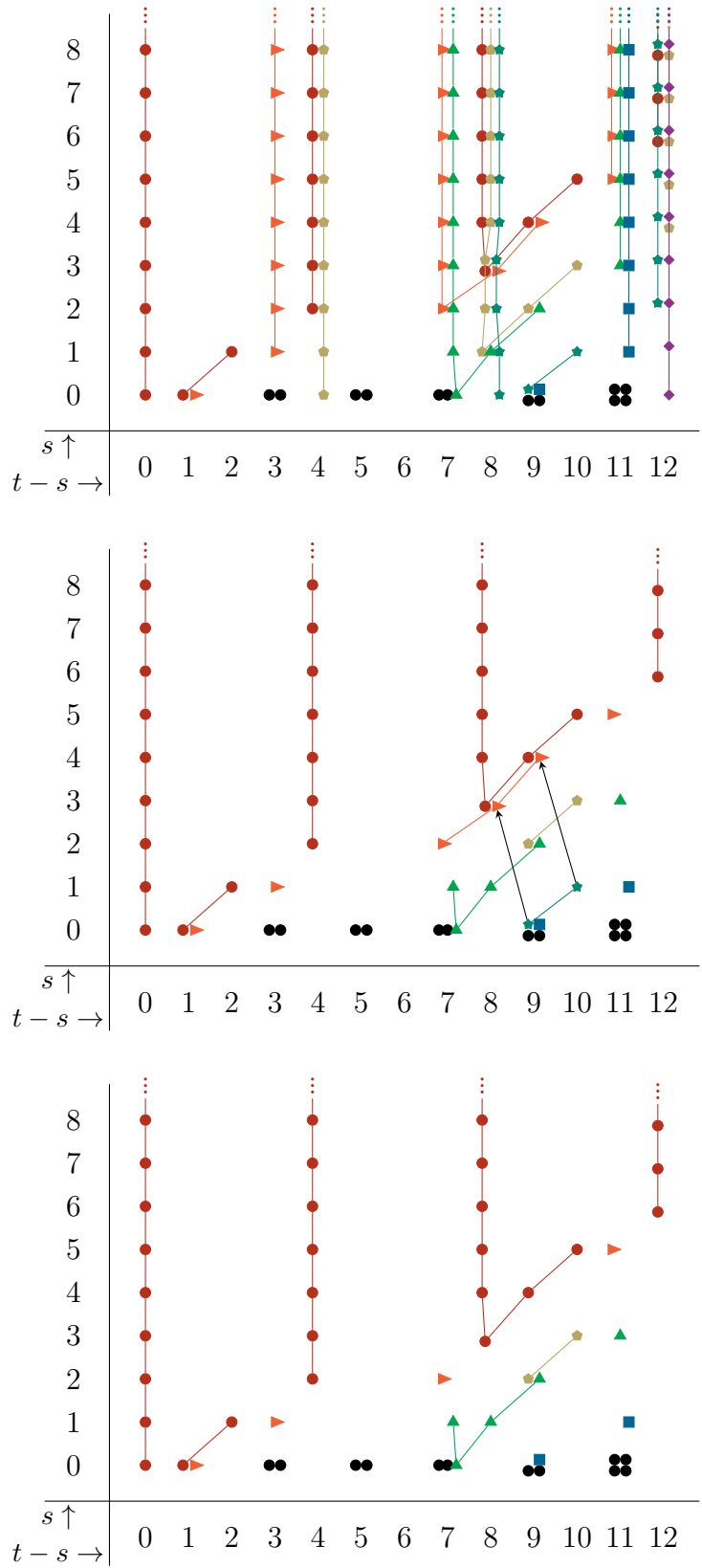

\centering
% first E2
\begin{sseqdata}[name=dihsseq, classes=fill, xrange={0}{12}, yrange={0}{8}, scale=0.6, Adams grading, >=stealth,
x label = {$\displaystyle{s\uparrow \atop t-s\rightarrow}$},
x label style = {font = \small, xshift = -26.5ex, yshift=5.5ex}]
% A(1) summands
\class(3, 0)
\class(3, 0)
\class(5, 0)
\class(5, 0)
\class(7, 0)
\class(7, 0)
\class(9, 0)
\class(9, 0)
\class(11, 0)
\class(11, 0)
\class(11, 0)
\class(11, 0)
\begin{scope}[draw=none, fill=none]
	\class(4, 0)
	\class(4, 1)
	\class(7, 1)
	\class(8, 0)
	\class(8, 0)
	\class(11, 1)\class(11, 1)
	\class(11, 2)\class(11, 2)
	\class(11, 3)
	\class(11, 4)
	\class(12, 0)%\class(12, 0)\class(12, 0)
	\class(12, 1)\class(12, 1)\class(12, 1)
	\class(12, 2)\class(12, 2)
	\class(12, 3)\class(12, 3)
	\class(12, 4)
	\class(12, 5)
\end{scope}

\begin{scope}[BrickRed]
	\class(0, 0)\AdamsTower{}
	\class(1, 0)
	\class(2, 1)\structline
	\class(4, 2)\AdamsTower{}
	\class(8, 3)\AdamsTower{}
	\class(9, 4)
		\structline(8, 3)(9, 4)
	\class(10, 5)\structline
	\class(12, 6)\AdamsTower{}
\end{scope}
\begin{scope}[RedOrange, isosceles triangle]
	\class(1, 0)
	\class(3, 1)\AdamsTower{}
	\class(7, 2)\AdamsTower{}
	\class(8, 3)
		\structline(7, 2)(8, 3, -1)
	\class(9, 4)\structline
	\class(11, 5)\AdamsTower{}
\end{scope}
\begin{scope}[Goldenrod!67!black, regular polygon,regular polygon sides=5]
	\class(4, 0)\AdamsTower{}
	\class(8, 1)\AdamsTower{}
	\class(9, 2)
		\structline(8, 1)(9, 2)
	\class(10, 3)\structline
	\class(12, 4)\AdamsTower{}
\end{scope}
\begin{scope}[Green, regular polygon, regular polygon sides=3, minimum width=1ex]
	\class(7, 0)\AdamsTower{}
	\class(8, 1)
		\structline(7, 0, -1)(8, 1, -1)
	\class(9, 2)\structline
	\class(11, 3)\AdamsTower{}
\end{scope}
\begin{scope}[PineGreen, star]
	\class(8, 0)\AdamsTower{}
	\class(9, 0)
	\class(10, 1)\structline
	\class(12, 2)\AdamsTower{}
\end{scope}
\begin{scope}[MidnightBlue, rectangle]
	\class(9, 0)
	\class(11, 1)\AdamsTower{}
\end{scope}
\begin{scope}[Fuchsia, diamond]
	\class(12, 0)\AdamsTower{}
\end{scope}
% out-of-scope classes needed to draw differentials w/ s large
\begin{scope}[draw=none, fill=none]
	\class(3, 10)
	\class(7, 10)\class(7, 10)
	\class(11, 10)\class(11, 10)\class(11, 10)
\end{scope}
% and now for some d2s
\foreach \x in {4, 8, 12} {
	\foreach \y in {0, ..., 8} {
		\d2(\x, \y, -1)(\x-1, \y+2, -1)
	}
}
\d2(8, 1, 1)(7, 3, 1)
\d2(8, 2, 1)(7, 4, 1)
\d2(8, 3, 3)(7, 5, 1)
\foreach \y in {4, ..., 8} {
	\d2(8, \y, 2)(7, \y+2, 1)
}
\foreach \y in {2, ..., 8} {
	\d2(12, \y, -2)(11, \y+2, -2)
}
\foreach \y in {4, ..., 8} {
	\d2(12, \y, -3)(11, \y+2, -3)
}
\d3(9, 0, 3)(8, 3, 2)
\d3(10, 1)(9, 4, 2)
\end{sseqdata}
\printpage[name=dihsseq, page=1]
\vspace{0.5cm}

% then E3
\printpage[name=dihsseq, page=3]
\vspace{0.5cm}

% then E4 = E_\infty
\printpage[name=dihsseq, page=4]
\caption{The Adams spectral sequence computing $\ko_*\big((BD_8)^{V + 3\Det(V)-5}\big)$.
Top: the $E_2$-page. There are
many nonzero $d_2$s, as we will establish in \cref{dihedral_d2s}.
Middle: the $E_3$-page; we will resolve these differentials in \cref{nonzero_d3}.
Bottom: the $E_4 = E_\infty$-page.}
\label{dihedral_E2_page}
\label{D8_E3}
\label{D8_Einf}
\end{figure}

There is potential for many differentials. We will deduce some of them using our previous work on Spin-$\Mp(2, \Z)$
bordism: the map $\Z/4\Z\inj D_8$ induces a map of Thom spectra $\Phi\colon (B\Z/4\Z)^{\rho-2}\to (BD_8)^{V +
3\Det(V) - 5}$. Upon taking Spin bordism, this map has a geometric description: sending a Spin-$\Z/8\Z$ manifold to
the bordism class of its canonical Spin-$D_{16}$ structure.

%First, though, some differentials vanish for even easier reasons.
%\begin{thm}[Margolis~\cite{Mar74}]
%If $\Z/2\in E_2^{0,t}$ arises from an $\cA(1)$ summand in $H^*(X;\Z/2)$, then all differentials to or from this
%$\Z/2$ summand vanish, and on the $E_\infty$-page, this summand does not combine non-trivially with any others in an
%extension.
%\end{thm}
%So we can ignore the black summands in \cref{dihedral_E2_page} for now.
%
For $2\le r\le\infty$, $\Phi$ induces a map $\Phi_*$ from the $E_r$-page of the Adams spectral sequence for
$\ko_*\big((B\Z/4\Z)^{\rho-2}\big)$ to the $E_r$-page of the Adams spectral sequence for $\ko_*(X)$,
and this map commutes with differentials. We will compute the image of $\Phi_*$ and see that it is large enough for
us to infer most of the differentials we need.

To compute $\Phi_*$, we can use a trick: we have decomposed the cohomology of the Thom spectra $(B\Z/4\Z)^{\rho-2}$
and $X$ into summands, and the map between them mostly respects this decomposition. For each summand, the map is
part of a short exact sequence of $\cA(1)$-modules, so we can compute using the long exact sequence in Ext.
See~\cite[Section 4.6]{BC18} for more information and examples of this technique. We run this computation one summand at
a time in \cref{image_from_spin_Z8}.
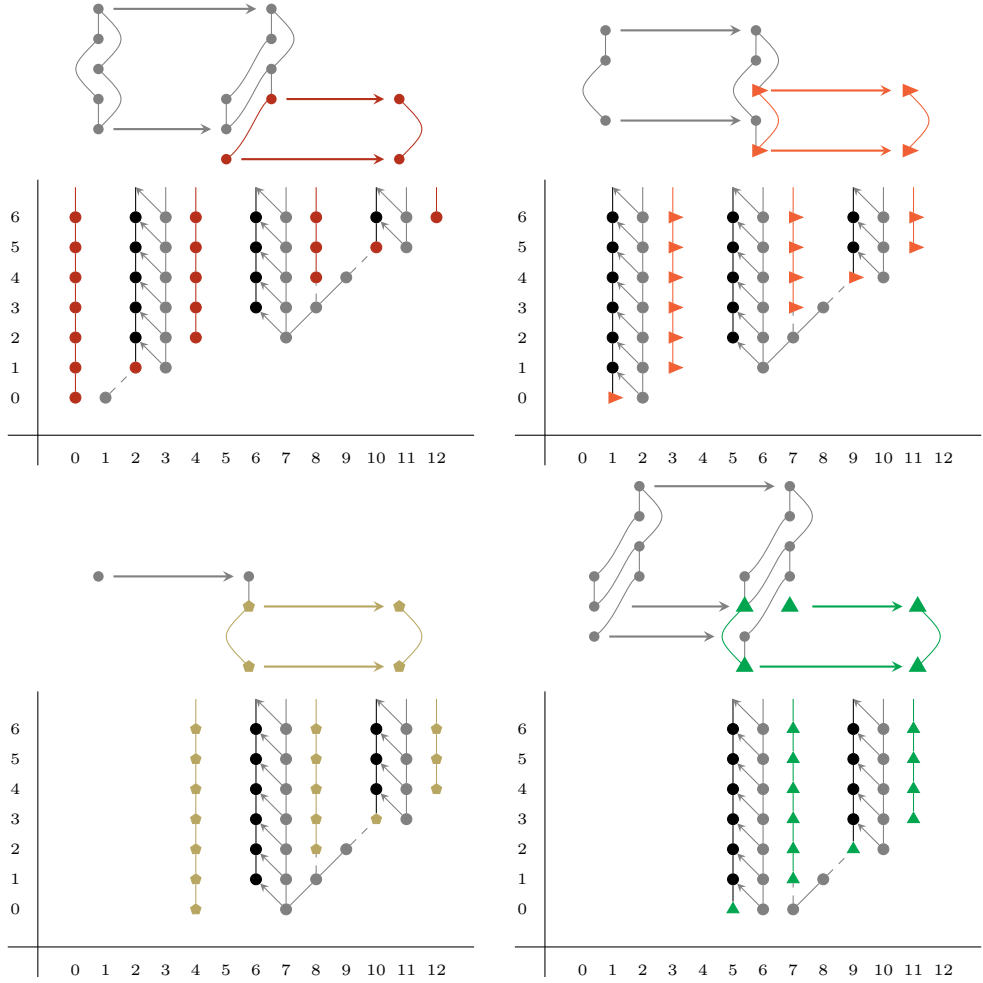
\begin{figure}
\centering
% R2
\begin{subfigure}[b]{0.4\textwidth}
\begin{subfigure}[c]{\textwidth}
\centering
\begin{tikzpicture}[scale=0.4, every node/.style = {font=\tiny}, >=stealth]
\begin{scope}[gray]
	\Joker{-5}{1}{}{}
	\tikzpt{-0.75}{1}{}{};
	\tikzpt{-0.75}{2}{}{};
	\tikzpt{0.75}{3}{}{};
	\tikzpt{0.75}{4}{}{};
	\tikzpt{0.75}{5}{}{};
	\sqone(-0.75, 1);
	\sqone(0.75, 2);
	\sqone(0.75, 4);
	\sqtwoCR(-0.75, 1);
	\sqtwoCR(-0.75, 2);
	\sqtwoR(0.75, 3);
	\draw[->, thick] (-4.5, 1) -- (-1.25, 1);
	\draw[->, thick] (-4.5, 5) -- (0.25, 5);
%       \sqone(0, 0);
%       \sqtwoR(0, 2);
\end{scope}
\begin{scope}[BrickRed]
	\tikzpt{-0.75}{0}{}{};
	\tikzpt{0.75}{2}{}{};
	\sqtwoCR(-0.75, 0);
	\tikzpt{5}{0}{}{};
	\tikzpt{5}{2}{}{};
	\sqtwoR(5, 0);
	\draw[thick, ->] (-0.25, 0) -- (4.5, 0);
	\draw[thick, ->] (1.25, 2) -- (4.5, 2);
\end{scope}
\end{tikzpicture}
\end{subfigure}
\begin{sseqdata}[name=R2Image, classes=fill, xrange={0}{12}, yrange={0}{6}, scale=0.4, Adams grading, >=stealth,
tick style={font=\tiny}, x axis tail=0.4cm, y axis tail=0.4cm, x tick gap=0.3cm, y tick gap=0.3cm]
\class(2, 2)\AdamsTower{}
\class(6, 3)\AdamsTower{}
\class(10, 6)\AdamsTower{}
\begin{scope}[gray] % TODO maybe another color?
	\class(1, 0)
	\class(3, 1)\AdamsTower{}
	\class(7, 2)\AdamsTower{}
	\class(8, 3)
		\structline(7, 2)(8, 3)
	\class(9, 4)\structline
	\class(11, 5)\AdamsTower{}
\end{scope}
\begin{scope}[BrickRed]
	\class(0, 0)\AdamsTower{}
	\class(2, 1)
	\class(4, 2)\AdamsTower{}
	\class(8, 4)\AdamsTower{}
	\class(10, 5)
	\class(12, 6)\AdamsTower{}
\end{scope}
\structline(2, 1)(2, 2)
\structline[dashed, gray](1, 0)(2, 1)
\structline[dashed, gray](8, 3)(8, 4)
\structline[dashed, gray](9, 4)(10, 5)
\structline(10, 5)(10, 6)
% now differentials
\foreach \y in {1, ..., 6} {
	\d[gray]1(3, \y)
}
\foreach \y in {2, ..., 6} {
	\d[gray]1(7, \y)
}
\foreach \y in {5, ..., 6} {
	\d[gray]1(11, \y)
}
\end{sseqdata}
\printpage[name=R2Image, page=1]
\end{subfigure}
% Joker
\begin{subfigure}[b]{0.4\textwidth}
\begin{subfigure}[c]{\textwidth}
\centering
\begin{tikzpicture}[scale=0.4, every node/.style = {font=\tiny}, >=stealth]
\begin{scope}[gray]
	\SpanishQnMark{-5}{1}{}{}
	\SpanishQnMark{0}{1}{}{}
	\draw[->, thick] (-4.5, 1) -- (-0.5, 1);
	\draw[->, thick] (-4.5, 4) -- (-0.5, 4);
	\sqone(0, 0);
	\sqtwoR(0, 2);
\end{scope}
\begin{scope}[RedOrange]
	\foreach \x in {0, 5} {
		\tikzpt{\x}{0}{}{isosceles triangle};
		\tikzpt{\x}{2}{}{isosceles triangle};
		\sqtwoR(\x, 0);
	}
	\draw[thick, ->] (0.5, 0) -- (4.5, 0);
	\draw[thick, ->] (0.5, 2) -- (4.5, 2);
\end{scope}
\end{tikzpicture}
\end{subfigure}
\begin{sseqdata}[name=jokerImage, classes=fill, xrange={0}{12}, yrange={0}{6}, scale=0.4, Adams grading, >=stealth,
tick style={font=\tiny}, x axis tail=0.4cm, y axis tail=0.4cm, x tick gap=0.3cm, y tick gap=0.3cm]
\class(1, 1)\AdamsTower{}
\class(5, 2)\AdamsTower{}
\class(9, 5)\AdamsTower{}
\begin{scope}[gray] % TODO maybe another color?
	\class(2, 0)\AdamsTower{}
	\class(6, 1)\AdamsTower{}
	\class(7, 2)
		\structline(6, 1)(7, 2)
	\class(8, 3)\structline
	\class(10, 4)\AdamsTower{}
\end{scope}
\begin{scope}[RedOrange, isosceles triangle]
	\class(1, 0)
	\class(3, 1)\AdamsTower{}
	\class(7, 3)\AdamsTower{}
	\class(9, 4)
	\class(11, 5)\AdamsTower{}
\end{scope}
\structline(1, 0)(1, 1)
\structline[dashed, gray](7, 2)(7, 3)
\structline[dashed, gray](8, 3)(9, 4)
\structline(9, 4)(9, 5)
% now differentials
\foreach \y in {0, ..., 6} {
	\d[gray]1(2, \y)
}
\foreach \y in {1, ..., 6} {
	\d[gray]1(6, \y)
}
\foreach \y in {4, ..., 6} {
	\d[gray]1(10, \y)
}
\end{sseqdata}
\printpage[name=jokerImage, page=1]
\end{subfigure}
% upside-down qn mark
\begin{subfigure}[b]{0.4\textwidth}
\begin{subfigure}[c]{\textwidth}
\centering
\begin{tikzpicture}[scale=0.4, every node/.style = {font=\tiny}, >=stealth]
\begin{scope}[gray]
	\tikzpt{-5}{3}{}{};
	\tikzpt{0}{3}{}{};
	\draw[->, thick] (-4.5, 3) -- (-0.5, 3);
	\sqone(0, 2);
\end{scope}
\begin{scope}[Goldenrod!67!black]
	\foreach \x in {0, 5} {
		\tikzpt{\x}{0}{}{regular polygon,regular polygon sides=5};
		\tikzpt{\x}{2}{}{regular polygon,regular polygon sides=5};
	}
	\sqtwoL(0, 0);
	\sqtwoR(5, 0);
	\draw[thick, ->] (0.5, 0) -- (4.5, 0);
	\draw[thick, ->] (0.5, 2) -- (4.5, 2);
\end{scope}
\end{tikzpicture}
\end{subfigure}
\begin{sseqdata}[name=qnImage, classes=fill, xrange={0}{12}, yrange={0}{6}, scale=0.4, Adams grading, >=stealth,
tick style={font=\tiny}, x axis tail=0.4cm, y axis tail=0.4cm, x tick gap=0.3cm, y tick gap=0.3cm]
\class(6, 1)\AdamsTower{}
\class(10, 4)\AdamsTower{}
\begin{scope}[gray] % TODO maybe another color?
	\class(7, 0)\AdamsTower{}
	\class(8, 1)
		\structline(7, 0)(8, 1)
	\class(9, 2)\structline
	\class(11, 3)\AdamsTower{}
\end{scope}
\begin{scope}[Goldenrod!67!black, regular polygon,regular polygon sides=5]
	\class(4, 0)\AdamsTower{}
	\class(8, 2)\AdamsTower{}
	\class(10, 3)
	\class(12, 4)\AdamsTower{}
\end{scope}
\structline(10, 3)(10, 4)
\structline[dashed, gray](8, 1)(8, 2)
\structline[dashed, gray](9, 2)(10, 3)
% now differentials
\foreach \y in {0, ..., 6} {
	\d[gray]1(7, \y)
}
\foreach \y in {3, ..., 6} {
	\d[gray]1(11, \y)
}
\end{sseqdata}
\printpage[name=qnImage, page=1]
\end{subfigure}
% F2
\begin{subfigure}[b]{0.4\textwidth}
\begin{subfigure}[c]{\textwidth}
\centering
\begin{tikzpicture}[scale=0.4, every node/.style = {font=\tiny}, >=stealth]
\begin{scope}[gray]
	\Rtwo{-5.75}{1}{}{}{};
	\tikzpt{-0.75}{1}{}{};
	\tikzpt{-0.75}{3}{}{};
	\tikzpt{0.75}{3}{}{};
	\tikzpt{0.75}{4}{}{};
	\tikzpt{0.75}{5}{}{};
	\tikzpt{0.75}{6}{}{};
	\sqone(-0.75, 0);
	\sqone(-0.75, 2);
	\sqone(0.75, 3);
	\sqone(0.75, 5);
	\sqtwoCR(-0.75, 1);
	\sqtwoCR(-0.75, 2);
	\sqtwoCR(-0.75, 3);
	\sqtwoR(0.75, 4);
	\draw[->, thick] (-5.25, 1) -- (-1.75, 1);
	\draw[->, thick] (-4.5, 2) -- (-1.25, 2);
	\draw[->, thick] (-3.75, 6) -- (0.25, 6);
\end{scope}
\begin{scope}[Green] % TODO: F2
	\foreach \x in {-0.75, 5} {
		\tikzpt{\x}{0}{}{regular polygon,regular polygon sides=3};
		\tikzpt{\x}{2}{}{regular polygon,regular polygon sides=3};
	}
	\tikzpt{0.75}{2}{}{regular polygon,regular polygon sides=3};
	\sqtwoL(-0.75, 0);
	\sqtwoR(5, 0);
	\draw[thick, ->] (-0.25, 0) -- (4.5, 0);
	\draw[thick, ->] (1.5, 2) -- (4.5, 2);
\end{scope}
\end{tikzpicture}
\end{subfigure}
\begin{sseqdata}[name=FtwoImage, classes=fill, xrange={0}{12}, yrange={0}{6}, scale=0.4, Adams grading, >=stealth,
tick style={font=\tiny}, x axis tail=0.4cm, y axis tail=0.4cm, x tick gap=0.3cm, y tick gap=0.3cm]
\class(5, 1)\AdamsTower{}
\class(9, 3)\AdamsTower{}
\begin{scope}[gray] % TODO maybe another color?
	\class(6, 0)\AdamsTower{}
	\class(7, 0)
	\class(8, 1)\structline
	\class(10, 2)\AdamsTower{}
\end{scope}
\begin{scope}[Green, regular polygon, regular polygon sides=3, minimum width=1ex]
	\class(5, 0)
	\class(7, 1)\AdamsTower{}
	\class(9, 2)
	\class(11, 3)\AdamsTower{}
\end{scope}
\structline(5, 0)(5, 1)
\structline(9, 2)(9, 3)
\structline[dashed, gray](7, 0)(7, 1)
\structline[dashed, gray](8, 1)(9, 2)
% now differentials
\foreach \y in {0, ..., 6} {
	\d[gray]1(6, \y)
}
\foreach \y in {2, ..., 6} {
	\d[gray]1(10, \y)
}
\end{sseqdata}
\printpage[name=FtwoImage, page=1]
\end{subfigure}
\caption{Part of the computation of the image of the map on the $E_2$-page induced from the map from Spin-$\Z/8\Z$
bordism to Spin-$D_{16}$ bordism. Color is the image; black is the kernel. In the lower-right (green triangles)
figure, both green summands in degree $2$ map non-trivially to the degree-$2$ summand in $\textcolor{Green}{C\eta}$;
the degree-two class in the kernel $\textcolor{gray}{\Sigma R_2}$ maps to their sum. The arguments for the teal
star, blue square, and purple diamond summands are analogous to the arguments for the red circle, orange triangle,
and yellow pentagon summands, respectively.}
\label{image_from_spin_Z8}
\end{figure}

\begin{defn}
We will say that a differential \term{vanishes for easy reasons} if it can be shown to vanish because either its
domain or codomain is zero; it must vanish in order to be equivariant for the $h_i$-action on the $E_2$-page; or it
is killed by Margolis' theorem because it comes from a $\Sigma^k\cA(1)$ summand in cohomology (see
\cref{Margolis_kills_differentials}).
\end{defn}
\begin{lem}
\label{dihedral_d2s}
In the range $t-s\le 12$, all $d_2$s either vanish for easy reasons or are in the image of $\Phi_*$.
\end{lem}
\begin{proof}
We computed $\Im(\Phi_*)$ in \cref{image_from_spin_Z8}; then one just checks the lemma directly. Specifically, one
observes that all $d_2$s that do not vanish for easy reasons are either between $h_0$-towers or from the $11$-line
to the $10$-line. All $h_0$-towers are in the image of $\Phi_*$,\footnote{$\Phi_*$ is not always surjective on
$h_0$-towers, e.g.\ the green $h_0$-tower in topological degree $7$. But it is surjective enough: for every
$h_0$-tower in range, the cokernel of $\Phi_*$ on that $h_0$-tower is finite-dimensional. Since $h_0$-towers are
generated by applying $h_0$ to some element of the $E_2$-page and differentials must be equivariant for
$h_0$-actions, this ``not quite surjectivity'' is good enough: it determines the values of differentials on all
elements of the $h_0$-tower.} as is the entire $10$-line, and the entire $11$-line except for the summands arising
from $\Sigma^{11}\cA(1)$ summands in $H^*\big((BD_8)^{V + 3\Det(V)-5};\Z/2\Z\big)$, and Margolis' theorem kills
differentials emerging from these latter summands. Thus $d_2$s from the $11$-line to the $10$-line either vanish
for easy reasons or are in the image of $\Phi_*$.
\end{proof}
Therefore we know all the $d_2$s in range: the differentials between $h_0$-towers are the same as they were for
$\ko_*\big((B\Z/4\Z)^{\rho-2}\big)$, and all other $d_2$s vanish. We draw the $E_3$-page in \cref{D8_E3}, middle.
Most $d_3$ differentials vanish for easy reasons, except for $d_3\colon E_3^{0,9}\to E_3^{3,11}$ and $d_3\colon
E_3^{1,11}\to E_3^{4,13}$. A combination of Margolis' theorem and $h_i$-equivariance means that the first $d_3$ is
either $0$ or carries the teal summand to the orange summand and vanishes on the remaining generators of
$E_3^{0,9}$; $h_1$-equivariance then forces the latter $d_3$ to either vanish or surject onto the orange summand in
$E_3^{4,13}$, and one vanishes iff the other does.

To address these $d_3$s, we need to show that they are the last possible differentials that can be nonzero. All
longer differentials vanish for easy reasons with the exception of $d_4\colon E_2^{1,12}\to E_2^{5,15}$.
\begin{lem}
\label{d4_vanishes}
This $d_4\colon E_4^{1,12}\to E_4^{5,15}$ vanishes.
\end{lem}
\begin{proof}
The generator $x$ of $E_4^{1,12}\cong\textcolor{MidnightBlue}{\Z/2\Z}$ is in $\Im(\Phi_*)$. Since $d_4$ commutes with
$\Phi_*$, $d_4(x)\in\Im(\Phi_*)$ too. However, in the proof of \cref{ko_spin_z8}, we saw that in the Adams spectral
sequence for $\ko_*\big((B\Z/4\Z)^{\rho-2}\big)$, $E_4^{5,15} = 0$, so $d_4(x) = 0$.
%%In \cref{ko_spin_z8}, we established that $\widetilde{\ko}_{11}((B\Z/4)^{\rho-2})\cong\Z/8$, and that
%there is a generator $x$ of that $\Z/8$ whose image in the Adams $E_\infty$-page is the nonzero element $\overline
%x$ of $E_\infty^{1,12}\cong \textcolor{MidnightBlue}{\Z/2}$. If $y\coloneqq 2x$, then the image of $y$ in the
%$E_\infty$ page, which we call $\overline y$, is the nonzero element of
%$E_\infty^{3,14}\cong\textcolor{Green}{\Z/2}$.
%
%\Cref{image_from_spin_Z8} implies that on the $E_\infty$-page, $\Phi_*(\overline y)\ne 0$ and $\Phi_*(\overline x)
%= 0$ iff our $d_4\ne 0$. Thus if $d_4\ne 0$, $\Phi(x) = 0$, so $\Phi(y) =0$ too. However, $\Phi(y)$ is a preimage
%of $\Phi_*(\overline y)\ne 0$, which is a contradiction.
\end{proof}
See \cref{d4_equals_zero} for a picture of this proof.

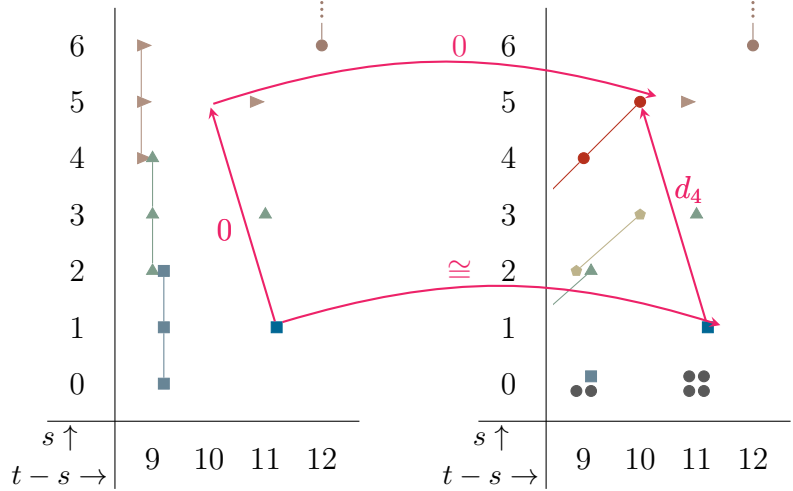
\begin{figure}[h!]
\centering
\begin{sseqdata}[name=d4MpPart, classes=fill, xrange={9}{12}, yrange={0}{6}, scale=0.75, Adams grading,
remember picture, >=stealth,
x label = {$\displaystyle{s\uparrow \atop t-s\rightarrow}$},
x label style = {font = \small, xshift = -13ex, yshift=5.5ex}]
\begin{scope}[BrickRed!30!gray]
	\class(12, 6)\AdamsTower{}
\end{scope}
\begin{scope}[RedOrange!30!gray, isosceles triangle]
	\class(9, 4)
	\class(9, 5)\structline
	\class(9, 6)\structline
	\class(11, 5)
\end{scope}
\begin{scope}[draw=none, fill=none]
	\class(9, 5)\class(9, 5)
	\class(9, 6)\class(9, 6)
	\class(9, 2)
	\class(9, 3)
	\class(11, 3)
\end{scope}
\begin{scope}[Green!30!gray, regular polygon, regular polygon sides=3, minimum width=1ex]
	\class(9, 2)
	\class(9, 3)\structline
	\class(9, 4)\structline
	\class(11, 3)
\end{scope}
\begin{scope}[draw=none, fill=none]
	\class(9, 4)
	\class(9, 3)
	\class(9, 0)\class(9, 0)
	\class(9, 1)\class(9, 1)
	\class(10, 5)
	\class(11, 1)\class(11, 1)
\end{scope}
\class[MidnightBlue, rectangle](11, 1)
\begin{scope}[MidnightBlue!30!gray, rectangle]
	\class(9, 0)
	\class(9, 1)\structline
	\class(9, 2)\structline
\end{scope}
\begin{scope}[draw=none, fill=none]
	\class(11, 3)
	\class(11, 5)\class(11, 5)
\end{scope}
\d["0" {xshift=-1.2em, yshift=-1.2em}, thick, WildStrawberry]4(11, 1, -1)
\tikzmarkinside{source1}{-2.6, -1.55}
\tikzmarkinside{source2}{-1.68, -4.45}
%\coordinate (source1) at (12.01, 4.01);
%\coordinate (source2) at (8.6, 0.01);
%\path[>->] (source1) edge (source2);
\end{sseqdata}
\begin{sseqdata}[name=Campbellsseq, classes=fill, xrange={9}{12}, yrange={0}{6}, scale=0.75, Adams grading,
remember picture, >=stealth,
x label = {$\displaystyle{s\uparrow \atop t-s\rightarrow}$},
x label style = {font = \small, xshift = -13ex, yshift=5.5ex}]
\begin{scope}[black!30!gray]
	\class(9, 0)\class(9, 0)
	\class(11, 0)\class(11, 0)\class(11, 0)\class(11, 0)
\end{scope}
\begin{scope}[BrickRed]
	\class(8, 3)
	\class(9, 4)\structline
	\class(10, 5)\structline
\end{scope}
\begin{scope}[BrickRed!30!gray]
	\class(12, 6)\AdamsTower{}
\end{scope}
\begin{scope}[RedOrange!30!gray, isosceles triangle]
	\class(11, 5)
\end{scope}
\begin{scope}[Goldenrod!40!gray, regular polygon,regular polygon sides=5]
	\class(9, 2)
	\class(10, 3)\structline
\end{scope}
\class[draw=none, fill=none](11, 3)
\begin{scope}[Green!30!gray, regular polygon, regular polygon sides=3, minimum width=1ex]
	\class(8, 1)
	\class(9, 2)\structline
	\class(11, 3)
\end{scope}
\class[draw=none, fill=none](11, 3)
\class[draw=none, fill=none](11, 5)\class[draw=none, fill=none](11, 5)
\class[draw=none, fill=none](11, 1)\class[draw=none, fill=none](11, 1)
\class[draw=none, fill=none](9, 0)
\class[MidnightBlue, rectangle](11, 1)
\class[MidnightBlue!30!gray, rectangle](9, 0)
\tikzmarkinside{targ1}{8.9, 5.2}
\tikzmarkinside{targ2}{9.75, 2.15}
\d["d_4" {xshift=-0.35em}, thick, WildStrawberry]4(11, 1, -1)
\end{sseqdata}
\printpage[name=d4MpPart, page=4]
\qquad
\printpage[name=Campbellsseq, page=4]
\begin{tikzpicture}[remember picture, overlay, >=stealth]
       \node (s1) at (source1) {};
       \node (t1) at (targ1) {};
       \node (s2) at (source2) {};
       \node (t2) at (targ2) {};
       \path[->, thick, WildStrawberry, bend left=17] (s1) edge (t1)
	       node at (-4.5, 3) {$\cong$};
       \path[->, thick, WildStrawberry, bend left=17] (s2) edge (t2)
	       node at (-4.5, 6) {$0$};
\end{tikzpicture}
\caption{Left: the $E_4$-page for $\ko_*\big((B\Z/4)^{V-2}\big)$ (the calculation for Spin-$\Mp(2, \Z)$
bordism). Right: the $E_4$-page for $\ko_*\big((BD_8)^{V + 3\Det(V) - 5}\big)$ (the calculation for
Spin-$\GL^+(2, \Z)$ bordism). The horizontal pink arrows are the map between these spectral sequences induced by
$\Phi$. This map commutes with differentials, so the pictured $d_4$ must vanish. This is a picture proof of
\cref{d4_vanishes}.}
\label{d4_equals_zero}
\end{figure}

Now we understand the $E_\infty$-page with the exception of the two indicated $d_3$ differentials.
\begin{lem}
\label{nonzero_d3}
The $d_3$s starting at $E_3^{0,9}$ and $E_3^{1,11}$ do not vanish.
\end{lem}
\begin{proof}
$E_3^{4,13}\cong\textcolor{BrickRed}{\Z/2\mathbb{Z}}\oplus \textcolor{RedOrange}{\Z/2 \mathbb{Z}}$, but because Adams differentials
commute with $h_1$, the image of $d_3\colon E_3^{1,11}\to E_3^{4,13}$ is contained in the orange triangle
$\textcolor{RedOrange}{\Z/2 \mathbb{Z}}$ summand, as opposed to the red circle. Therefore this differential vanishes if and
only if the nonzero element of that summand lives to the $E_\infty$-page. Call that nonzero element $\alpha$. Since
differentials commute with $h_1$, the destinies of the differentials starting in $E_3^{0,9}$ and $E_3^{1,11}$ are
bonded: one vanishes if and only if the other does. Thus, to prove this lemma, we will prove that $\alpha\in
E_3^{4,13}$ cannot survive to the $E_\infty$-page.

Suppose conversely that $\alpha$ \emph{does} survive. Looking at \cref{image_from_spin_Z8}, we see that $\alpha$ is
in the image of the map $\Phi_*$ of spectral sequences induced by $\Phi$; choose the preimage $\widetilde\alpha$
which is the nonzero element of the orange triangle $\textcolor{RedOrange}{\Z/2 \mathbb{Z}}$ summand of $E_3^{4,13}$ for
$\ko_*((B\Z/4\Z)^{\rho-2})$. Let $\widetilde\beta$ be the nonzero element of the green triangle
$\textcolor{Green}{\Z/2 \mathbb{Z}}$ summand of $E_3^{2,11}$ for $\ko_*((B\Z/4\Z)^{\rho-2})$ and $\beta$ be the nonzero
element of the green triangle $\textcolor{Green}{\Z/2 \mathbb{Z}}$ summand of $E_3^{2,11}$ for $\ko_*(X)$; then
$\Phi_*(\widetilde \beta) = \beta$. By our assumption, all four of these classes survive to their respective
$E_\infty$-pages, so define elements nonzero $a$, $b$, $\widetilde a$, and $\widetilde b$ in
$\ko_9((B\Z/4)^{\rho-2})$ and $\ko_*(X)$, where $a$ lifts $\alpha$, $b$ lifts $\beta$, $\widetilde a$ lifts
$\widetilde\alpha$, and $\widetilde b$ lifts $\widetilde\beta$, such that $\Phi(\widetilde a) = a$ and
$\Phi(\widetilde b) = b$. The hidden extension we deduced in the proof of \cref{ko_spin_z8} (see
\cref{spin_Z8_with_hidden_extensions}) implies we can choose $\widetilde a = 2\widetilde b$, and therefore $a = 2b$
as well.

Both $\alpha$ and $\beta$ are $h_1$ times some other classes $\gamma$ and $\delta$, respectively, and by our
assumption both $\gamma$ and $\delta$ survive to the $E_\infty$-page, hence lifting to define classes
$c,d\in\ko_8(X)$ with $\eta c = a$ and $\eta d = b$. That is, $a = 2\eta d$, but this is a contradiction, because
$2\eta = 0$ and $a\ne 0$.
\end{proof}

%\begin{figure}[h!]
%\centering
%\end{figure}

So $E_4 = E_\infty$ for $t-s < 12$; we draw this in \cref{D8_Einf}, bottom.

Now we have to address extension questions. By inspection, the only place a hidden extension could occur for $t-s <
6$ is in dimension $3$, and Margolis' theorem (\cref{Margolis_kills_differentials}) rules out a hidden extension
there. In topological degrees $7$ through $11$, however, we must address extension problems.
\begin{lem}
\label{deg_7_xtn_lemma}
The extension
\begin{equation}
\label{deg_7_xtn}
	\shortexact{\textcolor{RedOrange}{\Z/2\Z}}{\ko_7(X)}{\textcolor{Green}{(\Z/4\Z)}\oplus(\Z/2\Z)\oplus(\Z/2\Z)}{}
\end{equation}
splits, so $\ko_7(X)\cong \textcolor{Green}{(\Z/4\Z)}\oplus \textcolor{RedOrange}{(\Z/2\Z)}
\oplus(\Z/2\Z)\oplus(\Z/2\Z)$.
\end{lem}
\begin{proof}
Margolis' theorem (\cref{Margolis_kills_differentials}) splits off the black $\Z/2\Z$ summands. If~\eqref{deg_7_xtn}
does not split, then the orange rightward-pointing triangle $\textcolor{RedOrange}{\Z/2\Z}$ and green
upward-pointing triangles $\textcolor{Green}{\Z/4\Z}$ combine into a
$\Z/8\Z$. Let $x\in\ko_7(X)$ be an element whose image in the $E_\infty$-page is a generator of the green
upward-pointing triangle
$\textcolor{Green}{\Z/2\Z}\subset E_\infty^{0,7}$; then~\eqref{deg_7_xtn} splits iff $4x = 0$.

The computation in \cref{image_from_spin_Z8} shows that $2x\in\Im(\Phi_*)$. Let
$y\in\ko_7\big((B\Z/4\Z)^{\rho-2}\big)$ be a preimage of $2x$. Then $4x = \Phi_*(2y)$, and $2y = 0$,
so~\eqref{deg_7_xtn} splits.
\end{proof}
\begin{lem}
\label{ko8_xtn}
$\ko_8(X)\cong\textcolor{BrickRed}{\Z}\oplus\textcolor{Green}{(\Z/2\Z)}$.
\end{lem}
\begin{proof}
There are infinitely many extension questions to address here, but the isomorphism type of the final answer is
either $\textcolor{BrickRed}{\Z}\oplus (\textcolor{Green}{\Z/2\Z})$, or $\Z$, and the latter case can only occur if the
first extension
\begin{equation}
	\shortexact{\textcolor{BrickRed}{E_\infty^{3,11}}}{A}{\textcolor{Green}{E_\infty^{1,9}}}{}
\end{equation}
is non-trivial, but the non-trivial $h_1$-action
$\textcolor{BrickRed}{E_\infty^{3,11}}\to\textcolor{BrickRed}{E_\infty^{4,13}}$ lifts to a nonzero action by
$\eta$ in $\ko_*(X)$, which by \cref{2eta_lemma} splits this extension.
\end{proof}
\begin{lem}
\label{ko9_xtn}
$\ko_9(X)\cong \textcolor{BrickRed}{(\Z/2\Z)}\oplus \textcolor{Goldenrod!67!black}{(\Z/2\Z)} \oplus \textcolor{Green}{(\Z/2\Z)}
\oplus\textcolor{MidnightBlue}{(\Z/2\Z)} \oplus (\Z/2\Z) \oplus (\Z/2\Z)$.
\end{lem}
\begin{proof}
We need to determine the extension
\begin{equation}
\label{green_blue_extn}
	\shortexact{\textcolor{Green}{\Z/2\Z}}{A}{\textcolor{MidnightBlue}{\Z/2\Z}} \,.
\end{equation}
A priori, the extension could also involve the red circle and yellow pentagon summands, but \cref{2eta_lemma}
implies that restricted to those summands, the extension splits, so all we have to determine
is~\eqref{green_blue_extn}. Then, argue as in the proof of \cref{nonzero_d3} to conclude that if $x$ is a preimage
of the generator of the blue square summand, then $2x = 0$. This splits~\eqref{green_blue_extn}.
\end{proof}
\begin{lem}
$\ko_{10}(X)\cong \textcolor{BrickRed}{(\Z/2\Z)}\oplus \textcolor{Goldenrod!67!black}{(\Z/2\Z)}$.
\end{lem}
\begin{proof}
We must split the extension
\begin{equation}
\label{red_yellow_extn}
	\shortexact{\textcolor{BrickRed}{\Z/2\Z}}{A}{\textcolor{Goldenrod!67!black}{\Z/2\Z}}.
\end{equation}
On the $E_\infty$-page, the action of $h_1$ from the $9$-line to the $10$-line is surjective. This lifts to imply
that the action of $\eta\in\pi_1(\mathbb S)$, which defines a map $\ko_9(X)\to\ko_{10}(X)$, is surjective. This
means the orders of the elements in $\ko_{10}(X)$ are bounded above by the order of the largest element in
$\ko_9(X)$, and $2\ko_9(X) = 0$, so there can be no order-$4$ elements in $\ko_{10}(X)$.
\end{proof}
\begin{lem}
\label{deg_11_xtn_lemma}
$\ko_{11}(X)\cong (\Z/8\Z) \oplus (\Z/2\Z)^{\oplus 4}$.
\end{lem}
\begin{proof}
The four $\Z/2\Z$ summands in filtration $0$ split off by Margolis' theorem (\cref{margolis}), because they came from
$\Sigma^{11}\cA(1)$ summands in $H^*(X;\Z/2\Z)$. That leaves the orange, green, and blue summands in higher Adams
filtration; they are in the image of the map $\Phi$ from the $E_\infty$-page for Spin-$\Z/8\Z$ bordism to
Spin-$D_{16}$ bordism. For Spin-$\Z/8\Z$ bordism, we saw in the proof of \cref{ko_spin_z8} that these three
$\Z/2\Z$ summands combined to a $\Z/8\Z$, so that is also true here.
\end{proof}
\begin{rem}
Using this, we can give another proof of \cref{deg_7_xtn_lemma}. Recall from \cref{ext_Z2} the element
$v\in\Ext^{3,7}(\Z/2\Z)$, which acts on the pages of the Adams spectral sequence for $\ko_*(X)$ for a space or
spectrum $X$. In Spin bordism, this corresponds to taking the direct product with a K3 surface.

One can show that $v$ carries the green upward-pointing triangle $\textcolor{Green}{\Z/2\Z}\subset E_\infty^{0,7}$ isomorphically onto
$E_\infty^{3,14}\cong\textcolor{Green}{\Z/2\Z}$, and likewise carries
$E_\infty^{2,9}\cong\textcolor{RedOrange}{\Z/2\Z}$ isomorphically onto
$E_\infty^{5,16}\cong\textcolor{RedOrange}{\Z/2\Z}$.\footnote{This comes from a computation of the
$\Ext(\Z/2\Z)$-action on $\Ext(\textcolor{Green}{\Z/2\Z})$ and $\Ext(\textcolor{RedOrange}{J})$. For
$\Ext(\textcolor{Green}{\Z/2\Z})$ the claimed $v$-action is a consequence of the ring structure, which we gave in
\cref{ext_Z2}; for $\Ext(\textcolor{RedOrange}{J})$ one can produce the claimed $v$-action using the long exact
sequence in Ext groups associated to a short exact sequence of $\cA(1)$-modules, similarly to how we computed some
$v$-actions in \cref{LES_ext_exm}.} Let
$x\in\ko_7(X)$ be a class whose image in the $E_\infty$-page is the generator of the green upward-pointing triangle
$\textcolor{Green}{\Z/2\Z}\subset E_\infty^{0,7}$, and suppose that the extension in~\eqref{deg_7_xtn} does not
split. As we noted in the proof of \cref{deg_7_xtn_lemma}, Margolis' theorem implies the black $\Z/2\Z$ summands
split off, so the only way for~\eqref{deg_7_xtn} to not split is for the green and orange summands to combine,
which would imply that the image of $4x$ in the $E_\infty$-page is the nonzero element of the orange triangle $\textcolor{RedOrange}{\Z/2\Z}\subset
E_\infty^{2,9}$.

Now act by $v$. \Cref{deg_11_xtn_lemma} means the colored summands in the $11$-line assemble to a $\Z/8\Z$ with
generator $y$, and a generator of the green upward-pointing triangle $\textcolor{Green}{\Z/2\Z}\subset E_\infty^{3,14}$ lifts to $2y\in\Z/8\Z$.
The $v$-action from the $7$- to the $11$-line we noted above implies $\mathrm{K3}\cdot x = 2y$,\footnote{There is
more than one choice of generator $y$; a more careful way to say this is that we can choose a generator $y$ such
that $\mathrm{K3} \cdot x = 2y$.} so $4(\mathrm{K3}\cdot x) = 0$. But the non-trivial $v$-action on the orange
summands that we noted above implies $\mathrm{K3}\cdot(4x)\ne 0$, which is a contradiction.
\end{rem}
This finishes our proof of \cref{the_ko_part}.
\end{proof}
Collecting all the pieces we find the Spin-GL$^+(2,\Z)$ bordism groups summarized in Table~\ref{tab:GL_bordism}.
\renewcommand{\arraystretch}{1.5}
\begin{table}[h!]
\centering
\begin{tabular}{c c c}
\toprule
$k$ & $\Omega^{\Spin\text{-}\GL^+(2,\Z)}_k (\pt)$ & Generators \\
\midrule
$0$ & {\footnotesize $\Z$} & {\footnotesize $\pt_+$} \\ %\hline
$1$ & {\footnotesize $(\Z/2\Z) \oplus (\Z/2\Z)$} & {\footnotesize $(L_1^4 \,, S^1_R)$} \\ %\hline
$2$ & {\footnotesize $\Z/2\Z$} & {\footnotesize $S^1_p \times S^1_R$} \\ %\hline
$3$ & {\footnotesize $(\Z/2\Z)^{\oplus 3} \oplus (\Z/3\Z)$} & {\footnotesize $(L^3_4 \,, \RP^3 \,, \widetilde{\RP}^3 \,, L^3_3)$} \\ %\hline
$4$ & {\footnotesize $\Z$} & {\footnotesize $E$} \\ %\hline
$5$ & {\footnotesize $(\Z/2\Z)^{\oplus 2}$} & {\footnotesize $(L^5_4 \,, X_5)$} \\ %\hline
$6$ & {\footnotesize $0$} & \\ %\hline
$7$ & {\footnotesize $(\Z/2\Z)^{\oplus 3} \oplus (\Z/4\Z) \oplus (\Z/9\Z)$} & {\footnotesize $(\RP^7 \,, \widetilde{\RP}^7 \,, \orangeseven \,, \halfQseven \,, L^7_3)$} \\ %\hline
$8$ & {\footnotesize $\Z\oplus \Z\oplus (\Z/2\Z)$} & {\footnotesize $(
	\halfBott\,, \HP^2\,, \halfQseven \times S^1_p
)$} \\
$9$ & {\footnotesize $(\Z/2\Z)^{\oplus 8}$} & {\footnotesize $(W^7_1 \times S_p^1 \times S_p^1 \,, X_9 \,,
	\widetilde X_9 \,, L^9_4 \,, W^9_1 \,,$} \\
& &	 {\footnotesize $\phantom{(} B\times L^1_4 \,, \HP^2 \times L^1_4 \,, \HP^2\times S_R^1 )$} \\
$10$ & {\footnotesize $(\Z/2\Z)^{\oplus 4}$} & {\footnotesize $(B\times L^1_4 \times S_p^1 \,, W^9_1 \times S_p^1 \,, \HP^2 \times
	L^1_4 \times S_p^1\,, X_{10})$}\\
$11$ & {\footnotesize $(\Z/2\Z)^{\oplus 9} \oplus (\Z/8\Z) \oplus (\Z/3\Z) \oplus (\Z/27\Z)$} & {\footnotesize $(\RP^{11} \,, \widetilde{\RP}^{11} \,, X_{11} \,, \widetilde{X}_{11} \,, \HP^2 \times L^3_4 \,, \HP^2 \times \RP^3 \,, $} \\
& &	 {\footnotesize $\phantom{(} \HP^2\times \widetilde{\RP}^3 \,, X_{10} \times L_4^1 \,, X_{10}\times S^1_R \,,  Q_4^{11} \,,   \HP^2 \times L^3_3 \,, L^{11}_3 )$} \\
\bottomrule
\end{tabular}
\caption{Bordism groups $\Omega^{\text{Spin-GL}^+(2,\Z)}_k (\pt)$ and their generators (in the same order as group
summands) for $k \leq 11$.}
\label{tab:GL_bordism}
\end{table}
Again, we already include the generators which we will derive next.

\subsection{Determining the generators at $p = 2$}
\label{ss:gl2_gens}

The map from Spin-$\Mp(2, \Z)$ bordism to Spin-$\GL^+(2, \Z)$ bordism is surjective on $3$-torsion, as we saw in
Section \ref{odd_primary_glplus}. This map is not surjective on $2$-torsion, and we need to find generators for
everything in the (quite large) cokernel.

\subsubsection{Generators coming from Spin-$\Mp(2, \Z)$ bordism}
\label{sss:mp2}

We used the map from Spin-$\Mp(2, \Z)$ bordism to Spin-$\GL^+(2, \Z)$ bordism already to solve the Adams spectral
sequence for Spin-$\GL^+(2, \Z)$ bordism; we will use this map again to find some of our generators. We already saw
in Section \ref{odd_primary_glplus} that $3$-locally, this map is surjective, so for the rest of this subsection we focus
on the situation at the prime $2$, where this is the map from Spin-$\Z/8\Z$ bordism to Spin-$D_{16}$ bordism.
Looking at \cref{image_from_spin_Z8}, we can characterize the image of this map.
\begin{itemize}
	\item In dimensions $0$ and $4$, the map is surjective, so we learn that $\mathrm{pt}_+$ generates
	$\Omega_0^{\Spin\text{-}D_{16}} (\pt)$ and the Enriques surface generates $\Omega_4^{\Spin\text{-}D_{16}} (\pt)$.
	\item In dimensions $1$ and $3$, this map sees the orange triangle $\textcolor{RedOrange}{\Z/2\Z}$ summands, giving us
	lens spaces $L^k_4$ as generators of those summands in Spin-$\GL^+(2, \Z)$ bordism; however, there are other
	$\Z/2\Z$ summands that are not in the image.
	\item In dimensions $2$ and $10$, the image of the map is $0$.
	\item $\Omega_5^{\Spin\text{-}D_{16}} (\pt) \cong (\Z/2\Z) \oplus (\Z/2\Z)$; Margolis' theorem tells us these summands are
	detected by the mod $2$ characteristic classes $w^2x$ and $w^2y$. The $w^2y$ summand is in the image of the map
	from Spin-$\Z/8\Z$ bordism, and as such is generated by the lens space $L^5_4$.
	\item In dimension $7$, the image of the map contains none of the generators, though it does contain twice the
	generator of the green triangle $\textcolor{Green}{\Z/4\Z}$ summand.
	\item In dimension $8$, one of the generators is $\HP^2$, which is in the image of this map. The
	other generator $X$ is not in the image, but $2X$, which is bordant to the Bott manifold, is.
	\item In dimension $9$, $\Omega_9^{\Spin\text{-}D_{16}}(\pt)$ consists of eight $\Z/2\Z$ summands. Three are in the
	image of the map from Spin-$\Z/8\Z$ bordism: the green triangle $\textcolor{Green}{\Z/2\Z}$ and blue square
	$\textcolor{MidnightBlue}{\Z/2\Z}$ summands in $\ko_9$, and the orange triangle $\textcolor{RedOrange}{\Z/2\Z}$ in $\ko_1$,
	giving us three generators: $L^9_4$, $\widetilde{L}^9_4$, and $\HP^2\times L^1_4$, respectively.
	\item In dimension $11$, the image of the map contains the $\Z/8\Z$ summand, giving us $Q_4^{11}$ as a generator,
	but not the four $\Z/2\Z$ summands in $\ko_{11}$ which are detected by the mod $2$ cohomology classes $x^{11}$,
	$y^{11}$, $w^4x^3$, and $w^4y^3$. On $\ko_3$ it sees one of the three $\Z/2\Z$ summands like in dimension $3$,
	giving us the generator $\HP^2\times L_4^3$. On $\ko\ang 2_1$, it sees one of the generators, giving us
	$X_{10}\times L_4^1$.
\end{itemize}
See \cref{mp2_image_diagram} for a depiction of this map.

From the perspective of F-theory, classes in $\Omega_*^{\Spin\text{-}\Mp(2, \Z)} (\pt)$
correspond to defects that preserve supersymmetry, and classes in $\Omega_*^{\Spin\text{-}\GL^+(2, \Z)} (\pt)$ can break
supersymmetry. Thus the image of the map $\Omega_*^{\Spin\text{-}\Mp(2, \Z)}\to\Omega_*^{\Spin\text{-}\GL^+(2, \Z)}$
consists of the classes which, even though we do not care about preserving supersymmetry, can nevertheless be
matched to supersymmetry-preserving defects. See Section~\ref{subsec:GLcodim2} for more information. For example,
in dimension $1$, the map $\Omega_1^{\Spin\text{-}\Mp(2, \Z)}\to\Omega_1^{\Spin\text{-}\GL^+(2, \Z)}$ is a map
$\Z/24\Z\to (\Z/2\Z) \oplus (\Z/2\Z)$: each generator of $\Z/24\Z$ is sent to the generator of
one $\Z/2\Z$ summand, and the second $\Z/2\Z$ summand is new.
\begin{itemize}
	\item The $\Z/24\Z$ summand of $\Omega_1^{\Spin\text{-}\Mp(2, \Z)}(\pt)$ is generated by the lens space $L_{12}^1$,
	corresponding in string theory to a D7-brane (see Section~\ref{subsec:Mpcodim2}).
	Twice the image of this class in Spin-$\GL^+(2, \Z)$ bordism vanishes; this means that instead of needing a
	stack of $24$ D7-branes to be trivial, it suffices to take just two.
	\item The novel $\Z/2\Z$ summand in $\Omega_1^{\Spin\text{-}\GL^+(2, \Z)} (\pt)$, generated by $S_R^1$, corresponds
	to a new defect that we predict to not preserve supersymmetry. This is the reflection $7$-brane, as we discuss
	in Section~\ref{subsec:GLcodim2} (see also~\cite{Dierigl:2022reg}).
\end{itemize}

%From the perspective of the cobordism conjecture, the map $\Omega_*^{\Spin\text{-}\Mp(2,
%\Z)}\to\Omega_*^{\Spin\text{-}\GL^+(2, \Z)}$ has a physical interpretation.
%\textcolor{red}{Physics intuition of why that is?}

\begin{figure}[h!]
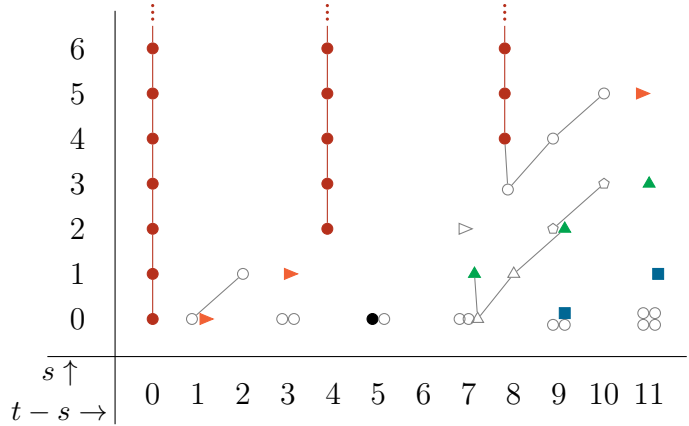

\centering
\begin{sseqdata}[name=mp2image, classes=fill, xrange={0}{11}, yrange={0}{6}, scale=0.6, Adams grading, >=stealth,
x label = {$\displaystyle{s\uparrow \atop t-s\rightarrow}$},
x label style = {font = \small, xshift = -25ex, yshift=5.5ex}
]
% A(1) summands
\class(5, 0)
\begin{scope}[fill=none, draw=gray]
	\class(3, 0)
	\class(3, 0)
	\class(5, 0)
	\class(7, 0)
	\class(7, 0)
	\class(9, 0)
	\class(9, 0)
	\class(11, 0)
	\class(11, 0)
	\class(11, 0)
	\class(11, 0)
\end{scope}
\begin{scope}[draw=none, fill=none]
	\class(4, 0)
	\class(4, 1)
	\class(7, 1)
	\class(8, 0)
	\class(8, 0)
	\class(11, 1)\class(11, 1)
	\class(11, 2)\class(11, 2)
	\class(11, 3)
	\class(11, 4)
	\class(12, 0)%\class(12, 0)\class(12, 0)
	\class(12, 1)\class(12, 1)\class(12, 1)
	\class(12, 2)\class(12, 2)
	\class(12, 3)\class(12, 3)
	\class(12, 4)
	\class(12, 5)
\end{scope}

\begin{scope}[draw=gray, fill=none]
	\class(1, 0)
	\class(2, 1)\structline[gray]
	\class(8, 3)
	\class(9, 4)
	\class(10, 5)\structline[gray]
\end{scope}

\begin{scope}[BrickRed]
	\class(0, 0)\AdamsTower{}
	
	\class(4, 2)\AdamsTower{}
		\class(8, 4)\AdamsTower{}
		\structline[gray](8, 3)(8, 4)
		\structline[gray](8, 3)(9, 4)
		\class(12, 6)\AdamsTower{}
\end{scope}
\begin{scope}[RedOrange, isosceles triangle]
	\class(1, 0)
	\class(3, 1)\AdamsTower{}
	\class[draw=gray, fill=none](7, 2)
	\class(7, 3)\AdamsTower{}
	\structline[gray](7, 2)(7, 3)
	\class[draw=gray, fill=none](8, 3)
		\structline[gray](7, 2)(8, 3, -1)
	\class(9, 4)\structline[gray]
	\class(11, 5)\AdamsTower{}
\end{scope}
\begin{scope}[Goldenrod!67!black, regular polygon,regular polygon sides=5]
	\class(4, 0)\AdamsTower{}
	\class[draw=gray, fill=none](8, 1)
	\class(8, 2)\AdamsTower{}
	\structline[gray](8, 1)(8, 2)
	\class[draw=gray, gray, fill=none](9, 2)
		\structline[gray](8, 1)(9, 2)
	\class[draw=gray, fill=none](10, 3)\structline[gray]
	\class(12, 4)\AdamsTower{}
\end{scope}
\begin{scope}[Green, regular polygon, regular polygon sides=3, minimum width=1ex]
	\class[draw=gray, fill=none](7, 0)
	\class(7, 1)\AdamsTower{}
	\structline[gray](7, 0, -1)(7, 1, -1)
	\class[draw=gray, fill=none](8, 1)
		\structline[gray](7, 0, -1)(8, 1, -1)
	\class(9, 2)\structline[gray]
	\class(11, 3)\AdamsTower{}
\end{scope}
\begin{scope}[PineGreen, star]
	\class(8, 0)\AdamsTower{}
	\class(9, 0)
	\class(10, 1)\structline
	\class(12, 2)\AdamsTower{}
\end{scope}
\begin{scope}[MidnightBlue, rectangle]
	\class(9, 0)
	\class(11, 1)\AdamsTower{}
\end{scope}
\begin{scope}[Fuchsia, diamond]
	\class(12, 0)\class(12, 1)\class(12, 2)\class(12, 3)\class(12, 4)\class(12, 5)
	\class(12, 6)\class(12, 7)\class(12, 8)
\end{scope}
% out-of-scope classes needed to draw differentials w/ s large
\begin{scope}[draw=none, fill=none]
	\class(3, 8)
	\class(7, 8)\class(7, 8)
	\class(11, 8)\class(11, 8)\class(11, 8)
	\class(12, 5)\class(12, 6)
	\class(12, 3)\class(12, 4)
	\class(12, 5)\class(12, 6)
\end{scope}
% and now for some d2s
\foreach \x in {4, 8, 12} {
	\foreach \y in {0, ..., 6} {
		\d2(\x, \y, -1)(\x-1, \y+2, -1)
	}
}
\d2(8, 1, 1)(7, 3, 1)
\d2(8, 2, 1)(7, 4, 1)
\d2(8, 3, 3)(7, 5, 1)
\foreach \y in {4, ..., 6} {
	\d2(8, \y, 2)(7, \y+2, 1)
}
\foreach \y in {2, ..., 6} {
	\d2(12, \y, -2)(11, \y+2, -2)
}
\foreach \y in {4, ..., 6} {
	\d2(12, \y, -3)(11, \y+2, -3)
}
\d[gray]3(9, 0, 3)(8, 3, 2)
\d[gray]3(10, 1)(9, 4, 2)
\end{sseqdata}
\printpage[name=mp2image, page=4]
\caption{Part of the image of the map $\Omega_*^{\Spin\text{-}\Mp(2, \Z)}(\pt) \to \Omega_*^{\Spin\text{-}\GL^+(2, \Z)} (\pt)$:
specifically, the part coming from $\ko_*$. The $\ko_{*-8}$ piece is similar; the $\ko\ang 2_{*-10}$ piece is
straightforward, thanks to \cref{ko2_is_homology}, and the $3$-local part is surjective, as we show in
Section \ref{odd_primary_glplus}.}
\label{mp2_image_diagram}
\end{figure}

\subsubsection{Generators coming from $\Omega_k^\Spin(B\Z/2\Z)$}
\label{Z2_to_D16}

Let $i_4, \tilde\imath_4\colon D_4\inj D_{16}$ be the two inclusions defined in Appendix \ref{app:embdihedral}:
specifically, these maps send rotations to rotations and reflections to reflections, but $i_4$ sends the standard
generating reflection $s$ of $D_4$ to the standard generating reflection of $D_8$, and $\tilde\imath_4$ sends $s$
to a different reflection. The maps $i_4$ and $\tilde\imath_4$ induce two maps from Spin-$D_4$ manifolds to
Spin-$D_{16}$ mani\-folds. Since $D_4\cong(\Z/2\Z)\oplus(\Z/2\Z)$, $\Spin\times_{\Z/2\Z}D_4\cong\Spin\times\Z/2\Z$ as
symmetry types: a Spin-$D_4$ manifold is the same thing as a Spin manifold with a principal $\Z/2\Z$-bundle. The
generators of $\Omega_k^\Spin(B\Z/2\Z)$ are well understood, so we will use them to find more Spin-$D_8$ manifolds.
See \cref{Z2_einf} for a depiction of what this method can see.

\begin{figure}[h!]
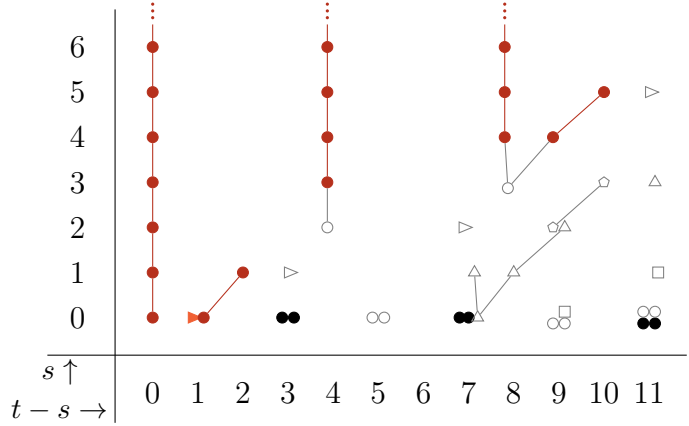

\centering
\begin{sseqdata}[name=dihsseqZ2, classes=fill, xrange={0}{11}, yrange={0}{6}, scale=0.6, Adams grading, >=stealth,
x label = {$\displaystyle{s\uparrow \atop t-s\rightarrow}$},
x label style = {font = \small, xshift = -25ex, yshift=5.5ex}
]
% A(1) summands
\class[RedOrange, isosceles triangle](1, 0)
\class(3, 0)
\class(3, 0)
\class[draw=gray, fill=none](5, 0)
\class[draw=gray, fill=none](5, 0)
\class(7, 0)
\class(7, 0)
\class[draw=gray, fill=none](9, 0)
\class[draw=gray, fill=none](9, 0)
\class(11, 0)
\class(11, 0)
\class[draw=gray, fill=none](11, 0)
\class[draw=gray, fill=none](11, 0)
\begin{scope}[BrickRed]
	\class(0, 0)\AdamsTower{}
	\class(1, 0)
	\class(2, 1)\structline
	\class(4, 3)\AdamsTower{}
	\class(8, 4)\AdamsTower{}
	\class(9, 4)
	\class(10, 5)\structline
	\class(12, 7)\AdamsTower{}
\end{scope}
\begin{scope}[draw=none, fill=none]
	\class(4, 0)
	\class(4, 1)
	\class(7, 1)
	\class(8, 0)
	\class(8, 0)
	\class(11, 1)\class(11, 1)
	\class(11, 2)\class(11, 2)
	\class(11, 3)
	\class(11, 4)
	\class(12, 0)%\class(12, 0)\class(12, 0)
	\class(12, 1)\class(12, 1)\class(12, 1)
	\class(12, 2)\class(12, 2)
	\class(12, 3)\class(12, 3)
	\class(12, 4)
	\class(12, 5)
\end{scope}

\begin{scope}[draw=gray, fill=none]
\class(4, 2)\structline(4, 2)(4, 3)
\class(8, 3)\structline(8, 3)(8, 4)
\class(12, 6)\structline(12, 6)(12, 7)
	\structline(8, 3)(9, 4)
	\class[isosceles triangle](3, 1)\AdamsTower{}
	\class[isosceles triangle](7, 2)\AdamsTower{}
	\class(8, 3)
		\structline(7, 2)(8, 3, -1)
	\class(9, 4)\structline
	\class[isosceles triangle](11, 5)
	\class(4, 0)\AdamsTower{}
	\class(8, 1)\AdamsTower{}
	\class[regular polygon,regular polygon sides=5](9, 2)
		\structline(8, 1)(9, 2)
	\class[regular polygon,regular polygon sides=5](10, 3)\structline
	\begin{scope}[regular polygon, regular polygon sides=3, minimum width=1ex]
		\class(7, 0)\AdamsTower{}
		\class(8, 1)
			\structline(7, 0, -1)(8, 1, -1)
		\class(9, 2)\structline
		\class(11, 3)
	\end{scope}
	\class(8, 0)\AdamsTower{}
	\class(9, 0)
	\class(10, 1)\structline
	\class[rectangle](9, 0)
	\class[rectangle](11, 1)
\end{scope}
% out-of-scope classes needed to draw differentials w/ s large
\begin{scope}[draw=none, fill=none]
	\class(3, 10)
	\class(7, 10)\class(7, 10)
	\class(11, 10)\class(11, 10)\class(11, 10)
\end{scope}
\class(3, 9)
\class(3, 8)
\class(7, 8)
\class(7, 9)
% and now for some d2s
\foreach \x in {4, 8} {
	\foreach \y in {0, ..., 7} {
		\d2(\x, \y, -1)(\x-1, \y+2, -1)
	}
}
\d2(8, 1, 1)(7, 3, 1)
\d2(8, 2, 1)(7, 4, 1)
\d2(8, 3, 3)(7, 5, 1)
\foreach \y in {4, ..., 7} {
	\d2(8, \y, 2)(7, \y+2, 1)
}

%\foreach \y in {2, ..., 8} {
%	\d2(12, \y, -2)(11, \y+2, -2)
%}
%\foreach \y in {4, ..., 8} {
%	\d2(12, \y, -3)(11, \y+2, -3)
%}
\d[gray]3(9, 0, 3)(8, 3, 2)
\d[gray]3(10, 1)(9, 4, 2)
\end{sseqdata}
\printpage[name=dihsseqZ2, page=4]
\caption{The colored and black summands are the image of $\Omega_*^\Spin(B\Z/2\Z)\to\Omega_*^{\Spin\text{-}D_{16}}(\pt)$
under the two maps $i_4$, $\tilde\imath_4$. See Section \ref{Z2_to_D16} for more information. Like in
\cref{mp2_image_diagram}, this picture just encodes the summands coming from $\ko_*$; $\ko_{*-8}$ is analogous and
$\ko\ang 2_{*-10}$ is straightforward thanks to \cref{ko2_is_homology}.}
\label{Z2_einf}
\end{figure}

Given a Spin manifold $M$ with a principal $\Z/2\Z$-bundle, we will let $M$ denote its Spin-$D_{16}$ structure
induced by $i_4$ and $\widetilde M$ denote its Spin-$D_{16}$ structure induced by $\tilde\imath_4$. We will find
several examples where one generator of $\Omega_*^\Spin(B\Z/2\Z)$ gives us two different generators of
$\Omega_*^{\Spin\text{-}D_{16}} (\pt)$ under $i_4$ and $\tilde\imath_4$.

Recall that $H^*(B\Z/2\Z;\Z/2\Z) \cong \Z/2\Z[a]$, with $a$ in degree $1$.
\begin{lem}
\label{spin_D4_pullback}
The pullback map $i_4^*\colon H^*(BD_8;\Z/2\Z) \to H^*(B\Z/2\Z;\Z/2\Z)$ sends $x \mapsto a$, $y \mapsto 0$, and $w
\mapsto 0$. $\tilde\imath_4^*$ sends $x\mapsto a$, $y\mapsto a$, and $w\mapsto 0$.
\end{lem}
\begin{proof}
Recall that we defined $x$, $y$, and $w$ as Stiefel-Whitney classes of vector bundles associated to representations
of $D_8$. Therefore to understand $i_4^*$ and $\tilde\imath_4^*$, it suffices to restrict those representations to
$\Z/2\Z \subset D_8$ and compute the Stiefel-Whitney classes of the associated vector bundles over $B\Z/2\Z$.
\begin{itemize}
	\item For $x$ and $w$, we used the standard representation $V$ of $D_8$ on $\mathbb R^2$ by rotations and
	reflections. For both $i_4$ and $\tilde\imath_4$, the pullback of $V$ is $\mathbb R\oplus\sigma$, where
	$\sigma$ is the sign representation, so the pullback of $x$ is $w_1(\sigma) = a$ and the pullback of $w$ is
	$w_2(\sigma) = 0$.
	\item For $y$, we used the real character $\chi_y$ of $D_8$ in which a quarter-turn is sent to $-1$ and
	the reflection $s$ is sent to $1$. Therefore $i_4^*(\chi_y)$ is trivial, so $i_4^*(y) = 0$ --- and since
	$\chi_y(rs) = -1$, $\tilde\imath_4^*(y)$ is non-trivial! Hence it is the unique nonzero character of $\Z/2\Z$,
	which has $w_1 = a$, so $\tilde\imath_4^*(y) = a$.
	\qedhere
\end{itemize}
\end{proof}
In particular, we hypothesize that if $c$ is a mod $2$ cohomology class corresponding to some generator of
$\Omega_*^{\Spin\text{-}D_{16}} (\pt)$ in Adams filtration $0$, and $i_4^*(c)$ or $\tilde\imath_4^*(c)$ is nonzero, then
perhaps that generator can be realized by an element of $\Omega_*^\Spin(B\Z/2\Z)$. This guess applies to
$\textcolor{BrickRed}{x}$, $\textcolor{RedOrange}{y}$, $x^3$, $y^3$, $x^7$, $y^7$, $w_4^2x$, $w_4^2y$, $x^{11}$,
$y^{11}$, $w_4^2x^3$, $w_4^2y^3$, $w_4w_6x$, and $w_4w_6y$ --- and in all of these cases, we can find a generator.
\begin{itemize}
	\item For $\textcolor{BrickRed}{x}$, $\textcolor{RedOrange}{y}$, $x^3$, $y^3$, $x^7$, $y^7$, $x^{11}$, and
	$y^{11}$, the generator is the real projective space of the appropriate dimension, with principal $\Z/2\Z$-bundle
	$S^k\to\RP^k$. For $x^k$, use $\RP^k$, and for $y^k$, use $\widetilde{\RP}{}^k$. We use $S_R^1$ to denote
	$\RP^1$, as we often think of this generator in Part~\ref{p:physics} as a circle with duality bundle induced by
	the Möbius bundle and a map $\Z/2\Z\inj\GL^+(2, \Z)$ defined by a reflection.
	\item For $w_4^2x$ and $w_4^2y$, use $\HP^2\times S^1_R$, resp.\ $\HP^2\times\widetilde S^1_R$, with principal
	$\Z/2\Z$-bundle the Möbius bundle on the $S^1$ factor. Likewise, for $w_4^2x^3$ and $w_4^2y^3$, use $\HP^2\times
	\RP^3$ and $\HP^2\times \widetilde{\RP}{}^3$, with principal $\Z/2\Z$-bundle $\HP^2\times
	S^3\to\HP^2\times\RP^3$.
	\item For $w_4w_6x$, use $X_{10}\times S^1_R$, and for $w_4w_6y$, use $X_{10}\times\widetilde S^1_R$; in both cases
	the principal $\Z/2\Z$-bundle is the Möbius bundle on the $S^1_R$ factor.
\end{itemize}
We can also use $\Omega_*^\Spin(B\Z/2\Z)$ to see three classes in positive Adams filtration. The $h_1$-action
$E_\infty^{0,1}\to E_\infty^{1,3}$ lifts to the product with $S_p^1$, as we discussed in Section \ref{sss:non_hidden},
implying that $\Omega_2^{\Spin\text{-}D_{16}} (\pt)$ is
generated by $S_p^1\times S^1_R$, with the principal $\Z/2\Z$-bundle non-trivial on the second factor; it
does not matter whether we use $i_4$ or $\tilde\imath_4$ here. The action of the Bott element $w\in\Ext(\Z/2\Z)$
defines an isomorphism $E_\infty^{0,1}\to E_\infty^{4,13}$; as we mentioned in Section \ref{sss:non_hidden}, this
means the red circle summand $\textcolor{BrickRed}{\Z/2\Z}$ of
$\Omega_9^{\Spin\text{-}D_{16}} (\pt)$ (corresponding to $E_\infty^{4,13}$)
is generated by $B \times S^1_R$, with the principal $\Z/2\Z$-bundle non-trivial on
the second factor. Finally, the $h_1$-action $E_\infty^{4,13}\to E_\infty^{5,15}$ lifts to tell us that the red
circle
$\textcolor{BrickRed}{\Z/2\Z}$ summand of $\Omega_{10}^{\Spin\text{-}D_{16}}(\pt)$ is generated by $B\times
S_p^1\times S^1_R$, with the principal $\Z/2\Z$-bundle non-trivial on the last factor.
%\begin{rem}
%The representation \textcolor{red}{TODO} pulls back under $D_8 \surj \Z/2\Z$ to define a bordism invariant of Spin-$D_{16}$ manifolds,
%and the value of the corresponding $\eta$-invariant on $\HP^2 \times S_p^1\times S^1_R$ (with the Möbius
%bundle on the last factor) is \textcolor{red}{TODO}. This provides another proof of \cref{d4_vanishes}: if that $d_4$ were nonzero,
%this manifold would bound as a Spin-$D_{16}$ manifold.
%%
%%\textcolor{red}{(TODO: the fact that the Spin-$\Z/8\Z$ $\eta$-invariant makes sense as a Spin-$D_{16}$ $\eta$-invariant also provides a
%%third proof)}
%\end{rem}
%
%The quotient
%$q\colon D_8\surj\Z/2$ means that a representation of $\Z/2$ induces a representation of $D_8$, and the
%$\eta$-invariants associated to those two representations coincide. Pulling back along $i_4$ or $\tilde\imath_4$,
%we obtain the identity map $\Z/2\inj D_8\surj\Z/2$, which means that if an $\eta$-invariant is a bordism invariant
%of spin manifolds with a principal $\Z/2$-bundle, and it is valid for spin-$D_8$ manifolds (\TODO: fill in what
%this means precisely), then $(q^*\eta)(i_4(M)) = \eta(M)$, and similarly for $\tilde\imath_4$. That is, we
%

\subsubsection{Spin-$D_8$ generators and the Arcana}
\label{arcana}

We can do better by using $D_4$ instead of $\Z/2\Z$. Analogous to $i_4$ and $\tilde\imath_4$, there are two maps
$i_8,\tilde\imath_8\colon D_4\inj D_{8}$, which we also discuss in Appendix \ref{app:embdihedral}. Since $D_4 \cong
(\Z/2\Z) \oplus (\Z/2\Z)$, $H^*(BD_4;\Z/2\Z)\cong\Z/2\Z[a, b]$ with $\abs a = \abs b = 1$.
\begin{lem}
\label{D8_pullback}
$i_8^*(x) = \tilde\imath_8^*(x) = \tilde\imath_8^*(y) = a$, $i_8^*(y) = 0$, and
$i_8^*(w) = \tilde\imath_8^*(w) = ab + b^2$.
\end{lem}
\begin{proof}
Recall that we defined $x$, $y$, and $w$ as Stiefel-Whitney classes of vector bundles associated to
$D_8$-representations. Therefore their pullbacks to $BD_4$ are the corresponding Stiefel-Whitney
classes of the bundles associated to the restrictions of those representations to $D_4$. Since $D_4$ is Abelian,
these representations split into sums of one-dimensional representations, allowing one to calculate their
Stiefel-Whitney classes in terms of $a$ and $b$.
\end{proof}
We say that a \term{Spin-$D_8$ structure} is analogous to a Spin-$D_{16}$ structure, but with $D_8$ in place of
$D_{16}$. Giambalvo~\cite{Gia76} and Pedrotti~\cite[Section 6]{Ped17} study Spin-$D_8$ bordism, and
Barkeshli-Chen-Hsin-Manjunath~\cite[Section VII.D]{BCHM21}, Ning-Qi-Gu-Wang~\cite[Section II]{NQGW21}, and
Manjunath-Calvera-Barkeshli~\cite{MCB22} apply this symmetry type to physics. \Cref{D8_pullback} says that a
Spin-$D_8$ structure is data of two principal $\Z/2\Z$-bundles (equivalent to a principal $(D_4 = (\Z/2\Z) \oplus
(\Z/2\Z)$-bundle) and a trivialization of $w_2(TM) - ab - b^2$, where $a$ and $b$ are the degree-$1$ characteristic
classes of the principal $\Z/2\Z$-bundles. Given a Spin-$D_8$ manifold $M$, $i_8$ and $\tilde\imath_8$ endow it
with two Spin-$D_{16}$ structures, which we denote $M$ and $\widetilde M$, respectively.

The shearing argument in \cref{spin_D16_shear} applies \textit{mutatis mutandis} with $D_4$ in place of $D_8$: if
$V$ is the (associated vector bundle to the) standard two-dimensional real representation of $D_4$, we obtain a
homotopy equivalence
\begin{equation}
	\mathit{MT}(\Spin\text{-}D_8)\overset\simeq\longrightarrow \MTSpin\wedge (BD_4)^{V + 3\Det(V) - 5}.
\end{equation}
So as usual, we need information about the $\cA(1)$-module structure on cohomology of $M_4\coloneqq (BD_4)^{V +
3\Det(V) - 5}$ to run the Adams spectral sequence for Spin-$D_8$ bordism. By pulling back from $H^*(BD_8;\Z/2\Z)$,
one learns that in $H^*(BD_4;\Z/2\Z)$, $w_1(V + 3\Det(V) - 5) = 0$ and $w_2(V + 3\Det(V) - 5) = ab + b^2$. Therefore
in $H^*(M_4;\Z/2\Z)$, $\Sq^1(U) = 0$ and $\Sq^2(U) = U(ab + b^2)$. The Steenrod squares of $a$ and $b$ can be
determined directly from the axioms to be $\Sq(a) = a + a^2$ and $\Sq(b) = b + b^2$. Therefore using the Cartan
formula we can compute Steenrod squares of arbitrary elements of $H^*(M_4;\Z/2\Z)$.
\begin{lem}
\label{triple_Sq2_D8}
Suppose $x\in H^k(M_4;\Z/2\Z)$ is such that $\Sq^2\Sq^2\Sq^2(x)\ne 0$. Then there is a $\Z/2\Z$ summand in
$\Omega_k^{\Spin\text{-}D_8} (\pt)$ detected by $x$.
\end{lem}
\begin{proof}
The fact that $\Sq^2\Sq^2\Sq^2(x)\ne 0$ implies that $x$ generates a free rank-$1$ summand in
$H^*(M_4;\Z/2\Z)$~\cite[Lemma D.8]{FH16}. Margolis' theorem, in the form of \cref{Margolis_kills_differentials}, then
implies the element of $\Ext^{0,k}\big(H^*(M_4;\Z/2\Z)\big)$ corresponding to this summand survives to the $E_\infty$-page
and generates a $\Z/2\Z$ summand of $\Omega_k^{\Spin\text{-}D_8} (\pt)$. Finally, classes on the $E_\infty$-page in
filtration $0$ are detected by the corresponding classes in mod $2$ cohomology (see~\cite[Section 8.4]{FH21}),
so this summand is detected by $x$.
\end{proof}
We will use this several times below to show that various characteristic classes detect elements in Spin-$D_8$
bordism; in all cases, we verified $\Sq^2\Sq^2\Sq^2(x)\ne 0$ using a computer, though these computations are not
too tedious to do by hand.

Giambalvo~\cite[Lemma 4.2]{Gia76} shows that the Adams spectral sequence for Spin-$D_8$ bordism
collapses.\footnote{What Giambalvo computes is slightly different: the characteristic class condition is $w_2(M) =
ab$, rather than $w_2(M) = ab + b^2$. These two conditions are related by an automorphism of $D_4 =
(\Z/2\Z) \oplus (\Z/2\Z)$, so their notions of bordism are the same.}
Now we characterize the image of Spin-$D_8$ bordism in Spin-$D_{16}$ bordism. Giambalvo's computations \cite[Section 4]{Gia76} show that there is an $R_2$ summand in $H^*(M_4;\Z/2\Z)$, and the
$\textcolor{BrickRed}{R_2}$ summand in $H^*\big((BD_8)^{V + 3\Det(V)-5};\Z/2 \Z\big)$ pulls back isomorphically to this
$R_2$,
%The
%$\textcolor{BrickRed}{R_2}$ summand in $\tH^*((BD_8)^{V+3\Det(V)-5};\Z/2)$ pulls back to the
%$\textcolor{BrickRed}{R_2}$ summand in $\tH^*((BD_4)^{V+3\Det(V)-5};\Z/2)$,
so the map on the $E_2$-pages between the Ext groups of these summands is an isomorphism. The rest of the
$E_2$-page of the Adams spectral sequence is in filtration $0$ \cite[Section 4]{Gia76}, which is canonically a
subspace of mod $2$ cohomology, so we can just compute how mod $2$ cohomology classes pull back along the two maps
$D_4\to D_8$.
\begin{itemize}
	\item Since $i_4$ and $\tilde\imath_4$ factor through $D_4 \inj D_8$ followed by $i_8,\tilde\imath_8\colon
	D_8 \to D_{16}$, the image of Spin-$D_8$ bordism inside Spin-$D_{16}$ bordism contains everything we found in
	\cref{Z2_to_D16}.
	\item In degree $5$, we want to detect the cohomology classes $w^2x$ and $w^2y$. We get $i_8^*(w^2x) =
	\tilde\imath_8^*(w^2x) = a^3b^2 +ab^4$, $i_8^*(w^2y) = 0$, and $\tilde\imath_8^*(w^2y) = a^3b^2 + ab^4$. One can
	calculate that $\Sq^2\Sq^2\Sq^2(U(a^3b^2 + ab^4)) \ne 0$, so by \cref{triple_Sq2_D8}, there is a closed
	Spin-$D_8$ $5$-manifold detected by $a^3b^2 + ab^4$. Applying $i_8^*$ and $\tilde\imath_8^*$, $X_5$ is the
	generator corresponding to $w^2x$ and $\widetilde X_5$ is the generator corresponding to $w^2y$.
	 \item In degree $7$, we want to detect the cohomology class $w^3y$ (we already hit $x^7$ and $y^7$ using
	$\Omega_7^{\Spin}(B\Z/2\Z)$). Even though $i_8^*(w^3y) = 0$, $\tilde\imath_8^*(w^3y)\ne 0$:
	\begin{equation}
	\label{imathw8}
		\tilde\imath_8^*(w^3y) = a^4b^3 + a^3b^4 + a^2b^5 + ab^6,
	\end{equation}
	but $\Sq^2\Sq^2\Sq^2$ vanishes on this class. In fact, the bordism invariant
	$\Omega_7^{\Spin\text{-}D_8} (\pt) \to\Z/2\Z$ given by integrating this class vanishes! One can deduce this a priori by noting:
	$2 \Omega_7^{\Spin\text{-}D_8} (\pt) = 0$, and if $M$ is a closed Spin-$D_{16}$ $7$-manifold with $w^3y\ne 0$ then
	$2 [M] \ne 0$, as $[M]$ generates the green triangle $\textcolor{Green}{\Z/4\Z}\subset\Omega_7^{\Spin\text{-}D_{16}}(\pt)$. So if
	$\tilde\imath_8^*(w^3y)$ were nonzero on some manifold $M$, then $2[M] = 0$ in $\Omega_7^{\Spin\text{-}D_8} (\pt)$
	and $2[\widetilde M]\ne 0$ in $\Omega_7^{\Spin\text{-}D_{16}} (\pt)$, which would be no good.
%	We will prove this in
%	\cref{no_w3y_d8}.
%	(Details: We know this bordism group is $(\Z/2)^{\oplus 3}$ detected by three specific mod $2$ cohomology
%	classes. The generators are in \verb+spin_d8.pdf+, and this class vanishes on all of the generators.)
	\item In degree $9$, we want to detect $w^2x^5$ and $w^2y^5$. As usual, $i_8^*(w^2y^5) = 0$ and
	\begin{equation}
		i_8^*(w^2x^5) = \tilde\imath_8^*(w^2x^5) = \tilde\imath_8^*(w^2y^5) = a^7b^2 + a^5b^4.
	\end{equation}
	$\Sq^2\Sq^2\Sq^2(U(a^7b^2 + a^5b^4))\ne 0$, so by \cref{triple_Sq2_D8}, some closed Spin-$D_8$ manifold $X_9$ is
	detected by $a^7b^2 + a^5b^4$, and $X_9$, respectively\ $\widetilde X_9$ give us $w^2x^5$, respectively $w^2y^5$.

	There is another summand of $\Omega_9^{\Spin\text{-}\GL^+(2, \Z)} (\pt)$ detected by the mod $2$ cohomology
	class $w^4y$. One can use Wu's theorem to show that this class vanishes on all Spin-$D_8$ manifolds. This class
	is in the image of the map from Spin-$\Mp(2, \Z)$ bordism and we already know a representative for it, so we
	will not worry about it.
	\item Finally, degree $11$. We want to detect the classes $w^4x^3$ and $w^4y^3$, which pull back as
	$i_8^*(w^4y^3) = 0$ and
	\begin{equation}
		i_8^*(w^4x^3) = \tilde\imath_8^*(w_4x^3) = \tilde\imath_8^*(w_4y^3) = a^7b^4 + a^3b^8.
	\end{equation}
	As usual, check that $\Sq^2\Sq^2\Sq^2$ of this is nonzero, which by \cref{triple_Sq2_D8} implies there is a
	closed Spin-$D_8$ $11$-manifold such that $X_{11}$ and $\widetilde X_{11}$ are detected by $w^4x^3$ and $w^4y^3$.
	This fulfills a promise we made in~\cite[Appendix C]{Debray:2021vob}.
\end{itemize}
\begin{figure}[h!]
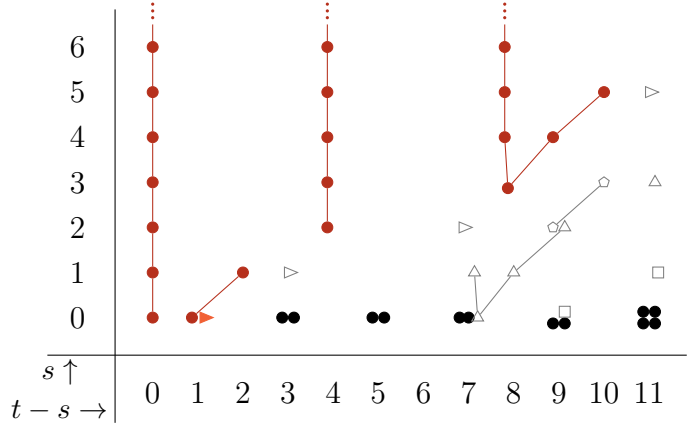

\centering
%\begin{subfigure}[c]{0.49\textwidth}
%\begin{sseqdata}[name=spinD8Adams, classes=fill, xrange={0}{11}, yrange={0}{6}, scale=0.5, Adams grading, >=stealth]
%% A(1) summands
%\class(3, 0)
%\class(3, 0)
%\class(5, 0)
%\class(5, 0)
%\class(7, 0)
%\class(7, 0)
%\class(7, 0)
%\class(9, 0)
%\class(9, 0)
%\class(9, 0)
%\class(11, 0)
%\class(11, 0)
%\class(11, 0)
%\class(11, 0)
%\begin{scope}[BrickRed]
%	\class(0, 0)\AdamsTower{}
%	\class(1, 0)
%	\class(2, 1)\structline
%	\class(4, 2)\AdamsTower{}
%	\class(8, 3)\AdamsTower{}
%	\class(9, 4)\structline(8, 3)(9, 4)
%	\class(10, 5)\structline
%	\class(12, 7)\AdamsTower{}
%\end{scope}
%\class(1, 0)
%\end{sseqdata}
%\printpage[name=spinD8Adams, page=2]
%\end{subfigure}
%\begin{subfigure}[c]{0.49\textwidth}
\begin{sseqdata}[name=spinD8image, classes=fill, xrange={0}{11}, yrange={0}{6}, scale=0.6, Adams grading,
>=stealth,
x label = {$\displaystyle{s\uparrow \atop t-s\rightarrow}$},
x label style = {font = \small, xshift = -25ex, yshift=5.5ex}
]
% A(1) summands
\class(3, 0)
\class(3, 0)
\class(5, 0)
\class(5, 0)
\class(7, 0)
\class(7, 0)
\class(9, 0)
\class(9, 0)
\class(11, 0)
\class(11, 0)
\class(11, 0)
\class(11, 0)
\begin{scope}[BrickRed]
	\class(0, 0)\AdamsTower{}
	\class(1, 0)
	\class(2, 1)\structline
	\class(4, 2)\AdamsTower{}
	\class(8, 3)\AdamsTower{}
	\class(9, 4)\structline(8, 3)(9, 4)
	\class(10, 5)\structline
	\class(12, 7)\AdamsTower{}
\end{scope}
\class[RedOrange, isosceles triangle](1, 0)
\begin{scope}[draw=none, fill=none]
	\class(4, 0)
	\class(4, 1)
	\class(7, 1)
	\class(8, 0)
	\class(8, 0)
	\class(11, 1)\class(11, 1)
	\class(11, 2)\class(11, 2)
	\class(11, 3)
	\class(11, 4)
	\class(12, 0)%\class(12, 0)\class(12, 0)
	\class(12, 1)\class(12, 1)\class(12, 1)
	\class(12, 2)\class(12, 2)
	\class(12, 3)\class(12, 3)
	\class(12, 4)
	\class(12, 5)
\end{scope}

\begin{scope}[draw=gray, fill=none]
\class(12, 6)\structline(12, 6)(12, 7)
	\begin{scope}[isosceles triangle]
		\class(3, 1)\AdamsTower{}
		\class(7, 2)\AdamsTower{}
		\class(8, 3)
		\class(11, 5)
		\structline(7, 2)(8, 3, -1)
	\end{scope}
	\class(9, 4)\structline
	\class(4, 0)\AdamsTower{}
	\class(8, 1)\AdamsTower{}
	\begin{scope}[regular polygon,regular polygon sides=5]
		\class(9, 2)
			\structline(8, 1)(9, 2)
		\class(10, 3)\structline
	\end{scope}
	\begin{scope}[regular polygon, regular polygon sides=3, minimum width=1ex]
		\class(7, 0)\AdamsTower{}
		\class(8, 1)
			\structline(7, 0, -1)(8, 1, -1)
		\class(9, 2)\structline
		\class(11, 3)
	\end{scope}
	\class(8, 0)\AdamsTower{}
	\class(9, 0)
	\class(10, 1)\structline
	\begin{scope}[rectangle]
		\class(9, 0)
		\class(11, 1)
	\end{scope}
\end{scope}
% out-of-scope classes needed to draw differentials w/ s large
\begin{scope}[draw=none, fill=none]
	\class(3, 10)
	\class(7, 10)\class(7, 10)
	\class(11, 10)\class(11, 10)\class(11, 10)
\end{scope}
\class(3, 9)
\class(3, 8)
\class(7, 8)
\class(7, 9)
% and now for some d2s
\foreach \x in {4, 8} {
	\foreach \y in {0, ..., 7} {
		\d2(\x, \y, -1)(\x-1, \y+2, -1)
	}
}
\d2(8, 1, 1)(7, 3, 1)
\d2(8, 2, 1)(7, 4, 1)
\d2(8, 3, 3)(7, 5, 1)
\foreach \y in {4, ..., 7} {
	\d2(8, \y, 2)(7, \y+2, 1)
}

%\foreach \y in {2, ..., 8} {
%	\d2(12, \y, -2)(11, \y+2, -2)
%}
%\foreach \y in {4, ..., 8} {
%	\d2(12, \y, -3)(11, \y+2, -3)
%}
\d[gray]3(9, 0, 3)(8, 3, 2)
\d[gray]3(10, 1)(9, 4, 2)
\end{sseqdata}
\printpage[name=spinD8image, page=4]
%
%\end{subfigure}
\caption{%Left: The Adams spectral sequence computing $\ko_*((BD_4)^{V+3\Det(V)-5})$, which collapses in this range.
%See \cref{D8_Adams_Cor}. Right:
The image of Spin-$D_8$ bordism in Spin-$D_{16}$ bordism; see Section \ref{arcana} for more information. Like in
\cref{mp2_image_diagram}, this picture just encodes the summands coming from $\ko_*$; $\ko_{*-8}$ is analogous and
we already accounted for $\ko\ang 2_{*-10}$ in \S\ref{Z2_to_D16}.}
\label{spinD8_figs}
\end{figure}

Now we need to actually write down $X_5$, $X_9$, and $X_{11}$. We call these manifolds ``Arcana manifolds.'' We
wrote down $X_{11}$ in~\cite[Appendix C.2]{Debray:2021vob}, but the other two are new. The Arcana are generalized Dold
manifolds in the sense of Nath-Sankaran~\cite{NS19}, and are a special case of a construction studied by
Sarkar-Zvengrowski~\cite[Example 3.1]{SZ20}.
\begin{defn}
The ``Arcanum XI'' manifold $X_{11}$ is
$(S^6\times S^5)/\big( (\Z/2\Z) \oplus (\Z/2\Z) \big)$, where if we regard $S^6\times S^5\subset\mathbb
R^7\times\mathbb R^6$ with coordinates $(x_1,\dotsc,x_7, y_1,\dotsc,y_6)$, the generators $\alpha$ and $\beta$ of
$(\Z/2\Z) \oplus (\Z/2\Z)$ act by
\begin{equation}
\begin{aligned}
	\alpha(x_1, \dotsc, x_7, y_1, \dotsc, y_6) &= (-x_1, \dotsc, -x_7, -y_1, y_2, \dotsc, y_6)\\
	\beta(x_1, \dotsc, x_7, y_1, \dotsc, y_6)  &= (x_1, \dotsc, x_7, -y_1, \dotsc, -y_6).
\end{aligned}
\end{equation}
\end{defn}
In~\cite[Appendix C.2]{Debray:2021vob}, we show this is a Spin-$D_8$ manifold with $a^7b^4 + a^3b^8\ne 0$, so that it is the
manifold we were looking for in dimension $11$.

We define $X_5$ and $X_9$ similarly.
\begin{defn}
The ``Arcanum V'' manifold $X_5$ is a quotient $(S^2\times S^3)/\big( (\Z/2\Z) \oplus (\Z/2\Z) \big)$. As above, we embed $S^2\times
S^3\subset\R^3\times\R^4$ with coordinates $(x_1,x_2,x_3,y_1,y_2,y_3,y_4)$ and let the generators $\alpha$ and
$\beta$ of $(\Z/2\Z) \oplus (\Z/2\Z)$ act by
\begin{equation}
\begin{aligned}
	\alpha(x_1, x_2, x_3, y_1, y_2, y_3, y_4)  &= (x_1, x_2, x_3, -y_1, -y_2, -y_3, -y_4)\\
	\beta(x_1, x_2, x_3, y_1, y_2, y_3, y_4) &= (-x_1, -x_2, -x_3, y_1, -y_2, -y_3, -y_4).
\end{aligned}
\end{equation}
The ``Arcanum IX'' manifold $X_9$ is a quotient $(S^6\times S^3)/ \big( (\Z/2\Z) \oplus (\Z/2\Z) \big)$. As above, we embed $S^6\times
S^3\subset\R^7\times\R^4$ with coordinates $(x_1,\dots, x_7,y_1,y_2,y_3,y_4)$ and let the generators $\alpha$ and
$\beta$ of $(\Z/2\Z) \oplus (\Z/2\Z)$ act by
\begin{equation}
\begin{aligned}
	\alpha(x_1, \dots, x_7, y_1, y_2, y_3, y_4)  &= (x_1, \dots, x_7, -y_1, -y_2, -y_3, -y_4)\\
	\beta(x_1, \dots, x_7, y_1, y_2, y_3, y_4) &= (-x_1, \dots, -x_7, y_1, -y_2, -y_3, -y_4).
\end{aligned}
\end{equation}
\end{defn}
Projection onto the $x_i$ coordinates makes $X_5$ and $X_9$ into fiber bundles over $\RP^2$ and $\RP^6$,
respectively, with fiber $\RP^3$ in both cases.
\begin{rem}
All three of these manifolds are projectivizations of vector bundles: if $\sigma\to\RP^n$ denotes the tautological
line bundle, then there are diffeomorphisms $X_5 \cong \mathbb P(\underline\R\oplus\sigma^{\oplus 3})\to \RP^2$,
$X_9\cong \mathbb P(\underline\R\oplus \sigma^{\oplus 3})\to\RP^6$, and $X_{11}\cong \mathbb
P(\underline\R^5\oplus\sigma)\to\RP^6$.
%I think $X_5$ is the projectivization of a vector bundle; specifically, $X_5 = \mathbb P(\underline\R\oplus
%3\sigma)\to\RP^2$. Here's why: $\beta$ is the data making $\underline\R^4\to S^2$ into a $\Z/2$-equivariant vector
%bundle, where $\Z/2$ acts antipodally on $S^2$ (the $x_i$-coordinates are $S^2$ and the $y_i$-coordinates are the
%vector bundle). Taking the quotient, this descends to the vector bundle $\underline\R \oplus 3\sigma\to\RP^2$,
%where $\sigma\to\RP^2$ is the tautological line bundle. Now we need to quotient by $\alpha$, which amounts to
%identifying antipodal points in the fiber of $S(\underline\R\oplus 3\sigma)$, giving us the projectivization.
%\TODO: does this simplify the calculations below?
%
%Probably something analogous is true for $X_9$, and maybe for $X_{11}$.
\end{rem}
Now we want to show that $X_5$ and $X_9$ are Spin-$D_8$ manifolds which represent the bordism classes that we are
looking for.
To perform this computation, we will \emph{stably split} the tangent bundles of the Arcana manifolds: after adding
on two trivial summands, $TX_5$, $TX_9$, and $TX_{11}$ become isomorphic to direct sums of line bundles, allowing an
easier calculation of their Stiefel-Whitney classes. We will use this technique several times in the next few
sections, so we take a moment to go over how this technique works.

Recall that the normal bundle to $S^n\hookrightarrow \R^{n+1}$ is trivial, with trivialization specified by the
outward unit normal vector field, and this specifies a stable splitting
\begin{equation}
\label{Sn_stable_splitting}
	TS^n\oplus\underline\R \overset\cong\longrightarrow T\underline\R^{n+1}|_{S^n} = \underline\R^{n+1}.
\end{equation}
But better than that, the normal bundle is \emph{equivariantly} trivial for the $\O (n+1)$-action on $S^n$,
lifting~\eqref{Sn_stable_splitting} to an isomorphism of equivariant vector bundles. If $G\subset\O (n+1)$ acts
freely on $S^n$, we can therefore descend~\eqref{Sn_stable_splitting} to the quotient $S^n/G$.
\begin{exm}[Real projective spaces]
\label{RP_tangent}
Let $\Z/2\Z$ act as $\{\pm 1\}$ on $\R^{n+1}$, which restricts to the antipodal action on $S^n$. The outward unit
normal vector field is invariant under this action, so~\eqref{Sn_stable_splitting} upgrades to an equivalence of
$\Z/2\Z$-equivariant vector bundles, where $\Z/2\Z$ acts trivially on the normal bundle and non-trivially on
$\underline{\R}^{n+1}$: specifically, on each $\underline\R$ summand, $\Z/2\Z$ acts by $-1$.

Therefore this stable splitting descends to the quotient: if $\sigma\to\RP^n$ denotes the quotient by $\Z/2\Z$ of
one of the $\underline\R$ summands on the right of~\eqref{Sn_stable_splitting}, then we have obtained an
isomorphism of (non-equivariant) vector bundles
\begin{equation}
	T\RP^n\oplus\underline\R\overset\cong\longrightarrow \sigma^{\oplus(n+1)}.
\end{equation}
This is a common way to calculate the Stiefel-Whitney classes of $\RP^n$: since $\sigma$ is a non-trivial line
bundle, $w(\sigma) = 1 + x$, where $x\in H^1(\RP^2;\Z/2\Z)$ is the generator, and then the Whitney sum formula tells
us
\begin{equation}
	w(T\RP^n) = w(T\RP^n\oplus\underline\R) = w(\sigma^{\oplus(n+1)}) = (1+x)^{n+1}.
\end{equation}
Pretty much by definition, $\sigma\to\RP^n$ is the line bundle associated to the principal $\Z/2\Z$-bundle
$S^n \to \RP^n$ and the $\Z/2\Z$-representation we put on $\R$ to make $\underline\R$ equivariant. This will not be
the last time something like that happens.
\end{exm}
\begin{exm}[Lens spaces]
\label{lens_rotate}
For lens spaces, we want $S^{2n-1}\subset\C^n$, so~\eqref{Sn_stable_splitting} looks instead like
$TS^{2n-1}\oplus\underline\R\cong\underline\C^n$. Now choose a primitive $k^{\mathrm{th}}$ root of unity $\zeta$
and let $\Z/k\Z$ act on $\C^n$ by having the generator act as multiplication by $\zeta$. As above, this leaves
invariant the outward unit normal vector field, so the stable splitting is $\Z/k\Z$-equivariant, and therefore it
descends to the quotient~\cite{Szc64, Kam66}:
\begin{equation}
\label{lens_stable_splitting}
	TL_k^{2n-1}\oplus\underline\R \overset\cong\longrightarrow \mathcal{L}^{\oplus n},
\end{equation}
where $\mathcal{L}\to L_k^{2n-1}$ is the quotient of $\underline\C\to S^{2n-1}$ by the $\Z/k\Z$-action, i.e.\ the associated
complex line bundle bundle to the principal $\Z/k\Z$-bundle $S^{2n-1}\to L_k^{2n-1}$ and the standard representation
of $\Z/k\Z$ on $\C$.
\end{exm}
%\textcolor{red}{TODO: move this example back to the section on $\ninedimgen$? Or move the definition of complex conjugation to
%here\dots}
%\begin{exm}[Complex conjugation on a lens space]
%\label{lens_reflect}
%Let $\Z/k\Z$ act on $S^n$ by $\zeta$ as above, producing a lens space quotient. Complex conjugation on $L_k^{2n-1}$
%is not free; nevertheless, the outer unit normal vector field on $S^{2n-1}$ is invariant under complex conjugation,
%so~\eqref{lens_stable_splitting} promotes to an isomorphism of $\Z/2\Z$-equivariant vector bundles.
%\end{exm}
Now back to the show.
\begin{lem}
\label{arcana_decomp}
There is a stable splitting of the tangent bundle $\pi\colon \mathbb P(\underline\R^p\oplus\sigma^{\oplus
q})\to\RP^n$ as
\begin{equation}
\label{general_arcana_splitting}
	T\mathbb P(\underline\R^p\oplus\sigma^{\oplus q})\oplus\underline\R^2 \overset\cong\longrightarrow
	\pi^\ast\sigma^{\oplus(n+1)}\oplus \tau^{\oplus p} \oplus (\pi^\ast\sigma\otimes\tau)^{\oplus q},
\end{equation}
where $\tau\to\mathbb P(\underline\R^p\oplus\sigma^{\oplus q})$ is the tautological line bundle in the fiber
direction (see Section \ref{ss:taut_bundle}).% \textcolor{red}{Maybe we should explain what the tautological bundle is.}.
\end{lem}
Our argument follows that of~\cite[Section 5.1]{SZ20}.
\begin{proof}
First consider the sphere bundle $S(\underline\R^p\oplus\sigma^{\oplus q})\to\RP^n$; choosing a connection for this
fiber bundle splits $TS(\underline\R^p\oplus\sigma^{\oplus q})\cong V\oplus\pi^\ast(T\RP^n)$, where $V\to
S(\underline\R^p\oplus\sigma^{\oplus q})$ is the vertical tangent bundle. The fiberwise normal bundle is trivial,
as usual thanks to the outward unit normal vector field in the fiber direction, so we have a splitting
\begin{equation}
	TS(\underline\R^p\oplus\sigma^{\oplus q})\oplus\underline\R\overset\cong\longrightarrow
	\underline\R^p\oplus \pi^\ast\sigma^{\oplus q}\oplus\pi^\ast(T\RP^n).
\end{equation}
Use \cref{RP_tangent} to simplify:
\begin{equation}
	TS(\underline\R^p\oplus\sigma^{\oplus q})\oplus\underline\R^2\overset\cong\longrightarrow
	\underline\R^p\oplus\pi^\ast\sigma^{\oplus q} \oplus \pi^\ast\sigma^{\oplus(n+1)}.
\end{equation}
Now add in the $\Z/2\Z$ antipodal action on the fibers. This is trivial in the horizontal direction, in particular on
$\pi^\ast\sigma^{\oplus(n+1)}$, and the fiberwise normal bundle is equivariantly trivial like in the examples
above. This $\Z/2\Z$ acts by $-1$ on each summand of $\underline\R^p$ and $\pi^\ast\sigma^{\oplus q}$. Therefore when
we take the quotient, each of these bundles is tensored with $\tau$, producing the decomposition
in~\eqref{general_arcana_splitting}.
\end{proof}
\begin{cor}
$X_5$ and $X_9$ have Spin-$D_8$ structures whose $D_4$-bundles are their universal covers $S^p\times S^q\to X_n$.
\end{cor}
\begin{proof}
For $m = 5,9$, let $a\in H^1(X_m;\Z/2\Z)$ be the class dual to $\alpha\in (\Z/2\Z) \oplus (\Z/2\Z) = \pi_1(X_m)$ under the
identification $H^1(X_m;\Z/2\Z)\cong \Hom \big(\pi_1(X_m), \Z/2\Z \big)$, and likewise let $b$ be the class dual to $\beta$.
The way the $(\Z/2\Z) \oplus (\Z/2\Z)$-actions were defined on $X_m$ implies that taking the quotient by $\beta$ produces the
sphere bundle $S(\underline\R^p\oplus\sigma^{\oplus q})\to\RP^n$, and then taking the quotient by $\alpha$ is
the fiberwise antipodal map; therefore $w(\pi^\ast\sigma) = 1+b$ and $w(\tau) = 1+a$.

\Cref{arcana_decomp} tells us
\begin{equation}
	TX_m = \pi^\ast\sigma^{\oplus(m-3)} \oplus \tau \oplus (\pi^\ast\sigma\otimes\tau)^{\oplus 3},
\end{equation}
so $w(X_m) = (1+b)^{m-3}(1+a)(1+a+b)^3$; expanding out, $w_1(X_m) = 0$ and $w_2(X_m) = ab + b^2$, which is what we
needed.
\end{proof}
\begin{lem}
\label{arcana_suffice}
The class $a^3b^2 + ab^4$ is nonzero in $H^5(X_5;\Z/2\Z)$, and the class $a^7b^2 + a^5b^4$ is nonzero in
$H^9(X_9;\Z/2\Z)$.
\end{lem}
Therefore $X_5$ and $X_9$ represent the Spin-$D_8$ bordism classes that we have been looking for.
\begin{proof}
The proof resembles that of~\cite[Lemma C.12]{Debray:2021vob}: we want to compute cup products, so we pass through
Poincaré duality and compute the intersections of submanifolds of $X_m$ Poincaré dual to the cohomology
classes we want to multiply together. There is a transversality condition which is satisfied for generic choices of
submanifold representatives. The Arcana manifolds themselves are a little bit difficult to visualize, so we work
upstairs in $S^p\times S^q$, intersecting submanifolds which are sent to themselves by the $(\Z/2\Z) \oplus (\Z/2\Z)$-action
and then taking the quotient.

For example, consider $N\coloneqq \set{x_1 = 0}\subset S^p\times S^q$. This is a codimension-$1$ submanifold and
$(\Z/2\Z) \oplus (\Z/2\Z)$ sends $N$ to itself. Though $N$ is null-homologous in $S^p\times S^q$, there is no null-homology
that is equivariant for the $(\Z/2\Z) \oplus (\Z/2\Z)$-action, and therefore the quotient of $N$ in $X_m$ is not
null-homologous, and represents some nonzero class in $H^1(X_m;\Z/2\Z)$. Specifically, because $\alpha$ acts
trivially on $N$, but $\beta$ acts non-trivially, $N$ is Poincaré dual to $b$. Similarly, $\set{x_i = 0}$ is
Poincaré dual to $b$ for each $i$; $\set{y_1 = 0}$ is Poincaré dual to $a$, and $\set{y_i = 0}$ for $i > 1$ is
Poincaré dual to $a+b$.

On $S^2\times S^3$ consider the set of solutions $W$ to the equations $x_1 = x_2 = y_1 = y_2 = y_3 = 0$, which is
the intersection of the five codimension-$1$ submanifolds $\set{x_1 = 0}$, $\set{x_2 = 0}$, \dots, $\set{y_3 = 0}$.
The images of these submanifolds in $X_5$ are transverse, so $W$ is a zero-dimensional submanifold, and the
Poincaré dual of $W$ is the cup product of the classes these submanifolds define. Specifically, $W$ is Poincaré
dual to $b^2a(a+b)^2 = a^3b^2 + ab^4$, so we can check whether $a^3b^2 + ab^4$ is zero in $H^5(X_5;\Z/2\Z)$ by
checking whether $W$ is null-homologous. Since $W$ is zero-dimensional, this amounts to counting whether $W$ has an
odd or even number of points.

Upstairs, the set of solutions to $x_1 = x_2 = y_1 = y_2 = y_3 = 0$ has four solutions in $S^2\times S^3$:
$\set{x_3 = \pm 1, y_4 = \pm 1, \text{ all other } x_i,y_i = 0}$. Thus when we take the quotient by
$(\Z/2\Z) \oplus (\Z/2\Z)$, these four solutions are identified, so there is one solution in $X_5$, and we conclude that
$a^3b^2 + ab^4\ne 0$ in $H^5(X_5;\Z/2\Z)$.

For $X_9$, the argument is the same; the equations we want are $x_1 = x_2 = y_1 = y_2 = y_4 = 0$ and $x_3y_3 =
x_4y_3 = x_5y_3 = x_6y_3 = 0$.  Again, on $S^6\times S^3$, the solutions are $\{x_7 = \pm 1, y_3 = \pm 1\}$, so
there is one solution on $X_9$, and this corresponds to the cohomology class $b^2(a)(a+b)^2a^4 = a^7b^2 + a^5b^4$.
\end{proof}
The computation that $a^7b^4 + a^3b^8\ne 0$ on $X_{11}$, so that $X_{11}$ and $\widetilde X_{11}$ represent the
$11$-dimensional Spin-$D_{16}$ bordism classes we want to detect, is given in~\cite[Proposition
C.7]{Debray:2021vob}.

\label{remaining_generators}
At this point we have only four classes left to detect: generators of the green upward-pointing triangle $\textcolor{Green}{\Z/4\Z}$ and
orange rightward-pointing triangle $\textcolor{RedOrange}{\Z/2\Z}$ summands in dimension $7$; the generator of the $\textcolor{BrickRed}{\Z}$ in
dimension $8$; and the generator of the yellow pentagon $\textcolor{Goldenrod!67!black}{\Z/2\Z}$ in dimension $9$. All
remaining classes are products of these four classes and non-bounding circles or tori.
\subsubsection{The generator $\halfBott$ in dimension 8}
\label{8d_red}
% NOTE: I (Arun) made a copy to include this in Duality_Bordism.pdf
% then I changed a few things (e.g. references and notation, nothing huge)
%
%
First we tackle the generator of the red $\textcolor{BrickRed}{\Z}$ in dimension $8$. As we established in
Section \ref{arcana} (see especially \cref{spinD8_figs}), this bordism class is in the image of the map $i_8\colon
\Omega_8^{\Spin\text{-}D_8}\to\Omega_8^{\Spin\text{-}D_{16}}$, so we can find a representative whose Spin-$D_{16}$
structure comes from a Spin-$D_8$ structure. Moreover, looking at the inclusion of
Spin-$\Z/8\Z$ bordism into Spin-$D_{16}$ bordism, we learn that two copies of our desired generator must be bordant
to the Bott manifold; thus it suffices to construct a Spin-$D_8$ manifold whose Pontrjagin numbers are half those
of the Bott manifold. For example, a connected Calabi-Yau fourfold $X_4$ has $\hat A(M) = 2$, which is twice $\hat
A(\mathrm{Bott})$, so its Pontrjagin numbers are twice those of a Bott manifold, so any quotient of $X_4$ by an
orientation-preserving action of a group of order $4$ has the correct Pontrjagin numbers to be our generator. We
will construct such a quotient $\halfBott$ of a specific choice of $X_4$, and show $\halfBott$ is Spin-$D_8$, so
that it is the generator we want.
\begin{defn}
Let $X_4$ be the complete intersection on $\CP^7$ given by choosing projective coordinates
\begin{equation}
	(x_1:x_2:x_3:x_4:x_5:x_6:y_1:y_2)\text{ for } \CP^{7},
\end{equation}
and consider the complete intersection specified by the three equations
\begin{equation}
	\sum_{i=1}^6 x_i^2= x_1^2+P(y)=x_2^3+P'(y)=0,
\label{w23}
\end{equation}
where $P$ and $P'$ are cubic polynomials. For generic choices of $P$ and $P'$, the intersection is
complete, and we obtain a smooth Calabi-Yau fourfold $X_4$.\footnote{The criterion for a complete intersection to
be Calabi-Yau is that the sums of the degrees of the defining polynomials add up to one plus the dimension of the
ambient projective space. See, e.g.,~\cite{GHL13}.}
\end{defn}
Consider the following action of $(\Z/2\Z) \oplus (\Z/2\Z)$ on $\CP^7$:
\begin{equation}
\begin{aligned}
	\alpha(x_1:\dotsm :x_6: y_1:y_2) &= (-x_1:\dotsm: -x_6: y_1: y_2)\\
	\beta(x_1:\dotsm :x_6: y_1:y_2) &= (x_1^*:\dotsm: x_6^*: y_1: y_2).
\end{aligned}
\label{z2z2}
\end{equation}
This action leaves the equations~\eqref{w23} invariant, and so defines an action of $(\Z/2\Z) \oplus (\Z/2\Z)$ on
$X_4$. One can check that both $\alpha$ and $\beta$ are orientation-preserving and, restricted to $X_4$, lack fixed
points: in a fixed point for $\alpha$, the last two equations would not have a solution for $y$, while in a fixed
point for $\beta$, all the $x_i$ are real, and then the first equation implies they all vanish. Thus the quotient
%We can see that, taking into account \eqref{w23}, neither of these isometries have fixed points on $X_4$. The first
%one does not because then the last two equations do not have a solution for $y$, while in a fixed point for the
%second involution, all the $x_i$ are real, and then the first equation implies they all vanish. Also, both actions
%are orientation-preserving \textcolor{red}{(TODO check)}, so the manifold
$\halfBott \coloneqq X_4/ \big((\Z/2\Z) \oplus (\Z/2\Z) \big)$ is an oriented smooth manifold. $\halfBott$ cannot
be Spin: if it were, its $\hat A$-genus would be $1/2$, but the $\hat A$-genus is an integer on Spin manifolds. We
call this manifold the ``half Bott manifold'' because $\halfBott \sqcup \halfBott$ is bordant to a Bott manifold.
\begin{prop}
$\halfBott$ admits a Spin-$D_8$ structure.
\end{prop}
\begin{proof}
%It is not Spin, since the corresponding index density evaluates to a half-integer quantity. It is however a $\text{Spin}^{D_8}$ manifold. To see this,
Consider the standard Spin structure on the ambient $\CP^{7}$. Equivariant spinors on $\CP^{7}$ invariant under the
action \eqref{z2z2} will descend to spinors on $\halfBott$. To construct spinors on $\CP^{n}$ explicitly, we can
start with Dirac spinors $\Psi$ on $\C^{n+1}=\R^{2n+2}$. For $z\in\C^{n+1}$, the transformation $z\rightarrow
\lambda z$ that we quotient by to get complex projective space induces an infinitesimal action on spinor fields
containing a rotation and a scaling transformation~\cite{Osb13}:
%\url{https://www.damtp.cam.ac.uk/user/ho/CFTNotes.pdf} \textcolor{red}{Transform to citation},
\begin{equation}
\delta_\lambda \Psi=(M - \text{id})\Psi,\quad M=\sum_{i=1}^n \Gamma_{i}\Gamma_{i+1} \,,
\label{45rt}
\end{equation}
where the first part of the above transformation implements the infinitesimal rotation involved in a complex
rescaling, and the part proportional to the identity matrix implements the rescaling. Spinors on $\CP^{n}$
correspond precisely to those spinors on $\C^{n+1}$ which satisfy $\delta_\lambda \Psi=0$. As is clear from
\eqref{45rt}, this will only be possible if the matrix $M$ has a $+1$ eigenvalue. According to the
analysis around equation (D.6) of~\cite{Hsieh:2020jpj}, this is only possible if $n+1$ is even. This is a very elementary way to see that we only have a Spin structure on $\CP^n$ when $n$ is odd.

Taking now $n=7$, spinors on $\halfBott$ correspond to spinors on $\C^{8}$ satisfying $\delta_\lambda \Psi=0$ and which are also invariant under the Spin lifts of \eqref{z2z2},
\begin{equation}
	\Psi= \Big(\prod_{i=1}^{12} \Gamma_i \Big)\Psi,\quad  \Psi=\Big( \prod_{i=0}^{5}\Gamma_{2i+1}\Big)\Psi \,.
\end{equation}
For this to work, the two matrices
\begin{equation}
M_1=\prod_{i=1}^{12} \Gamma_i \,, \quad M_2= \prod_{i=0}^{5}\Gamma_{2i+1}
\end{equation}
should generate the group $(\Z/2\Z) \oplus (\Z/2\Z)$ when restricted to spinors with $\delta_\lambda \Psi=0$. But one can explicitly compute that instead the commutation relations are
\begin{equation}
M_1^2=\text{id} \,,\quad M_2^2=-\text{id} \,,\quad M_1 M_2 M_1= -M_2 \,,
\end{equation}
i.e.\ they generate a $D_8$ group. This is why there is no Spin structure on $\halfBott$. However, it is also clear
that, by tensoring the above matrices with an appropriate representation of $D_8$, two wrongs can make a right, and
we can obtain $D_4$ representations. So $\halfBott$ admits a Spin-$D_8$ structure, as claimed.
\end{proof}

\subsubsection{The generator $\ninedimgen$ in dimension 9}
\label{9d_yellow}
The next generator we address corresponds to the yellow pentagon $\textcolor{Goldenrod!67!black}{\Z/2\Z}$ summand
in $E_\infty^{2,11}$. The manifold we find resembles the generators we give for the remaining classes in dimension
$7$, but this one is simpler, so we start in dimension $9$ as a warm-up.

The $\textcolor{Goldenrod!67!black}{\Z/2\Z}$ class under investigation is not in the images of the maps from
Spin-$\Z/8\Z$ or Spin-$D_8$ bordism, so we need its representative $\ninedimgen$ to have a fundamental group at least
as complicated as $D_8$. We will define $\ninedimgen$ momentarily; the key results in this section are
\cref{9d_is_spD16}, which shows that $\ninedimgen$ is Spin-$D_{16}$, and \cref{W9_is_good}, which shows that it is
linearly independent from the other generators we have found.
%The goal of this (sub)section is to discuss a candidate for the last 9-dimensional generator. It is known to be
%spin-$D_{16}$, but I don't know whether it's the missing generator --- this probably amounts to computing
%$\eta$-invariants of it and the other generators we have in dimension $9$.
\begin{defn}
Let $\zeta$ be a primitive $k^{\mathrm{th}}$ root of unity and $L_k^{2n-1}$ denote the lens space which is the
quotient of $S^{2n-1}\subset\C^n$ by the $\Z/k\Z$-action on $\C^n$ which is multiplication by $\zeta$, which
preserves the unit sphere. Complex conjugation exchanges $\zeta$ with another primitive $k^{\mathrm{th}}$ root of
unity, and therefore the image of a $\Z/k\Z$-orbit of $S^{2n-1}$ under complex conjugation is another $\Z/k\Z$-orbit.
Therefore this involution descends to an involution on $L_k^{2n-1}$, which we also call \term{complex conjugation}.
\end{defn}
Since complex conjugation is orientation-reversing on $\C$, it is orientation-preserving on $\C^{2m}$ and
orientation-reversing on $\C^{2m+1}$. The $\Z/k\Z$-action on $S^{2n-1}$ is orientation-preserving, so we can
descend to $L_k^{2n-1}$ and infer that complex conjugation on $L_k^{2n-1}$ is orientation-preserving when $n$ is
even and is orientation-reversing when $n$ is odd.
\begin{defn}
Let $\Z/2\Z$ act on $L_4^5\times S^4$ by complex conjugation on the lens space and the antipodal map on $S^4$. This
is a free action; we define $\ninedimgen$ to be the quotient.
\end{defn}
As far as we can tell, this manifold was first studied by Kamata-Minami~\cite[Section 2]{KM73}, where it is called $D(5,
4)$.

The first thing to check is whether $\ninedimgen$ is orientable, or
equivalently, whether the involution we used to define it is orientation-preserving. This is true: the antipodal
map is orientation-reversing on $S^{2n}$, and we just saw that complex conjugation is orientation-reversing on
$L_k^5$. Therefore the combination of these two involutions is orientation-preserving. In fact, because $L_4^5$ and
$S^4$ come with canonical orientations, $\ninedimgen$ is also canonically oriented.

The complex conjugation involution passes to inversion on $\pi_1(L_4^5) = \Z/4\Z$, so $\pi_1(\ninedimgen) = D_8$;
concretely, $\ninedimgen$ is also the quotient of $S^5\times S^4$ by the $D_8$-action generated by two
diffeomorphisms $r$ and $s$: $r = (m, \mathrm{id})$ is multiplication by $i$ on $S^5$ and the identity on $S^4$ and
$s$ is reflection on $S^5$ and the antipodal map on $S^4$.
\begin{prop}
\label{9d_is_spD16}
$\ninedimgen$ with its $D_8$-bundle $S^5\times S^4\to \ninedimgen$ has a Spin-$D_{16}$ structure.
\end{prop}
To do this, we will stably split $T\ninedimgen$ in order to compute its low-degree Stiefel-Whitney classes,
following Fujino~\cite[Proposition 3.2]{Fuj75}.
\begin{proof}
Using \cref{RP_tangent,lens_rotate}, there is an isomorphism of vector bundles
\begin{equation}
\label{lens_sphere_stable_split}
	T(L_4^5\times S^4)\oplus\underline\R^2\overset\cong\longrightarrow \mathcal{L}^{\oplus 3}\oplus\underline\R^5.
\end{equation}
Let $\Z/2\Z$ act on $L_4^5\times S^4$ by complex conjugation on $L_4^5$ and the antipodal map on $S^4$, so that the
quotient is $\ninedimgen$. Then, \cref{RP_tangent} together with the analogous fact about complex conjugation tells
us that~\eqref{lens_sphere_stable_split} upgrades to an isomorphism of $\Z/2\Z$-equivariant vector bundles, where:
\begin{itemize}
	\item $\Z/2\Z$ acts trivially on the $\underline\R^2$ on the left-hand side of~\eqref{lens_sphere_stable_split}.
	\item $\Z/2\Z$ acts by complex conjugation on each $\mathcal{L}$ summand on the right-hand side.
	\item $\Z/2\Z$ acts by $-1$ on each factor of $\underline\R$ on the right-hand side.
\end{itemize}
Therefore when we take the quotient, we reprove a theorem of Fujino~\cite[Proposition 3.2]{Fuj75}: that there is an
isomorphism of vector bundles
\begin{equation}
\label{9d_stable_splitting}
	T\ninedimgen \oplus\underline\R^2 \overset\cong\longrightarrow V^{\oplus 3}\oplus \sigma^{\oplus 5},
\end{equation}
where $V$ and $\sigma$ are as follows: if $P\to \ninedimgen$ denotes the quotient $S^5\times S^4\to \ninedimgen$,
which is a principal $D_8$-bundle, then $\sigma$ is associated to $P$ and the sign representation $D_8\to \O(1)$
sending rotations to $1$ and reflections to $-1$, and $V$ is associated to the standard representation $D_8\to
\O(2)$ as rotations and reflections on $\R^2$.\footnote{What we call $V$, Fujino \cite[Proposition 3.1]{Fuj75}
calls $\eta_1$.} Thus $\sigma = \mathrm{Det}(V)$, so $w_1(V) = w_1(\sigma)$.

Throwing the Whitney sum formula at~\eqref{9d_stable_splitting}, and using the fact that $w_1(V) = w_1(\sigma)$ and
$w_2(\sigma) = 0$, because $\sigma$ is a line bundle, we learn that $w_1(T\ninedimgen) = 0$ and that
$w_2(T\ninedimgen) = w_2(V)$. Since $V$ is associated to a principal $D_8$-bundle via the standard representation
$D_8 \to \O(2)$, this is precisely the condition for $\ninedimgen$ to have a Spin-$D_{16}$ structure.
\end{proof}
We have found all but two of the generators of $\Omega_9^{\Spin\text{-}D_{16}}(\pt)$. The remaining two each
generate $\Z/2\Z$ summands. We will show that $\ninedimgen$ generates the yellow
$\textcolor{Goldenrod!67!black}{\Z/2\Z}$ summand, but at present we also do not know a generator for the green
triangle $\textcolor{Green}{\Z/2\Z}$ summand. However, looking at the $E_\infty$-page (\cref{D8_Einf}, bottom), we
know that there is some Spin-$D_{16}$ $7$-manifold $\halfQseven$ such that this green $\textcolor{Green}{\Z/2\Z}$
is generated by $\halfQseven\times S^1_p \times S^1_p$. For now, choose such a manifold $\halfQseven$; we will not
need to know anything more about it right now. We will choose an explicit example of $\halfQseven$ later in Section
\ref{7d_green}.
\begin{prop}
\label{W9_is_good}
There is a bordism invariant $\mu\colon\Omega_9^{\Spin\text{-}D_{16}} (\pt) \to \Z/2\Z$ which is $1$ on $\ninedimgen$ and
vanishes on the other generators (which are $L_4^9$, $\halfQseven\times
S^1_p \times S^1_p$, $\mathrm{Bott}\times S^1_R$, $\HP^2\times S^1_R$, $\HP^2\times\widetilde S^1_R$,
$X_{9}$, and $\widetilde X_{9}$).
\end{prop}
\begin{proof}
Let $M$ be a Spin-$D_{16}$ manifold with associated $D_8$-bundle $P\to M$. The quotient of $P$ by $D_4\subset D_8$
is a Spin-$D_8$ manifold $\widehat M$ which is a double cover of $M$: it is Spin-$D_8$ because the $D_8$-action on
$P$ lifts to a $D_{16}$-action on its spinor bundle $\mathbb S$, and pulling back along $D_4\inj D_8$, we get
the subgroup $D_8\subset D_{16}$ acting on $\mathbb S$. The double cover $\widehat M\to M$ corresponds to the class
$y(P)\in H^1(M;\Z/2\Z)$. The assignment $M\mapsto\widehat M$ is compatible with taking boundaries, and therefore
defines a map
\begin{equation}
	D\colon \Omega_k^{\Spin\text{-}D_{16}} (\pt) \longrightarrow \Omega_k^{\Spin\text{-}D_8} (\pt) \,.
\end{equation}
%This is analogous to the invariant we defined in (\TODO).
We use notation for Spin-$D_8$ characteristic classes from Section \ref{arcana}: $a$ and $b$ are the two
generators of $H^*(BD_4;\Z/2\Z)$, both in degree $1$. The class $a$ is detected by a reflection and $b$ is detected
by a generating rotation. Our invariant $\mu$ is defined to be
\begin{equation}
	\mu\colon\Omega_9^{\Spin\text{-}D_{16}} (\pt) \overset D\longrightarrow \Omega_9^{\Spin\text{-}D_8} (\pt) \xrightarrow{\int
	(a^5b^4 + a^3b^6)} \Z/2\Z.
\end{equation}
Recall the maps $i_8,\widetilde\imath_8\colon D_8\to D_{16}$ that we used to define two Spin-$D_{16}$ structures on
a Spin-$D_8$ manifold $M$. The composition $D\circ i_8 = 0$, because $i_8^*(y) = 0$, and therefore $\mu\circ i_8 =
0$. For $\widetilde\imath_8$, $\widetilde\imath_8^*(y) = \widetilde\imath_8^*(x) = a$, so manifolds in the image
of the map $D\circ\widetilde\imath_8$ have $a(P) = 0$, so $\mu\circ\widetilde\imath_8 = 0$ too. Therefore $\mu(N) =
0$ when $N$ is one of the following generators: $\mathrm{Bott}\times S^1_R$, $\HP^2\times S^1_R$,
$\HP^2\times\widetilde S^1_R$, $X_{9}$, and $\widetilde X_{9}$.

Next $L_9^4$. $D(L_9^4) = \RP^9$, which is Spin-$D_8$ but not Spin. Therefore $ab + b^2\ne 0$ for $\RP^9$, so $a\ne
b$; since $H^1(\RP^9;\Z/2\Z)$ has only one nonzero element, one of $a$ and $b$ vanishes, and therefore $a^5b^4 +
a^6b^3 = 0$ and $\mu(L_4^9) = 0$.

Finally, $\halfQseven\times S^1_p \times S^1_p$. The associated $D_4$-bundle is trivial in the torus
directions, so its characteristic classes pull back across the projection $\mathrm{pr}_1\colon \halfQseven\times
S^1_p \times S^1_p \to \halfQseven$, and the same is true of $D(\halfQseven\times S_p^1\times S_p^1) \cong
D(\halfQseven)\times S_p^1\times S_p^1$. The degree-$9$ class $a^5b^4 + a^3b^6$ vanishes on the $7$-manifold
$D(\halfQseven)$, so $\mu(\halfQseven\times S^1_p \times S^1_p) = 0$.

That accounts for all of the previously discovered generators, so it suffices to show that $\mu(\ninedimgen) = 1$.
First, we should characterize $D(\ninedimgen)$. This is a projective bundle over a projective space, like the
Arcana; as in Section \ref{arcana}, it is an $\RP^5$-bundle over $\RP^4$ which we specify as a $(\Z/2\Z) \oplus
(\Z/2\Z)$ quotient of $S^5\times S^4$ inside $\R^6\times\R^5$ with coordinates $(x_1,\dotsc,x_6, y_1,\dotsc,y_5)$.
Specifically, letting $\alpha$ and $\beta$ be the generators of $(\Z/2\Z) \oplus (\Z/2\Z)$, act as follows:
\begin{equation}
\begin{aligned}
	\alpha(x_1,x_2,x_3,x_4,x_5,x_6,y_1,\dotsc,y_5) &= (-x_1,-x_2,-x_3,-x_4,-x_5, -x_6, y_1,\dotsc,y_5) \,, \\
	\beta(x_1,x_2,x_3,x_4,x_5,x_6,y_1,\dotsc,y_5)  &= (-x_1,x_2,-x_3,x_4,-x_5, x_6, -y_1,\dotsc,-y_5) \,.
\end{aligned}
\end{equation}
Why is this? $D(\ninedimgen)$ is the quotient of $\RP^5\times S^4$ by an involution; $\beta$ is this involution,
and $\alpha$ is the antipodal action on $S^5$ to get us to $\RP^5\times S^4$.

As in \cref{arcana_suffice}, one can check that $w\big(D(\ninedimgen)\big) = ab + b^2$ and that the set of solutions to
the equations $x_2 = x_3 = x_4 = x_5 = x_6 = y_1 = y_2 = y_3 = y_4 = 0$ has four solutions $\set{x_1 = \pm 1, y_5 =
\pm 1}$ on $S^5\times S^4$, hence one solution on $D(\ninedimgen)$, so the corresponding cohomology class is
nonzero. Tracking the $(\Z/2\Z) \oplus (\Z/2\Z)$-actions on these coordinates like in \cref{arcana_suffice}, the class we get
is $(a+b)^2a^3b^4 = a^5b^4 + a^3b^6$, so $\mu(\ninedimgen) = 1$.
\end{proof}
\begin{rem}
The manifolds $(\CP^m\times S^n)/ (\Z/2\Z)$, where $\Z/2\Z$ acts by complex conjugation on $\CP^m$ and by the
antipodal map on $S^n$, are called \term{Dold manifolds}. Dold studied them in~\cite{Dol56}, showing that they
generate the unoriented bordism ring. These manifolds are quotients of $S^{2m+1}\times S^n$ by a free $\mathrm O
(2)$-action, and in this regard $\ninedimgen$ is an analogue of the Dold manifolds with $\O (2)$ replaced with
$D_8$. From this perspective, it is not so surprising that a manifold like $\ninedimgen$ shows up in our list of
generators.  Kamata-Minami~\cite{KM73} and Kamata~\cite{Kam74} use a family of manifolds including $\ninedimgen$ to
study $\Omega_\ast^{\mathrm U}(BD_{2k})$, providing an additional hint that $\ninedimgen$, or something like it,
would be a nonzero element of $\Omega_9^{\Spin\text{-}D_{16}} (pt)$.
\end{rem}
%Now we need to check that $\ninedimgen$ is linearly independent from the other generators we have found in
%dimension $9$.
%
%Ideas:
%\begin{itemize}
%	\item The Arcana are detected by $w^2x^5$ and $w^2y^5$, so maybe we can show these vanish on $\ninedimgen$.
%	\item $L_4^9$ is detected by $w^4y$, and maybe we can show this vanishes on $\ninedimgen$.
%	\item There are two generators of the form $\HP^2\times S^1$, detected by $w_4^2 x$ and $w_4^2y$, and maybe we
%	can show these vanish on $\ninedimgen$.
%\end{itemize}
%If some of these classes don't vanish, we'll have to work a little harder (compute more $\eta$-invariants,
%probably). The remaining generators are trickier.
%\begin{itemize}
%	\item The red generator $\mathrm{Bott}\times S^1$ (this is the reflection $S^1$) is detected by an $\eta$
%	invariant for a representation $D_8\to\Z/2$ (probably the sign representation). There's the question of how to
%	define this on manifolds which aren't spinc.
%	\item Finally, the green generator: $T^2$ times a $7$-manifold.
%\end{itemize}

\subsubsection{The generator $\halfQseven$ in dimension 7}
\label{7d_green}
Next, the $\textcolor{Green}{\Z/4\Z}$ in dimension $7$. Recall the definition of $Q_4^7$ from Appendix \ref{ss:eta_Q}: the
quotient of the unit sphere bundle in $V\coloneqq \mathcal{O}(2) \oplus \underline{\C}^3\to\CP^1$ by the
$\Z/4\Z$-action given by multiplication by $i$. This is a fiber bundle over $S^2 = \CP^1$ with fiber $L_4^5$.
\begin{defn}
Let $\tau$ denote the  the involution on $Q_4^7$ which is complex conjugation on the lens space fibers and the
antipodal map on $S^2$. This is a free involution; let $\halfQseven$ be the quotient.
\end{defn}
$\halfQseven$ is a fiber bundle over $\RP^2$ with fiber $L_4^5$. Reflection on $L_4^5$ is orientation-reversing,
and the antipodal map on $S^2$ is orientation-reversing, so $\halfQseven$ is orientable.

Complex conjugation acts on $\pi_1(Q_4^7)\cong\Z/4\Z$ by $-1$, so $\pi_1(\halfQseven)\cong D_8$. The universal cover
is $S(V)\to \halfQseven$, and this is a principal $D_8$-bundle.
\begin{prop}
\label{halfQseven_spD16}
$\halfQseven$, with principal $D_8$-bundle $S(V)\to\halfQseven$, has a Spin-$D_{16}$ structure.
\end{prop}
\begin{proof}
First, like in Section \ref{9d_yellow}, we split the stable tangent bundle of $\halfQseven$. The
splitting~\eqref{Sn_stable_splitting} behaves well in families; that is, if $\pi\colon E\to M$ is a Euclidean
vector bundle over a manifold $M$ and $i\colon S(E)\inj E$ denotes the unit sphere bundle, then the normal
bundle to the inclusion $i$ is trivializable using the outward unit normal vector field. That is,
\begin{equation}
\label{sphere_bundle_splitting}
	TS(E)\oplus\underline\R \overset\cong\longrightarrow TE|_{S(E)} \overset\cong\longrightarrow (\pi^*TM \oplus
	\pi^*E)|_{S(E)} \,,
\end{equation}
the latter isomorphism given by choosing a connection on $V$. We care about the vector bundle $V\coloneqq
\underline\C^2\oplus\mathcal O(2)\to S^2$; from~\eqref{sphere_bundle_splitting} we learn
\begin{equation}
	TS(V)\oplus\underline\R \overset\cong\longrightarrow \pi^*(TS^2) \oplus \underline\C^2\oplus \pi^*\mathcal
	O(2) \,.
\end{equation}
Add on the normal bundle to $S^2\subset\R^3$ to stably trivialize $TS^2$:
\begin{equation}
\label{T_S_7}
	TS(V)\oplus\underline\R^2 \overset\cong\longrightarrow \underline\R^3 \oplus \underline\C^2\oplus \pi^*\mathcal
	O(2) \,.
\end{equation}
Now we descend to $Q_4^7$, which means adding $\Z/4\Z$-actions to~\eqref{T_S_7}. $\Z/4\Z$ acts on $TS(V)$ by
differentiating the $\Z/4\Z$-action on $Q_4^7$. The $\Z/4\Z$-action on the normal bundle to $S(V)\subset V$ is
equivariantly trivial, just as in~\eqref{lens_rotate}, and the action on the normal bundle to $S^2\subset\R^3$ is
trivial, because $\Z/4\Z$ acts trivially on the base $S^2$. This takes care of the left-hand side of~\eqref{T_S_7};
for the right-hand side, $\Z/4\Z$ acts trivially on $\underline\R^3$, because this summand came from the base and
$\Z/4\Z$ acts fiberwise; and on each summand of $\underline\C^2\oplus \pi^*\mathcal O(2)$, $\Z/4\Z$ acts fiberwise by
the standard one-dimensional complex representation. Therefore when we take the quotient, we obtain a stable splitting
\begin{equation}
\label{T_Q47}
	TQ_4^7\oplus\underline\R^2 \overset\cong\longrightarrow \underline\R^3 \oplus \mathcal L_1^{\oplus 2} \oplus
	\mathcal L_2,
\end{equation}
where $\mathcal L_1,\mathcal L_2\to Q_4^7$ are complex line bundles. Explicitly, $\mathcal L_1$ is associated to
the principal $\Z/4\Z$-bundle $S(V)\to Q_4^7$ via the rotation representation $\rho\colon \Z/4\Z\to\mathrm U(1)$
and $\mathcal L_2\cong \mathcal L_1\otimes p^*\mathcal O(2)$, where $p\colon Q_4^7\to S^2$ is the bundle map. Using
these descriptions and the Serre spectral sequence for the fiber bundle $p\colon Q_4^7\to S^2$ one can compute the
Chern classes of $\mathcal L_1$ and $\mathcal L_2$, similarly to the calculation in the proof of
\cref{serre_for_green_7d} below.

Next we want to descend to $\halfQseven$, so we incorporate $\Z/2\Z$-actions into~\eqref{T_Q47}. Again, the normal
bundles on the left-hand side are equivariantly trivial. On the right-hand side, $\Z/2\Z$ acts on $\underline\R^3$
by the antipodal map, which is $-1$. On $\mathcal L_1^{\oplus 2}$ and $\mathcal L_2$, $\Z/2\Z$ acts by complex
conjugation. Thus we have a stable splitting analogous to~\eqref{9d_stable_splitting}:
\begin{equation}
\label{half_Qseven_stable_splitting}
	T\halfQseven \oplus\underline\R^2 \overset\cong\longrightarrow p^*\sigma^{\oplus 3} \oplus E_1^{\oplus
	2}\oplus E_2,
\end{equation}
where $p\colon\halfQseven\to\RP^2$ is induced from $\pi\colon Q_4^7\to S^2$, $E_1,E_2\to \halfQseven$ are real,
rank-two vector bundles, and $\sigma\to \RP^2$ is the tautological line bundle.
%\textcolor{red}{above the tautological bundle was denotes as $\tau$}.
By keeping track of which elements
of $D_8$ act orientation-reversingly on each summand, we learn $w_1(p^*\sigma) = w_1(E_1) = w_1(E_2)$. Call this
element $x$. Using this, and the fact that $w_2(\sigma) = 0$, apply the Whitney sum formula
to~\eqref{half_Qseven_stable_splitting}. We learn $w_1(T\halfQseven) = 0$ (indeed, we already checked that
$\halfQseven$ is orientable) and that $w_2(T\halfQseven) = w_2(E_2) + x^2$.

To proceed we need to better understand the mod $2$ cohomology ring of $\halfQseven$.
\begin{lem}
\label{serre_for_green_7d}
The Serre spectral sequence for the fiber bundle $L_4^5\to \halfQseven\to\RP^2$ collapses, providing an isomorphism
\begin{equation}
	H^*(\halfQseven;\Z/2\Z) \cong \Z/2\Z[x, y, w]/(x^3, xy + y^2, w^3) \,,
\end{equation}
where $\abs x = \abs y = 1$ and $\abs w = 2$.
\end{lem}
The class $x\in H^1(\halfQseven;\Z/2\Z)$ is the same $x$ we met above.
\begin{proof}
Draw the Serre spectral sequence in \cref{half_Qseven_Serre} and see what happens! For degree reasons,
the only differentials that can be nonzero are of the form $d_2\colon E_2^{0,q}\to E_2^{2,q-1}$. Since $E_2^{0,*} =
H^*(L_4^5;\Z/2\Z)$ is generated as a ring by $y\in H^1(L_4^5;\Z/2\Z)$ and $w\in H^2(L_4^5;\Z/2\Z)$, the Leibniz rule
tells us that all $d_2$s out of the $p = 0$ line are determined by their values on $y$ and $w$. We will show
$d_2(y) = 0$ and $d_2(w) = 0$.
\begin{figure}[h!]
\centering
\begin{sseqdata}[name=halfQsevenSerre, cohomological Serre grading, xrange={0}{2}, yrange={0}{5},
classes={draw=none}, >=stealth, xscale=2,
x label = {$\displaystyle{q\uparrow \atop p\rightarrow}$},
x label style = {font = \small, xshift = -17ex, yshift=5.4ex}
]
	\class["1"](0, 0)
	\class["x"](1, 0)
	\class["x^2"](2, 0)

	\class["y"](0, 1)
	\class["xy"](1, 1)
	\class["x^2y"](2, 1)

	\class["w"](0, 2)
	\class["wx"](1, 2)
	\class["wx^2"](2, 2)

	\class["wy"](0, 3)
	\class["wxy"](1, 3)
	\class["wx^2y"](2, 3)

	\class["w^2"](0, 4)
	\class["w^2x"](1, 4)
	\class["w^2x^2"](2, 4)

	\class["w^2y"](0, 5)
	\class["w^2xy"](1, 5)
	\class["w^2x^2y"](2, 5)

\d2(0, 1)
\d2(0, 2)
\end{sseqdata}
\printpage[name=halfQsevenSerre, page=2]
\caption{The Serre spectral sequence computing $H^*(\halfQseven; \Z/2\Z)$. In \cref{serre_for_green_7d}, we show the
two pictured differentials vanish, and that this implies the entire spectral sequence collapses.}
\label{half_Qseven_Serre}
\end{figure}

Suppose $d_2(y)\ne 0$. Then on the $E_\infty$-page, there is a single $\Z/2\Z$ summand in total degree $1$, forcing
$H^1(\halfQseven;\Z/2\Z)$ to be $\Z/2\Z$. However, we already know
\begin{equation}
H^1(\halfQseven;\Z/2\Z) = \Hom(\pi_1(\halfQseven), \Z/2\Z) = \Hom(D_8, \Z/2\Z) \cong (\Z/2\Z) \oplus (\Z/2\Z) \,,
\end{equation}
so $d_2(y) = 0$.

Now suppose $d_2(w)\ne 0$. This, together with $d_2(y) = 0$, implies $d_2(wy) = wyx^2$, $d_2(w^2) = 0$, and
$d_2(w^2y) = 0$. But then on the $E_\infty$-page, there are two $\Z/2\Z$ summands in total degree $2$ and three
$\Z/2\Z$ summands in total degree $5$, which violates Poincaré duality for $\halfQseven$.

The multiplicative structure is clear except for the relation $xy + y^2 = 0$. This follows from the fact that $x$
and $y$ pull back from $BD_8$, and $xy + y^2 = 0$ in $H^*(BD_8;\Z/2\Z)$.
\end{proof}
Recall that $w_2(T\halfQseven) = w_2(E_2) + x^2$, and that we want to show this equals $w_2$ of the associated
bundle to the principal $D_8$-bundle $S(V)\to\halfQseven$ and the standard two-dimensional real representation of
$D_8$. Call this bundle $E$.

\Cref{serre_for_green_7d} tells us that $w_2(E_2) = a_1 x^2 + a_2xy + a_3z$ and $w_2(E) = b_1 x^2 + b_2xy + b_3z$.
\begin{itemize}
	\item To compute $a_1$ and $b_1$, let $\widetilde W$ be the quotient of $S(V)$ by $\Z/2\Z$ acting by the
	antipodal map on the base and complex conjugation on the fibers. Then $\pi\colon \widetilde W\to W$ is a
	principal $\Z/4\Z$-bundle, and $x^2$ pulls back to the nonzero class in $H^2(\widetilde W;\Z/2\Z)\cong\Z/2\Z$. One
	can check that $\pi^*E_2$ is Spin, but $\pi^*E$ is not, so $a_1 = 0$ and $b_1 = 1$.
	\item To compute $a_2$ and $b_2$, first pull back $\halfQseven\to\RP^2$ along $\RP^1\inj\RP^2$; then pull back
	along the inclusion $L_4^1\inj L_4^5$, which is compatible with the fiber bundle structure. The resulting space
	is the Klein bottle $KB$, as a $L_4^1 = S^1$-bundle over $\RP^1 = S^1$. For both $E_2$ and $E$,
	$\pi_1(KB)\cong\Z$ acts in a way that may be lifted to the spinor bundle, so $a_2 = b_2 = 0$.
	\item Finally, $a_3 = b_3 = 1$ by restricting $E_2$ and $E$ to the fiber; as bundles over lens spaces, neither
	is Spin.
\end{itemize}
Thus $w_2(E_2) + x^2 = w_2(E)$.
\end{proof}
Recall that $\Omega_7^{\Spin\text{-}D_{16}} (\pt) \cong\textcolor{Green}{(\Z/4\Z)}\oplus\textcolor{RedOrange}{(\Z/2\Z)}$. We are
trying to generate the green $\textcolor{Green}{\Z/4\Z}$ summand (corresponding to the upward-pointing triangle pieces in the
$E_\infty$-page).

If $M$ is a Spin-$D_{16}$ manifold with associated principal $D_8$-bundle $P\to M$, let $\widehat M\coloneqq
P/\Z/4\Z$; then, $\widehat M$ is a Spin-$\Z/8\Z$ manifold which is a double cover of $M$: the Spin-$\Z/8\Z$ structure
comes from $P\to\widehat M$, which is a $\Z/4\Z$ bundle, and the action on spinors lifts to the pullback of
$D_{16}\to D_8$ along $\Z/4\Z\inj D_8$, giving us $\Z/8\Z$ as claimed. The double cover $\widehat M\to M$ is classified
by $x(P)\in H^1(M;\Z/2\Z)$. Like the Spin-$D_8$ double cover we used in Section \ref{9d_yellow}, this assignment defines a
bordism invariant
\begin{equation}
	\varpi\colon \Omega_k^{\Spin\text{-}D_{16}} (\pt) \longrightarrow\Omega_k^{\Spin\text{-}\Z/8\Z} (\pt) \,.
\end{equation}
\begin{prop}
\label{half_Qseven_order_4}
$\varpi(\halfQseven)$ is a generator of $\Omega_7^{\Spin\text{-}\Z/8\Z} (\pt) \cong\Z/4\Z$.
\end{prop}
This implies that $\halfQseven$ represents either $(1, 0)$ or $(1, 1)$ in
$\Omega_7^{\Spin\text{-}D_{16}} (\pt) \cong\textcolor{Green}{(\Z/4\Z)} \oplus \textcolor{RedOrange}{(\Z/2\Z)}$, which is what we
wanted.
%\begin{prop}
%\label{halfQ_detects}
%$\halfQseven$ generates the $\textcolor{Green}{\Z/4}$ summand in $\Omega_7^{\Spin\text{-}D_{16}}$. In particular,
%$\halfQseven$ has a spin-$D_{16}$ structure.
%\end{prop}
%Most of the effort in this proof goes into proving $\halfQseven$ admits a spin-$D_{16}$ structure, and we factor
%that out as \cref{halfQ_D16}. Once we have that, the bordism invariant detecting it is not so bad.
%\begin{lem}
%\label{D16toZ8_cover}
%Let $M$ be a manifold with a spin-$D_{16}$ structure and $M'\to M$ be a double cover with $w_1(M') = x\in
%H^1(M;\Z/2)$. Then $M'$ has a canonical spin-$\Z/8$ structure, and $M\mapsto M'$ passes to a homomorphism
%$\Omega_k^{\Spin\text{-}D_{16}}\to\Omega_k^{\Spin\text{-}\Z/8}$.
%\end{lem}
%\begin{proof}[Proof of \cref{halfQ_detects}, assuming \cref{halfQ_D16} and \cref{D16toZ8_cover}]
\begin{proof}
Essentially by definition, the double cover of $\halfQseven$ we get from $\varpi$ is $Q_4^7$, and $Q_4^7$ with any
of its Spin-$\Z/8\Z$ structures is order $4$ in $\Omega_7^{\Spin\text{-}\Z/8\Z} (\pt)$.
%Therefore $[\halfQseven]\in
%\Omega_7^{\Spin\text{-}D_{16}}\cong\textcolor{Green}{\Z/4} \oplus \textcolor{RedOrange}{\Z/2}$ must have order at
%least $4$. The only way for this to happen is for $[\halfQseven]$ to generate the $\textcolor{Green}{\Z/4}$
%summand.
\end{proof}

\subsubsection{The orange generator in dimension 7}
\label{7d_orange}
\label{s_7dorange}
%Outline:
%\begin{itemize}
%	\item Sketch of the argument --- introduce $\orangeseven$ and its principal $D_8$-bundle, explain how the proof
%	will go
%	\item Show that $\orangeseven$ is Spin-$D_{16}$
%	\item Recall relevant facts on the Smith map from the appendix. Key fact: as soon as we obtain
%	$S^3/Q_{16}\times T^2$, we're done.
%	\item Show that the Poincaré dual to $x^2$ is the fiber
%\end{itemize}
We are missing one generator in dimension $7$, corresponding in the $E_\infty$-page to the orange
(rightward-pointing triangle)
$\textcolor{RedOrange}{\Z/2\Z}$ summand in topological degree $7$. Let $Q_{16}$ be the Pin$^-$ cover of $D_8$ with
reference to the standard, two-dimensional representation; that is, $Q_{16}$ is the pullback of $D_8\inj \mathrm
O(2)$ by the double cover $\Pin^-(2)\to \O(2)$.\footnote{Another way $Q_{16}$ arises is as the spin cover of the
representation $V\oplus\Det(V)\colon D_8\to\SSO(3)$, i.e.\ the pullback of this representation by the double cover
$\mathrm{SU}(2)\to\SSO(3)$.} $Q_{16}$ is called the \term{generalized quaternion group} of
order $16$, or the \term{binary dihedral group} $\mathit{2D}_8$, or the \term{dicyclic group} $\mathrm{Dic}_{16}$.
$Q_{16}$ acts freely on $S^3$; the quotient is called a \term{prism manifold}.

Inside $Q_{16}$ there are classes $\hat r$ and $j$ such that $Q_{16}$ has the presentation
%\TODO: explicit presentation of $Q_{16}$ with generators ?? [say $\hat r$ for now] and $j$.
%
%\textcolor{red}{MD: Here is one possible presentation, where $\hat{r}$ is the lift of rotations and $j$ the lift of reflections
\begin{equation}
	Q_{16} = \langle \hat{r}\mid \hat{r}^{8} = \text{id} \,, \enspace j^2 = \hat{r}^4 \,, \enspace
	j^{-1} \, \hat{r} \, j = \hat{r}^{-1} \rangle.
\end{equation}
Under the double cover $Q_{16}\to D_8$, $\hat r$ is sent to a rotation and $j$ is sent to a reflection.

Let $\zeta$ be the involution on the quaternions $\mathbb H$ which is $1$ in the $1$ and $i$ directions and $-1$ in
the $j$ and $k$ directions. Then, up to a factor of $-1$, $\zeta$ commutes with the standard action of $Q_{16}$ on
$\mathbb H$: if we think of $\mathbb H$ as $\C^2$, $\zeta(z, w) = (z, -w)$. This commutes with multiplication by a
unit complex number on both factors (the action of $\hat r$), but $j\cdot (z, w) = (w, -z)$, so $\zeta\circ j =
-j\circ\zeta$.
%\textcolor{red}{Is this really the action of $j$?}.
Since $\zeta$ is an isometry, it restricts to an involution on
$S^3\subset\mathbb H$. We would like to descend it to an involution on $S^3/Q_{16}$; a priori, this is not
possible, because $\zeta$ does not commute with the $Q_{16}$-action on $S^3$, but the two actions commute up to a
sign, and since $S^3/Q_{16}$ is a quotient of $\RP^3$, sign discrepancies do not matter, allowing $\zeta$ to induce
a (very much non-free) $\Z/2\Z$-action $\overline\zeta$ on $S^3/Q_{16}$.

The final generator, which we call $\orangeseven$, is the quotient of
\begin{equation}
	(S^3/Q_{16})\times T_p^2 \times S^2 %\textcolor{red}{(S^3/Q_{16})\times S^1_p \times S^1_p \times S^2}
\end{equation}
where $T^2_p = S^1_p \times S^1_p$, by the involution which is:
\begin{itemize}
	\item $\overline\zeta$ on $S^3/Q_{16}$,
	%``complex conjugation'' on $S^3/Q_{16}$. Specifically, complex conjugation acts on $S^3\subset\C^2$. Since
	%$Q_{16}\subset \mathrm{SU}(2)$ is preserved by complex conjugation, complex conjugation sends $Q_{16}$-orbits of $S^3$ to
	%$Q_{16}$-orbits, like we saw for lens spaces, so descends to a (non-free, orientation-preserving) involution on
	%$S^3/Q_{16}$.
	\item the involution on $T_p^2$ whose quotient is the Klein bottle $\mathit{KB}$, and
	\item the antipodal action on $S^2$.
\end{itemize}
The quotient is an $S^3/Q_{16}\times T^2$-bundle over $\RP^2$; it can also be realized as a fiber bundle over
$\mathit{KB}$
whose fiber is a $S^3/Q_{16}$-bundle over $\RP^2$. The quotient is oriented. Using the long exact sequence in
homotopy groups associated to a fibration, we learn
\begin{equation}
	\pi_1(\orangeseven) \cong (Q_{16}\times\Z^2)\rtimes\Z/2\Z,
\end{equation}
where $\Z/2\Z$ acts separately on $Q_{16}$ and on $\Z^2$; therefore $\Z^2$ is a normal subgroup of
$\pi_1(\orangeseven)$.
\begin{defn}
\label{orange_seven_bundle}
We define a principal $D_8$-bundle $P\to \orangeseven$ by specifying a map from the fundamental group $\lambda\colon \pi_1(\orangeseven) \to
D_8$; then $P$ is defined so that its monodromy around a loop $\gamma$ in $\orangeseven$ is $\lambda([-\gamma])$,
which uniquely specifies a principal $D_8$-bundle up to isomorphism.

First, let $q_1\colon \pi_1(\orangeseven)\to Q_{16}\rtimes\Z/2\Z$ be the quotient by the normal $\Z^2\subset
\pi_1(\orangeseven)$ noted above. Then let $q_2$ be the quotient by $-1\in Q_{16}$; because $\zeta$ is
multiplication by $-1$ in some coordinate directions, once we take this quotient, the semidirect product untwists,
so $q_2$ is a homomorphism
\begin{equation}
	Q_{16}\rtimes\Z/2\Z\to D_8\times\Z/2\Z.
\end{equation}
Finally, let $q_3\colon D_8\times\Z/2\Z\to D_8$ be the identity on the first factor and the inclusion of a
reflection subgroup on the second factor.\footnote{The choice of specific reflection subgroup does not matter here,
even though it does elsewhere in this paper.} Then define
\begin{equation}
	\lambda\coloneqq q_3\circ q_2\circ q_1.
\end{equation}
\end{defn}
To show $\orangeseven$ is the last generator, we must show it is Spin-$D_{16}$ with the principal $D_8$-bundle in
\cref{orange_seven_bundle}, and that the bordism class of $(\orangeseven, P)$ is linearly independent of that of
$\halfQseven$. The first fact, which we prove in \cref{seven_orange_is_spin_d16}, is similar to our proofs for
$\halfQseven$ and $\ninedimgen$ in \cref{9d_is_spD16,halfQseven_spD16}: stably split the tangent bundle and
compute. Our approach showing that $\orangeseven$ is linearly independent from $\halfQseven$, however, is different
from previous approaches: we use a \term{Smith homomorphism} $S_{x^2}\colon
\Omega_7^{\Spin\text{-}D_{16}}\to\Omega_5^{\Spin\text{-}Q_{16}}$ and show that the images of $\halfQseven$ and
$\orangeseven$ are linearly independent. For background information on the Smith homomorphism, see
Appendix \ref{Smith_general}.

Because we constructed the principal $D_8$-bundle $P\to \orangeseven$ using $Q_{16}$, it will be helpful in the
rest of this section to know $H^*(BQ_{16};\Z/2)$, so that we can assign characteristic classes to principal
$Q_{16}$-bundles.
\begin{prop}[{\cite[Lemma 4.1.1b]{Tei92} and~\cite[Theorem 4.40ii]{Sna13}}]
\label{Q16_coh}
There is an isomorphism $H^\ast(BQ_{16};\Z/2\Z)\cong\Z/2\Z[\hat x, \hat y, p]/(\hat x\hat y+\hat y^2, \hat x^3)$,
where $\hat x$ and $\hat y$ are the pullbacks of $x$, respectively $y$, by the map $Q_{16}\to D_8$, so
have degree $1$; and $\abs p = 4$.
\end{prop}
\begin{prop}
\label{seven_orange_is_spin_d16}
$\orangeseven$ with its principal $D_8$-bundle $P\to\orangeseven$ from \cref{orange_seven_bundle} has a
Spin-$D_{16}$ structure.
\end{prop}
\begin{proof}
As usual, we begin by analyzing the $\Z/2\Z$-action on the tangent bundle of $S^3/Q_{16}\times T_p^2\times S^2$. The
tangent bundle to $T_p^2$ is trivial, because $T^2$ admits a Lie group structure; the tangent bundle to $S^2$ is
trivial after appending on a single trivial summand, as we saw in~\eqref{Sn_stable_splitting}. The tangent bundle
of $S^3/Q_{16}$ is actually trivializable, like for any compact, oriented $3$-manifold, but for us it is more useful
to have an explicit isomorphism to a different vector bundle.

Like in \cref{lens_rotate}, consider the $Q_{16}$-action on $TS^3\oplus\underline\R\cong\underline\C^2$. This
action is the standard representation of $Q_{16}$ on $\C^2$, and the outward unit normal vector field is invariant,
so we learn that $T(S^3/Q_{16})\oplus\underline\R$ is a rank-two complex vector bundle $E\to S^3/Q_{16}$ induced
from the principal $Q_{16}$-bundle $S^3\to S^3/Q_{16}$.

Now the $\Z/2\Z$-actions. From \cref{RP_tangent} we know that $\Z/2\Z$ acts on $TS^2\oplus\underline\R$ by $-1$. For
$T_p^2$, use the fact that the tangent space of the Klein bottle splits as $\underline\R$ plus a non-trivial line
bundle to infer how $\Z/2\Z$ acts on $T(T_p^2)$. And for $T(S^3/Q_{16})\oplus\underline\R\cong E$, $\Z/2\Z$ acts by
$\zeta$. Therefore% \textcolor{red}{MD: $\underline{\R}^2$ on the right-hand side?}
\begin{equation}
\label{orange_7_stable_splitting}
	T\orangeseven\oplus\underline\R^2 \overset\cong\longrightarrow E'\oplus\underline\R\oplus
	\sigma^{\oplus 4},
\end{equation}
where $E'\to\orangeseven$ is a real, rank-$4$ vector bundle and $\sigma\to\orangeseven$ is the pullback of the
tautological bundle on $\RP^2$ along $\orangeseven\to\RP^2$. Applying the Whitney sum formula
to~\eqref{orange_7_stable_splitting}, we learn that $w_1(T\orangeseven) = w_1(E')$ and $w_2(T\orangeseven) =
w_2(E')$. The action of $\zeta$ on $\mathbb H$ is orientation-preserving, so $w_1(E') = 0$, so $\orangeseven$ is
orientable. To compute $w_2$, first observe that $E\to S^3/Q_{16}\times T_p^2\times S^2$ is Spin, because the action
of $Q_{16}$ on $\mathbb H$ is Spin; thus we just have to ask how $\Z/2\Z$ acts on $E$. In coordinates, we obtain
two trivial actions and two reflection actions on the four real coordinates, and this can be used to show $w_2(E')
= \hat x^2$. Thus $w_2(T\orangeseven) = \hat x^2$.

Next we need to compute $w(P)$, which we do by tracing $w$ through the pullbacks by $q_3$, $q_2$, and then $q_1$.
In \cref{spin_D4_pullback} we proved that for any inclusion $\Z/2\Z\inj D_8$ given by a reflection, $w$ pulls back
to $0$, so $q_3^*(w) = w$; a priori one has to verify there is no term such as $x\hat x$ or $y\hat x$ coming from
the cross term in the Künneth formula, but all such terms vanish in $q_3^*(w)$. Next, since $Q_{16}$ is the Pin$^-$
cover of $D_8$ with respect to the standard representation $V$, $w_2(V)$ and $w_1(V)^2$ are identified in
$H^2(BQ_{16};\Z/2\Z)$. Therefore $q_2^*(w) = q_2^*(x^2) = \hat x^2$. Finally, $q_1^*$ pulls back $\hat x\mapsto\hat
x$, so $w(P) = \lambda^*(w) = \hat x^2$.

We have shown that $\orangeseven$ is orientable and $w_2(T\orangeseven) = w(P)$, so $(\orangeseven, P)$ is
Spin-$D_{16}$.
\end{proof}
Now we show that $\halfQseven$ and $\orangeseven$ are linearly independent in
$\Omega_7^{\Spin\text{-}D_{16}} (\pt) \cong\textcolor{Green}{\Z/4\Z}\oplus\textcolor{RedOrange}{\Z/2\Z}$. Since we already
established that $\halfQseven$ has order $4$ in \cref{half_Qseven_order_4}, this will suffice to show that
$\halfQseven$ and $\orangeseven$ together generate $\Omega_7^{\Spin\text{-}D_{16}} (\pt)$.

Let $\xi\colon H\to\O$ be a tangential structure and $c\in H^k(BH;\Z/2\Z)$, so we may interpret $c$ as a
characteristic class for manifolds with $H$-structure. Let $M$ be an $H$-manifold and $N\subset M$ be a closed,
$(n-k)$-dimensional submanifold such that the image of the mod $2$ fundamental class of $N$ in $H_{n-k}(M;\Z/2\Z)$ is
Poincaré dual to $c(M)\in H^k(M;\Z/2\Z)$. In this case we say that $N$ is a \term{smooth representative of the
Poincaré dual to $c(M)$}. In some settings, $N$ acquires a canonical $H'$-structure, where $H'\to\O$ is some other
tangential structure. In this case, the class of $N$ in $\Omega_{n-k}^{H'}$ depends only on the class of $M$ in
$\Omega_n^H$, and this construction lifts to a homomorphism $S_c\colon \Omega_n^H\to\Omega_{n-k}^{H'}$. This is
called a \term{Smith homomorphism}.

We will say more about Smith homomorphisms in \cref{Smith_appendix}; what we need to know right now is that
%if $Q_{16}$ denotes the quaternion group of order $16$, which is the Pin$^-$ cover of $D_8$, then
there is a Smith homomorphism $S_{x^2}\colon \Omega_n^{\Spin\text{-}D_{16}}\to\Omega_{n-2}^{\Spin\text{-}Q_{16}}$
defined by taking a smooth representative of the Poincaré dual of $x(P)^2$, where $P\to M$ is the principal
$D_8$-bundle associated to a Spin-$D_{16}$ structure and $\Spin\text{-}Q_{16}\coloneqq \Spin\times_{\Z/2}Q_{16}$.
We prove this in Appendix \ref{smith_GL}.

The following three propositions jointly imply $\halfQseven$ and $\orangeseven$ are linearly independent in
$\Omega_7^{\Spin\text{-}D_{16}} (\pt)$.
\begin{prop}
\label{7d_lens_integral}
Let $\nu\colon\Omega_7^{\Spin\text{-}D_{16}}$ be the bordism invariant
\begin{equation}
	\nu\colon \Omega_7^{\Spin\text{-}D_{16}}\overset{S_{x^2}}{\longrightarrow}
	\Omega_5^{\Spin\text{-}Q_{16}}\overset{\int w^2 y}{\longrightarrow} \Z/2.
\end{equation}
Then $\nu(\halfQseven) = 1$ and $\nu(\orangeseven) = 0$.
\end{prop}
\begin{prop}
\label{smith_or_nontriv}
$S_{x^2}(\orangeseven)\ne 0$ in $\Omega_5^{\Spin\text{-}Q_{16}}$.
\end{prop}
\begin{prop}
\label{split_Q16_5}
$\Omega_5^{\Spin\text{-}Q_{16}}\cong (\Z/2\Z)^{\oplus k}$ for some $k$.
\end{prop}
We prove \cref{split_Q16_5} in Appendix \ref{smith_GL} in the form of \cref{quater_hidden}; we prove
\cref{7d_lens_integral,smith_or_nontriv} here. First, though, we compute $x(P)^2\in H^2(\orangeseven;\Z/2\Z)$.
\begin{lem}
\label{x_of_orange_seven}
Let $\alpha\in H^1(\orangeseven;\Z/2\Z)$ be the element of $H^1(\orangeseven;\Z/2\Z)$ corresponding to the
principal $\Z/2\Z$-bundle $S^3/Q_{16}\times T_p^2\times S^2\to \orangeseven$. Then $x(P)^2 = \hat x^2 + \hat x\alpha +
\alpha^2$.
\end{lem}
A quick look at the Serre spectral sequence for the fiber bundle
\begin{equation}
\label{orange_seven_Serre}
	S^3/Q_{16}\times T_p^2 \longrightarrow \orangeseven\longrightarrow \RP^2
\end{equation}
shows that the three monomials $\hat x^2$, $\hat x\alpha$, and $\alpha^2$ are linearly independent in
$H^2(\orangeseven;\Z/2\Z)$.
\begin{proof}[Proof of \cref{x_of_orange_seven}]
The characteristic class $x(P)$ is the pullback of the nonzero class in $H^1(B\Z/2\Z;\Z/2\Z)$ by the
composition $\orangeseven\to BD_8\to B\Z/2\Z$ in which the first map is the classifying map for $P$ and the second
map is induced from the quotient by the subgroup of rotations. So we need to track how $\lambda$ pulls back $x\in
H^1(BD_8;\Z/2\Z)$. Applying $q_3^*$, if $\alpha'$ is the generator of $H^1(B\Z/2\Z;\Z/2\Z)$ for the $\Z/2\Z$
summand in $BD_8\times B\Z/2\Z$, then $q_3^*(x) = x+\alpha'$, because $x$ is nontrivial on reflections (as we
discussed in \cref{spin_D4_pullback}). Then $q_2^*$ sends $x\mapsto\hat x$ and $q_1^*$ does not affect this class,
and we learn $q_1^*q_2^*(\alpha') = \alpha$. Therefore $x(P) = \hat x + \alpha$ and $x(P)^2 = \hat x^2 + \hat
x\alpha + \alpha^2$. \end{proof} Now the propositions we promised.
\begin{proof}[Proof of \cref{7d_lens_integral}]
First we show $\nu(\halfQseven) = 1$. For this, it suffices to show $S_{x^2}(\halfQseven) = L_4^5$ with the duality
bundle induced from $S^5\to L_4^5$, because for this duality bundle
$\int_{L_5^4} w^2y = 1$. To see this, use that a smooth submanifold $i\colon N\inj M$ is
a smooth representative of the Poincaré dual of a class $\alpha$ if for any $\beta\in H^*(M;\Z/2\Z)$,
\begin{equation}
\label{PD_integrate}
	\int_N i^*(\beta) = \int_M \alpha\beta.
\end{equation}
This equation is only interesting if $\abs\beta = \dim(N)$ and $\abs\alpha = \dim(M) - \dim(N)$; otherwise both
sides vanish. Let $N\subset\halfQseven$ be the $L_4^5$ fiber for $M = \halfQseven$ over some point in $\RP^2$; we
want to verify~\eqref{PD_integrate} in this example. This is taken care of by \cref{serre_for_green_7d}: the only
nonzero case is $\beta = w^2y$, and
\begin{equation}
	\int_{L_4^5} w^2y = \int_{\halfQseven} w^2x^2y = 1.
\end{equation}
The second half of this proposition is showing that $\nu(\orangeseven) = 0$. We first evaluate the Smith
homomorphism on $\orangeseven$; \cref{x_of_orange_seven} tells us to compute the Poincaré dual of $\hat x^2 + \hat
x\alpha + \alpha^2$. We may do this one monomial at a time and take the disjoint union.
\begin{itemize}
	\item Because $\hat x^2$ comes from the fiber in~\eqref{orange_seven_Serre}, the Poincaré dual of $\hat x^2$
	can also be taken to be a fiber bundle, specifically one of the form $S^1\times T^2\to Y_1\to \RP^2$. One finds
	that the duality bundle is trivial when restricted to any fiber $S^1$, and that these $S^1$s carry the
	nonbounding Spin structure; therefore $Y_1$ bounds by filling in the fiberwise $S^1$s with copies of $D^2$.
	\item Since $\alpha^2$ is a top-degree cohomology class for the base $\RP^2$, the Poincaré dual of $\alpha^2$
	is the fiber $S^3/Q_{16}\times T_p^2$; the argument is essentially the same as the way we calculated
	$S_{x^2}(\halfQseven)$ above. The duality bundle on $S^3/Q_{16}\times T_p^2$ pulls back across the quotient
	$S^3/Q_{16}\times T_p^2\to S^3/Q_{16}$, so its characteristic classes also pull back from $S^3/Q_{16}$ and
	therefore vanish in degrees above $3$. Thus $\int_{S^3/Q_{16}\times T^2}w^2y = 0$.
	\item Finally, one representative of the Poincaré dual of $x\alpha$ is a manifold which factors as a fiber
	bundle $\Sigma\times T_p^2\to Y_2\to \RP^1$, where $\Sigma$ is a closed, unorientable surface. The class $w^2$ on
	this manifold pulls back from a $3$-manifold obtained by quotienting $Y_2$ by $T_p^2$, so $w^2 = 0$ for $Y_2$,
	and therefore $\int_{Y_2}w^2y = 0$ too.
\end{itemize}
Therefore $\nu(\orangeseven) = 0$.
\end{proof}
\begin{proof}[Proof of \cref{smith_or_nontriv}]
In the proof of the previous proposition, we found a representative of $S_{x^2}(\orangeseven)$ that is a disjoint
union of three $5$-manifolds $S^3/Q_{16}\times T_p^2$, $Y_1$, and $Y_2$, and we saw that $Y_1$ bounds as a
Spin-$Q_{16}$ manifold. In \cref{after_Smith}, we prove that $[S^3/Q_{16}\times T_p^2]\ne 0$ in
$\Omega_5^{\Spin\text{-}Q_{16}}$; here is where it is crucial that we gave $T_p^2$ the nonbounding Spin structure
when we defined $\orangeseven$, so that we obtain $S^3/Q_{16}\times T_p^2$ after the Smith
homomorphism. Finally, $Y_2$ is not bordant to $S^3/Q_{16}\times T_p^2$, so we can conclude that
$S_{x^2}(\orangeseven)\ne 0$.
\end{proof}
This finishes the determination of the generators of Spin-$\GL^+(2, \Z)$ bordism in dimensions $11$ and below.
\subsection{Multiplicative structure}
\label{mult_str_D16}

We saw in Section \ref{mult_Z4} that Spin-$\Z/8\Z$ bordism is a ring. Spin-$D_{16}$ bordism is not a ring, ultimately
because there is no non-trivial way to tensor together two $D_8$-bundles into a single $D_8$-bundle. Spin-$D_{16}$
bordism is still an $\Omega_\ast^\Spin$-module, just as for any kind of twisted Spin bordism, and most of this
information can be read off of the $E_\infty$-page of the Adams spectral sequence.

One surprise is that even though Spin-$\Z/8\Z$ bordism is a ring, and we have a map from Spin-$\Z/8\Z$ bordism to
Spin-$D_{16}$ bordism, Spin-$D_{16}$ bordism is not a module over Spin-$\Z/8\Z$ bordism! This is because in
$\Omega_\ast^{\Spin\text{-}\Z/8\Z} (\pt)$, the class of $S^1_p$ is divisible by $4$. However, we saw that the
generator $x$ of $\Omega_2^{\Spin\text{-}D_{16}} (\pt) \cong\Z/2\Z$ is $S_p^1$ times another class, but $x$ is
not divisible by $4$. Therefore there is not a natural Spin-$D_{16}$ structure on the product of a Spin-$\Z/8\Z$
manifold and a Spin-$D_{16}$ manifold. This is similar to (and related to) the fact that the inclusion $\text{U}(1)\to\mathrm O(2)$ defines a map from Spin\textsuperscript{$c$} bordism to Spin-$\mathrm O(2)$ bordism, and
Spin$^c$ bordism is a ring, but Spin-$\mathrm O(2)$ bordism is not a module over this ring (in this case, because
taking the product with $S_p^1$ is not always zero, which cannot happen in any
$\Omega_*^{\Spin^c}$-module).

This concludes our computations. Additional details are presented in the Appendices.

%%%%%%%%%%%%%%%%%%%%%%%%%%%%%%%%%%%%%%%%%%%%%%%%%%%%
%%%%%%%%%%%%%%% APPENDIX %%%%%%%%%%%%%%%%%%%%%%%%%%%%%%%
%%%%%%%%%%%%%%%%%%%%%%%%%%%%%%%%%%%%%%%%%%%%%%%%%%%%

\pagebreak
\begin{appendix}

\section{Collection of bordism groups}
\label{app:collection}

In this Appendix we collect some of the bordism groups calculated in this paper up to dimension eleven. We also list a set of generators.

\begin{table}[h!]
\centering
\begin{tabular}{c c c}
\toprule
$k$ & $\Omega_k^{\Spin} (\pt)$ & Generators\\
\midrule
$0$ & {\footnotesize $\Z$} & {\footnotesize $ \pt_+$} \\
$1$ & {\footnotesize $\Z/2\Z$} & {\footnotesize $S^1_p$}\\
$2$ &  {\footnotesize $\Z/2\Z$} & {\footnotesize $S^1_p \times S^1_p$}\\
$3$ & {\footnotesize $0$} & \\
$4$ & {\footnotesize $\Z$} & {\footnotesize K3}\\
$5$ & {\footnotesize $0$} & \\
$6$ & {\footnotesize $0$} & \\
$7$ & {\footnotesize $0$} & \\
$8$ & {\footnotesize $\Z \oplus \Z$} & {\footnotesize $(B \,, \HP^2)$}\\
$9$ & {\footnotesize $(\Z/2\Z) \oplus (\Z/2\Z)$} & {\footnotesize $(B \times S^1_p \,, \HP^2 \times S^1_p )$} \\
$10$ & {\footnotesize $(\Z/2\Z) \oplus (\Z/2\Z) \oplus (\Z/2\Z)$} & {\footnotesize $(B \times S^1_p \times S^1_p \,, \HP^2 \times S^1_p \times S^1_p \,, X_{10} )$}\\
$11$ & {\footnotesize $0$} & \\
\bottomrule
\end{tabular}
\caption{Bordism groups for Spin manifolds in dimension smaller than twelve and a full set of generators.}
\label{apptab:spin}
\end{table}

\begin{table}[h!]
\centering
\begin{tabular}{c c c}
\toprule
$k$ & $\Omega_k^{\Spin} (B \Z/3\Z)$ & Generators\\
\midrule
$0$ & {\footnotesize $\Z$} & {\footnotesize $ \pt_+$} \\
$1$ & {\footnotesize $(\Z/2\Z) \oplus (\Z/3\Z)$} & {\footnotesize $(S^1_p \,, L^1_3)$}\\
$2$ &  {\footnotesize $\Z/2\Z$} & {\footnotesize $S^1_p \times S^1_p$}\\
$3$ & {\footnotesize $\Z/3\Z$} & {\footnotesize $L^3_3$}\\
$4$ & {\footnotesize $\Z$} & {\footnotesize K3}\\
$5$ & {\footnotesize $\Z/9\Z$} & {\footnotesize $L^5_3$}\\
$6$ & {\footnotesize $0$} & \\
$7$ & {\footnotesize $\Z/9\Z$} & {\footnotesize $L^7_3$}\\
$8$ & {\footnotesize $\Z \oplus \Z$} & {\footnotesize $(B \,, \HP^2)$}\\
$9$ & {\footnotesize $(\Z/2\Z) \oplus (\Z/2\Z) \oplus (\Z/3\Z) \oplus (\Z/27\Z)$} & {\footnotesize $(B \times S^1_p \,, \HP^2 \times S^1_p \,, \HP^2 \times L^1_3 \,, L^9_3)$} \\
$10$ & {\footnotesize $(\Z/2\Z) \oplus (\Z/2\Z) \oplus (\Z/2\Z)$} & {\footnotesize $(B \times S^1_p \times S^1_p \,, \HP^2 \times S^1_p \times S^1_p \,, X_{10} )$}\\
$11$ & {\footnotesize $(\Z/3\Z) \oplus (\Z/27\Z)$} & {\footnotesize $(\HP^2 \times L^3_3 \,, L^{11}_3)$}\\
\bottomrule
\end{tabular}
\caption{Bordism groups for Spin manifolds with $\Z/3\Z$ bundle in dimension smaller than twelve and a full set of generators. This contains the Spin bordism groups above.}
\label{apptab:spinZ3}
\end{table}

\begin{table}[h!]
\centering
\begin{tabular}{c c c}
\toprule
$k$ & $\Omega_k^{\Spin} (B \Z /4\Z)$ & Generators\\
\midrule
$0$ & {\footnotesize $\Z$} & {\footnotesize $ \pt_+$} \\
$1$ & {\footnotesize $(\Z/2\Z) \oplus (\Z/4\Z)$} & {\footnotesize $(S^1_p \,, L^1_4)$}\\
$2$ &  {\footnotesize $(\Z/2\Z) \oplus (\Z/2\Z)$} & {\footnotesize $(S^1_p \times S^1_p \,, S^1_p \times L^1_4)$}\\
$3$ & {\footnotesize $(\Z/2\Z) \oplus (\Z/8\Z)$} & {\footnotesize $(S^1_p \times S^1_p \times L^1_4 \,, L^3_4)$}\\
$4$ & {\footnotesize $\Z$} & {\footnotesize K3}\\
$5$ & {\footnotesize $\Z/4\Z$} & {\footnotesize $Q^5_4$}\\
$6$ & {\footnotesize $0$} & \\
$7$ & {\footnotesize $(\Z/2\Z) \oplus (\Z/32\Z)$} & {\footnotesize $(\widetilde{L}^7_4 \,, L^7_4)$}\\
$8$ & {\footnotesize $\Z \oplus \Z$} & {\footnotesize $(B \,, \HP^2)$}\\
$9$ & {\footnotesize $(\Z/2\Z)^{\oplus 3} \oplus (\Z/4\Z) \oplus (\Z/8\Z)$} & {\footnotesize $(B \times S^1_p \,,
\HP^2 \times S^1_p \,, B\times L_4^1 \,, \HP^2 \times L^1_4 \,, Q^9_4)$} \\
$10$ & {\footnotesize $(\Z/2\Z)^{\oplus 5}$} & {\footnotesize $(B \times S^1_p \times S^1_p \,, \HP^2 \times S^1_p \times S^1_p \,, X_{10} \,,$}\\
 & & {\footnotesize $(B \times S^1_p \times L^1_4 \,, \HP^2 \times S^1_p \times L^1_4)$}\\
$11$ & {\footnotesize $(\Z/2\Z)^{\oplus 2}$} & {\footnotesize $(X_{10} \times L^1_4 \,, \HP^2 \times S^1_p \times S^1_p \times L^1_4 \,, $}\\
 & {\footnotesize $ \oplus (\Z/8\Z)^{\oplus 2} \oplus (\Z /128\Z)$} & {\footnotesize $\HP^2 \times L^3_4 \,, \widetilde{L}^{11}_4 \,, L^{11}_4)$}\\
\bottomrule
\end{tabular}
\caption{Bordism groups for Spin manifolds with $\Z/4\Z$ bundle in dimension smaller than twelve and a full set of
generators. This contains the Spin bordism groups above. In dimension $9$, to get a set of generators which produce
the given direct-sum decomposition, as opposed to merely a generating set, replace $B\times L_4^1$ with $(B\times
L_4^1) \mathbin{\#} (-2 Q_4^9)$.}
\label{apptab:spinZ4}
\end{table}

\begin{table}[h!]
\centering
\begin{tabular}{c c c}
\toprule
$k$ & $\Omega_k^{\Spin\text{-}(\Z/8\Z)} (\pt)$ & Generators\\
\midrule
$0$ & {\footnotesize $\Z$} & {\footnotesize $ \pt_+$} \\
$1$ & {\footnotesize $\Z/8\Z$} & {\footnotesize $L^1_4$}\\
$2$ &  {\footnotesize $0$} & \\
$3$ & {\footnotesize $\Z/2\Z$} & {\footnotesize $L^3_4$}\\
$4$ & {\footnotesize $\Z$} & {\footnotesize $E$}\\
$5$ & {\footnotesize $(\Z/2\Z) \oplus (\Z/32\Z)$} & {\footnotesize $(\widetilde{L}^5_4 \,, L^5_4)$}\\
$6$ & {\footnotesize $0$} & \\
$7$ & {\footnotesize $\Z/4\Z$} & {\footnotesize $Q^7_4$}\\
$8$ & {\footnotesize $\Z \oplus \Z$} & {\footnotesize $(B \,, \HP^2)$}\\
$9$ & {\footnotesize $(\Z/4\Z) \oplus (\Z/8\Z) \oplus (\Z/128\Z)$} & {\footnotesize $(\widetilde{L}^9_4 \,, \HP^2 \times L^1_4 \,, L^9_4)$} \\
$10$ & {\footnotesize $\Z/2\Z$} & {\footnotesize $X_{10}$}\\
$11$ & {\footnotesize $(\Z/2\Z) \oplus (\Z/2\Z) \oplus (\Z /8\Z)$} & {\footnotesize $(\HP^2 \times L^3_4 \,, X_{10} \times L^1_4 \,, Q^{11}_4)$}\\
\bottomrule
\end{tabular}
\caption{Bordism groups for Spin-$\Z/8\Z$ manifolds with dimension lower than twelve and a full set of generators.}
\label{apptab:spinZ8}
\end{table}

\begin{table}[h!]
\centering
\begin{tabular}{c c c}
\toprule
$k$ & $\Omega_k^{\Spin\text{-}D_{16}} (\pt)$ & Generators\\
\midrule
$0$ & {\footnotesize $\Z$} & {\footnotesize $ \pt_+$} \\
$1$ & {\footnotesize $(\Z/2\Z) \oplus (\Z/2\Z)$} & {\footnotesize $(L^1_4 \,, S^1_R)$}\\
$2$ &  {\footnotesize $\Z/2\Z$} & {\footnotesize $S^1_p \times S^1_R$} \\
$3$ & {\footnotesize $(\Z/2\Z)^{\oplus 3}$} & {\footnotesize $(L^3_4 \,, \RP^3 \,, \widetilde{\RP}^3 )$}\\
$4$ & {\footnotesize $\Z$} & {\footnotesize $E$}\\
$5$ & {\footnotesize $(\Z/2\Z) \oplus (\Z/2\Z)$} & {\footnotesize $(L^5_4 \,, X_5)$}\\
$6$ & {\footnotesize $0$} & \\
$7$ & {\footnotesize $(\Z/2\Z)^{\oplus 3} \oplus (\Z/4\Z)$} & {\footnotesize $(\RP^7 \,, \widetilde{\RP}^7 \,, \orangeseven  \,, \halfQseven)$}\\
$8$ & {\footnotesize $\Z\oplus \Z\oplus (\Z/2\Z)$} & {\footnotesize $(\halfBott \,, \HP^2 \,, \halfQseven \times S_p^1)$}\\
$9$ & {\footnotesize $(\Z/2\Z)^{\oplus 8}$} & {\footnotesize $(\halfQseven \times S_p^1 \times S_p^1 \,, X_9 \,,
	\widetilde X_9 \,, L^9_4 \,, W^9_1 \,,$} \\
& &	 {\footnotesize $\phantom{(} B\times L_4^1 \,, \HP^2 \times L^1_4 \,, \HP^2\times S_R^1 )$} \\
$10$ & {\footnotesize $(\Z/2\Z)^{\oplus 4}$} & {\footnotesize $(B\times L_4^1\times S_p^1 \,, W^9_1 \times S_p^1 \,, \HP^2 \times
	L_4^1 \times S_p^1\,, X_{10})$}\\
$11$ & {\footnotesize $(\Z/2\Z)^{\oplus 9} \oplus (\Z/8\Z)$} & {\footnotesize $(\RP^{11} \,, \widetilde{\RP}^{11} \,, X_{11} \,, \widetilde{X}_{11} \,, \HP^2 \times L^3_4 \,, \HP^2 \times \RP^3 \,, $} \\
& &	 {\footnotesize $\phantom{(} \HP^2\times \widetilde{\RP}^3 \,, X_{10} \times L_4^1 \,, X_{10}\times S_R^1 \,,  Q_4^{11})$} \\
\bottomrule
\end{tabular}
\caption{Bordism groups for Spin-$D_{16}$ manifolds with dimension lower than twelve and a full set of
generators.
}
\label{apptab:spinD16}
\end{table}

We list some of the generators and explain the notation.
\begin{itemize}
	\item  $S_p^1$ is a circle with Spin structure induced from the Lie group
	framing, i.e., periodic boundary condition for fermions.
	\item K3 denotes a K3 surface and $E$ denotes an Enriques surface.
	\item $B$ describes a Bott manifold and $\HP^2$ quaternionic
	projective space. The name ``Bott manifold'' does not fix $B$ up to diffeomorphism; we never need to pick a
	specific Bott manifold, but explicit examples are constructed in~\cite[Section 7]{Lau04} and~\cite[Section 5.3]{FH21}.
	\item $X_{10}$ is a Milnor hypersurface.
	\item $L^{2n-1}_k$ are lens spaces, see also Appendix \ref{subapp:etalens}. The discrete $\Z/k\Z$-bundle is induced by the fibration $S^{2n-1} \rightarrow L^{2n-1}_k$. The notation $\widetilde{L}^{2n-1}_k$ denotes the choice of a distinct Spin or Spin-$\Z/8\Z$ structure, whose realization can be extracted from the lift of the $\Z/k\Z$ action on fermions given in \eqref{eq:tangact}.
	\item The spaces $Q^{2n-1}_k$ are lens space bundles of $L^{2n-3}_k$ over $\CP^1$ and $\Z/k\Z$ bundle on the lens space fiber, see also Appendix \ref{ss:eta_Q}.
	\item $S^1_R$ is a circle with a duality bundle involving a reflection.
	\item $\RP^{3 + 4m}$ and $\widetilde{\RP}^{3 + 4m}$ are real projective spaces with $\Z/2\Z$ duality bundle associated to different embeddings of the reflection into $D_8$ (or its lift to $D_{16}$).
	\item $\halfQseven$ is a $\Z/2\Z$ quotient of $Q^7_4$.
	\item $\orangeseven$ is a $\Z/2\Z$ quotient, involving reflections, of $S^3/Q_{16} \times S^2\times T^2$, where
	the quotient acts via reflections in the $S^3/Q_{16}$ factor, as antipodal mapping on the $S^2$, and as the
	fixed-point-free isometry that would yield the Klein Bottle when acting on $T^2$.
	\item $\halfBott$ is a half Bott manifold that can be constructed as a $(\Z/2\Z) \oplus (\Z/2\Z)$ quotient of a Calabi-Yau 4-fold.
	\item $W^9_1$ is a $\Z/2\Z$ quotient of $L^5_4 \times S^4$.
	\item $X_5$, $X_9$, and $X_{11}$ and their tilded versions are fibrations of real projective spaces over real projective spaces with appropriate duality bundle distinguished by the embedding of $D_8$ into $D_{16}$, that we refer to as Arcana.
\end{itemize}

\newpage

\section{Duality groups}
\label{app:groups}

In this Appendix we summarize the various duality groups and their individual group structure, see also \cite{Debray:2021vob}.

The first group we discuss is $\SL(2,\Z)$ which can be understood as the bosonic duality group of type IIB. It is generated by two elements
\renewcommand{\arraystretch}{1}
\begin{align}
U = \begin{pmatrix} 0 & -1 \\ 1 & 1\end{pmatrix} \,, \quad S = \begin{pmatrix} 0 & -1 \\ 1 & 0 \end{pmatrix} \,,
\end{align}
of finite order
\begin{align}
S^4 = \text{id} \,, \quad U^6 = \text{id} \,.
\end{align}
It can be defined as
\begin{align}
\SL (2,\Z) = \langle U, S | \, S^4 = \text{id} \,, \enspace S^2 = U^3 \rangle = (\Z/6\Z) \ast_{(\Z/2\Z)} (\Z/4\Z) \,,
\end{align}
which makes the structure as an amalgamated free product evident, see also \cite{Seiberg:2018ntt,Hsieh:2019iba,Hsieh:2020jpj}. The usual generator $T$ of infinite order is obtained as follows
\begin{align}
T = S^{-1} U = \begin{pmatrix} 1 & 1 \\ 0 & 1 \end{pmatrix} \,.
\end{align}
The Abelianization of $\SL(2,\Z)$ is given by
\begin{align}
\text{Ab}\big( \SL (2,\Z) \big) = \Z/{12}\Z \,,
\end{align}
which is generated by the image of $T$ in $\text{Ab}\big( \SL (2,\Z) \big)$.

After we include fermions we need to lift the action of the duality group to the fermionic degrees of freedom. This leads to the second group $\Mp(2,\Z)$ the metaplectic group, see \cite{Pantev:2016nze, Hsieh:2020jpj}. Defining the generators as $\hat{U}$ (order $12$) and $\hat{S}$ (order $8$) it is given by
\begin{align}
\label{mp_amalg}
\Mp(2,\Z) = \langle \hat{U}, \hat{S} | \, \hat{S}^8 = \text{id} \,, \enspace \hat{S}^2 = \hat{U}^3 \rangle = (\Z/{12}\Z) \ast_{(\Z/4\Z)} (\Z/8\Z) \,.
\end{align}
The Spin lift is encoded in
\begin{align}
\hat{S}^4 = \hat{U}^6 = (-1)^F \,.
\end{align}
Its Abelianization is given by
\begin{align}
\text{Ab}\big( \Mp (2,\Z) \big) = \Z/{24}\Z \,,
\end{align}
generated by the image of $\hat{T} = \hat{S}^{-1} \hat{U}$.

Finally, we also include the reflection operator
\begin{align}
R = \begin{pmatrix} 0 & 1 \\ 1 & 0 \end{pmatrix}
\end{align}
to define the group
\begin{align}
\label{GL_amalgamation}
\GL(2,\Z) = \langle U, S, R \, | \, S^4 = \text{id} \,, \, R S R^{-1} = S^{-1} \,, \, R U R^{-1} = U^{-1} \,, \, S^2 = U^3 \rangle = D_{12} \ast_{D_4} D_{8} \,,
\end{align}
where $D_{2n}$ denotes the dihedral (not binary dihedral) group of $2n$ elements. For example, one has
\begin{align}
D_8 = \langle S, R \, | S^4 = \text{id} \,, \enspace R^2 = \text{id} \,, \enspace R S R^{-1} = S^{-1} \rangle \,.
\end{align}
The Pin$^+$ cover of $\GL(2,\Z)$ is defined as
\begin{align}
\begin{split}
\GL^+(2,\Z) &= \langle \hat{U}, \hat{S}, \hat{R} \, | \, \hat{S}^8 = \text{id}, \enspace \hat{R}^2 = \text{id} \,, \enspace \hat{R} \hat{S} \hat{R}^{-1} = \hat{S}^{-1} \,, \enspace \hat{R} \hat{U} \hat{R}^{-1} = \hat{U}^{-1} \,, \enspace \hat{S}^2 = \hat{U}^3 \rangle \\
&= D_{16} \ast_{D_{8}} D_{24} \,,
\end{split}
\label{eq:GL+group}
\end{align}
where $\hat{R}$ denotes the Pin$^+$ lift of $R$.
Note that the Pin$^+$ cover precisely means that the lift of the reflections $\hat{R}$ squares to $\text{id}$ and not $(-1)^F$ as would be the case for the Pin$^-$ cover. In both cases the Abelianization is given by
\begin{align}
\text{Ab} \big( \GL (2,\Z) \big) = \text{Ab} \big( \GL^+ (2,\Z) \big) = (\Z/2\Z) \oplus (\Z/2\Z) \,,
\end{align}
generated by the images of $R, S$ and $\hat{R}, \hat{S}$, respectively.

\subsection{Fermion parity and worldsheet orientation reversal}

One can also consider the action of worldsheet orientation reversal $\Omega$ given by
\begin{align}
\Omega = \begin{pmatrix} 1 & 0 \\ 0 & -1 \end{pmatrix} .
\end{align}
This can be deduced from the action on the 2-form fields of type IIB that transform as a $\GL(2,\Z)$ doublet $(C_2, B_2)^T$. Similarly, one finds the representation of the left-handed fermion parity $(-1)^{F_L}$ in terms of a $\GL(2,\Z)$ element
\begin{align}
(-1)^{F_L} = \begin{pmatrix} -1 & 0 \\ 0 & 1\end{pmatrix} \,.
\end{align}
With this one can see that
\begin{align}
S \, \Omega \, S^{-1} = (-1)^{F_L} \,,
\end{align}
implying that $\Omega$ and $(-1)^{F_L}$ are S-dual. The relation with the reflection $R$ is given by
\begin{align}
S \, R = (-1)^{F_L} \,, \quad R \, S = \Omega \,,
\end{align}
from which we can define the associated Spin-lifts to elements in $\GL^+ (2,\Z)$
\begin{align}
\hat{S} \, \hat{R} = (-1)^{\hat{F}_L} \,,\quad \hat{R} \, \hat{S} = \hat{\Omega} \,,
\end{align}
which both square to the identity because of \eqref{eq:GL+group}.

Both these elements leave the dilaton field invariant but produce a sign flip for the RR 0-form so that on the axio-dilaton,
we have the transformation rule
\begin{align}
\Omega \text{ or } (-1)^{F_L}: \quad \tau \rightarrow - \overline{\tau} \,.
\end{align}
The fact that this involves complex conjugation generates obstructions for the definition of local Calabi-Yau manifolds involving reflections.

Note also that while the chiral 4-form $C_4$ of type IIB string theory is invariant under $\SL(2,\Z)$ transformations, it changes sign under the reflection in $\GL(2,\Z)$. Moreover, as mentioned above the two 2-form field of type IIB transform as doublet under the $\GL(2,\Z)$ symmetry and therefore $(p,q)$-strings and their magnetic dual 5-branes will be affected by non-trivial backgrounds.

\subsection{Embeddings of dihedral groups}
\label{app:embdihedral}

Some of the generators for Spin-$\GL^+ (2,\Z)$ manifolds originate from Spin-$D_4$ or Spin-$D_8$ manifolds via embedding of these groups in $\GL^+(2,\mathbb{Z})$. There is more than one possible embedding of $D_4$ and $D_8$ into $D_{16}$; each will lead to different Spin-$D_{16}$ structures and consequently to different Spin-$\GL^+ (2, \Z)$ manifolds. This results in different bordism classes realized by the same underlying manifold with different bundles on top. In this Appendix we describe these embeddings and how to characterize them.

Let us first focus on the embedding of $D_4 \cong (\Z/2\Z) \oplus (\Z/2\Z)$ into $D_{16}$. The latter can be realized as the symmetry group of
the octagon. When embedding $(\Z/2\Z) \oplus (\Z/2\Z)$ here, one  $\Z/2\Z$ will always be associated with a $180^\circ$ rotation, which we will call $s_4$, while the other $\Z/2\Z$ describes reflections and is denoted by $r_4$. The first embedding $i_4$ is  given by
\begin{align}
i_4\colon \quad s_4 \mapsto \hat{S}^4 \,, \enspace r_4 \mapsto \hat{\Omega} \,.
\end{align}
One can easily show that this has the correct group structure by applying the defining identities in
\eqref{eq:GL+group}.  Geometrically, this is embedding the group as that generated by reflection along one of the
octagon sides and rotation by 180$^\circ$. The second embedding modifies the reflections as follows
\begin{align}
\tilde{\imath}_4\colon \quad {s}_4 \mapsto \hat{S}^4 \,, \enspace {r}_4 \mapsto \hat{S} \, \hat{\Omega} \,,
\end{align}
with the correct group structure
\begin{align}
r_4^2 = \hat{S} \, \hat{\Omega} \, \hat{S} \, \hat{\Omega} = \hat{S} \, \hat{S}^{-1} = \text{id} \,, \enspace s_4 r_4 = \hat{S}^5 \, \hat{\Omega} = \hat{S} \, \hat{\Omega} \, \hat{\Omega} \, \hat{S}^4 \, \hat{\Omega} = \hat{S} \, \hat{\Omega} \, \hat{S}^{-4} = \hat{S} \, \hat{\Omega} \, \hat{S}^{4} = r_4 s_4 \,,
\end{align}
where we use that
\begin{align}
\hat{\Omega} \, \hat{S} \, \hat{\Omega} = \hat{R} \, \hat{S}^2 \, \hat{R} \, \hat{S} = \hat{S}^{-1} \,,
\end{align}
and therefore can replace $\hat{R}$ in the definition of the dihedral group. This second embedding is similar, but
sends $r_4$ to a reflection along an octagonal diagonal, rather than a side.

In similar fashion, we find two different embeddings of $D_8$ into $D_{16}$, given by
\begin{align}
i_8\colon \quad s_8 \mapsto \hat{S}^2 \,, \enspace r_8 \mapsto \hat{\Omega} \,,
\end{align}
and
\begin{align}
\tilde{\imath}_8\colon \quad s_8 \mapsto \hat{S}^2 \,, \enspace r_8 \mapsto \hat{S} \, \hat{\Omega} \,,
\end{align}
with reflection $r_8$ and rotation $s_8$, respectively. For $\tilde{\imath}_8$ we verify the defining property
\begin{align}
\tilde{r}_8 \, \tilde{s}_8 \, \tilde{r}_8 = \hat{S} \, \hat{\Omega} \, \hat{S}^3 \, \hat{\Omega} = \hat{S}^{-2} = \tilde{s}_8^{-1} \,.
\end{align}
Of course we can also define intermediate maps of the form
\begin{align}
\begin{split}
& s_4 \mapsto s_8^2 \mapsto \hat{S}^4 \,, \enspace r_4 \mapsto r_8 \mapsto \hat{\Omega} \,, \\
& s_4 \mapsto s_8^2 \mapsto \hat{S}^4 \,, \enspace r_4 \mapsto r_8 \mapsto \hat{S} \, \hat{\Omega} \,,
\end{split}
\end{align}
to define embeddings of $D_4$ into $D_8$.

Interpreting $D_{2n}$ as the symmetries of an $n$-gon we can picture the different embeddings of the reflection as a reflection with respect to a line through two sides or two corners. This is depicted in Figure \ref{fig:D8emb}, which we reproduce here (see Figure \ref{fig:EMBERS}).
\begin{figure}\centering
\includegraphics[width = 0.65 \textwidth]{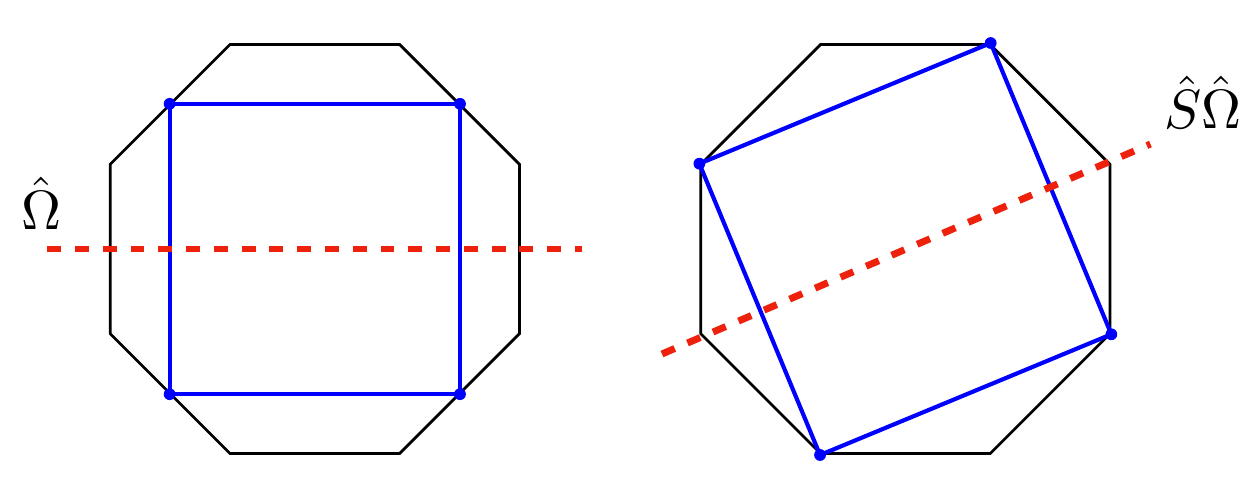}
\caption{Two embeddings of $D_8$ into $D_{16}$ with the action of $\hat{\Omega}$ and $\hat{S} \, \hat{\Omega}$ indicated.}
\label{fig:EMBERS}
\end{figure}
From this we also see that all other embeddings that map the reflection to an element $\hat{S}^k \, \hat{\Omega}$ can be derived from the embeddings above by conjugation with a rotation, i.e., shifting the axis of reflection by a rotation defined by a multiple of $\hat{S}$. As we have seen this is the case for $(-1)^{\hat{F}_L}$ as S-dual of $\hat{\Omega}$.

\section{$\boldsymbol{\eta}$-invariants}
\label{app:eta}

In this Appendix we summarize the calculation of $\eta$-invariants on lens spaces and lens space bundles. These are
often useful in order to guarantee that certain manifolds generate the non-trivial bordism classes we are looking
for and can resolve certain extension questions in spectral sequence computation. They can also be used in order to
demonstrate whether different lens spaces, i.e., lens spaces derived from the quotient of a different group action,
are bordant. For that we need the $\eta$-invariants to be bordism invariants, which is the case for dimension $(4n+1)$ but not in dimension $(4n-1)$, since the $(4n)$-dimensional index density has pure gravitational contributions. For that reason we use the difference of two $\eta$-invariants with different charges in dimension $(4n-3)$ as bordism invariants for which the pure gravitational contribution to the index density vanishes, see also the discussion in Section \ref{subsec:exteta} below.

\subsection[$\eta$-invariants for lens spaces]{$\boldsymbol{\eta}$-invariants for lens spaces}
\label{subapp:etalens}

We start with the discussion of $\eta$-invariants for lens spaces. We will closely follow the description in \cite{Hsieh:2020jpj} with generalizations to include the $\eta$-invariants for Rarita-Schwinger fields of spin $\tfrac{3}{2}$ as well as general $\Z / k \Z$ actions. In most of the discussion we are interested in the fractional part of the $\eta$-invariant only and the results are understood mod $\Z$.

In general we define the lens spaces as
\begin{align}
L^{2n-1}_k (j_1 \,, \dots \,, j_n) = S^{2n-1} / (\Z/k\Z) \,,
\end{align}
where $j_i$ are integers co-prime to $k$; these manifolds can be understood as the asymptotic boundary of $\mathbb{C}^n / (\Z/k\Z)$. The $\Z/k\Z$ acts on the complex coordinates $z_i$ of $\mathbb{C}^n$ as
\begin{align}
z_i \mapsto \exp (2 \pi i j_i / k) z_i \,.
\end{align}
This can be summarized in an action on the tangent bundle of $\R^{2n} \cong \mathbb{C}^n$ by a $(2n \times 2n)$-matrix $\tau$ with
\begin{align}
\det \big(\text{id} - \tau(\ell) \big) = \big|1 - e^{2 \pi i \ell j_1 / k}\big|^2 \dots \big| 1 - e^{2 \pi i \ell j_n / k} \big|^2
\label{eq:tangact}
\end{align}
Moreover, this naturally defines a $\Z/k\Z$ bundle given by the fibration structure
\begin{align}
\begin{split}
\Z/k\Z \hookrightarrow \, & S^{2n-1}  \\
& \downarrow \\
& L^{2n-1}_k(j_1\,,\dots\,, j_n)
\end{split}
\label{eq:natZkbundle}
\end{align}
which in many of our generators can be interpreted as a non-trivial duality bundle. Dirac fermions $\psi_q$ of charge $q$ under this $\Z/k\Z$ bundle transform as follows
\begin{align}
\Z/k\Z: \quad \psi_q \mapsto e^{- 2 \pi i q/ k} \exp \big(- \tfrac{\pi}{k} (j_1 \, \Gamma^1 \Gamma^2 + \dots + j_n \, \Gamma^{2n-1} \Gamma^{2n}) \big) \, \psi_q \,.
\label{eq:Dact}
\end{align}
defining a Spin lift of the group action. Here, $\Gamma^\ell$ are $\gamma$-matrices on $\mathbb{C}^n \cong \R^{2n}$ and we further define
\begin{align}
\overline{\Gamma} = i^{-n} \Gamma^1 \dots \Gamma^{2n} \,.
\end{align}
For this action to be well-defined one must have
\begin{align}
e^{- 2\pi i q} \exp \big( - \pi (j_1 \, \Gamma^1 \Gamma^2 + \dots + j_n \, \Gamma^{2n-1} \Gamma^{2n}) \big) = 1 \,,
\end{align}
which demands $q$ to potentially be half-integer. This can also be summarized in a matrix $\rho(\ell)$ acting on the Spin bundle.

\subsubsection*{Spin structures on lens spaces}

Before we describe the calculation of various $\eta$-invariants for lens spaces in detail we want to describe the different Spin structures on them. For that we look at the lens spaces of the form
\begin{align}
	L_4^{4n-1}\coloneqq L^{4n-1}_4 (1 \,, \dots \,, 1) \,.
\end{align}
These are Spin, but allow for two different Spin structures. These can be identified by considering the lift of the rotational action as discussed in \eqref{eq:Dact} to spinor fields. Switching off the duality bundle for now (i.e., setting $q = 0$), we see that there is more than one lift that defines a well-defined $\Z / 4 \Z$ action
\begin{align}
\psi \mapsto \alpha \, \exp \big( - \tfrac{\pi}{4} (\Gamma^1 \Gamma^2 + \dots + \Gamma^{4n-1} \Gamma^{4n}) \big) \, \psi \,,
\end{align}
where $\alpha$ takes into account an ambiguity that was absorbed in the duality piece $e^{- 2 \pi i q /k}$, in \eqref{eq:Dact}. In order to define a $\Z / 4 \Z$ action $\alpha$ has to satisfy
\begin{align}
\alpha^4 = 1 \quad \rightarrow \alpha  = e^{- 2 \pi i \kappa / 4} \,.
\end{align}
The choice of $\kappa$ defines the different possible lifts of the rotational action to the spinor bundle, i.e.,
the different Spin structures. Reinstating the duality bundle we further see that this choice can be incorporated
in the shift of the charge of the associated fermion. On a manifold with Spin structure specified by $\kappa$ a
fermion of charge $q$ thus transforms as if it had charge $\kappa + q$ on the manifold specified by Spin structure
$\kappa = 0$. Since the Spin structures are classified by $H^1 (L^{4n-1}_4; \Z/2\Z)$ one further has to restrict to $\alpha$'s that are $\pm 1$ ,i.e., only the two values $\kappa = \{0,2\}$ correspond to distinct Spin structures.

Very similar statements hold for lens space bundles as well as lens spaces that do not allow for a Spin structure but a Spin-$\Z / 8\Z$ structure, as for example $L^5_4$. In all of these cases the various structures can be probed by varying the effective charge of the spinors, which in the case of twisted Spin structures might be subject to consistency constraints.

\subsubsection*{Dirac fields}

Let us start the discussion with the usual Dirac fields with spin $\tfrac{1}{2}$, whose $\eta$-invariants we denote as $\eta^\text{D}$. The equivariant index theorem \cite{Don78,Hsieh:2020jpj} allows to deduce the fractional part of the $\eta$ invariants in terms of transformation property of the fermions at the singular point in $\mathbb{C}^n/(\Z/k\Z)$
\begin{align}
\eta^{\text{D}}_q \big(L^{2n-1}_k \big) = - \frac{1}{k} \bigg( \sum_{\ell = 1}^{k-1} \frac{\text{tr} \big( \overline{\Gamma} \rho(\ell) \big)}{\det (\text{id} - \tau(\ell))} \bigg) \,,
\end{align}
where the matrices $\tau$ and $\rho$ are the ones defined above. For the action \eqref{eq:Dact} one has
\begin{align}
\text{tr} \big( \overline{\Gamma} \rho (\ell) \big) = e^{- 2 \pi i \ell q/k} \big( e^{- \pi i \ell j_1/k} - e^{\pi i \ell j_1/k}\big) \dots \big( e^{- \pi i \ell j_n/k} - e^{\pi i \ell j_n/k}\big) \,.
\end{align}
Together with \eqref{eq:tangact} this implies
\begin{align}
\eta^{\text{D}}_q \big(L^{2n-1}_k \big) = - \frac{1}{(2i)^n k} \sum_{\ell = 1}^{k-1} \frac{e^{- 2 \pi i \ell q/k}}{\sin (\pi \ell j_1/k) \big) \dots \sin (\pi \ell j_n/k)} \,,
\label{eq:etaDlens}
\end{align}
which coincides with the formula in \cite{Hsieh:2020jpj} for all $j_i$ equal.

\subsubsection*{Rarita-Schwinger fields}

We present two ways to determine $\eta^{\text{RS}}_q$, the $\eta$-invariant for the Rarita-Schwinger operators of $\Z/k\Z$ charge $q$. One can deduce it from the fact that a Rarita-Schwinger field $\Psi_{\mu}$ transforms in the tensor product of the Spin and the tangent bundle. This modifies the matrix $\rho_{\text{RS}}(\ell)$ to also incorporate the rotation of the vector index, giving
\begin{align}
\begin{split}
\text{tr} \big( \overline{\Gamma} \rho_{\text{RS}} (\ell) \big) &= e^{- 2 \pi i \ell q / k} \big( e^{- \pi i \ell j_1/k} - e^{\pi i \ell j_1/k}\big) \dots \big( e^{- \pi i \ell j_n/k} - e^{\pi i \ell j_n/k}\big) \\
& \enspace \times \big( 2 \cos(2 \pi \ell j_1/k) + \dots + 2 \cos(2 \pi i \ell j_n/k) - 1\big) \,.
\end{split}
\label{rer44}\end{align}
The minus one in the last factor roughly accounts for the fact that the tangent space of $L^{2n-1}_k$ is one dimension smaller than that of $\mathbb{C}^n$, see also \eqref{eq:Ltangdecomp} below. More formally, calculations using the equivariant index theorem such as these essentially compute the $\eta$ invariant by finding a one-higher dimensional manifold with boundary in which the index can be computed via the APS index theorem. A Rarita-Schwinger operator in $(d+1)$ dimensions restricts to a $d$-dimensional Rarita-Schwinger + a Dirac contribution. The extra $-1$ in \eqref{rer44} is precisely removing the additional Dirac contribution. Putting everything together, the $\eta$-invariant for the Rarita-Schwinger operator is then given by
\begin{align}
\begin{split}
\eta^{\text{RS}}_q \big( L^{2n-1}_k \big) = - \frac{1}{(2i)^n k} \sum_{\ell = 1}^{k-1} \bigg( & \frac{e^{- 2 \pi i \ell q/k}}{\sin (\pi \ell j_1/k) \big) \dots \sin (\pi \ell j_n/k)}  \\ & \times \big( 2 \cos(2 \pi \ell j_1/k) + \dots + 2 \cos(2 \pi i \ell j_n/k) - 1\big) \bigg) \,.
\end{split}
\label{eq:etaRSlens}
\end{align}
Alternatively, one can derive the RS $\eta$-invariant by noting that the complexification of the tangent bundle of the lens space decomposes as
\begin{align}
\big( TL^{2n-1}_k \otimes \mathbb{C} \big) \oplus \mathbb{C} \simeq \bigoplus_{i = 1}^n \big( \mathcal{L}^{j_i} \oplus \mathcal{L}^{-j_i} \big) \,,
\label{eq:Ltangdecomp}
\end{align}
where $\mathcal{L}$ is the complex line bundle associated to a $\Z/k\Z$ bundle of the form defined in \eqref{eq:natZkbundle} with all $j_i$ set to $1$. Tensoring with $\mathcal{L}^q$, see also \cite{Debray:2021vob}, we have
\begin{align}
\big( TL^{2n-1}_k \otimes \mathcal{L}^q \big) \oplus \mathcal{L}^q \simeq \bigoplus_{i = 1}^n \big( \mathcal{L}^{q+j_i} \oplus \mathcal{L}^{q-j_i} \big) \,,
\end{align}
which implies
\begin{align}
\eta^{\text{RS}}_q \big( L^{2n-1}_k \big) = \sum_{i =1}^{n} \Big( \eta^{\text{D}}_{q+j_i} \big( L^{2n-1}_k \big) + \eta^{\text{D}}_{q-j_i} \big( L^{2n-1}_k \big) \Big) - \eta^{\text{D}}_{q} \big( L^{2n-1}_k \big) \,.
\end{align}
Plugging in \eqref{eq:etaDlens} we recover the result in \eqref{eq:etaRSlens}. Setting all $j_i$ to $1$ we once more find a perfect agreement with the formulae in \cite{Debray:2021vob}.

\subsubsection*{Signature operator}

For completeness we also state the formula for the $\eta$-invariant associated to the signature, which can be derived from the bi-spinor field, as discussed in \cite{Hsieh:2020jpj}.

The $\rho$ matrix for the signature with charge $q$ is again modified according to its action on bi-spinor fields and one has
\begin{align}
\text{tr} \big( \overline{\Gamma} \rho_{\text{sig}} (\ell) \big) =  e^{- 2 \pi i \ell q /k} \prod_{i = 1}^{n} \big( e^{- \pi i \ell j_1 / k} - e^{\pi i \ell j_1 / k} \big) \big( e^{- \pi i \ell j_1 / k} + e^{\pi i \ell j_1 / k} \big) \,,
\end{align}
with which one finds
\begin{align}
\eta^{\text{sig}}_{q} \big( L^{2n-1}_k \big) = - \frac{1}{(i)^n \, k} \sum_{\ell = 1}^{k-1} \frac{e^{2 \pi i \ell q / k}}{\tan (\pi \ell j_1/ k) \dots \tan (\pi \ell j_n/ k)} \,,
\end{align}
in agreement with results for $j_i = 1$.

\subsection[$\eta$-invariants for lens space bundles]{$\boldsymbol{\eta}$-invariants for lens space bundles}
\label{ss:eta_Q}

We can also describe the fractional part of $\eta$-invariants for lens space bundles described as manifolds $Q^{2n-1}_k$ in the main text. Let us briefly recall the construction of these lens space bundles.

Starting with the base manifold $B = \CP^1$ we define the direct sum of $(n-1)$ complex line bundles, which can be described as tensor product of the hyperplane bundle $H = \mathcal{O}(1)$
\begin{align}
\mathcal{L}_1 \oplus \dots \oplus \mathcal{L}_{n-1} = H^{m_1} \oplus \dots H^{m_{n-1}} \,.
\end{align}
At each point on $\CP^1$ the fiber is therefore given by $\C^{n-1}$. Next, we take the sphere bundle of this sum of line bundles and define a $\Z/k\Z$ action as usual, acting fiberwise as
\begin{align}
\Z/k\Z: \quad z_i \mapsto e^{2 \pi i j_i/k} z_i \,,
\end{align}
with $z_i$ describing the fiber coordinate of the $i^{\text{th}}$ line bundle. The $\eta$-invariant of such a space is given by \cite{BG87b,Hsi18}
\begin{align}
\begin{split}
\eta^{\text{D}}_q \big( Q^{2n-1}_k (j_1\,, \dots j_{n-1}; m_1 \,, \dots m_{n-1}) \big) = \, - \frac{1}{k (-2 i)^n} \sum_{\ell = 1}^{k-1} &  \frac{e^{2 \pi i q \ell / k}}{\text{sin} (\pi \ell j_1/k) \dots \text{sin} (\pi \ell j_{n-1} / k) } \\
& \times \sum_{r = 1}^{n-1} m_r \, \text{cot} (\pi \ell j_r/k) \,.
\end{split}
\end{align}
Applied to the case $(m_1, \dots, m_{n-1}) = (\pm 2, 0, \dots, 0)$ this yields
\begin{align}
\eta^{\text{D}}_q \big( Q^{2n-1}_k (j_1\,, \dots j_{n-1}) \big) = \mp \frac{2}{k (-2i)^n} \sum_{\ell = 1}^{k-1} \frac{e^{2 \pi i q \ell/k} \, \text{cos}(\pi \ell j_1/k)}{\text{sin} (\pi \ell j_1/k)^2 \, \text{sin} (\pi \ell j_2/k) \dots \text{sin} (\pi \ell j_{n-1}/k)} \,,
\label{eq:Qeta}
\end{align}
with which we can resolve some extension problems.

Alternative methods to evaluate these $\eta$-invariants include adiabatic limits of equivariant $\eta$-invariants, see e.g., \cite{5f7204f19b854838af285f506c7aad44, 2002math......3269G, +2000+181+236, 2010arXiv1011.4766G}.

\subsection[Evaluation of $\eta$-invariants to resolve extension questions]{Evaluation of $\boldsymbol{\eta}$-invariants to resolve extension questions}
\label{subsec:exteta}

Since the spectral sequences describe a filtration of the groups under investigation one is often left with extension problems. One way to resolve them is to find bordism invariants such as (the difference of) $\eta$-invariants that are sufficiently fine to account for the non-trivial extensions. For that we evaluate the bordism invariants on our candidate for the generator, if one obtains a result of the form $1/k$ one can be certain that this manifold generates at least a subgroup of order $k$, i.e., $\Z/k\Z$, which often fixes the group completely.

As mentioned above to apply this techniques it is crucial to use bordism invariants. From the APS index theorem \cite{Atiyah:1975jf, APS2, Atiyah:1976qjr} one sees that for $X = \partial Y$
\begin{align}
\eta^{\text{op}}_q (X) = \text{Index}^{\text{op}}_q (Y) - \int_Y I^{\text{op}} \,,
\end{align}
where $\text{op}$ can stand for D, RS, or sig, and $I^{\text{op}}$ denotes the usual index density of the
associated operator. This means that the fractional part of the $\eta$-invariants for two bordant manifolds, i.e.,
described as boundaries of a manifold in one higher dimension $Y$, can differ by contributions from the index
density. Since all background symmetries besides those for the tangent bundle are discrete, the index density is
given in terms of pure gravitational contributions. Since these gravitational contributions are non-vanishing only in dimension $4n$, $n \in \Z$, we see that $\eta$-invariants are bordism invariants in dimension other than $4n-1$. In dimension $4n-1$, one can consider the difference of two $\eta$-invariants with different charges for which the index density contribution cancels.

We can apply this to the generators of $\Omega^{\Spin}_k (B \Z / 3 \Z)$ generated by the lens spaces $L^{2n-1}_3$ described above. For even $n$ one has to consider the difference of two $\eta$-invariants, since a single one can have contributions from $\hat{A}$ in $4 m$ dimensions. For odd $n$ one further needs to use charges of the form $q = Q + \tfrac{1}{2}$, with $Q \in \Z$ to account for a properly defined $\Z / 3 \Z$ action on the fermions. The results are summarized in Table \ref{tab:L3bordinv}.

\renewcommand{\arraystretch}{1.5}
\begin{table}[h!]
\centering
\begin{tabular}{c c}
\toprule
Generator & Bordism invariant \\
\midrule
$L^3_3$  & $\eta^{\text{D}}_1 - \eta^{\text{D}}_0 = - \tfrac{1}{3}$ \\
$L^5_3$  & $\eta^{\text{D}}_{1/2} = - \tfrac{1}{9}$\\
$L^7_3$  & $\eta^{\text{D}}_1 - \eta^{\text{D}}_0 = \tfrac{1}{9}$\\
$L^9_3$  & $\eta^{\text{D}}_{1/2} = \tfrac{1}{27}$\\
$L^{11}_3$  & $\eta^{\text{D}}_1 - \eta^{\text{D}}_0 = - \tfrac{1}{27}$ \\
\bottomrule
\end{tabular}
\caption{Bordism invariants of some lens spaces of the form $L^{2n-1}_3$ which can be used to determine some extension problems in finding $\Omega_k^\Spin(B\Z/3 \mathbb{Z})$.}
\label{tab:L3bordinv}
\end{table}
With these we see that the lens spaces indeed generate the full summands in the associated bordism groups.

The same consideration can be applied for the lens spaces $L^{2n-1}_4$ that generate part of the bordism groups at prime $2$. The corresponding bordism invariants are summarized in Table \ref{tab:L4bordinv}.
\begin{table}[h!]
\centering
\begin{tabular}{c c}
\toprule
Generator & Bordism invariant \\
\midrule
$L^3_4$  & $\eta^{\text{D}}_1 - \eta^{\text{D}}_0 = - \tfrac{3}{8}$ \\
$L^5_4$  & $\eta^{\text{D}}_{1/2} = - \tfrac{5}{32} \,, \enspace \eta^{\text{RS}}_{1/2} = \tfrac{11}{32}$\\
$\widetilde{L}^5_4$  & $\eta^{\text{D}}_{3/2} = - \tfrac{3}{32} \,, \enspace \eta^{\text{RS}}_{3/2} = - \tfrac{3}{32}$\\
$L^9_4$  & $\eta^{\text{D}}_{1/2} = \tfrac{9}{128} \,, \enspace \eta^{\text{RS}}_{1/2} = - \tfrac{19}{128}$\\
$\widetilde{L}^9_4$  & $\eta^{\text{D}}_{3/2} = \tfrac{7}{128} \,, \enspace \eta^{\text{RS}}_{3/2} = \tfrac{3}{128}$ \\
\bottomrule
\end{tabular}
\caption{Bordism invariants of some lens spaces of the form $L^{2n-1}_4$ which can be used to determine some extension problems in finding $\Omega^{\text{Spin-Mp}(2,\Z)}_k (\pt)$ at prime $2$.}
\label{tab:L4bordinv}
\end{table}
Note that in five and nine dimensions we have two different Spin structures. In order to form two linearly
independent generators one needs a second bordism invariant, e.g., the Rarita-Schwinger invariant, as done in
Sections \ref{subsec:Mpcodim6} and \ref{subsec:Mpcodim10}. From the values of the $\eta$-invariants above one can already see that the bordism group contains a $\Z/32\Z$ and $\Z / 128 \Z$ summand, respectively, which already solves the extension problem in Section \ref{mp_spin} for dimension five. In dimension nine one can define a linearly independent bordism invariant that evaluates to zero on one of the Spin structures on the lens spaces and demonstrates that the other Spin structure generators a $\Z/4\Z$ as needed; this is done in the final paragraphs of Section \ref{subsec:Mpcodim10}.

We are left with the generators $Q^{2n-1}_4$ in dimension seven and eleven.
\begin{table}[h!]
\centering
\begin{tabular}{c c}
\toprule
Generator & Bordism invariant \\
\midrule
$Q^7_4$  & $\eta^{\text{D}}_{3/2} - \eta^{\text{D}}_{1/2} = \tfrac{1}{4}$ \\
$Q^{11}_4$  & $\eta^{\text{D}}_{3/2} - \eta^{\text{D}}_{1/2} = - \tfrac{1}{8}$\\
\bottomrule
\end{tabular}
\caption{Bordism invariants of some lens space bundles of the form $Q^{2n-1}_4$ which can be used to determine some extension problems in finding $\Omega^{\text{Spin-Mp}(2,\Z)}_k (\pt)$ at prime $2$.}
\label{tab:Q4bordinv}
\end{table}
Evaluating associated bordism invariants, see Table \ref{tab:Q4bordinv}, we are again able to fix the extension question.

For the determination of the bordism groups of Spin manifolds with SL$(2,\Z)$ bundle we also need some further bordism invariants which are summarized in Table \ref{tab:SL2inv} (the tilde on $\widetilde{L}^7_4$ denotes the alternative Spin structure for lens space as discussed above),
\begin{table}[h!]
\centering
\begin{tabular}{c c}
\toprule
Generator & Bordism invariant \\
\midrule
$Q^5_4$  & $\eta^{\text{D}}_{1} = - \tfrac{1}{4}$ \\
$L^7_4$  & $\eta^{\text{D}}_{1} - \eta^{\text{D}}_{0} = \tfrac{5}{32} \,, \enspace \eta^{\text{D}}_2 - \eta^{\text{D}}_0 = \tfrac{1}{4}$ \\
$\widetilde{L}^7_4$  & $\eta^{\text{D}}_{3} - \eta^{\text{D}}_{2} = - \tfrac{3}{32} \,, \enspace \eta^{\text{D}}_0 - \eta^{\text{D}}_2 = - \tfrac{1}{4}$ \\
$Q^9_4$  & $\eta^{\text{D}}_{1} = \tfrac{1}{8}$ \\
$L^{11}_4$  & $\eta^{\text{D}}_{1} - \eta^{\text{D}}_{0} = - \tfrac{9}{128} \,, \enspace \tfrac{1}{2} \big( \eta^{\text{D}}_2 - \eta^{\text{D}}_0 \big) = - \tfrac{1}{16}$ \\
$\widetilde{L}^{11}_4$  & $\eta^{\text{D}}_{3} - \eta^{\text{D}}_{2} = \tfrac{7}{128} \,, \enspace \tfrac{1}{2} \big( \eta^{\text{D}}_0 - \eta^{\text{D}}_2 \big) = \tfrac{1}{16}$ \\
\bottomrule
\end{tabular}
\caption{Further bordism invariants for the determination of $\Omega^{\Spin}_k \big( B \SL(2,\Z) \big)$.}
\label{tab:SL2inv}
\end{table}
which can be evaluated by the application of the same formulas as above.\footnote{Note that in twelve dimensions the index of a fermion in a real representation is even, see e.g.\ \cite{Hsieh:2020jpj}, and therefore we see that $\tfrac{1}{2} \big( \eta^{\text{D}}_2 - \eta^{\text{D}}_0 \big)$ mod $\Z$ is indeed a bordism invariant.}
\renewcommand{\arraystretch}{1}

\section{The May-Milgram theorem}
\label{s:MM_appendix}
In this Appendix, we discuss the May-Milgram theorem, a key tool for determining differentials in Adams spectral
sequences. We highlight an important nuance that occurs when applying this theorem to the Adams spectral sequence
over $\cA(1)$; this nuance has not always been treated carefully in recent such calculations, and we hope this
appendix provides a clear guide for how to work with the May-Milgram theorem in the setting of the Adams spectral
sequence over $\cA(1)$.\footnote{Though we focused on $\cA(1)$ in this paper, there are other subalgebras of the
Steenrod algebra such that the Adams spectral sequences over those subalgebras compute interesting generalized
homology theories, such as $2$-completed $\mathit{ku}$- or $\mathit{tmf}$-homology: see
Beaudry-Campbell~\cite[Remark 3.2.2]{BC18}. The nuance we highlight regarding the May-Milgram theorem and the
techniques we use to resolve it should apply for these more general examples as well.}

\begin{defn}
Choose an $r\ge 1$. The short exact sequence
\begin{equation}
\label{MM_Bock_SES}
	\shortexact[\cdot 2^r][\bmod 2^r]{\Z/2^r\Z}{\Z/2^{2r}\Z}{\Z/2^r\Z}{}
\end{equation}
induces a long exact sequence in cohomology:
\begin{equation}
	% https://q.uiver.app/?q=WzAsNixbMCwwLCJcXGRvdHNiIl0sWzEsMCwiSF5cXGFzdChcXGJsO1xcWi8yXnIpIl0sWzIsMCwiSF5cXGFzdChcXGJsO1xcWi8yXnsycn0pIl0sWzMsMCwiSF5cXGFzdChcXGJsO1xcWi8yXnIpIl0sWzQsMCwiSF57XFxhc3QrMX0oXFxibDtcXFovMl5yKSJdLFs1LDAsIlxcZG90c2IiXSxbMSwyXSxbMiwzXSxbMyw0LCJcXGJldGFfciJdLFswLDFdLFs0LDVdXQ==
\begin{tikzcd}
	\phantom{.} & {H^\ast(\bl;\Z/2^r\Z)} & {H^\ast(\bl;\Z/2^{2r}\Z)} & {H^\ast(\bl;\Z/2^r\Z)} & {H^{\ast+1}(\bl;\Z/2^r\Z)} & \phantom{.}
	\arrow[from=1-2, to=1-3]
	\arrow[from=1-3, to=1-4]
	\arrow["{\beta_r}", from=1-4, to=1-5]
	\arrow[from=1-1, to=1-2]
	\arrow[from=1-5, to=1-6]
\end{tikzcd}
\end{equation}
We let $\beta_r$ be the connecting homomorphism in this long exact sequence; it is called a \term{Bockstein
homomorphism}.
\end{defn}
\eqref{MM_Bock_SES} is compatible with the multiplication-by-$2^r$ short exact sequence
\begin{equation}
\label{Bock_sseq}
	\shortexact[\cdot 2^r][\bmod 2^r]{\Z}{\Z}{\Z/2^r\Z},
\end{equation}
so there is a commutative diagram of long exact sequences
\begin{equation}
\label{LES_diag}
\begin{gathered}
% https://q.uiver.app/?q=WzAsMTIsWzEsMCwiSF5rKFxcdGV4dHstLX07XFxaKSJdLFsyLDAsIkheayhcXHRleHR7LS19O1xcWikiXSxbMywwLCJIXmsoXFx0ZXh0ey0tfTtcXFovMl5yXFxaKSJdLFs0LDAsIkheayhcXHRleHR7LS19O1xcWikiXSxbMSwxLCJIXmsoXFx0ZXh0ey0tfTtcXFovMl5yXFxaKSJdLFsyLDEsIkheayhcXHRleHR7LS19O1xcWi8yXnsycn1cXFopIl0sWzMsMSwiSF5rKFxcdGV4dHstLX07XFxaLzJeclxcWikiXSxbNCwxLCJIXmsoXFx0ZXh0ey0tfTtcXFovMl5yXFxaKSJdLFs1LDAsIlxcZG90c2IiXSxbNSwxLCJcXGRvdHNiIl0sWzAsMSwiXFxkb3RzYiJdLFswLDAsIlxcZG90c2IiXSxbMTEsMF0sWzEwLDRdLFswLDEsIlxcY2RvdCAyXnIiXSxbNCw1LCJcXGNkb3QgMl5yIl0sWzEsMiwiXFxibW9kIDJeciJdLFs1LDYsIlxcYm1vZCAyXnIiXSxbMiwzXSxbNiw3LCJcXGJldGFfciJdLFszLDhdLFs3LDldLFswLDRdLFsxLDVdLFsyLDZdLFszLDddXQ==
\begin{tikzcd}
	\dotsb & {H^k(\text{--};\Z)} & {H^k(\text{--};\Z)} & {H^k(\text{--};\Z/2^r\Z)} & {H^k(\text{--};\Z)} & \dotsb \\
	\dotsb & {H^k(\text{--};\Z/2^r\Z)} & {H^k(\text{--};\Z/2^{2r}\Z)} & {H^k(\text{--};\Z/2^r\Z)} & {H^k(\text{--};\Z/2^r\Z)} & \dotsb
	\arrow[from=1-1, to=1-2]
	\arrow[from=2-1, to=2-2]
	\arrow["{\cdot 2^r}", from=1-2, to=1-3]
	\arrow["{\cdot 2^r}", from=2-2, to=2-3]
	\arrow["{\bmod 2^r}", from=1-3, to=1-4]
	\arrow["{\bmod 2^r}", from=2-3, to=2-4]
	\arrow[from=1-4, to=1-5]
	\arrow["{\beta_r}", from=2-4, to=2-5]
	\arrow[from=1-5, to=1-6]
	\arrow[from=2-5, to=2-6]
	\arrow["{\bmod 2^r}", from=1-2, to=2-2]
	\arrow["{\bmod 2^{2r}}", from=1-3, to=2-3]
	\arrow["{\bmod 2^r}", from=1-4, to=2-4]
	\arrow["{\bmod 2^r}", from=1-5, to=2-5]
\end{tikzcd}
\end{gathered}
\end{equation}
Using this diagram, one can show that $\beta_r\ne 0$ acts on $H^{k-1}(X;\Z/2^r\Z)$ precisely when $H^k(X;\Z)$ has
$2^r$-torsion elements. We want to use this to detect direct summands.
\begin{defn}
A \term{key Bockstein class for $2^r$} is an element $x\in H^*(X;\Z/2^r\Z)$ in the image of $\beta_r$, and such that
for all $s < r$, $x\bmod s$ is not in the image of $\beta_s$.
\end{defn}
\begin{lem}
Degree-$k$ key Bockstein classes for $2^r$ are in bijective correspondence with $\Z/2^r\Z$ direct summands in
$H^{k+1}(X;\Z)$.
\end{lem}
The proof is a diagram chase, so we give a quick summary. If $x$ is in the image of $\beta_r$, the preimage of $x$
in $\Z$ cohomology generates a $\Z/2^r\Z$ subgroup. This subgroup might not be a direct summand, but the mod $s$
condition fixes that problem: if $x = 2y$, then $x\bmod 2 = 0$, and $0 = \beta_2(0)$. Thus key Bockstein classes
correspond to direct summands.
\begin{defn}
An \term{$h_0$-tower} in an Adams spectral sequence is an infinite sequence of nonzero elements $x_i\in
E_2^{s+i,t+i}$ such that $h_0 x_i = x_{i+1}$. We do not require this sequence to begin at $i = 0$.
\end{defn}
For example, looking at $\Ext_{\cA(1)}(\Z/2\Z, \Z/2\Z)$ in \cref{pic_ext_Z2}, there are $h_0$-towers in all degrees
$4k$.

Because differentials commute with the $h_0$-action on the $E_r$-page, the behavior of differentials on
$h_0$-towers is very restricted. If $T$ is an $h_0$-tower, one of two things occurs:
\begin{enumerate}
	\item every element of $T$ survives to the $E_\infty$-page, or
	\item there is a single number $r$ such that $d_r$ differentials kill all but finitely many elements of $T$.
\end{enumerate}
In the second situation, $d_r$ goes between $T$ and exactly one other $h_0$-tower $T'$. Because of these two cases,
we often refer to $d_r$s between $h_0$-towers or to entire $h_0$-towers surviving or not surviving to the
$E_\infty$-page. Now we can state the main theorem:
\begin{thm}[May-Milgram~\cite{MM81}]
\label{MM_appendix}
Consider the Adams spectral sequence computing the $2$-completed stable homotopy groups of a space or spectrum $X$.
There is a bijection between nonzero $d_r$ differentials out of $h_0$-towers in topological degree $k$ and key
Bockstein classes for $2^r$ in $H^k(X;\Z/2^r\Z)$.
\end{thm}
The slogan is ``Bocksteins tell us everything about differentials between $h_0$-towers.''
\begin{rem}
The May-Milgram theorem is usually stated differently, identifying Adams differentials between $h_0$-towers and
differentials in the \term{Bockstein spectral sequence} associated to the short exact sequence~\eqref{Bock_sseq}.
The pathway to the way we state it in \cref{MM_appendix} is to solve said Bockstein spectral sequence, as
in~\cite[Section 24.2]{MP12}.
\end{rem}
We would like to use the May-Milgram theorem in the computation of twisted Spin bordism, or, thanks to
Anderson-Brown-Peterson's theorem~\cite{ABP67}, $\ko$-homology. Thus we should take a look at Bocksteins in
$H^*(\ko\wedge X;\Z/2^r\Z)$. This is the important nuance we mentioned above: previous work in the mathematical
physics literature applying the May-Milgram theorem, such as~\cite{Deb21}, has not always been clear about passing
from Bocksteins in $H^*(X)$ to Bocksteins in $H^*(\ko\wedge X)$.

The Künneth formula says that $H^*(\ko\wedge X;\Z/2^r\Z)$ is at least as complicated as the combination $H^*(\ko;\Z/2^r\Z)\otimes
H^*(X;\Z/2^r\Z)$. Mahowald-Milgram~\cite[Corollary 1.3]{MM76} compute $H^*(\ko;\Z_{(2)})$, which is complicated
enough to make direct computation of Bocksteins in $H^*(\ko\wedge X;\Z/2^r\Z)$ look imposing. Bayen~\cite[Chapter
1]{Bay94} provides a general approach using the Hopf algebra structure of $\cA$, but we can get away with an ad hoc
approach.

Let $\cA(0)$ be the subalgebra of $\cA$ generated by $\Sq^1$; then $\cA(0)\cong\Z/2\Z[\Sq^1]/\big((\Sq^1)^2\big)$, i.e.\
this is an exterior algebra. There is an isomorphism\footnote{We do not know who originally proved this.
Ravenel~\cite[Lemma 3.1.11]{Rav86} remarks that it is ``standard.''}
\begin{equation}
	H^*(H\Z;\Z/2\Z)\overset\cong\longrightarrow \cA\otimes_{\cA(0)}\Z/2\Z,
\end{equation}
so just like the change-of-rings theorem provides an Adams spectral sequence beginning with Ext over $\cA(1)$ and
computing $\ko$-homology, there is an even simpler Adams spectral sequence of the form\footnote{In fact, this
Adams spectral sequence is isomorphic to a certain presentation of the Bockstein spectral sequence~\cite[Section
1.4]{BG03}, and May-Milgram's proof of their theorem can be interpreted from this point of view.}
\begin{equation}
	E_2^{s,t} = \Ext_{\cA(0)}^{s,t}(H^*(X;\Z/2\Z), \Z/2\Z) \Longrightarrow H_{t-s}(X;\Z)_2^\wedge.
\end{equation}
The idea behind our ad hoc approach to the May-Milgram theorem is that differentials between Adams towers are easy
to compute over $\cA(0)$, because we generally know the integral homology of $X$; then we pull them back by the map
of Adams spectral sequences induced by the quotient $\cA(1)\to\cA(0)$.

If $M$ is an $\cA(0)$-module which is a finite-dimensional $\Z/2\Z$-vector space in each degree, then $M$ is a
direct sum of shifts of the two $\cA(0)$-modules $\Z/2\Z$ and $\cA(0)$. The Ext groups of these modules are simple
to compute.
\begin{itemize}
	\item $\Ext_{\cA(0)}\big(\cA(0), \Z/2\Z\big)$ consists of a single $\Z/2\Z$ in bidegree $(0, 0)$, as is always the case for
	Ext of an algebra over itself.
	\item $\Ext_{\cA(0)}(\Z/2\Z, \Z/2\Z)$ consists of a single $h_0$-tower in topological degree $0$, which is a
	consequence of Koszul duality~\cite[Example 4.5.5]{BC18}.
\end{itemize}
Therefore a $d_r$ differential between $h_0$-towers in the $\cA(0)$-Adams spectral sequence for $X$ corresponds to
a $\Z/2^r\Z$ summand in the homology of $X$, as we draw in \cref{dr_A0}, and as we noted above such a summand is
equivalent data to a key Bockstein class.
\begin{figure}[h!]
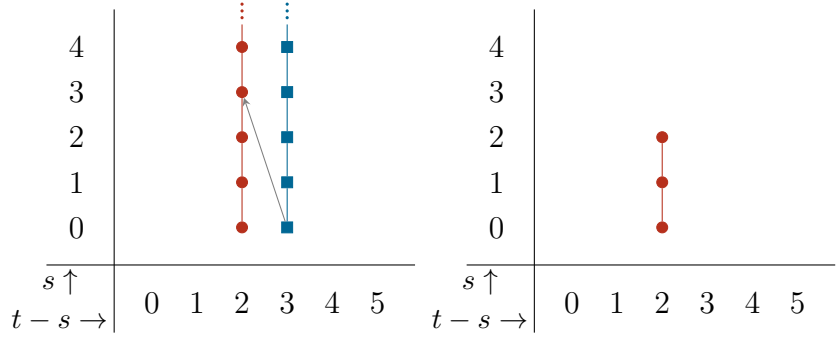

\centering
\begin{subfigure}[c]{0.33\textwidth}
\begin{sseqdata}[name=A0diff, classes=fill, xrange={0}{5}, yrange={0}{4}, scale=0.6, Adams grading, >=stealth,
x label = {$\displaystyle{s\uparrow \atop t-s\rightarrow}$},
x label style = {font = \small, xshift = -15ex, yshift=5.5ex}]
\begin{scope}[BrickRed]
	\class(2, 0)\AdamsTower{}
\end{scope}
\begin{scope}[MidnightBlue, rectangle]
	\class(3, 0)\AdamsTower{}
\end{scope}
\begin{scope}[draw=none, fill=none]
	\class(2, 6)
	\class(2, 7)
\end{scope}
\d[gray]3(3, 0)
\d[draw=none]3(3, 1)
\d[draw=none]3(3, 2)
\d[draw=none]3(3, 3)
\d[draw=none]3(3, 4)
\end{sseqdata}
	\printpage[name=A0diff, page=3]
\end{subfigure}
\begin{subfigure}[c]{0.33\textwidth}
	\printpage[name=A0diff, page=4]
\end{subfigure}
\caption{A differential in an Adams spectral sequence over $\cA(0)$ between two $h_0$-towers, arising from two
$\Z/2\Z$ summands in the $\cA(0)$-module structure on cohomology. The $d_r$ differential on the $E_r$-page (left)
leaves a $\Z/2^r\Z$ summand in the $E_\infty$-page (right). Here $r = 3$.}
\label{dr_A0}
\end{figure}

Recall that $C\eta$ denotes the $\cA$-module consisting of two $\Z/2\Z$ summands in degrees $0$ and $2$ linked by a
$\Sq^2$.
\begin{prop}[May-Milgram for copies of $C\eta$]
\label{MM_ceta}
Let $X$ be a space or spectrum and $\textcolor{BrickRed}{M_0}$ and $\textcolor{MidnightBlue}{M_1}$ be two
$\cA(1)$-module summands of $H^*(X;\Z/2\Z)$ isomorphic to shifts of $C\eta$. Suppose there exist $x\in
\textcolor{BrickRed}{M_0}$ and $y\in\textcolor{MidnightBlue}{M_1}$ of degrees $n$ and $n+1$, respectively, and
classes $\widetilde x, y\in H^*(X;\Z/2^r\Z)$ such that $\widetilde x\bmod 2 = x$, $\widetilde y\bmod 2 = y$,
$\beta_r(\widetilde x) = \widetilde y$, and $\widetilde y$ is a key Bockstein class. Then for all $k\ge 0$, there
is a $d_r$ differential from the $h_0$-tower in degree $n+4k+1$ coming from
$\Ext_{\cA(1)}(\textcolor{MidnightBlue}{M_1})$ to the $h_0$-tower in degree $n+4k$ coming from
$\Ext_{\cA(1)}(\textcolor{BrickRed}{M_0})$.
\end{prop}
The existence of these $h_0$-towers in degrees $n+4k$ and $n+4k+1$ follows from the computation we did in
\cref{LES_ext_exm}.
\begin{proof}
Recall from \cref{LES_ext_exm} that the $h_0$-towers in $\Ext_{\cA(1)}(\Sigma^n C\eta)$ in topological degrees
$n+4k$ and $n+4k+4$ are linked by the action of $v\in\Ext_{\cA(1)}^{3,7}(\Z/2\Z, \Z/2\Z)$: all but finitely many
elements of the higher-degree tower are in the image of this $v$-action applied to the lower-degree tower. This
$v$-action commutes with differentials, so if we have proven the theorem for $k$ we can deduce the theorem for
$k+1$. Thus in the rest of the proof we assume $k= 0$.

As an $\cA(0)$-module, $C\eta\cong \Z/2\Z\oplus \Sigma^2\Z/2\Z$, so $\Ext_{\cA(0)}(C\eta, \Z/2\Z)$ consists of
$h_0$-towers in degrees $0$ and $2$. The map $\cA(1)\to\cA(0)$ induces a map $\phi\colon \Ext_{\cA(1)}(C\eta,
\Z/2\Z)\to \Ext_{\cA(0)}(C\eta, \Z/2\Z)$; thinking of elements of Ext as equivalence classes of extensions of
$\cA(1)$-modules, like we discussed in Section \ref{ss:ext}, $\phi$ takes an extension of $\cA(1)$-modules and forgets
the $\Sq^2$-action to obtain the same extension but of $\cA(0)$-modules. This model for $\phi$ lends itself to
computations: for example, in \cref{from_A1_to_A0}, we show that
\begin{equation}
	\phi^{1,3}\colon \Ext_{\cA(1)}^{1,3}(C\eta, \Z/2\Z)\longrightarrow \Ext_{\cA(0)}^{1,3}(C\eta, \Z/2\Z)
\end{equation}
is nonzero; since both of these Ext groups are isomorphic to $\Z/2\Z$, this means $\phi^{1,3}$ is an isomorphism.
From the Yoneda product description of the action by $h_0$~\eqref{Yoneda_product}, it is possible to show that
$\phi$ commutes with $h_0$-actions, so the calculation in \cref{from_A1_to_A0} implies that in degrees $s\ge 1$,
$\phi$ maps the $h_0$-tower in topological degree $2$ in $\Ext_{\cA(1)}(C\eta)$ isomorphically onto the $h_0$-tower
in topological degree $2$ in $\Ext_{\cA(0)}(C\eta)$.
\begin{figure}[h!]
\centering
% the A1-module xtn
\begin{subfigure}[c]{0.45\textwidth}
\begin{tikzpicture}[scale=0.6, every node/.style = {font=\small}]
\sqone(0, 2);
\begin{scope}[BrickRed]
	\foreach \x in {-4, 0} {
		\tikzpt{\x}{3}{}{};
	}
	\draw[thick, ->] (-3.5, 3) -- (-0.5, 3);
	\node[below] at (-4, -0.25) {$\Sigma^3 \Z/2\Z$};
\end{scope}
\node[below] at (0, -0.25) {$\uQ$};
\begin{scope}[MidnightBlue]
	\foreach \x in {0, 4} {
		\tikzpt{\x}{0}{}{rectangle, minimum size=3.5pt};
		\tikzpt{\x}{2}{}{rectangle, minimum size=3.5pt};
		\sqtwoR(\x, 0);
	}
	\draw[thick, ->] (1, 0) -- (3.5, 0);
	\draw[thick, ->] (1, 2) -- (3.5, 2);
	\node[below] at (4, -0.25) {$C\eta$};
\end{scope}
\end{tikzpicture}
\end{subfigure}
% the A0-module extn
\begin{subfigure}[c]{0.45\textwidth}
\begin{tikzpicture}[scale=0.6, every node/.style = {font=\small}]
\sqone(0, 2);
\begin{scope}[BrickRed]
	\foreach \x in {-4, 0} {
		\tikzpt{\x}{3}{}{};
	}
	\draw[thick, ->] (-3.5, 3) -- (-0.5, 3);
	\node[below] at (-4, -0.25) {$\Sigma^3 \Z/2\Z$};
\end{scope}
\node[below] at (0, -0.25) {$\Z/2\Z\oplus \Sigma^2\cA(0)$};
\begin{scope}[MidnightBlue]
	\foreach \x in {0, 4} {
		\tikzpt{\x}{0}{}{rectangle, minimum size=3.5pt};
		\tikzpt{\x}{2}{}{rectangle, minimum size=3.5pt};
%		\sqtwoR(\x, 0);
	}
	\draw[thick, ->] (0.5, 0) -- (3.5, 0);
	\draw[thick, ->] (0.5, 2) -- (3.5, 2);
	\node[below] at (4, -0.25) {$C\eta$};
\end{scope}
\end{tikzpicture}
\end{subfigure}
\caption{Left: an extension of $\cA(1)$-modules representing the nonzero element $x$ of $\Ext_{\cA(1)}^{1,3}(C\eta,
\Z/2\Z)\cong\Z/2\Z$. Right: if we only remember the action of $\cA(0)$, this extension is still nonsplit, so
$\phi(x)\ne 0$. This computation is part of the proof of \cref{MM_ceta}.}
\label{from_A1_to_A0}
\end{figure}
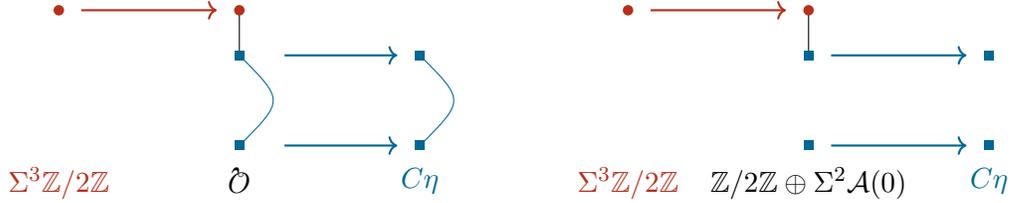
The same is true for the $h_0$-towers in topological degree $0$, which we
leave as an exercise for the reader.\footnote{Another way to compute $\phi$ in both degrees $0$ and $2$ is to
observe that $\Sigma^2 C\eta \cong \widetilde H^*(\CP^2;\Z/2\Z)$, so $\phi$ is the morphism of Adams spectral
sequences corresponding to the map $\psi\colon \ko_*(\CP^2)\to H_*(\CP^2)$, and $\psi$ can be computed with the
Atiyah-Hirzebruch spectral sequence.} Therefore there is a $d_r$ differential
from the $h_0$-tower in degree $n+1$ in $\Ext_{\cA(1)}(\textcolor{MidnightBlue}{M_1})$ to the $h_0$-tower in degree
$n$ in $\Ext_{\cA(1)}(\textcolor{BrickRed}{M_0})$ exactly when the same is true for Ext over $\cA(0)$, and the
latter is true when preimages of $x$ and $y$ in $\Z/2^r\Z$ cohomology are linked by a $\beta_r$ Bockstein and the
lift of $y$ is a key Bockstein class.
\end{proof}
See \cref{MM_ceta_picture} for a picture of the key idea in the proof of \cref{MM_ceta}.
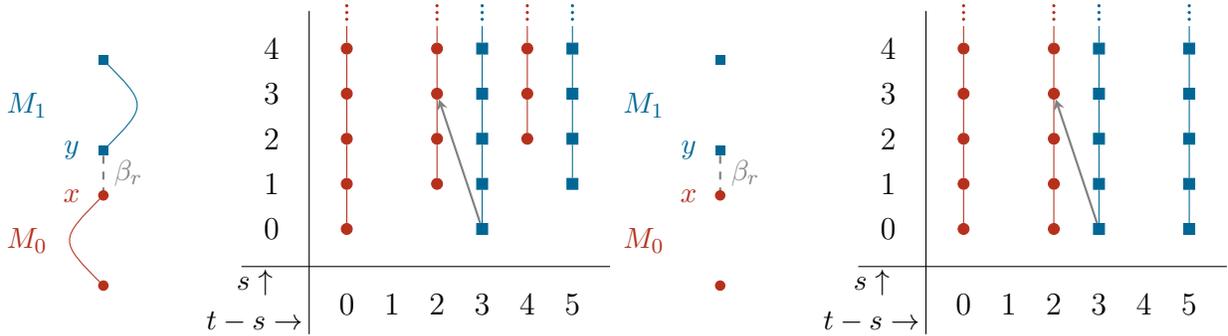
\begin{figure}[h!]
\centering
\begin{subfigure}[c]{0.15\textwidth}
	\begin{tikzpicture}[scale=0.6, every node/.style = {font=\small}]
		\draw[thick, gray, dashed] (0, 2) -- (0, 3);
		\node[right, gray] at (0, 2.5) {$\beta_r$};
		\begin{scope}[BrickRed]
			\tikzpt{0}{0}{}{};
			\tikzpt{0}{2}{$x$}{};
			\sqtwoL(0, 0);
			\node[left] at (-1, 1) {$M_0$};
		\end{scope}
		\begin{scope}[MidnightBlue]
			\tikzpt{0}{3}{$y$}{rectangle, minimum size=3.5pt};
			\tikzpt{0}{5}{}{rectangle, minimum size=3.5pt};
			\sqtwoR(0, 3);
			\node[left] at (-1, 4) {$M_1$};
		\end{scope}
	\end{tikzpicture}
\end{subfigure}
\begin{subfigure}[c]{0.33\textwidth}
\begin{sseqdata}[name=MMA1, classes=fill, xrange={0}{5}, yrange={0}{4}, scale=0.6, Adams grading, >=stealth,
x label = {$\displaystyle{s\uparrow \atop t-s\rightarrow}$},
x label style = {font = \small, xshift = -15ex, yshift=5.5ex}
]
	\begin{scope}[BrickRed]
		\class(0, 0)\AdamsTower{}
		\class(2, 1)\AdamsTower{}
		\class(4, 2)\AdamsTower{}
	\end{scope}
	\begin{scope}[MidnightBlue, rectangle]
		\class(3, 0)\AdamsTower{}
		\class(5, 1)\AdamsTower{}
	\end{scope}
	\d[thick, gray]3(3, 0)
\end{sseqdata}
\printpage[name=MMA1, page=3]
\end{subfigure}
% and now over A0
\begin{subfigure}[c]{0.15\textwidth}
	\begin{tikzpicture}[scale=0.6, every node/.style = {font=\small}]
		\draw[thick, gray, dashed] (0, 2) -- (0, 3);
		\node[right, gray] at (0, 2.5) {$\beta_r$};
		\begin{scope}[BrickRed]
			\tikzpt{0}{0}{}{};
			\tikzpt{0}{2}{$x$}{};
			\node[left] at (-1, 1) {$M_0$};
		\end{scope}
		\begin{scope}[MidnightBlue]
			\tikzpt{0}{3}{$y$}{rectangle, minimum size=3.5pt};
			\tikzpt{0}{5}{}{rectangle, minimum size=3.5pt};
			\node[left] at (-1, 4) {$M_1$};
		\end{scope}
	\end{tikzpicture}
\end{subfigure}
% and then the sseq
\begin{subfigure}[c]{0.33\textwidth}
\begin{sseqdata}[name=MMA0, classes=fill, xrange={0}{5}, yrange={0}{4}, scale=0.6, Adams grading, >=stealth,
x label = {$\displaystyle{s\uparrow \atop t-s\rightarrow}$},
x label style = {font = \small, xshift = -15ex, yshift=5.5ex}
]
	\begin{scope}[BrickRed]
		\class(0, 0)\AdamsTower{}
		\class(2, 0)\AdamsTower{}
	\end{scope}
	\begin{scope}[MidnightBlue, rectangle]
		\class(3, 0)\AdamsTower{}
		\class(5, 0)\AdamsTower{}
	\end{scope}
	\d[thick, gray]3(3, 0)
\end{sseqdata}
\printpage[name=MMA0, page=3]
\end{subfigure}
\caption{Part of the proof of \cref{MM_ceta}. The left-hand side is part of the module structure and Adams spectral
sequence over $\cA(1)$ for $\textcolor{BrickRed}{M_0}\oplus \textcolor{MidnightBlue}{M_1}\subset H^*(X;\Z/2\Z)$; the
right-hand side is the corresponding data over $\cA(0)$. The gray $d_r$ differential on the right can be deduced from the
existence of a $\beta_r$ Bockstein, and implies the gray $d_r$ differential on the left.}
\label{MM_ceta_picture}
\end{figure}
\begin{rem}
\Cref{MM_ceta} is quite narrow in scope, but simple to prove. We suggest a similar approach to other applications
of the May-Milgram theorem to twisted spin bordism questions: compute differentials between $h_0$-towers over
$\cA(0)$, where they correspond to Bocksteins in the cohomology of $X$, then pull them back to Ext over $\cA(1)$ to
learn about $\ko_*(X)$.
\end{rem}

%aWhen running the Adams spectral sequence over $\cA(1)$, as we do in the body of this article, there is
%an
%
%
%The slogan behind this theorem is that ``a nonzero $\Z/2^r$ Bockstein in $H^\ast(X)$ is equivalent to a
%differential between towers in the Adams spectral sequence for $X$;'' Bocksteins in cohomology are usually pretty
%easy to work out. But this slogan has an important caveat: it is true as stated only for the Adams spectral
%sequence computing stable homotopy groups of $X$; when invoking a change-of-rings theorem to compute e.g.\
%$\ko$-homology, the

\section{Some Smith homomorphisms}
\label{Smith_appendix}

Smith homomorphisms are certain maps between bordism groups in which both the dimension and the tangential
structure vary. Specialized to the cases of Spin-$\Mp(2, \Z)$ and Spin-$\GL^+(2, \Z)$ bordism, there are a few
Smith homomorphisms that are interesting and helpful, and the purpose of this appendix is to discuss them. In
Appendix \ref{Smith_general} we introduce Smith homomorphisms and generalities; in Appendix \ref{smith_Mp} we discuss a
sequence of Smith maps between $\Omega_k^\Spin \big(B\SL(2, \Z)\big)$ and $\Omega_k^{\Spin\text{-}\Mp(2, \Z)} (\pt)$; and in
Appendix \ref{smith_GL} we discuss a Smith map for Spin-$\GL^+(2, \Z)$ bordism. This last Smith map plays a key role in our
proof in Section \ref{s_7dorange} that $[\orangeseven] \in \Omega_7^{\Spin\text{-}\GL^+(2, \Z)} (\pt)$ is linearly
independent from $[\halfQseven]$.

\subsection{This charming man(ifold): generalities on the Smith(homomorphism)s}
\label{Smith_general}

In this subsection we review the general setup for Smith homomorphisms in twisted Spin bordism, so that we can
study examples involving Spin-$\Mp(2, \Z)$ and Spin-$\GL^+(2, \Z)$ bordism in the rest of this appendix.
See~\cite{COSY20, HKT20, DDKLPT22} for more detail with an eye towards physics applications.
%\textcolor{red}{TODO: signpost, e.g.\
%notation for twisted $\xi$-bordism}.

Recall from \cref{twisted_defn} that for a space $X$ and vector bundle $V\to X$, an $(X, V)$-twisted Spin structure on a
vector bundle $E\to M$ is a map $f\colon M\to X$ and a Spin structure on $E\oplus f^\ast V$. Suppose $M$ is a
manifold and $TM$ has an $(X, V)$-twisted Spin structure. Choose another vector bundle $W\to X$ of rank $r$ and let
$N\subset M$ be a \term{smooth representative} for the Poincaré dual of $w_r(f^\ast W)$; that is, the image of the
mod $2$ fundamental class of $N$ in $H_\ast(M; \Z/2\Z)$ is Poincaré dual to $w_r(f^\ast W)\in H^\ast(M; \Z/2\Z)$.

The key fact underlying the Smith homomorphism is that in this situation, $N$ has a canonical $(X, V\oplus
W)$-twisted Spin structure, and that this construction factors through bordism classes to define a map
\begin{equation}
	S_W\colon \Omega_{n+V}^\Spin(X)\longrightarrow \Omega_{n-r+V\oplus W}^\Spin(X).
\end{equation}
Now the explanation: if $N$ is a smooth representative of the Poincaré dual of the top Stiefel-Whitney class of a
vector bundle $E\to M$, then the normal bundle $\nu_N$ of $N\inj M$ is isomorphic to $E|_N$. Using the isomorphism
$TN\oplus \nu_N \cong TM|_N$, the $(X, V)$-twisted Spin structure on $M$ induces a Spin structure on
\begin{equation}
	(TM\oplus f^\ast V)|_N \cong TN \oplus\nu_N\oplus f^\ast V|_N \cong TN \oplus E|_N \oplus f^\ast V|_N,
\end{equation}
and plugging in $E = f^\ast W$, we have found an $(X, V\oplus W)$-twisted Spin structure on $N$.

A theorem of~\cite{DDKLPT22} puts the Smith homomorphism into a long exact sequence, making calculating kernels
and images of Smith homomorphisms easier.
\begin{thm}[\cite{DDKLPT22}]
Let $p\colon S(V)\to X$ be the sphere bundle of the vector bundle $V\to X$. Then there is a long exact sequence
\begin{equation}
\label{smith_LES}
	\dotsb \longrightarrow \Omega_{k+ p^*W}^\Spin \big(S(V)\big) \overset{p_\ast}{\longrightarrow}
	\Omega_{k+V}^\Spin(X)\overset{S_W}{\longrightarrow} \Omega_{k-r+V\oplus W}^\Spin(X)\longrightarrow \Omega_{k-1
	+p^*W}^\Spin \big(S(V)\big) \longrightarrow\dotsb,
\end{equation}
\end{thm}
Often, $S(V)$ is something like a relatively simple
classifying space and the entire long exact sequence can be worked out using techniques similar to the ones in this
paper. We discuss a few examples below.

\subsection{Smith homomorphisms for Spin-$\Mp(2, \Z)$ bordism}
\label{smith_Mp}
Recall from Section \ref{mp_spin} that Spin-$\Mp(2, \Z)$ structures are equivalent in a natural way to $\big( B\SL(2, \Z), V \big)$-twisted
Spin structures, where $V \to B\SL(2, \Z)$ is induced from the standard two-dimensional real representation of
$\SL(2, \Z)$.\footnote{Strictly speaking, we only proved this after localizing at $2$ and at $3$, but the same
argument building on \cref{w2_twisted_cor} holds before localizing.} Therefore there are codimension-$2$ Smith
homomorphisms exchanging Spin bordism of $B\SL(2, \Z)$ and Spin-$\Mp(2, \Z)$ bordism.

To compute these Smith homomorphisms, we would like to identify the third term in the long exact
sequence~\eqref{smith_LES}. The sphere bundle of $V\to \SL(2, \Z)$ is homotopy equivalent to $\SL(2, \R)/\SL(2,
\Z)$, which is a little complicated to describe, so instead we will work at the primes $2$ and $3$ separately,
where it is easier.

For $p = 2$, this Smith homomorphism pulls back to a Smith homomorphism exchanging Spin-$\Z/8\Z$ bordism and the Spin
bordism of $B\Z/4\Z$. The sphere bundle of $V\to B\Z/4\Z$ is homotopy equivalent to $S^1$~\cite{DDKLPT22}, and every
vector bundle trivializes when pulled back to its sphere bundle, so there are long exact sequences
\begin{subequations}
\label{smith_Z8_LES}
\begin{gather}
	\dotsb\longrightarrow \Omega_k^\Spin(S^1) \longrightarrow \Omega_k^{\Spin\text{-}\Z/8\Z} (\pt)
	\overset{S_V}{\longrightarrow} \Omega_{k-2}^\Spin(B\Z/4\Z)\longrightarrow \Omega_{k-1}^\Spin(S^1)
	\longrightarrow\dotsb\\
	\dotsb\longrightarrow \Omega_k^\Spin(S^1) \longrightarrow \Omega_k^{\Spin}(B\Z/4\Z)
	\overset{S_V}{\longrightarrow} \Omega_{k-2}^{\Spin\text{-}\Z/8\Z} (\pt) \longrightarrow \Omega_{k-1}^\Spin(S^1)
	\longrightarrow\dotsb
\end{gather}
\end{subequations}
\begin{lem}
\label{S1_split}
For any generalized homology theory $E$, $E_k(S^1)\cong E_k(\pt)\oplus E_{k-1}(\pt)$.
\end{lem}
\begin{proof}[Proof sketch]
This is equivalent to the collapse of the Atiyah-Hirzebruch spectral sequence computing the $E$-homology of $S^1$
and the absence of hidden extensions. One can show this by considering the maps $\pt\to S^1\to\pt$ and their effect
on differentials and extension questions.
\end{proof}
Using this lemma and the calculations of $\Omega_*^\Spin(B\Z/4\Z)$ and $\Omega_*^{\Spin\text{-}\Z/8\Z}(\pt)$ we
made in \cref{tab:SLprime2data} and Section \ref{ss:spin_Z8}, we can make the long exact sequences~\eqref{smith_Z8_LES}
explicit. For example:
\begin{prop}
The Smith homomorphism $S_V\colon\Omega_5^{\Spin\text{-}\Z/8\Z}(\pt)\to\Omega_3^\Spin(B\Z/4\Z)$, which is a map
$(\Z/32\Z) \oplus (\Z/2\Z) \to (\Z/8\Z) \oplus (\Z/2\Z)$, is surjective, sending $(x, y)\mapsto (x\bmod 8, y)$.
\end{prop}
\begin{proof}
Once we know this map is surjective, its value on generators follows, possibly after an automorphism of
$(\Z/4\Z) \oplus (\Z/2\Z)$. The next map after our Smith homomorphism $S_V$ in the long exact sequence is a map
$\Omega_3^\Spin(B\Z/4\Z)\to\Omega_4^\Spin(S^1)$; the domain of this map is torsion, as we computed in
Section \ref{ss:spin_Z8}, and the codomain is free, so this map is $0$, so by exactness $S_V$ is surjective.
\end{proof}
One could alternatively prove this by computing Poincaré duals of characteristic classes for generators of
$\Omega_5^{\Spin\text{-}\Z/8\Z}(\pt)$; we find it helpful to have both methods available.

For $p = 3$, the argument will have a similar feel: the sphere bundle of $V\to B\Z/3\Z$ is again homotopy
equivalent to $S^1$~\cite{DDKLPT22}, and $3$-locally, both $\Omega_*^\Spin \big(B\SL(2, \Z)\big)$ and
$\Omega_*^{\Spin\text{-}\Mp(2, \Z)} (\pt)$ are equivalent to $\Omega_*^\SSO(B\Z/3\Z)$, which we proved in
\cref{sl2_at_3_lem,mp2_at_3}, respectively. Therefore our Smith long exact sequence has the form
\begin{equation}
	\dotsb\longrightarrow \Omega_k^\SSO(S^1)\longrightarrow \Omega_k^\SSO(B\Z/3\Z)\overset{S_V}{\longrightarrow}
	\Omega_{k-2}^\SSO(B\Z/3\Z)\longrightarrow \Omega_{k-1}^\SSO(S^1)\longrightarrow\dotsb
\end{equation}
\Cref{S1_split} and~\eqref{SO_BZ3} allow us to populate this long exact sequence. We learn, for example, that the
Smith homomorphism $\Omega_3^\SSO(B\Z/3\Z)\to\Omega_1^\SSO(B\Z/3\Z)$ is an isomorphism $\Z/3\Z\to\Z/3\Z$, because
the two terms surrounding this map in the long exact sequence are $\Omega_k^\SSO(S^1)$ for $k = 2,3$, which both
vanish.  In a similar way one can show that $\tOmega_5^\SSO(B\Z/3\Z)\to\tOmega_3^\SSO(B\Z/3\Z)$ is a surjective map
$\Z/9\Z\to\Z/3\Z$, because the next map in the long exact sequence is from a torsion Abelian group to a free one,
and hence vanishes.

\subsection{Smith homomorphisms for Spin-$\GL^+(2, \Z)$ bordism}
\label{smith_GL}

As we noted above, the possible Smith homomorphisms out of $\Omega_k^{\Spin\text{-}\GL^+(2, \Z)}(\pt)$ are given by the
isomorphism classes of vector bundles over $B\GL(2, \Z)$, and this gives us a lot of options, even in codimension
$1$ and $2$. We focus on a particular example which we need in Section \ref{s_7dorange}: a codimension-$2$ map exchanging
the $\Pin^+$ and $\Pin^-$ covers of $\GL(2, \Z)$.

Let $\GL^-(2, \Z)$ denote the $\Pin^-$ cover of $\GL(2, \Z)$ and $\Spin\text{-}\GL^-(2, \Z)\coloneqq
\Spin\times_{\set{\pm 1}}\GL^-(2, \Z)$. Our arguments above studying $\Spin\text{-}\GL^+(2,
\Z)$ bordism also apply in this case, allowing one to show that:
\begin{itemize}
	\item The map $\Omega_k^{\Spin\text{-}\GL^-(2, \Z)} (\pt) \to \Omega_k^\SSO(BD_{12})$ is a $p$-local
	isomorphism for any odd prime $p$. The proof is the same as the argument given at the beginning of
	Section \ref{odd_primary_glplus}.
	\item The $\Pin^-$ cover of $D_8$ is the quaternion group $Q_{16}$, implying the central extension $0\to\Z/2\to
	Q_{16}\to D_8\to 0$ is classified by $w + x^2\in H^2(BD_8;\Z/2\Z)$. The inclusion $D_8 \inj \GL(2, \Z)$ lifts
	to an inclusion $Q_{16} \inj \GL^-(2, \Z)$ inducing a $2$-local isomorphism $\Omega_k^{\Spin\text{-}Q_{16}}
	(\pt) \to \Omega_k^{\Spin\text{-}\GL^-(2, \Z)} (\pt)$.
	\item A $\Spin\text{-}Q_{16}$ structure on a vector bundle $V\to M$ is equivalent data to a $\big(BD_8,
	V \oplus \Det(V) \big)$-twisted Spin structure, i.e.\ a principal $D_8$-bundle $P\to M$ and data of an identification
	$w_2(V) = w(P) + x(P)^2$. The proof is similar to that of \cref{spin_D16_shear} with $V\oplus\Det(V)$ in place
	of $V\oplus 3\Det(V)$ and $w+x^2$ in place of $w$.
\end{itemize}

We consider the Smith homomorphism given by the vector bundle $2\Det(V)$, which lowers the dimension by $2$. This
goes from $\big(B\GL(2, \Z), V\oplus \Det(V) \big)$-twisted Spin bordism to $\big(B\GL(2, \Z), V\oplus 3\Det(V) \big)$-twisted Spin
bordism, i.e.\ Spin-$\GL^-(2, \Z)$ bordism to Spin-$\GL^+(2, \Z)$ bordism. We can also use it to go from $\big(B\GL(2,
\Z), V\oplus 3\Det(V)\big)$-twisted Spin bordism to $\big(B\GL(2, \Z), V\oplus 5\Det(V)\big)$-twisted Spin bordism, but since
four copies of any vector bundle is Spin, this can be identified with $(B\GL(2, \Z), V\oplus\Det(V))$-twisted Spin
bordism, and the Smith homomorphism goes from Spin-$\GL^+(2, \Z)$ bordism to Spin-$\GL^-(2, \Z)$
bordism.\footnote{Kirby-Taylor~\cite[Lemma 7]{KT90} study a similar pair of Smith homomorphisms which
exchange Pin$^-$ and Pin$^+$ bordism.}

Like in the previous section, it is easier to work at $p = 2$ and $p = 3$ separately. After $3$-localizing, this
Smith homomorphism is a map $\Omega_k^\SSO(BD_{12})\to\Omega_{k-2}^\SSO(BD_{12})$, but in \cref{the_thm} we saw that
at least for all $k$ within the range we care about, at least one of $\tOmega_k^\SSO(BD_{12})$ and
$\tOmega_{k-2}^\SSO(BD_{12})$ vanishes, so this Smith homomorphism is trivial. Using unreduced bordism does not
make the map much more interesting.

At $p = 2$, though, this Smith homomorphism is more useful. The whole story with $\GL^\pm(2, \Z)$ pulls back: the
Smith homomorphism exchanges Spin-$D_{16}$ and Spin-$Q_{16}$ bordism.  However, the third term in these Smith long
exact sequences is a little complicated: twisted Spin bordism groups of $B(\Z \ltimes \Z/4\Z)$, where $\Z$ acts on
$\Z/4\Z$ by $\alpha\cdot \beta = (-1)^\alpha\beta$, and twisted by the pullback of $V\oplus\Det(V)$. Proving this
amounts to calculating the sphere bundle of $2\Det(V)\to BD_{8}$, which we found a fun exercise similar to some of
the sphere bundle calculations in~\cite{DDKLPT22}, and to~\cite[Lemma 7]{KT90}.

These twisted Spin bordism groups are not so hard to calculate using similar techniques to the ones
in this paper, and we invite the reader to try some of these computations.\footnote{Alternately, one could apply
the Smith homomorphism for just $\Det(V)$, and do it twice, e.g.\ getting one long exact sequence for $V\oplus
\Det(V)$ to $V\oplus 2\Det(V)$ and a second one for $V\oplus 2\Det(V)$ to $V\oplus 3\Det(V)$. In these cases the
third term in the long exact sequence is $\Omega_k^\Spin(B\Z/4\Z)$.}

We use the map $S_{x^2}\colon \Omega_7^{\Spin\text{-}D_{16}} (\pt) \to\Omega_5^{\Spin\text{-}Q_{16}} (\pt)$ in
Section \ref{s_7dorange}, so we take the opportunity here to investigate $\Omega_5^{\Spin\text{-}Q_{16}} (\pt)$.
Spin-$Q_{16}$ bordism has also been studied in dimension $4$ by Pedrotti~\cite[Theorem 9.0.14]{Ped17} using other
methods; he proves that $\Omega_4^{\Spin\text{-}Q_{16}}(\pt) \cong\Z \oplus (\Z/2\Z)$, which we will rediscover
during the proof of \cref{after_Smith}.
%Recall that
%the lens space $L_4^5$ with principal $\Z/4\Z$-bundle $S^5\to L_4^5$ admits a Spin-$\Z/8\Z$ structure, and
%therefore with the induced principal $D_8$-bundle, $L_4^5$ has a Spin-$Q_{16}$ structure; choose one.

The group $Q_{16}$ acts freely on $S^3$, where $Q_{16}$ is considered as a subgroup of $\text{SU} (2)$ acting on the unit
sphere in $\C^2$. The quotient $S^3/Q_{16}$ is called a \term{prism manifold}, and the quotient $\RP^3\to
S^3/Q_{16}$ is a principal $D_8$-bundle.
\begin{lem}
\label{prism_is_pinm_d}
$S^3/Q_{16}$ with the $D_8$-bundle $P\coloneqq \RP^3\to S^3/Q_{16}$ admits a Spin-$Q_{16}$ structure, and
$\int_{S^3/Q_{16}} w(P)y(P) = 1$.
\end{lem}
Before we prove this, we need to better understand the mod 2 cohomology of $S^3/Q_{16}$. Recall from \cref{Q16_coh}
that $H^*(BQ_{16};\Z/2\Z)\cong\Z/2\Z[\hat x, \hat y, p]/(\hat x\hat y + \hat y^2, \hat x^3)$, where $\abs{\hat x} =
\abs{\hat y} = 1$ and $\abs p = 4$, and $\hat x$, resp.\ $\hat y$ are the pullbacks of $x,y\in H^1(BD_8;\Z/2\Z)$
under $BQ_{16}\to BD_8$.
%\begin{prop}[{\cite[Lemma 4.1.1b]{Tei92} and~\cite[Theorem 4.40ii]{Sna13}}]
%There is an isomorphism $H^\ast(BQ_{16};\Z/2\Z)\cong\Z/2\Z[\hat x, \hat y, p]/(\hat x\hat y+\hat y^2, \hat x^3)$,
%where $\hat x$ and $\hat y$ are the pullbacks of $x$, respectively $y$, by the map $Q_{16}\to D_8$, so
%have degree $1$; and $\abs p = 4$.
%\end{prop}
\begin{lem}[{Tomoda-Zvengrowski~\cite[Section 2 and Theorem 2.2(1)]{TZ08}}]
\label{prism_coh_lem}
There is an isomorphism $H^\ast(S^3/Q_{16};\Z/2\Z)\cong\Z/2\Z[\hat x, \hat y]/(\hat x\hat y + \hat y^2, \hat x^3)$,
and if $f\colon S^3/Q_{16}\to BQ_{16}$ is the classifying map of the principal $Q_{16}$-bundle $S^3\to S^3/Q_{16}$,
then $f^\ast\colon H^\ast(BQ_{16};\Z/2\Z)\to H^\ast(S^3/Q_{16};\Z/2\Z)$ is the quotient by $(p)$, i.e.\ it sends
$\hat x\mapsto\hat x$, $\hat y\mapsto\hat y$, and $p\mapsto 0$.
\end{lem}
\begin{proof}[Proof of \cref{prism_is_pinm_d}]
The relations in the cohomology ring in \cref{prism_coh_lem} imply that $\hat x^2\hat y = \hat x\hat y^2 = \hat
y^3$ is the nonzero element in the top degree mod $2$ cohomology of $S^3/Q_{16}$.

Recall that the extension
\begin{equation}
	\shortexact*[\alpha][\beta]{\Z/2\Z}{Q_{16}}{D_8}{}
\end{equation}
is classified by $w + x^2\in H^2(BD_8;\Z/2\Z)$, so $\beta^*\colon BD_8\to BQ_{16}$ sends $w(ED_8) + x(ED_8)^2\mapsto
0$.  Thus for the $D_8$-bundle $\beta^*ED_8\to BQ_{16}$, $w = \hat x^2$. This bundle is the fiber $\alpha\colon B\Z/2\Z \to
BQ_{16}$ of $\beta$, and when we pull back further along $f\colon S^3/Q_{16}\to BQ_{16}$, $\beta^*ED_{16}$ pulls
back to $\RP^3\to Q_{16}$. Therefore $w(\RP^3) = \hat x^2$. Likewise, $y\in H^1(BD_{16};\Z/2\Z)$ pulls back to $\hat
y\in H^1(S^3/Q_{16};\Z/2\Z)$, so
\begin{equation}
	\int_{S^3/Q_{16}} w(P)y(P) = \int_{S^3/Q_{16}} \hat x^2\hat y = 1.
	\qedhere
\end{equation}
\end{proof}
We will use this result in Section \ref{s_7dorange}.
\begin{prop}
\label{after_Smith}
$[S^3/Q_{16}\times T^2]\ne 0$  in $\Omega_5^{\Spin\text{-}Q_{16}} (\pt)$.
\end{prop}
\begin{proof}
%First consider the bordism invariant $\int w^2y\colon\Omega_5^{\Spin\text{-}Q_{16}} (\pt) \to\Z/2\Z$; it equals $1$ on
%$L_4^5$ and $0$ on $S^3/Q_{16}\times T^2$, the latter because the $D_8$-bundle is trivial in the
%$T^2_{\mathit{nb,nb}}$-direction, so it pulls back across the projection to $S^3/Q_{16}$, and therefore its
%characteristic classes vanish above degree $3$. Therefore to prove the proposition, it suffices to show
%$[S^3/Q_{16}\times T^2] \ne 0$ in $\Omega_5^{\Spin\text{-}Q_{16}} (\pt)$.
We use the Adams spectral sequence. Since we remain below degree $8$, there is no difference between Spin
bordism and $\ko$-homology. Similarly to the case of Spin-$D_{16}$ bordism, the fact that a Spin-$Q_{16}$
structure is equivalent to a $\big(BD_8, V\oplus\Det(V)\big)$-twisted Spin structure allows us to identify
$\Omega_k^{\Spin\text{-}Q_{16}} (\pt) \cong\Omega_k^\Spin \big((BD_8)^{V + \Det(V)-3}\big)$. In the cohomology of
the Thom spectrum $(BD_8)^{V + \Det(V)-3}$, $\Sq^1(U) = 0$ and $\Sq^2(U) = U(w+x^2)$; using this, there is an
$\cA(1)$-module isomorphism
\begin{equation}
\begin{aligned}
	H^*\big((BD_8)^{V + \Det(V)-3};\Z/2\Z \big) &\cong \textcolor{BrickRed}{\uQ} \oplus
		\textcolor{Green}{\Sigma^3\Z/2\Z} \oplus
		\textcolor{MidnightBlue}{\Sigma^4 R_2} \oplus
		\textcolor{Fuchsia}{\Sigma^5 J} \mathop{\oplus}\\
		& \enspace \phantom{\cong}
		\Sigma\cA(1) \oplus \Sigma\cA(1) \oplus
		 \Sigma^5 \cA(1) \oplus \Sigma^5\cA(1) \oplus P,
\end{aligned}
\end{equation}
where $P$ is concentrated in degrees $7$ and above. We draw this in \cref{Adams_for_spin_Q8}, left. We already know
the Ext of each of these summands already (except for $P$, which as usual does not matter in this proof), so we can
draw the $E_2$-page of the Adams spectral sequence in \cref{Adams_for_spin_Q8}, right.
\begin{figure}[h!]
\begin{subfigure}[c]{0.55\textwidth}
\begin{tikzpicture}[scale=0.6, every node/.style = {font=\tiny}]
	\foreach \y in {0, ..., 11} {
		\node at (-2, \y) {$\y$};
	}
	\Aone{0}{1}{$Ux$};
	\Aone{2.25}{1}{$Uy$};
	\Aone{5.5}{5}{$Ux^5$};
	\Aone{7.5}{5}{$Uy^5$};

	\begin{scope}[BrickRed]
		\SpanishQnMark{4.75}{0}{$U$}{star};
	\end{scope}
	\tikzptR{6}{3}{$Uwy$}{Green, regular polygon, regular polygon sides=3, minimum width=1ex};
	\begin{scope}[MidnightBlue]
		\Rtwo{9}{4}{$Uw^2$}{$Uw^2x$}{rectangle, minimum size=3.5pt};
	\end{scope}
	\begin{scope}[Fuchsia]
		\Joker{12.1}{5}{$Uw^2y$}{diamond};
	\end{scope}
\end{tikzpicture}
\end{subfigure}
\begin{subfigure}[c]{0.4\textwidth}
\begin{sseqdata}[name=AdamsQ16, Adams grading, classes=fill, xrange={0}{6}, yrange={0}{3}, scale=0.6,
	x label = {$\displaystyle{s\uparrow \atop t-s\rightarrow}$},
	x label style = {font = \small, xshift = -16.5ex, yshift=5.75ex}, >=stealth]
\class(1, 0)
\class(1, 0)
\class(5, 0)
\class(5, 0)
\begin{scope}[BrickRed, star]
	\class(0, 0)\AdamsTower{}
	\class(4, 1)\AdamsTower{}
	\class(5, 2)\structline(4, 1)(5, 2)
	\class(6, 3)\structline
\end{scope}
\begin{scope}[Green, regular polygon, regular polygon sides=3]
	\class(3, 0)\AdamsTower{}
	\class(4, 1)\structline(3, 0)(4, 1, -1)
	\class(5, 2)\structline
\end{scope}
\begin{scope}[draw=none, fill=none]
	\class(4, 0)\class(4, 0)
	\class(4, 2)
	\class(4, 3)\class(4, 4)
\end{scope}
\begin{scope}[MidnightBlue, rectangle]
	\class(4, 0)\AdamsTower{}
	\class(5, 0)
	\class(6, 1)\structline
\end{scope}
\class[Fuchsia, diamond](5, 0)
\d[gray]2(4, 0, -1)
\end{sseqdata}
\printpage[name=AdamsQ16, page=2]
\end{subfigure}
\caption{Left: the $\cA(1)$-module structure on $H^*\big((BD_8)^{V\oplus\Det(V)-3}; \Z/2\Z \big)$ in low degrees. The pictured
submodule contains all elements in degrees $6$ and below. Right: the $E_2$-page of the corresponding Adams spectral
sequence computing $\ko_*\big((BD_8)^{V\oplus\Det(V)-3}\big)$; in \cref{after_Smith}, we use this to study
Spin-$Q_{16}$ bordism.}
\label{Adams_for_spin_Q8}
\end{figure}
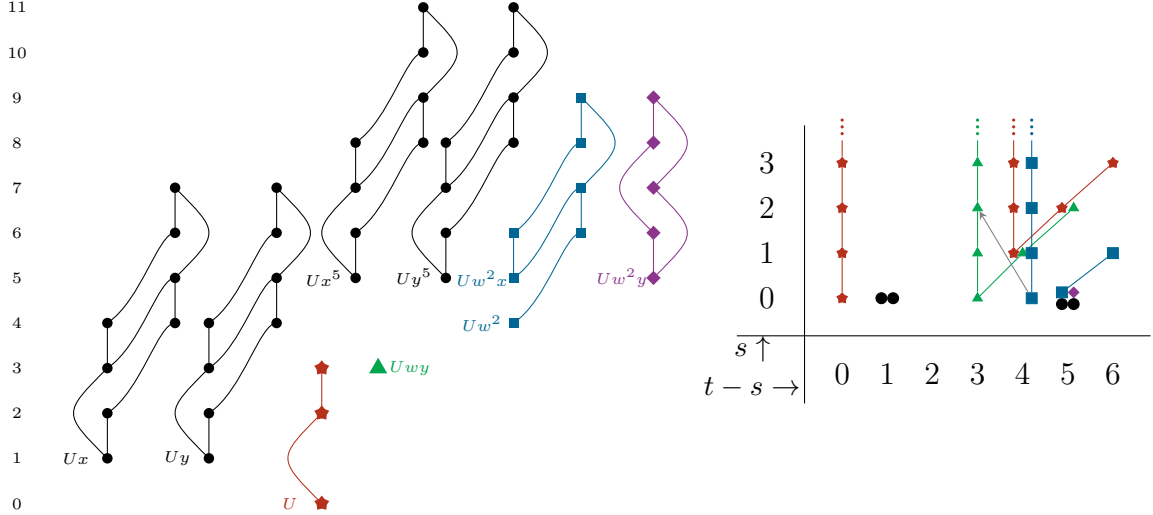

In the range depicted, most differentials vanish either for degree reasons or by $h_i$-equivariance. The exceptions
are the differentials from the red and blue towers in the $4$-line to the green tower in the $3$-line. Like in the
case of Spin-$\Z/8\Z$ bordism that we treated in Section \ref{ss:spin_Z8}, the May-Milgram theorem implies there is a $d_2$
from the blue tower to the green tower.

Therefore we learn that the green triangle $\textcolor{Green}{\Z/2\Z}$ summands in $E_2^{0,3}$, $E_2^{1,5}$, and
$E_2^{2,7}$ all survive to the $E_\infty$-page. They are linked by $h_1$-actions, so we have a diagram
% https://q.uiver.app/?q=WzAsNixbMCwwLCJcXHRleHRjb2xvcntHcmVlbn17XFxaLzR9Il0sWzEsMCwiXFx0ZXh0Y29sb3J7R3JlZW59e1xcWi8yfSJdLFsyLDAsIlxcdGV4dGNvbG9ye0dyZWVufXtcXFovMn0iXSxbMCwxLCJcXE9tZWdhXzNee1xcU3BpblxcdGV4dHstfVFfezE2fX0iXSxbMSwxLCJcXE9tZWdhXzRee1xcU3BpblxcdGV4dHstfVFfezE2fX0iXSxbMiwxLCJcXE9tZWdhXzVee1xcU3BpblxcdGV4dHstfVFfezE2fX0iXSxbMSwyLCJcXGV0YSIsMCx7InN0eWxlIjp7ImhlYWQiOnsibmFtZSI6ImVwaSJ9fX1dLFswLDMsIiIsMix7InN0eWxlIjp7InRhaWwiOnsibmFtZSI6Imhvb2siLCJzaWRlIjoidG9wIn19fV0sWzAsMSwiXFxldGEiLDAseyJzdHlsZSI6eyJoZWFkIjp7Im5hbWUiOiJlcGkifX19XSxbMSw0LCIiLDIseyJzdHlsZSI6eyJ0YWlsIjp7Im5hbWUiOiJob29rIiwic2lkZSI6InRvcCJ9fX1dLFsyLDUsIiIsMCx7InN0eWxlIjp7InRhaWwiOnsibmFtZSI6Imhvb2siLCJzaWRlIjoidG9wIn19fV0sWzMsNCwiXFxldGEiXSxbNCw1LCJcXGV0YSJdXQ==4
\begin{equation}
\label{greens_and_etas}
\begin{tikzcd}
	{\textcolor{Green}{\Z/4\Z}} & {\textcolor{Green}{\Z/2\Z}} & {\textcolor{Green}{\Z/2\Z}} \\
	{\Omega_3^{\Spin\text{-}Q_{16}} (\pt)} & {\Omega_4^{\Spin\text{-}Q_{16}} (\pt)} & {\Omega_5^{\Spin\text{-}Q_{16}} (\pt) \,,}
	\arrow["\eta", two heads, from=1-2, to=1-3]
	\arrow[hook, from=1-2, to=2-2]
	\arrow["\eta", two heads, from=1-1, to=1-2]
	\arrow["\cong", from=1-1, to=2-1]
	\arrow[hook, from=1-3, to=2-3]
	\arrow["\eta", from=2-1, to=2-2]
	\arrow["\eta", from=2-2, to=2-3]
\end{tikzcd}\end{equation}
with the takeaway that $\eta^2\colon\Omega_3^{\Spin\text{-}Q_{16}}(\pt)\to\Omega_5^{\Spin\text{-}Q_{16}}(\pt)$ is nonzero
precisely on the odd elements of $\Omega_3^{\Spin\text{-}Q_{16}} (\pt)\cong\textcolor{Green}{\Z/4\Z}$. The image of a
generator in the $E_\infty$-page is in Adams filtration zero, where it corresponds to the mod $2$ cohomology class
$Uwy$, which has the convenient consequence that $\int wy$ is also nonzero exactly on the odd elements of
$\Omega_3^{\Spin\text{-}Q_{16}} (\pt)$. So if we can show that this integral is nonzero on $S^3/Q_{16}$, then we also
deduce that $\eta^2\cdot[S^3/Q_{16}]$, i.e.\ $[S^3/Q_{16} \times T_p^2]$, is nonzero, which is what we want to prove.
By \cref{prism_is_pinm_d}, $S^3/Q_{16}$ is an odd element of $\Omega_3^{\Spin\text{-}Q_{16}} (\pt)$, so
$[S^3/Q_{16} \times T_p^2]\ne 0$ in $\Omega_5^{\Spin\text{-}Q_{16}} (\pt)$, as we wanted.
\end{proof}
Our Adams computation in the proof of \cref{after_Smith} leaves open the possibility of a hidden extension in
degree $5$. We will show that there is no such hidden extension, so that $\Omega_5^{\Spin\text{-}Q_{16}}$ is a
direct sum of copies of $\Z/2\Z$; we stated but did not prove this in Section \ref{s_7dorange} as \cref{split_Q16_5}.
\begin{prop}
\label{quater_hidden}
$\Omega_5^{\Spin\text{-}Q_{16}}\cong (\Z/2\Z)^{\oplus 6}$.
\end{prop}
\begin{proof}
We let $X_8'\coloneqq (BD_8)^{V\oplus\Det(V)-3}$ for readability.
%The proof is complicated and it would be nice if there were an easier way.
%For example, it may be possible to use the map from spin-$Q_8$ bordism to spin-$Q_{16}$ bordism to show that $2X =
%0$ when $X$ is any of the elements coming from Adams filtration $0$. One of these is the lens space $L_4^5$,
%detected by $w^2y$; the other should be detected by $w^2x$. There are two more summands coming from $x^5$ and
%$y^5$, but they split off by Margolis' theorem.
The multiplication-by-$2$ map $2\colon\ko_k(X)\to\ko_k(X)$ factors through $\ku$-theory; specifically, as we
discussed in \cref{ku_remark}, $2 = R\circ b\circ c$.
%$r\colon\ku\to\ko$ is the ``realification'' map forgetting the complex structure on a vector bundle and
%$c\colon\ko\to\ku$ is the complexification map, $2 = r\circ c$. We can factor $r$ further as $r = R\circ b$, where
%$b\colon \ku\to\Sigma^2\ku$ is the Bott map. We introduced $c$ and $R$ in the proof of \cref{Z4_5_extn}; for a
%reference that $r = R\circ b$, see~\cite[Theorem 1]{Bru12}.
We will show that $b\circ
c\colon\ko_5(X_8')\to\ku_7(X_8')$ vanishes, which implies multiplication by $2$ vanishes on $\ko_5(X_8')$. To do
this, we need to study $\ku_*(X_8')$ with the Adams spectral sequence.

The change-of-rings trick that enables one to work over $\cA(1)$ for $\ko$-theory has an analogue for $\ku$-theory.
The catalyst is Adams' calculation~\cite{Ada61} that as an $\cA$-module,
$H^*(\ku;\Z/2\Z)\cong\cA\otimes_{\cE(1)}\Z/2\Z$, where $\cE(1)$ is the algebra generated by $Q_0 \coloneqq \Sq^1$
and $Q_1 \coloneqq \Sq^2\Sq^1 + \Sq^1\Sq^2$. Using the change-of-rings theorem~\cite[Section 4.5]{BC18}, the Adams
$E_2$-page for $\ku_*(X)$ simplifies:
\begin{equation}
\begin{aligned}
	E_2^{s,t} &= \Ext_\cA^{s,t}\big(H^*(\ku\wedge X;\Z/2\Z), \Z/2\Z \big)\\
		&\cong \Ext_{\cA}^{s,t}\big(H^*(\ku;\Z/2\Z)\otimes_{\Z/2} H^*(X;\Z/2\Z), \Z/2\Z \big)\\
		&\cong \Ext_\cA^{s,t} \big(\cA\otimes_{\cE(1)} \Z/2\otimes_{\Z/2} H^*(X;\Z/2\Z), \Z/2\Z \big)\\
		&\cong\Ext_{\cE(1)}^{s,t}(H^*\big(X;\Z/2\Z), \Z/2\Z\big) \,.
\end{aligned}
\end{equation}
As $\cE(1)$ is an exterior algebra on two generators, this reformulation of the $E_2$-page is much easier to
calculate. This technique also appears in~\cite{Bay94, Deb21}.
\begin{lem}
%Let $\uQ$ denote the $\cE(1)$-module which is the kernel of the
%augmentation map $\Sigma^{-1}\cE(1)\to\Sigma^{-1}\Z/2\Z$. Then
There are isomorphisms of $\cE(1)$-modules
\begin{enumerate}
	\item $\textcolor{MidnightBlue}{R_2}\cong \uQ\oplus \Sigma \cE(1)$,
	\item $\textcolor{Fuchsia}{J}\cong\cE(1)\oplus\Sigma^2\Z/2\Z$, and
	\item $\cA(1)\cong\cE(1)\oplus \Sigma^2\cE(1)$.
\end{enumerate}
\end{lem}
One can prove this by directly checking how $Q_0$ and $Q_1$ act on these $\cA(1)$-modules, though this also appears
in~\cite[Section 4.4.3]{Deb21} (specifically in the proof of Theorem 4.53). Using this, there is an isomorphism of
$\cE(1)$-modules
\begin{equation}
\label{E1_decomp}
\begin{aligned}
	H^*(X_8';\Z/2\Z) \cong\, &\textcolor{BrickRed}{\uQ} \oplus
		\textcolor{Green}{\Sigma^3\Z/2\Z} \oplus
		\textcolor{MidnightBlue}{\Sigma^4\uQ} \oplus
		\Sigma\cE(1)\oplus\Sigma\cE(1) \oplus \Sigma^3\cE(1)\oplus \Sigma^3\cE(1)\, \oplus\\
		&\Sigma^5\cE(1) \oplus\Sigma^5\cE(1) \oplus
		\textcolor{MidnightBlue}{\Sigma^5\cE(1)} \oplus \textcolor{Fuchsia}{\Sigma^5\cE(1)} \oplus P',
\end{aligned}
\end{equation}
where $P'$ is concentrated in degrees $7$ and above (hence is irrelevant for our computations). We draw this
decomposition in \cref{spin_Q8_E1}: straight lines denote $Q_0$-actions and dashed curved lines denote
$Q_1$-actions.
\begin{figure}[h!]
\centering
\begin{tikzpicture}[scale=0.6, every node/.style = {font=\tiny}]
	\foreach \y in {0, ..., 9} {
		\node at (-2, \y) {$\y$};
	}
	\begin{scope}[BrickRed]
		\EoneQnMark{0}{0}{$U$}{star};
	\end{scope}
	\begin{scope}[MidnightBlue]
		\EoneQnMark{0}{4}{$Uw^2$}{rectangle, minimum size=3.5pt};
		\Eone{11}{5}{$Uw^2x$}{rectangle, minimum size=3.5pt};
	\end{scope}
	\tikzptB{1}{3}{$Uwy$}{Green, regular polygon, regular polygon sides=3, minimum width=1ex};
	\Eone{3}{1}{$Ux$}{};
	\Eone{5}{1}{$Uy$}{};
	\Eone{6.875}{3}{$Ux^3$}{};
	\Eone{8.75}{3}{$Uy^3$}{};
	\Eone{2.3}{5}{$Ux^5$}{};
	\Eone{4.3}{5}{$Uy^5$}{};
	\begin{scope}[Fuchsia]
		\Eone{13}{5}{$Uw^2y$}{diamond};
	\end{scope}
\end{tikzpicture}
\caption{The $\cE(1)$-module structure on $H^*(X_8';\Z/2\Z)$ in low degrees. The pictured submodule contains all
elements in degrees $6$ and below. See the proof of \cref{quater_hidden} for more information.}
\label{spin_Q8_E1}
\end{figure}
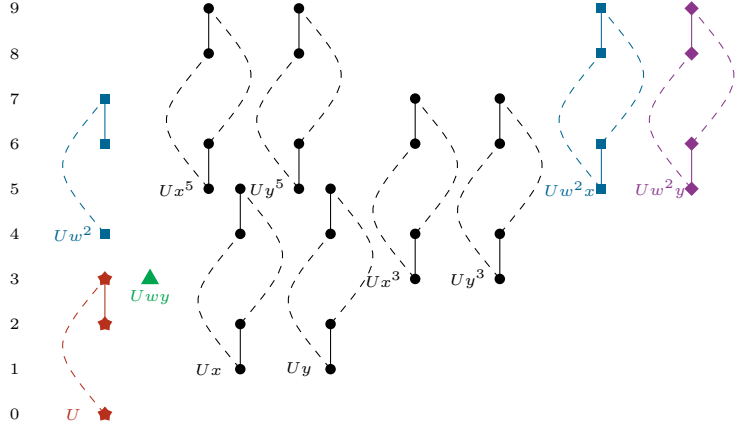

To determine the $E_2$-page of the Adams spectral sequence, we need to compute $\Ext_{\cE(1)}$ of the summands
in~\eqref{E1_decomp} as modules over $\Ext_{\cE(1)}(\Z/2\Z)$; the formal structure of the argument is the same as
over $\cA(1)$, but the specific computations are different, leading to $\ku$-homology rather than $\ko$-homology.

So, there is an isomorphism
\begin{equation}
        \Ext_{\cE(1)}^{s,t}(\Z/2\Z, \Z/2\Z)\cong\Z/2\Z[h_0, v_1]
\end{equation}
with $\abs{h_0} = (1, 1)$ and $\abs{v_1} = (1, 3)$~\cite[Example 4.5.6]{BC18}, and this algebra acts canonically on
$\Ext_{\cE(1)}^{s,t}\big(H^*(X;\Z/2\Z), \Z/2\Z \big)$, providing information about the $\ku_*$-action on $\ku_*(X)$.
Specifically, $h_0$-actions on the $E_\infty$-page lift to multiplication by $2$, and $v_1$-actions on the
$E_\infty$-page lift to an action by the Bott element $b\in\ku_2$. Just as in the $\cA(1)$ case, there can be
``hidden extensions,'' where $2$ or $b$ act non-trivially on $\ku_*(X)$, but in a way not detected by $h_0$ or $v_1$
acting on the $E_\infty$-page.

We want to know the Ext groups of $\uQ$ and $\cE(1)$ as modules over this algebra. These are known:
$\Ext_{\cE(1)}^{s,t}(\uQ, \Z/2\Z)$ is in~\cite[Section 3]{AP76}, and $\Ext_{\cE(1)}^{s,t}(\cE(1), \Z/2\Z)$ consists of a
single $\Z/2\Z$ in bidegree $(0, 0)$, since $\cE(1)$ is a free $\cE(1)$-module. Now we can draw the $E_2$-page of the
Adams spectral sequence for $\ku_*(X_8')$, which we do in \cref{the_ku_homology_sseq}, left. Like over $\cA(1)$,
$h_0$-actions are written with vertical lines; $v_1$-actions are drawn with lighter diagonal lines.
\begin{figure}[h!]
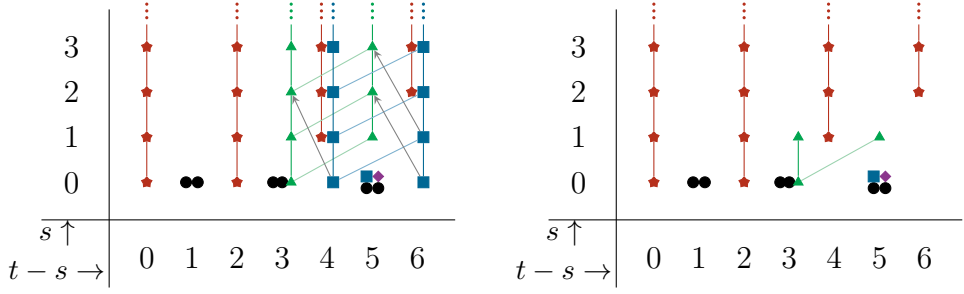

\centering
\begin{subfigure}[c]{0.4\textwidth}
\begin{sseqdata}[name=AdamskuQ16, Adams grading, classes=fill, xrange={0}{6}, yrange={0}{3}, scale=0.6,
	x label = {$\displaystyle{s\uparrow \atop t-s\rightarrow}$},
	x label style = {font = \small, xshift = -16.5ex, yshift=5.75ex}, >=stealth]
\class(1, 0)
\class(1, 0)
\class(3, 0)
\class(3, 0)
\class(5, 0)
\class(5, 0)
\begin{scope}[BrickRed, star]
	\class(0, 0)\AdamsTower{}
	\class(2, 0)\AdamsTower{}
	\class(4, 1)\AdamsTower{}
	\class(6, 2)\AdamsTower{}
\end{scope}
\begin{scope}[draw=none, fill=none]
	\class(3, 1)\AdamsTower{}
	\class(3, 1)\AdamsTower{}
	\class(4, 0)
	\class(6, 0)
	\class(6, 1)
	\class(3, 5)
	\class(5, 5)
\end{scope}
\begin{scope}[Green, regular polygon, regular polygon sides=3]
	\class(3, 0)\AdamsTower{}
	\class(5, 1)\AdamsTower{}
\end{scope}
\begin{scope}[Green!40!white]
	\structline(3, 0, -1)(5, 1, -1)
	\structline(3, 1, -1)(5, 2, -1)
	\structline(3, 2, -1)(5, 3, -1)
\end{scope}
\begin{scope}[MidnightBlue, rectangle]
	\class(4, 0)\AdamsTower{}
	\class(6, 0)\AdamsTower{}
	\class(5, 0)
\end{scope}
\begin{scope}[MidnightBlue!40!white]
	\structline(4, 0, -1)(6, 1, -1)
	\structline(4, 1, -1)(6, 2, -1)
	\structline(4, 2, -1)(6, 3, -1)
\end{scope}

\class[Fuchsia, diamond](5, 0)
\d[gray]2(4, 0, -1)(3, 2, -1)
\d[gray]2(6, 0, -1)(5, 2, -1)
\d[gray]2(6, 1, -1)(5, 3, -1)
% the spectralsequences package takes care of drawing the E_3 page automatically if we specify these differentials,
% even if we don't actually draw them!
\begin{scope}[draw=none]
	\d[gray]2(4, 1, -1)(3, 3, -1)
	\d[gray]2(4, 2, -1)(3, 4, -1)
	\d[gray]2(4, 3, -1)(3, 5, -1)
	\d[gray]2(6, 2, -1)(5, 4, -1)
	\d[gray]2(6, 3, -1)(5, 5, -1)
\end{scope}

\end{sseqdata}
	\printpage[name=AdamskuQ16, page=2]
\end{subfigure}
\begin{subfigure}[c]{0.4\textwidth}
	\printpage[name=AdamskuQ16, page=3]
\end{subfigure}
\caption{The Adams spectral sequence computing $\ku_*(X_8')$. Some $v_1$-actions are hidden for readability. Left:
the $E_2$-page. Right: the $E_3 = E_\infty$-page. Margolis' theorem lifts $E_\infty^{0,5}\cong (\Z/2\Z)^{\oplus 4}$
to a $(\Z/2\Z)^{\oplus 4}$ subspace $\mathcal S$ of $\ku_5(X_8')$.}
\label{the_ku_homology_sseq}
\end{figure}

There is potential for differentials, and indeed they are present. All of them are either zeroed out by Margolis'
theorem or are determined by the differentials from the $4$-line to the $3$-line using that the
$\Ext_{\cE(1)}(\Z/2\Z)$-action commutes with differentials. And there is a differential between the $h_0$-towers on
the $4$- and $3$-lines, which we deduce from the corresponding differential in the $\ko$-homology Adams
spectral sequence (\cref{Adams_for_spin_Q8}, right), or alternatively by comparing with the Atiyah-Hirzebruch
spectral sequence. Once we take this differential, there can be no others in range; we display the $E_\infty$-page
in \cref{the_ku_homology_sseq}, right. We see from this $E_\infty$-page that $\ku_5 \big((BD_8)^{V\oplus\Det(V)-3} \big)\cong
(\Z/2\Z)^{\oplus 5}$, with a four-dimensional subspace $\mathcal S$ coming from the $\big(\Sigma^5\cE(1)\big)^{\oplus 4}$
summand in $H^*\big((BD_8)^{V\oplus\Det(V)-3};\Z/2\Z \big)$ and Margolis' theorem.
\begin{lem}
$\mathcal S$ is precisely the kernel of
\begin{equation}
	R\colon\ku_5 \big((BD_8)^{V\oplus\Det(V)-3} \big)\longrightarrow\ko_3 \big((BD_8)^{V\oplus\Det(V)-3} \big) \,.
\end{equation}
\end{lem}
\begin{proof}
Recall $\ko_3(X_8')\cong \textcolor{Green}{\Z/4\Z}$, and let $x$ be a generator. Then $\eta(kx) \ne 0$ iff $k$ is
odd, which we discovered in the proof of \cref{after_Smith} (specifically see~\eqref{greens_and_etas}). Recall from
the proof of \cref{Z4_5_extn} that the maps $\eta$, $c$, and $R$ fit together into a long exact sequence; thus if
$R(z) \ne 0$, $R(z)\not\in\Im(\eta)$, so $R(z) = 2x$.

The connective cover map $\ku\to\KU$ induces a map $i_c\colon \ku_*(X)\to\KU_*(X)$ which is the
localization of the $\ku_*$-module $\ku_*(X)$ away from the ideal generated by the Bott element $b$. That is, for
every $x\in\ku_*(X)$, if $b^mx = 0$ for some $m$ , $i_c(x) = 0$; if $b^mx$ is never zero, then $i_c(x)\ne 0$, and
$i_c(x)$ is divisible by arbitrarily high powers of $b$. In the former case, $x$ is called \term{$b$-torsion}; in
the latter case, $x$ is called \term{$b$-periodic}. An analogous statement is true with
$i_r\colon\ko_*(X)\to\KO_*(X)$ in place of $i_c$; the Bott element $w$ in $\ko_*$ is the image of the
Bott manifold $B$ under the Atiyah-Bott-Shapiro map $\MTSpin\to\ko$.

The class $2x\in\ko_3(X_8')$ is $w$-periodic. One way to see this is to lift to its preimage in Spin bordism, which is
represented by $2S^3/Q_{16}$, as we showed in the proof of \cref{after_Smith}. This manifold is detected by some
kind of twisted $\eta$-invariant. $w$-periodicity of this class is equivalent to $B^m\times
(2S^3/Q_{16})$ being nonbounding for all $m$, which one can prove by computing the same $\eta$-invariant on
that class and using the factorization $\eta(M\times N) = \mathrm{Index}(M)\eta(N)$ when $m$ is
$4\ell$-dimensional. Thus $i_r(2x)\ne 0$. If $z\in\mathcal S$, $bz = 0$ by Margolis' theorem,\footnote{There is a
subtlety here that we want to be careful about.  Margolis proved that free $\cA$-module summands in $H^*(X;\Z/2\Z)$
lift to split copies of $H\Z/2\Z$ off of $X$ as a spectrum, which in particular prevents non-trivial differentials and
hidden extensions involving those $H\Z/2\Z$ factors in the Adams spectral sequence. In particular,
$\pi_*(\Sph)_2^\wedge$ acts trivially on the copies of $\Z/2\Z$ in $\pi_*(X)$ coming from these $H\Z/2\Z$ summands.
Using the change-of-rings theorem, one learns that free $\cE(1)$-module summands in $H^*(X;\Z/2\Z)$ correspond to
$H\Z/2\Z$ summands in $\ku\wedge X$, but a priori we do not know that this splitting is as $\ku$-modules, only as
spectra; in principle $b\in\ku_2$ could act non-trivially. \label{marg_ku_footnote}

Fortunately, though, this problem does not occur: $b$ acts trivially on the $\Z/2\Z$ summands in $\ku_*(X)$ coming
from free $\cE(1)$-submodules of $H^*(X;\Z/2\Z)$. A reference for this is Bruner-Greenlees~\cite[Section 2.1]{BG03}.} so
$i_c(z) = 0$.

Suppose $z\in\mathcal S$ is such that $R(z)\ne 0$; then $R(z) = 2x$. The map $R$ is the connective cover of the
realification map $R\colon\KU\simeq\Sigma^2\KU\to\KO$, implying there is a commutative diagram
% https://q.uiver.app/?q=WzAsNCxbMCwwLCJcXGt1XzUoWF84JykiXSxbMCwxLCJcXEtVXzUoWF84JykiXSxbMSwwLCJcXGtvXzMoWF84JykiXSxbMSwxLCJcXEtPXzMoWF84JykiXSxbMSwzLCJSIl0sWzAsMiwiUiJdLFswLDEsImlfYyIsMl0sWzIsMywiaV9yIiwyXV0=
\begin{equation}
\begin{gathered}
\begin{tikzcd}
	{\ku_5(X_8')} & {\ko_3(X_8')} \\
	{\KU_5(X_8')} & {\KO_3(X_8').}
	\arrow["R", from=2-1, to=2-2]
	\arrow["R", from=1-1, to=1-2]
	\arrow["{i_c}"', from=1-1, to=2-1]
	\arrow["{i_r}"', from=1-2, to=2-2]
\end{tikzcd}
\end{gathered}
\end{equation}
Traveling along the upper right, $i_r(R(z)) = i_r(2x)\ne 0$, but along the lower left, $R(i_c(z)) = 0$, which is a
contradiction, so $\mathcal S\subset\ker(R)$. To turn this into an equality, use the $\eta$, $c$, $R$ long exact
sequence to show $\ker(R)$ is a four-dimensional $\Z/2\Z$-vector space.
\end{proof}
By exactness of the $\eta$, $c$, $R$ sequence, we conclude $\Im(c\colon \ko_5(X_8')\to\ku_5(X_8')) = \mathcal S$.
Finally, Margolis' theorem (or a slight strengthening, as we discussed in Footnote~\ref{marg_ku_footnote}) implies
$b$ acts trivially on classes in $\ku_*(X_8')$ coming from free $\cE(1)$-module summands in cohomology. Therefore
$b\circ c = 0$ on $\ko_5(X_8')$ as promised, so multiplication by $2$ is the zero map on $\ko_5(X_8')$.
\end{proof}
%basis consisting of
%four elements $z_1,\dotsc,z_4$ whose images in the $E_\infty$-page have filtration zero and which split off thanks
%to Margolis' theorem; together with a class $y$ such that $y = vy'$ for some
%$y'\in\ku_3((BD_8)^{V\oplus\Det(V)-3})$. Further, by comparing [\TODO: $\ko$ and $\ku$ Adams spectral sequences],
%we can see that $y' = c(x)$ for $x$ the generator of $\ko_3((BD_8)^{V\oplus\Det(V)-3})\cong\Z/4$. Since $R\circ
%b\circ c = 2$, we learn $R(y) = 2x$.
%
%
%Therefore it
%suffices to produce a basis of $\ku_5((BD_8)^{V\oplus\Det(V)-3})$ such that either $b$ or $r$ vanishes on each
%basis element; this implies $r\colon\ku_5((BD_8)^{V\oplus\Det(V)-3}) \to \ko_5((BD_8)^{V\oplus\Det(V)-3})$
%vanishes, so multiplication by $2$, which factors through $r$, is also trivial, implying the proposition statement.
%
%To calculate the actions of $b$ and $r$ on a basis of $\ku_5((BD_8)^{V\oplus\Det(V)-3})$, we apply the Adams
%spectral sequence over $\cE(1)$. On the $E_2$-page, there are five summands in topological degree $5$; four
%come from $\Sigma^5\cE(1)$ summands in cohomology, so by Margolis' theorem $b$ acts trivially on them
%[\TODO: this goes slightly beyond Margolis' theorem; cite Bruner-Greenlees].
%
%The fifth class is harder. [\TODO: changes the proof strategy slightly.] Use the $\eta cR$ long exact sequence to
%show that the image of $R\colon\ku_5\to\ko_3$ is nonzero iff the coefficient of the not-in-filtration-zero
%element is not zero; therefore by exactness such elements cannot be in the image of $c$.
%c$x
This proof, elaborate as it was, is the simplest argument we found. We would be interested in learning of a simpler
way to address this extension question.

\end{appendix}

\newpage
\bibliographystyle{utphys}
\bibliography{References}

\end{document}